\documentclass[12pt]{article} % Aptara syntax
\usepackage{fullpage}
\usepackage{bbm,amsmath,amssymb,amsthm}
\usepackage{mathrsfs}
\usepackage[normalem]{ulem}
\usepackage[titles]{tocloft}
\usepackage{authblk}

\newtheorem{example}{Example}[section]
\newtheorem{definition}{Definition}[section]
\newtheorem{theorem}[definition]{Theorem}
\newtheorem{corollary}[definition]{Corollary}
\newtheorem{proposition}[definition]{Proposition}
\newtheorem{lemma}[definition]{Lemma}

\usepackage{pdfpages}

\usepackage{xcolor}
\definecolor{monvert}{rgb}{0.05,0.54,0.05}
\usepackage[
    linktocpage=true,
    hyperfootnotes=false,
    pdftex,                %
    bookmarks         = true,%     % Signets
    bookmarksnumbered = true,%     % Signets numerotes
    pdfpagemode       = None,%     % Signets/vignettes fermes a l'ouverture
    pdfstartview      = FitH,%     % La page prend toute la largeur
    pdfpagelayout     = SinglePage,% Vue par page
    colorlinks        = true,%     % Liens en couleur
    linkcolor= red, %    % couleur des liens internes
    citecolor         =blue,
    urlcolor          = magenta,%  % Couleur des liens externes
    pdfborder         = {0 0 0}%   % Style de bordure : ici, pas de bordure
    ]{hyperref}%                   % Utilisation de HyperTeX

\usepackage{tikz}
\usetikzlibrary{trees}
\usetikzlibrary{arrows}
\tikzset{
 treenode/.style = {align=center, inner sep=4pt, text centered,
    font=\sffamily},
  arn_n/.style = {treenode, rounded corners, rectangle, black, font=\sffamily, draw=black,minimum height=2em,
    fill=white}, 
leaf/.style = {treenode, rectangle, draw=black,
    minimum width=1em, minimum height=1em}
}

\def\new#1{{\textcolor{black}{#1}}}

\def\fbr{\ensuremath{\mathrm{fbr}}}
\def\lag{\ensuremath{\mathsf{Lagrange}}}
\def\reg{\ensuremath{\mathrm{reg}}}
\def\sing{\ensuremath{\mathrm{sing}}}
\def\oreg{\ensuremath{V^\circ_{\rm reg}}}
\def\freg{\ensuremath{V_{\rm reg}}}
\def\GL{\ensuremath{\mathrm{GL}}}
\def\C {\ensuremath{\mathbf{C}}}
\def\R {\ensuremath{\mathbf{R}}}
\def\QQ {\ensuremath{\mathbf{Q}}}
\def\scrQ{\ensuremath{\mathscr{Q}}}
\def\scrC{\ensuremath{\mathscr{C}}}
\def\Zeroes{\ensuremath{\mathsf{Z}}}
\def\Ideal{\ensuremath{I}}
\def\IdealV{{\Ideal}(V)}
\def\openpolar{{W^\circ}}
\def\polar{{W}}
\def\Kpolar{{K}}
\def\Cons{\ensuremath{\mathscr{D}}}
\def\Proj{\ensuremath{\mathscr{U}}}
\def\lnf{\ensuremath{\mathsf{L}}}
\def\gnf{\ensuremath{\mathsf{G}}}
\def\Fiberlag{\ensuremath{\mathrm{F}_{\mathrm{Lagrange}}}}
\def\Polarlag{\ensuremath{\mathrm{W}_{\mathrm{Lagrange}}}}
\def\DF{\ensuremath{\mathrm{Dg}}}
\def\sfH{\ensuremath{\mathsf{H}}}
\def\sfC {\ensuremath{\mathsf{C}}}
\def\sfA {\ensuremath{\mathsf{A}}}
\def\T {\ensuremath{T}}
\def\chartpolar {\ensuremath{\mathrm{W}_{\rm chart}}}
\def\atlaspolar {\ensuremath{\mathrm{W}_{\rm atlas}}}
\def\atlasfiber {\ensuremath{\mathrm{F}_{\rm atlas}}}
\def\Open {\ensuremath{\mathcal{O}}}

\def\ZOffp{\ensuremath{\tilde{\mathscr{G}}}} % Zariski open finite fibers of polar varieties
\def\ZOlp{\ensuremath{\mathscr{G}'}} % Zariski open set ; prop. of local polar varieties
\def\ZOfinite{\ensuremath{\mathscr{G}_2}} % Zariski ouvert finiteness properties

\def\pollambda{{\textcolor{black}{\mathfrak{l}}}}
\def\polmu{{\textcolor{black}{\mathfrak{m}}}}
\def\linearmu{{\textcolor{black}{\mathfrak{h}}}}

\def\poldelta{{\textcolor{black}{\mathfrak{d}}}}
\def\singSBasicOpen{{\textcolor{black}{\ensuremath{\mathcal{U}}}}}
\def\BasicOpen{{\textcolor{black}{\ensuremath{\mathcal{O}}}}}
\def\UOpen{{\textcolor{black}{\ensuremath{\mathcal{U}}}}}
\def\openB{{\textcolor{black}{\ensuremath{\mathcal{B}}}}}
\def\lcs{{\textcolor{black}{\ensuremath{V^\circ}}}}
\def\lcswi{{\textcolor{black}{\ensuremath{Y_i^\circ}}}}
\def\lcsw{{\textcolor{black}{\ensuremath{Y^\circ}}}}
\def\SlocallyclosediA{{\textcolor{black}{\ensuremath{S_{i,\mA}^\circ}}}}
\def\Slocallyclosed{{\textcolor{black}{\ensuremath{S^\circ}}}}
\def\algxAX{{\textcolor{black}{\mathfrak{X}_{\mA,\x}^\circ}}}
\def\algiAX{{\textcolor{black}{\mathfrak{X}_{i,\mA}^\circ}}}
\def\algAX{{\textcolor{black}{\mathfrak{X}_\mA^\circ}}}
\def\algX{{\textcolor{black}{\mathfrak{X}^\circ}}}
\def\lcLambda{{\textcolor{black}{Y^\circ}}}
\def\closedX{{\textcolor{black}{X}}}
\def\closedZ{{\textcolor{black}{Z}}}
\def\algZ{{\textcolor{black}{Z}}}
\def\alg2Z{{\textcolor{black}{Z}}}
\def\Vmp{{\textcolor{black}{V^{mp}}}}

\def\Prime{{\ensuremath{\mathcal{P}}}}

\def\roottau{{\alpha}}
\def\nodetau{{\tau}}
\def\nodeoherletter{{\tilde\tau}}
\def\rootrho{{\xi}}
\def\depth{r}

\def\ZOA12{\ensuremath{\mathscr{E}}}
\def\ZOpropRK{\ensuremath{\mathscr{A}}}
\def\ZOdeltapropRK{\ensuremath{\mathcal{U}}}
\def\ZOK{\ensuremath{\mathscr{K}}}
\def\ZOomega{\textcolor{black}{\ensuremath{\mathscr{M}}}}

\def\scrGfiberchart{\ensuremath{\mathscr{G}_{3}^{\mathrm{chart}}}}
\def\scrGpolarchart{\ensuremath{\mathscr{G}_{1}^{\mathrm{chart}}}}
\def\scrGfiber{{\ensuremath{\mathscr{G}_3}}}
\def\scrGpolar{\ensuremath{\mathscr{G}_1}}

\def\Clos#1{\ensuremath{\overline{\mathscr{U}#1}}}

\def\scrOpen{{\textcolor{black}{\ensuremath{\mathscr{G}}}}}
\def\scrIopen{{\textcolor{black}{\ensuremath{\mathscr{I}}}}}

\def\sfTa{{\ensuremath{\mathsf{T}}}}
\def\sfTpa{{\ensuremath{\mathsf{T}'}}}

\def\scrS{\ensuremath{\mathscr{S}}}

\def\Klastindex{3}
\def\Klastlastindex{4}

\def\scrX{\textcolor{black}{\ensuremath{\mathscr{X}}}}

\def\degQ{\ensuremath{\kappa}}
\def\degR{\ensuremath{\gamma}}
\def\degS{\ensuremath{\sigma}}
\def\degC{\ensuremath{\mu}}
\def\degB{\ensuremath{\beta}}
\def\degfiber{\ensuremath{\gamma}}

\def\softO{\ensuremath{{O}{\,\tilde{ }\,}}}
\def\dalgo{{\ensuremath{\tilde d}}}
\def\minor{{\ensuremath{m'}}}
\def\pminor{{\ensuremath{m''}}}

\def\dinit{\ensuremath{{d_{\rho}}}}

\def\A {\ensuremath{\mathbb{A}}}
\def\B {\ensuremath{\mathbb{B}}}
\def\Co {\ensuremath{\mathbb{C}}}
\def\CC {\ensuremath{\overline{\mathbf{K}}}}

\def\P {\ensuremath{\mathbb{P}}}
\def\Q {\ensuremath{\mathbb{Q}}}
\def\N {\ensuremath{\mathbb{N}}}
\def\Re {\ensuremath{\mathbb{R}}}

\def\Z {\ensuremath{\mathbb{Z}}}

\def\GG {\ensuremath{\mathbf{G}}}
\def\KK {\ensuremath{\mathbf{K}}}

\def\n {\ensuremath{\mathbf{n}}}

\def\p {\ensuremath{\mathbf{p}}}

\def\mS {\ensuremath{\mathbf{S}}}
\def\mzero {\ensuremath{\mathbf{0}}}
\def\mT {\ensuremath{\mathbf{T}}}

\def\JJ {\ensuremath{\mathbf{J}}}

\def\mm {\ensuremath{\mathbf{m}}}
\def\v {\ensuremath{\mathbf{v}}}

\def\mA{\ensuremath{\mathbf{A}}}

\def\mB{\ensuremath{\mathbf{B}}}

\def\mM{\ensuremath{\mathbf{M}}}
\def\G {\ensuremath{\mathrm{GL}}}

\def\grad{\ensuremath{{\rm grad}}}
\def\ker{\ensuremath{{\rm ker}}}
\def\rank{\ensuremath{{\rm rank}}}
\def\jac{\ensuremath{{\rm jac}}}

\def\scrT{\ensuremath{\mathscr{T}}}

\def\scrR{\ensuremath{\mathscr{R}}}
\def\scrY{\ensuremath{\mathscr{Y}}}
\def\scrB{\ensuremath{\mathscr{B}}}

\def\y {\ensuremath{\mathbf{y}}}

\def\a {\ensuremath{\mathbf{a}}}
\def\b {\ensuremath{\mathbf{b}}}
\def\z {\ensuremath{\mathbf{z}}}
\def\x {\ensuremath{\mathbf{x}}}

\def\f {\ensuremath{\mathbf{f}}}
\def\g {\ensuremath{\mathbf{g}}}
\def\h {\ensuremath{\mathbf{h}}}

\def\L {\ensuremath{\mathbf{L}}}
\def\F {\ensuremath{\mathbf{F}}}
\def\G {\ensuremath{\mathbf{G}}}

\def\X {\ensuremath{\mathbf{X}}}
\def\Y {\ensuremath{\mathbf{Y}}}
\def\H {\ensuremath{\mathbf{H}}}

\def\m {\ensuremath{\mathfrak{m}}}
\def\v {\ensuremath{\mathbf{v}}}
\def\u {\ensuremath{\mathbf{u}}}

\def\U {\ensuremath{\mathbf{U}}}

\def\bdelta{\mbox{\boldmath$\delta$}}
\def\bzeta{\mbox{\boldmath$\zeta$}}

\def\bpsi{\mbox{\boldmath$\psi$}}
\def\bphi{\mbox{\boldmath$\phi$}}

\def\brho{\mbox{\boldmath$\rho$}}
\def\bell{\mbox{\boldmath$\ell$}}

\def\Vfiber{\ensuremath{V''}}
\def\fiber2{\ensuremath{Q''}}
\def\fibersing2{\ensuremath{S''}}
\def\param2{\ensuremath{\mathscr{Q}''}}
\def\paramsing2{\ensuremath{\mathscr{S}''}}

\def\gcd{\ensuremath{\mathrm{GCD}}}

\begin{document}

%\markboth{M. Safey El Din and \'E. Schost}{Nearly optimal roadmap algorithm}

% Title portion
\title{A nearly optimal algorithm for deciding connectivity queries
in smooth and bounded real algebraic sets}
\author[1]{Mohab Safey El Din}
\author[2]{\'Eric Schost}
\affil[1]{Sorbonne Universit\'es, UPMC Univ. Paris 06, CNRS, INRIA Paris Center, LIP6, PolSys Team, France}
\affil[2]{David Cheriton School of Computer Science, University of Waterloo, Waterloo, ON, Canada}
% \author{Mohab Safey El Din \\
% Sorbonne Universit\'es, UPMC Univ. Paris 06, \\
% CNRS, INRIA Paris Center, LIP6, \\
% PolSys Team, France \\
% \'Eric Schost \\
% David Cheriton School of Computer Science, \\University of Waterloo, \\Waterloo, ON, Canada
% }

% \category{I.1.2}{Computing Methodologies}{Symbolic and algebraic manipulation}

% \keywords{Non-linear computational geometry, connectivity queries, algorithms, complexity}

% \acmformat{Mohab Safey El Din and \'Eric Schost, 2015. A nearly
%   optimal algorithm for deciding connectivity queries in smooth and
%   bounded real algebraic sets.}

\maketitle
\begin{abstract}
  A roadmap for a semi-algebraic set $S$ is a curve which has a
  non-empty and connected intersection with all connected components
  of $S$. Hence, this kind of object, introduced by Canny, can be used
  to answer connectivity queries (with applications, for instance, to
  motion planning) but has also become of central importance in
  effective real algebraic geometry, since it is used in higher-level
  algorithms.

  In this paper, we provide a probabilistic algorithm which computes
  roadmaps for smooth and bounded real algebraic sets. Its output size
  and running time are polynomial in $(nD)^{n\log(d)}$, where $D$ is
  the maximum of the degrees of the input polynomials, $d$ is the
  dimension of the set under consideration and $n$ is the number of
  variables. More precisely, the running time of the algorithm is
  essentially subquadratic in the output size. Even under our
  assumptions, it is the first roadmap algorithm with output size and
  running time polynomial in $(nD)^{n\log(d)}$.
\end{abstract}

\section{Introduction}

Roadmaps were introduced by Canny~\cite{CannyThese,Canny} as a
means to decide connectivity properties for semi-algebraic
sets. Informally, a roadmap of a semi-algebraic set $S$ is a
semi-algebraic curve in $S$, whose intersection with each connected
component of $S$ is non-empty and connected: connecting points on $S$
can then be reduced to connecting them to the roadmap and moving along
it.  The initial motivation of Canny's work was to motion planning,
but computing roadmaps actually became the key to further algorithms
in semi-algebraic geometry, such as computing a decomposition of a
semi-algebraic set into its semi-algebraically connected
components~\cite{BaPoRo06}.

This paper presents an algorithm that computes a roadmap of a real
algebraic set, under some regularity, smoothness and compactness
assumptions. In all this work, we work over a real field $\QQ$ with
real closure $\R$ and algebraic closure $\C$ (the reader may replace
$\QQ$ by the field of rational numbers $\Q$, $\R$ by the field of
reals $\mathbb{R}$ and $\C$ by the field of complex numbers
$\mathbb{C}$). To estimate running times, we count arithmetic
operations $(+,-,\times,\div)$ in $\QQ$ at unit cost.

\subsection{Prior results} 
Let $S \subset \R^n$ be a semi-algebraic set. If $S$ is defined by $s$
equations and inequalities with coefficients in $\QQ$ of degree
bounded by $D$, the cost of Canny's algorithm is
$s^n\log(s)D^{O(n^4)}$ operations in $\QQ$ \cite{Canny}; a Monte Carlo
version of it runs in time $s^n\log(s)D^{O(n^2)}$. Subsequent
contributions \cite{HRSRoadmap,GRRoadmap} gave algorithms of cost
$(sD)^{n^{O(1)}}$; they culminate with the algorithm of Basu, Pollack
and Roy~\cite{BaPoRo96,BPRRoadmap} of cost $s^{d+1}D^{O(n^2)}$,
where $d\le n$ is the dimension of the algebraic set defined by all
equations in the system.

None of these algorithms has cost lower than $D^{O(n^2)}$ and none of
them returns a roadmap of degree lower than $D^{O(n^2)}$. Yet, \new{in
  the case of real algebraic sets}, one would expect that a much
better cost $D^{O(n)}$ be achievable, since this is an upper bound on
the number of connected components of $S$, and many other questions
(such as finding at least one point per connected component) can be
solved within that cost.

In~\cite{SaSc11}, we proposed a probabilistic algorithm for the
hypersurface case that extended Canny's original approach; under
smoothness and compactness assumptions, the cost of that algorithm is
$(nD)^{O(n^{1.5})}$. In a nutshell, the main new idea introduced in
that paper is the following. Canny's algorithm and his successors,
including that in~\cite{SaSc11}, share a recursive structure, where
the dimension of the input drops through recursive calls; the main
factor that determines their complexity is the depth $\depth$ of the
recursion, since the cost grows roughly like $D^{O(\depth n)}$ for
inputs of degree $D$. In Canny's version, the dimension drops by one
at each step, so the recursion depth $\depth$ can reach $n-1$.

In~\cite{SaSc11}, we introduced new proof techniques for connectivity
results that leave more freedom in the construction of a roadmap,
allowing us to decrease the depth of the recursion. The algorithm
in~\cite{SaSc11} used baby-steps / giant-steps techniques, combining
steps of size $O(\sqrt{n})$ \new{(where the dimension decreases
  by roughly $\sqrt{n}$)} and steps of unit size, leading to an overall
recursion depth of $O(\sqrt{n})$.

The results in~\cite{SaSc11} left many questions open, such as making
the algorithm deterministic, removing the smoothness-compactness
assumptions or generalizing the approach from hypersurfaces to systems
of equations. In~\cite{BRSS}, we answered these questions, while still
following a baby-steps / giant-steps strategy: we showed how to obtain
a deterministic algorithm for computing a roadmap of a general real
algebraic set within a cost of $D^{O(n^{1.5})}$ operations in $\QQ$.

The next step is obviously to use a divide-and-conquer strategy, that
would divide the current dimension by two at every recursive step,
leading to a recursion tree of depth $O(\log(n))$.  In \cite{BaRo13},
Basu and Roy recently obtained such an important result: given $f$ in
$\QQ[X_1,\dots,X_n]$, their algorithm computes a roadmap for
$V(f)\cap \R^n$ in time polynomial in $n^{n\log^3(n)} D^{n \log^2(n)}$
while the output has size polynomial in
$n^{n\log^2(n)} D^{n \log(n)}$. Note that this algorithm is not
polynomial in its output size; the extra logarithmic factors appearing
in the exponents reflect the cost of computing with $O(\log(n))$
infinitesimals. Since that algorithm makes no smoothness assumption on
$V(f)$, it can as well handle the case of a system of equations
$f_1=\cdots=f_s=0$ by taking $f=\sum_i f_i^2$. Note also that this
algorithm is deterministic.

In this paper, we present as well a divide-and-conquer roadmap
algorithm. Compared to Basu and Roy's recent work, our algorithm is
probabilistic and handles less general situations (we still rely on
smoothness and compactness). However, it features a better running
time for such inputs: both output degree and running time are
polynomial in $(nD)^{n\log(d)}$ (where $d$ is the dimension of the
algebraic set we consider), the running time of our algorithm is
subquadratic in the size of the output, and the complexity constants
that lie in the exponent are made explicit.

%%%%%%%%%%%%%%%%%%%%%%%%%%%%%%%%%%%%%%%%%%%%%%%%%%%%%%%%%%%%

\subsection{Roadmaps: definition and data representation}\label{ssec:data}%\label{sec:intro}

\paragraph*{Definition}
Our definition of a roadmap in the algebraic case is as follows.  Let
$V \subset \C^n$ be an algebraic set (the set of common solutions in
$\C^n$ to some polynomial equations). An
algebraic set ${R} \subset \C^n$ is a {\em roadmap of $V$} if the
following holds:
\begin{itemize}
\item ${R}$ is either an algebraic curve, or empty;

\smallskip

\item ${R}$ is contained in $V$;

\smallskip

\item each semi-algebraically connected component of $V\cap \R^n$ has
  a non-empty and semi-algebraically connected intersection with
  ${R}\cap \R^n$.
\end{itemize}
Finally, if $C$ is a finite subset of $\C^n$, we say that ${R}$ is a
{\em roadmap of $(V,C)$} if we have in addition:
\begin{itemize}
\item ${R}$ contains $C\cap V \cap \R^n$.
\end{itemize}
The set $C$ will be referred to as {\em control points}. For instance,
computing a roadmap of $(V,\{P_1,P_2\})$ enables us to test if the
points $P_1,P_2$ are on the same connected component of $V\cap \R^n$.

This definition is from~\cite{SaSc11}; it slightly differs from the
one in e.g.~\cite{BaPoRo06}, but serves the same purpose: compared
to~\cite{BaPoRo06}, our definition is coordinate-independent, and does
not involve a condition (called ${\rm RM}_3$ in~\cite{BaPoRo06}) that
is specific to the algorithm used in that reference. Most importantly,
we do not deal here with semi-algebraic sets, but with algebraic sets
only.

\paragraph*{Straight-line programs}  Our algorithms handle
mainly multivariate polynomials, as well as finite sets of points
and algebraic curves.

The input polynomials will be given by {\em straight-line programs}.
Informally, this is a representation of polynomials by means of a
sequence of operations $(+,-,\times)$, without test or division.
Precisely, a straight-line program $\Gamma$ computing polynomials in
$\QQ[X_1, \ldots, X_N]$ is a sequence $\gamma_1,\dots,\gamma_E$, where
for $i \ge 1$, we require that one of the following holds:
\begin{itemize}
\item $\gamma_i =\lambda_i$, with $\lambda_i \in \QQ$;

\smallskip

\item $\gamma_i = ({\sf op}_i,\lambda_i,a_i)$, with ${\sf op}_i \in
  \{+,-,\times\}$, $\lambda_i \in \QQ$ and $-N+1 \leq a_i < i$
(non-positive indices will refer to input variables);

\smallskip

\item  $\gamma_i = ({\sf op}_i,a_i,b_i)$, with   ${\sf op}_i \in
  \{+,-,\times\}$ and $-N+1 \leq a_i,b_i < i$.
\end{itemize}
To $\Gamma$, we can associate polynomials $G_{-N+1},\dots,G_{E}$
defined in the following manner: for $-N+1 \leq i \leq 0$, we take
$G_i=X_{i+N}$; for $i \ge 1$, $G_i$ is defined inductively in the
obvious manner, as either $G_i = \lambda_i$, $G_i = \lambda_i ~{\sf
  op}_i~ G_{a_i}$ or $G_i = G_{a_i} ~{\sf op}_i~ G_{b_i}$. We say that
$\Gamma$ {\em computes} some polynomials $f_1,\dots,f_s$ if all $f_i$
belong to $\{G_{-N+1},\dots,G_E\}$.  Finally, we call $E$ the
\emph{length} of $\Gamma$.

The reason for this choice is that we will use algorithms for solving
polynomial systems that originate in the
references~\cite{GiHeMoPa95,GiHeMoPa97,GiHeMoMoPa98,GiLeSa01,Lecerf2000},
where such an encoding is used. This is not a restriction, since any polynomial
of degree $D$ in $n$ variables can be computed by a straight-line program
of length $O(D^n)$, obtained by evaluating and summing all its monomials.

\paragraph*{Representing the output}
To represent finite algebraic sets and algebraic curves, we respectively use {\em
  zero-dimensional} and {\em one-dimensional} parametrizations. 

A zero-dimensional parametrization
$\scrQ=((q,v_1,\dots,v_n),\pollambda)$ with coefficients in $\QQ$
consists in polynomials $(q,v_1,\dots,v_n)$, such that $q\in \QQ[T]$
is squarefree and all $v_i$ are in $\QQ[T]$ and satisfy $\deg(v_i) <
\deg(q)$, and in a $\QQ$-linear form $\pollambda$ in the variables
$X_1,\dots,X_n$, such that $\pollambda(v_1,\dots,v_n)=T$. The
corresponding algebraic set, denoted by $\Zeroes(\scrQ)\subset \C^n$,
is defined in a parametric manner by
$$q(\roottau) = 0, \qquad X_i = v_i(\roottau) \ \ (1 \le i \le n);$$ it is
thus a finite set of points parametrized by the finitely many roots of
$q$. The constraint on $\pollambda$ says that the roots of $q$ are the
values taken by $\pollambda$ on $\Zeroes(\scrQ)$. The {\em degree} of
$\scrQ$ is defined as $\deg(q)=|\Zeroes(\scrQ)|$. By convention, the
sequence $(1)$ is considered as a zero-dimensional parametrization
that defines the empty set. 

Any finite subset $Q$ of $\C^n$ defined over $\QQ$ ({\it i.e.}, which
can be written as the zero-set of polynomials in $\QQ[X_1,\dots,X_n]$)
can be represented as $Q=\Zeroes(\scrQ)$, for a suitable~$\scrQ$. This kind
of description goes back to work of Kronecker and
Macaulay~\cite{Kronecker82,Macaulay16}, and has been used in computer
algebra since the
1980's~\cite{GiMo89,GiHeMoPa95,ABRW96,GiHeMoPa97,GiHeMoMoPa98,Rouillier99,GiLeSa01,Lecerf2000}.

\new{Next, we discuss the extension of this idea to algebraic curves.
  A {\em one-dimensional parame\-trization}
  $\scrQ=((q,v_1,\dots,v_n),\pollambda,\pollambda')$ with coefficients
  in $\QQ$ consists in polynomials $(q,v_1,\dots,v_n)$, such that
  we have:
  \begin{itemize}
  \item $q\in \QQ[U,T]$ is squarefree and monic in $U$ and $T$, with $\deg(q, U)=\deg(q, T)=\deg(q)$,
\smallskip
  \item $v_i$ are in $\QQ[U,T]$ and satisfy $\deg(v_i,T) < \deg(q,T)$,
  \end{itemize}
  and in linear forms $\pollambda,\pollambda'$ in $X_1,\dots,X_n$,
  such that
  $$\pollambda\left (v_1,\dots,v_n\right)=T \frac{\partial q}{\partial T}
  \bmod q \quad\text{and}\quad \pollambda'\left (v_1,\dots,v_n\right)=U
  \frac{\partial q}{\partial T} \bmod q.$$ The
corresponding algebraic set, denoted by $\Zeroes(\scrQ)\subset \C^n$, is now
defined as the smallest algebraic set containing the curve defined in
 a parametric manner by
 \begin{equation}\label{eq:def_one_dim_param}
q(\eta,\rootrho) = 0, \qquad \frac{\partial q}{\partial T}(\eta,\rootrho)
\ne 0, \qquad X_i = \frac{v_i(\eta, \rootrho)}{\frac{\partial q}{\partial
    T}(\eta, \rootrho)} \ \ (1 \le i \le n).   
 \end{equation}
}

The {\em degree} $\delta$ of $\Zeroes(\scrQ)$ is the maximum of the
cardinalities of the finite sets obtained by intersecting
$\Zeroes(\scrQ)$ with a hyperplane (whenever such sets are finite).
In all cases we use one-dimensional parametrizations, we request
additionally that $\delta=\deg(q)$.

Using for instance~\cite[Theorem~1]{Schost03}, we deduce that all
polynomials $q,v_1,\dots,v_n$ have total degree at most $\delta$; this
is the reason why we use these polynomials: if we were to invert the
denominator $\partial q/\partial T$ modulo $q$ in $\QQ(U)[T]$
in~\eqref{eq:def_one_dim_param}, thus involving rational functions in
$U$, the degree in $U$ would be quadratic in $\delta$.

Thus, we are now using the points of the plane curve $V(q) \subset
\C^2$ defined by $q(\eta,\rootrho)=0$ to parametrize the space curve
$\Zeroes(\scrQ)$; the condition on $\pollambda$ and $\pollambda'$
means that the plane curve $V(q)$ is the smallest algebraic set
containing the image of $\Zeroes(\scrQ)$ through the projection $\x
\mapsto (\pollambda'(\x),\pollambda(\x))$.

Any algebraic curve in $\C^n$ defined by polynomials with coefficients
in $\QQ$ can be written as $\Zeroes(\scrQ)$, for some
one-dimensional parametrization  $\scrQ$, by choosing $\pollambda$ and
$\pollambda'$ as random linear forms in $\QQ[X_1, \ldots, X_n]$ (this
is classical; see for instance~\cite{GiLeSa01}). For a curve of degree
$\delta$, such a description involves $O(n\delta^2)$ monomials.

The output of our algorithm is a roadmap $R$ of an algebraic set $V$:
it will thus be represented by a one-dimensional
parametrization. Given such a data structure, we explained
in~\cite{SaSc11} how to construct paths between points in $V \cap
\R^n$, so as to answer connectivity queries.

%%%%%%%%%%%%%%%%%%%%%%%%%%%%%%%%%%%%%%%%%%%%%%%%%%%%%%%%%%%%

\subsection{Main result}

With these definitions, our main result is the following theorem. The
input polynomials are given by means of a straight-line program, whose
length will be called $E$; as said above, we can always use a trivial
straight-line program of length $O(D^n)$ to encode a polynomial of
degree $D$, so in the worst case we can take $E = O(n D^n)$. We make a
regularity assumption on these polynomials, that they should form a
{\em reduced regular sequence}. This means that for all $i$ in
$\{1,\dots,s\}$, $V(f_1,\dots,f_i)$ is equidimensional of dimension
$n-i$ and the ideal $\langle f_1,\dots,f_i \rangle$ is {\em radical},
in the sense that any polynomial vanishing on $V(f_1,\dots,f_i)$ must
belong to that ideal (in the next section, we review basic concepts of
algebraic geometry along these lines).

In all this work, the soft-O notation $\softO(g)$ denotes the class $g\log(g)^{O(1)}$.

\begin{theorem}\label{theo:main}
  Consider $\f=(f_1,\dots,f_p)$ of degree at most $D$ in
  $\QQ[X_1,\dots,X_n]$, given by a straight-line program of length
  $E$. Suppose that $V(\f) \subset \C^n$ has finitely many singular
  points, that $V(\f)\cap \R^n$ is bounded, and that the polynomials
  $\f$ form a reduced regular sequence. Given a zero-dimensional
  parametrization $\scrC$ of degree $\degC$, one can compute a roadmap
  of $(V(\f),\Zeroes(\scrC))$ of degree
\[
\softO \left (
\degC 16^{3d}  (n \log_2(n))^{2(2d+12\log_2(d))(\log_2(d)+6)}D^{(2n+1)(\log_2(d)+4)}
\right )\]
using 
\[
\softO \left (
\degC^3 16^{9d} E (n \log_2(n))^{6(2d+12\log_2(d))(\log_2(d)+7)}D^{3(2n+1)(\log_2(d)+5)}
\right ) 
\]
arithmetic operations in $\QQ$, with $d=n-p$.
\end{theorem}

In other words, both output degree and running time are polynomial in
the quantity $\degC\, (nD)^{n \log(d)}$; the running time is essentially
cubic in the output degree, and subquadratic in the output size ---
recall that if the bivariate polynomials returned as output have
degree $\delta$, the output size, in terms of number of coefficients
in $\QQ$, is essentially $n\delta^2$.

The algorithm is probabilistic in the following sense: at several
steps, we have to choose random elements from the base field,
typically in the form of matrices or vectors. Every time a random
element $\gamma$ is chosen in a parameter space such as $\QQ^i$, there
will exist a non-zero polynomial $\Delta$ such that success is
guaranteed as soon as $\Delta(\gamma)\ne 0$. 

To our knowledge, this is the best known result for this question;
compared to the recent result in~\cite{BaRo13}, the exponents
appearing here are better. Even under our assumptions, Basu and Roy's
algorithm relies on the introduction of several infinitesimals, which
allow them to alleviate problems such as the presence of
singularities; our algorithm avoids introducing infinitesimals, which
improves running times and output degree but requires stronger
assumptions.

%%%%%%%%%%%%%%%%%%%%%%%%%%%%%%%%%%%%%%%%%%%%%%%%%%%%%%%%%%%%

\subsection{Structure of the paper}

This paper is accompanied by an electronic appendix. The goal of the
main text is to give the reader a global view and understanding of the
objects and properties that are used; most proofs are postponed to the
appendix.  Sections in the main text are indexed as 1, 2, \dots;
sections in the appendix as A, B, \dots

We start with a short section of notation and background definitions.
In Section~\ref{sec:dim:smooth:finite}, we introduce the notions of
polar varieties and fibers that will play a crucial role in our
algorithm. 
Geometric properties of polar varieties and fibers allow us to give an
abstract version of our algorithm in Section~\ref{ssec:abstractalgo},
where data representation is not discussed yet. 

We then introduce in Section~\ref{sec:GLS} a construction based on
Lagrange systems, that we call generalized Lagrange system, to
represent all intermediate data (as the more standard techniques using
minors of Jacobian matrices to describe polar varieties do not lead to
acceptable complexity results), from which the final form of our
algorithm follows.

Properties of generalized Lagrange systems and their connection with
polar varieties and fibers are summarized in Section~\ref{sec:GLS}.
The description of our concrete algorithm and its complexity analysis
are given in Section~\ref{sec:mainalgo}; they are based on several
subroutines which are presented in Section~\ref{sec:solvingglag}.

%%%%%%%%%%%%%%%%%%%%%%%%%%%%%%%%%%%%%%%%%%%%%%%%%%%%%%%%%%%%
%%%%%%%%%%%%%%%%%%%%%%%%%%%%%%%%%%%%%%%%%%%%%%%%%%%%%%%%%%%%
%%%%%%%%%%%%%%%%%%%%%%%%%%%%%%%%%%%%%%%%%%%%%%%%%%%%%%%%%%%%

\section{Algebraic sets}\label{sec:prelim}

In this section, we first recall some basic definitions related to
algebraic sets, that is, zero-sets of systems of polynomial equations
(for proofs and standard notions not recalled here, see for
instance~\cite{ZaSa58,Mumford76,Shafarevich77,Eisenbud95}). The last
subsection introduces the concepts of charts and atlases, which will
form the basis of the correctness proofs of our algorithms.

%%%%%%%%%%%%%%%%%%%%%%%%%%%%%%%%%%%%%%%%%%%%%%%%%%%%%%%%%%%%

\subsection{Generalities on algebraic sets}

An {\em algebraic set} $V \subset \C^n$ is the set of common zeros of
some polynomials $\f=(f_1,\dots,f_s)$ in $\C[X_1,\dots,X_n]$; we write
$V=V(f_1,\dots,f_s)=V(\f)$.  We denote by $\IdealV$ the {\em ideal} of
$V$, that is, the set of polynomials in $\C[X_1,\dots,X_n]$ that
vanish at all points of $V$; the set $V$ is said to be {\em defined
  over $\QQ$} if $\IdealV$ can be generated by polynomials with
coefficients in~$\QQ$.

Two fundamental integer quantities associated to algebraic sets are
dimension and degree. Before defining them, let us mention that an
algebraic set $V$ can be uniquely decomposed into a finite union of
{\em irreducible} algebraic sets (that is, algebraic sets which
themselves cannot be written as a finite union of proper algebraic
subsets); they will be called the irreducible components of $V$.
\begin{itemize}
\item The {\em dimension} $\dim(V)$ of an algebraic set $V\subset
  \C^n$ can be defined either as the Krull dimension of
  $\C[X_1,\dots,X_n]/\IdealV$, or equivalently as the number of
  generic hyperplanes needed to obtain a finite set after intersection
  with $V$. We often write $d=\dim(V)$, and the {\em codimension} of
  $V$ is defined as $c=n-\dim(V)$.\smallskip

  For instance, an algebraic set $V\subset \C^n$ defined by a single
  equation $f=0$ (where $f$ is not a constant) has dimension $n-1$:
  intersecting $V$ with $n-1$ generic hyperplanes (defined by generic
  linear equations) and eliminating $n-1$ variables thanks to the
  linear equations leads to a univariate polynomial which has finitely
  many roots.\smallskip

  When all irreducible components of $V$ have the same dimension, we
  say that $V$ is {\em equidimensional}, or $d$-{\em equidimensional}
  if we want to make it clear that this dimension is $d$.

\smallskip

\item The {\em degree} of an irreducible algebraic set $V \subset
  \C^n$ is the number of intersection points between $V$ and $\dim(V)$
  generic hyperplanes (this is also the {\em maximal} number of such
  intersection points); the degree of an arbitrary algebraic set is
  defined as the sum of the degrees of its irreducible
  components~\cite{Heintz83}.
  For instance, the degree of an algebraic set $V\subset \C^n$ defined
  by a single squarefree equation $f=0$ equals the degree of the
  polynomial $f$. 

\smallskip

Crucial for us will be the {\em B\'ezout bound}~\cite{Heintz83}:
  if polynomials $\f=(f_1,\dots,f_s)$ have degree at most $D$, their zero-set
  $V(\f)$ has degree at most $D^s$.
\end{itemize}
Most important for our purposes will be algebraic sets of dimension
zero, and equidimensional algebraic sets of dimension $1$. The former 
are thus finite sets of points, for which degree equals cardinality;
the latter are algebraic curves, for which the degree is the number of 
intersection points with a generic hyperplane.

Finally, we mention that algebraic sets are the closed sets for the
so-called Zariski topology on $\C^n$; the Zariski closure
\new{$\overline{S}$} of an arbitrary subset \new{$S$} of $\C^n$ is thus the
smallest algebraic set that contains it.  For $\f=(f_1,\dots,f_s)$ as
above, the complement $\C^n-V(\f)$ will be written $\Open(\f)$;
it is open for the Zariski topology.

%%%%%%%%%%%%%%%%%%%%%%%%%%%%%%%%%%%%%%%%%%%%%%%%%%%%%%%%%%%%

\subsection{Local properties}

Next, we discuss regular and singular points of an algebraic set. Let
thus $V$ be an algebraic set in $\C^n$. For $f$ in $\C[X_1, \ldots,
  X_n]$ and $\x$ in $\C^n$, we denote by $\grad_\x(f)$ the evaluation
of the gradient vector of $f$ at $\x$. Then, the {\em tangent space to
  $V$} at $\x \in V$ is the vector space $\T_\x V$ defined by the
equations $\grad_\x(f)\cdot\v =0$, for all polynomials $f$ in the
ideal $\IdealV$.

If $V$ is equidimensional, we define {\em regular points} on $V$ as
those points $\x$ where $\dim(T_\x V)=\dim(V)$ and {\em singular
  points} as all other points in $V$. The set of regular, resp.\ singular,
points is denoted by $\reg(V)$, resp.\ $\sing(V)$; the latter is an
algebraic subset of $V$, of smaller dimension than $V$. An
equidimensional algebraic set $V$ is said to be smooth when $\sing(V)$
is empty.
%% \cite[Chap. 8,  pp. 168]{Shafarevich77}.

For polynomials $\f=(f_1,\dots,f_s)$ in $\C[X_1,\dots,X_n]$,
$\jac(\f)$ denotes the Jacobian matrix of $(f_1,\dots,f_s)$ with
respect to $X_1,\dots,X_n$; later on, we will also use the notation
$\jac(\f,i)$, which for $i\le n$ denotes the matrix obtained by
removing the first $i$ columns from $\jac(\f)$. As for gradients,
$\jac_\x(\f)$ and $\jac_\x(\f,i)$ denote the same matrices, with
entries evaluated at a point $\x$ in $\C^n$.

The following lemma is a direct consequence
of~\cite[Corollary~16.20]{Eisenbud95}, and gives us a more concrete
description of the objects defined above.

\begin{lemma}\label{sec:prelim:lemma:sing}
  If $V\subset \C^n$ is a $d$-equidimensional algebraic set, whose ideal $\IdealV$ is generated by polynomials
  $\f= (f_1,\dots,f_s)$, then we have the following:
  \begin{itemize}
  \item at any point $\x$ of $\reg(V)$, $\jac_\x(\f)$ has full rank
    $c=n-\dim(V)$ and its kernel is $T_\x V$;

\smallskip

  \item $\sing(V)$ is the zero-set of $\f$ and all
    $c$-minors of $\jac(\f)$.
  \end{itemize}
\end{lemma}

%%%%%%%%%%%%%%%%%%%%%%%%%%%%%%%%%%%%%%%%%%%%%%%%%%%%%%%%%%%%

\subsection{Changes of variables}\label{ssec:preliminaries:changevar}

Several statements will depend on linear changes of variables. If
$\KK$ is a field (typically for us $\C$ or $\QQ$), we denote by $\GL(n,
\KK)$ the set of $n\times n$ invertible matrices with entries in
$\KK$; when $\KK=\C$, we simply write $\GL(n)$ for $\GL(n, \C)$.  The
subset of matrices in $\GL(n, \KK)$ which leave invariant the first
$e$ coordinates and which act only on the last $n-e$ ones is denoted
by $\GL(n,e,\KK)$; such matrices have a $2 \times 2$ block diagonal
structure, the first block being the identity. {\em If extra variables
  are added on top of $\X=X_1,\dots,X_n$, these matrices will act only
  on the $\X$ variables}.

Given $f$ in $\C[X_1,\dots,X_n]$, and $\mA$ in $\GL(n)$, $f^\mA$
denotes the polynomial $f(\mA\X)$ and for $V \subset \C^n$, $V^\mA$
denotes the image of $V$ by the map $\phi_\mA: \x \mapsto
\mA^{-1}\x$. Thus, we have that for polynomials $\f=(f_1,\dots,f_s)$,
$V(\f^\mA)=\phi_\mA(V(\f))=V(\f)^\mA$.

The success of our algorithms will depend on our change of variables
being ``lucky'', in a sense that will always be made explicit. Our
statements will take the form: ``there exists a non-empty Zariski open
subset $\mathscr{O}$ of $\GL(n)$ such that for $\mA$ in $\mathscr{O}$,
\dots (some desirable properties are guaranteed)''. Strictly speaking,
we have only defined Zariski open and closed sets in $\C^n$, but the
definition carries over to subsets of $\GL(n)$ (which itself is open
in $\C^{n^2}$) by considering the induced topology.

%%%%%%%%%%%%%%%%%%%%%%%%%%%%%%%%%%%%%%%%%%%%%%%%%%%%%%%%%%%%

\subsection{Fixing coordinates}\label{ssec:preliminaries:fixing}

\new{The structure of the main algorithm will require us to constantly consider
  situations where the first coordinates are fixed.
 For a fixed ambient dimension $n$ (which will always be clear 
from the context) and integers $0 \le e \le n$ and $0 \le d \le n-e$, we
denote by $\pi_{e,d}$ the projection
$$\begin{array}{cccc}
  \pi_{e,d}: & \C^n & \to & \C^d \\
  & \x =(x_1,\dots,x_n) & \mapsto & (x_{e+1},\dots,x_{e+d}).
\end{array}$$
For
$e=0$, $\pi_{0,d}$ is the projection on the space of the first $d$
coordinates; in this case, we simply write $\pi_{d}$.}

For $d=0$, we let $\C^0$ be a singleton of the form $\C^0=\{\bullet
\}$, and $\pi_{e,0}$ is the constant map $\x \mapsto \bullet$ (in this
respect, we also make the convention that the empty sequence $(\,)$ is
seen as a zero-dimensional parametrization encoding the singleton
$\{\bullet \}$).

\new{Consider a set $V$ in $\C^n$ and a subset $Q$ of $\C^d$, for some $d\in
\{1, \ldots, n\}$. Then, the {\em fiber} of $V$ above $Q$ for the
projection $\pi_d$ is the set $\fbr(V,Q)=V \cap {\pi_d}^{-1}(Q)$;
we say that $V$ {\em lies over $Q$} if $\pi_d(V)$ is contained
in~$Q$. For $\y$ in $\C^d$, we will further write $\fbr(V,\y)$
instead of the more formally correct $\fbr(V,\{\y\})$.}

%%%%%%%%%%%%%%%%%%%%%%%%%%%%%%%%%%%%%%%%%%%%%%%%%%%%%%%%%%%%

\subsection{Charts and atlases}

An equidimensional algebraic set $V \subset \C^n$ is a {\em complete
  intersection} if it can be defined by a number of equations equal to
its codimension. This is a particularly convenient situation, as many
geometric properties are easier to comprehend in such a case.

\new{We will not be able to ensure this property throughout our
  algorithm, so we will replace it by a local version. We will also
  impose a smoothness property, leading us to the following notion of
  {\em chart}.  This definition applies to an algebraic set $V$ lying
  over a finite set $Q$, together with a set $S$ lying over $Q$ that
  we wish to exclude (this will be typically the set of singular
  points of $V$, or a superset of it).}

\begin{definition}\label{def:chart}
\new{ Let $n,e$ be integers, with $e \le n$, let $Q \subset \C^e$ be a
  finite set, and let $V \subset \C^n$ and $S\subset \C^n$ be algebraic
  sets lying over $Q$.}

\new{  We say that a pair of the form $\psi=(m,\h)$, with $m$ and
  $\h=(h_1,\dots,h_c)$ in $\C[X_1,\dots,X_n]$, is a {\em chart} of
  $(V,Q,S)$ if the following properties hold:
  \begin{enumerate}
  \item[${\sfC_1}.$] $\Open(m)\cap V-S$ is not empty;
\smallskip
  \item[${\sfC_2}.$] $\Open(m)\cap V-S = \Open(m) \cap \fbr(V(\h),Q)-S$;
\smallskip
  \item[${\sfC_3}.$] the inequality $c+e \le n$ holds;
\smallskip
  \item[${\sfC_4}.$] for all $\x$ in $\Open(m)\cap V-S$, the
    Jacobian matrix $\jac(\h,e)$ has full rank $c$ at $\x$.
  \end{enumerate}
}\end{definition}

This definition is inspired by the construction
in~\cite[Proposition~3.3.8]{BoCoRo98}. The salient points are the set
equality ${\sfC_2}$, together with the rank condition ${\sfC_4}$. To
understand the latter, consider the particular case where the finite
set $Q$ is a single point $(x_1,\dots,x_e)$. Then, the fiber
$\fbr(V(\h),Q)$ in ${\sfC_2}$ is defined by the equations
$(X_i-x_i)_{1\le i\le e}$ and $\h$, and 
the rank condition in $\sfC_4$ says that the Jacobian
matrix of these equations has full rank at $\x$.

An easy consequence of this definition is that when $V$ is
equidimensional of dimension $d$, if $\psi=(m,(h_1,\dots,h_c))$ is a
chart of $(V,Q,S)$, then as one would expect, $c=n-e-d$. This result
is proved as Lemma~\ref{sec:lemma:singS} in the electronic appendix.

Continuing the analogy with differential geometry, we will also rely
on the notion of {\em atlas} of $(V,Q,S)$.

\begin{definition} \label{def:atlas} 
\new{Let $n,e$ be integers, with $e \le n$, let $Q \subset \C^e$ be a
  finite set, let $V \subset \C^n$ and $S\subset \C^n$ be algebraic
  sets lying over $Q$.}

\new{An {\em atlas} of $(V, Q,S)$ is the data of $\bpsi=(\psi_i)_{1 \le i
    \le s}$, with $\psi_i=(m_i,\h_i)$ for all $i$, such that:
  \begin{enumerate}
  \item[${\sfA_1}.$] each $\psi_i$ is a chart of $(V, Q,S)$;
\smallskip
  \item[${\sfA_2}.$] $s \ge 1$ ({\it i.e.}, $\bpsi$ is not the empty sequence);
\smallskip
  \item[${\sfA_3}.$] the open sets $\Open(m_i)$ cover $V-S$. 
  \end{enumerate}
}
\end{definition}

If $V$ is equidimensional, there always exists an atlas for
$(V,Q,\sing(V))$. Conversely, the existence of an atlas for
$(V,Q,S)$, for some set $S$,
 is not enough to ensure that $V$ is equidimensional.
However, if this is known to be the case, and if $(V,Q,S)$ admits an
atlas, then all singular points of $V$ are in $S$.  As another example
of a useful property, if $(V,Q,S)$ admits an atlas
$\bpsi=(\psi_i)_{1 \le i \le s}$, with $\psi_i=(m_i,\h_i)$ for all
$i$, and if all $\h_i$ have the same cardinality $c$, then $V$ is
$d$-equidimensional, with $d=n-e-c$. These properties are proved in
Section~\ref{sec:preliminaries} of the electronic appendix.

Given a matrix $\mA$ in $\GL(n,e)$, and an atlas $\bpsi=(\psi_i)_{1
  \le i \le s}$ of $(V,Q,S)$, with $V$ and $S$ in $\C^n$, $Q$ in $\C^e$ and
$\psi_i=(m_i,\h_i)$ for all $i$, we write $\bpsi^\mA=(\psi^\mA_i)_{1
  \le i \le s}$, with $\psi_i^\mA=(m_i^\mA,\h_i^\mA)$ for all
$i$. Then, all $\psi_i^\mA$ are charts of $(V^\mA,Q,S^\mA)$, and
$\bpsi^\mA$ is an atlas of $(V^\mA,Q,S^\mA)$ (note that such a matrix
$\mA$ leaves $Q$ invariant, so $Q^\mA=Q$).

It is worth noting that the algorithms will never explicitly compute
any chart or atlas; however, we will rely on the properties of these
objects to establish correctness.

%%%%%%%%%%%%%%%%%%%%%%%%%%%%%%%%%%%%%%%%%%%%%%%%%%%%%%%%%%%%
%%%%%%%%%%%%%%%%%%%%%%%%%%%%%%%%%%%%%%%%%%%%%%%%%%%%%%%%%%%%
%%%%%%%%%%%%%%%%%%%%%%%%%%%%%%%%%%%%%%%%%%%%%%%%%%%%%%%%%%%%

\section{Fibers and polar varieties}\label{sec:dim:smooth:finite}

The basic geometric constructions on which our algorithm relies are
fibers, already described above, and polar varieties. In this section, 
we state the main geometric properties (dimension, smoothness)
of these objects.

%%%%%%%%%%%%%%%%%%%%%%%%%%%%%%%%%%%%%%%%%%%%%%%%%%%%%%%%%%%%

\subsection{Polar varieties}

Let $Q$ be a finite subset of $\C^e$, and let $V$ be an algebraic
subset of $\C^n$ lying over $Q$. If $V$ is $d$-equidimensional, for
any integer $\dalgo$ in $\{1,\dots,d\}$ the {\em open polar variety}
$\openpolar(e,\dalgo,V)$
is defined as the set of
critical points of $\pi_{e,\dalgo}$ on $\reg(V)$, that is, the set of
points $\x$ in $\reg(V)$ such that $\pi_{e,\dalgo}(T_\x V)$ has
dimension less than $\dalgo$. We further define the following objects:
\begin{itemize}
\item $\polar(e,\dalgo,V)$
  is the Zariski closure of $\openpolar(e,\dalgo,V)$;
\smallskip
\item $\Kpolar(e,\dalgo,V)=\openpolar(e,\dalgo,V) \cup \sing(V)$.
\end{itemize}
The set $\Kpolar(e,\dalgo,V)$ turns out to be closed for the Zariski
topology.  For instance, if $e=0$ and if the defining ideal of $V$ is
generated by polynomials $\f= (f_1,\dots,f_s)$, using Lemma
\ref{sec:prelim:lemma:sing}, we can deduce that $\Kpolar(0,\dalgo,V)$
is the subset of $V$ where $\jac(\f,\dalgo)$ has rank less than $c$,
where $c=n-d$ is the codimension of $V$ (this is proved as
Lemma~\ref{sec:prelim:lemma:Ksing} in appendix).

Since $\Kpolar(e,\dalgo,V)$ contains $\openpolar(e,\dalgo,V)$, and
since it is Zariski closed, it must contain $\polar(e,\dalgo,V)$ as
well. Although we will be mostly interested in $\polar(e,\dalgo,V)$,
the superset $\Kpolar(e,\dalgo,V)$ will turn out to be slightly simpler to compute, as
suggested by the remark above.
In cases where $V$ has no singular point, this distinction
becomes irrelevant, as the sets $\openpolar(e,\dalgo,V)$,
$\polar(e,\dalgo,V)$ and $\Kpolar(e,\dalgo,V)$ all coincide. 

Polar varieties as considered for instance in
references~\cite{BaGiHeMb97,BaGiHeMb01} and their successors
correspond to $e=0$.

Polar varieties were introduced by algebraic geometers Severi and Todd
in the 1930's, as a means to define characteristic classes, and they
played an important role in singularity theory in the 1970's and
1980's; see~\cite{Piene78,Teissier88} for a history of this
subject. They were used for algorithmic purposes in real geometry by
Bank, Giusti, Heintz {\it et al.} in a series of papers starting in
1997~\cite{BaGiHeMb97}, whose goal was to compute sample points on
real algebraic sets \cite{BaGiHeMb01,BaGiHePa05,BGHSS,SaSc03} and for
polynomial optimization~\cite{BGHS14,GS14}. While these ideas are close in
essence to other forms of {\em critical point
  methods}~\cite{BaPoRo06}, the rich geometry underlying the
construction of polar varieties is the key to many useful results (see
also~\cite{RRS,ARS}).

\begin{figure}[ht]
\centerline{\includegraphics[width=5cm]{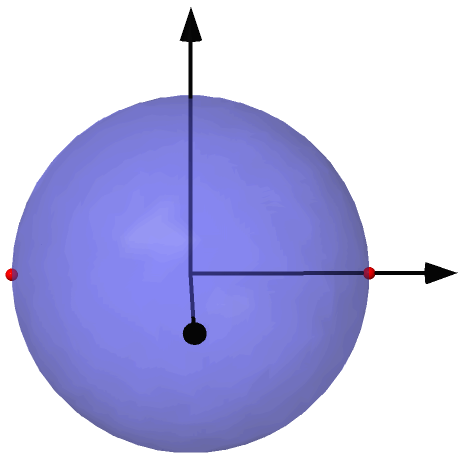}\qquad\quad
\includegraphics[width=5cm]{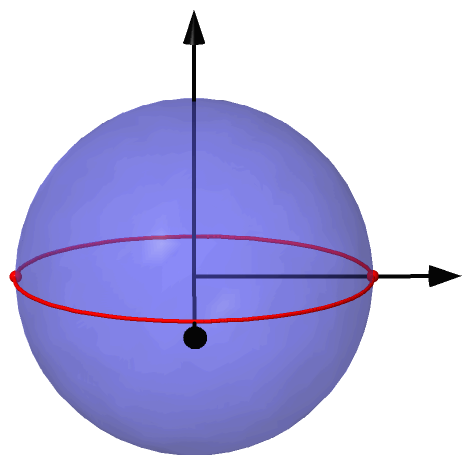}}
\caption{The polar varieties $\polar(0, 1, V)$ and $\polar(0, 2, V)$, where $V=V(X_1^2+X_2^2+X_3^2-1)$}
\label{fig:polar1}
\end{figure}

\begin{example}\label{example1}
  Figure~\ref{fig:polar1} shows the real points of the polar varieties
  $\polar(0, 1, V)$ and $\polar(0, 2, V)$, where $V \subset \Co^3$ is
  the $2$-dimensional sphere defined by $X_1^2+X_2^2+X_3^2-1=0$; these
  polar varieties
  correspond to critical points of projections on respectively a line
  and a plane. In this particular case, we see that $\polar(0, 1, V)$
  is defined by
$$
X_1^2+X_2^2+X_3^2-1=X_2=X_3=0,
$$ 
and that it has dimension zero.  The polar variety $\polar(0, 2, V)$ is defined by 
$$
X_1^2+X_2^2+X_3^2-1=X_3=0
$$ 
and it has dimension one. 
\end{example}
This example suggests that when $V$ is smooth and
equidimensional, $W(0, \dalgo, V)$ has dimension $\dalgo-1$. The next
proposition will show that this dimension property indeed holds,
provided we are in generic coordinates. In this respect, one should
notice that in general, $\polar(e,d,V^\mA)$ differs from
$\polar(e,d,V)^\mA$: the geometry of polar varieties, in particular
their dimension, may change when one applies a linear change of variables
to $V$.

The precise form of this dimension statement (which will be required in the
proof of Proposition~\ref{SEC:LAGRANGE:PROP:TRANSFER:POLAR} below) is
constructive: given an atlas for $V$, we build atlases for its polar
varieties.

Let $Q \subset \C^e$ be a finite set and let $V \subset \C^n$ and
$S\subset \C^n$ be algebraic sets lying over $Q$.  Suppose that $V$ is
equidimensional of dimension $d$ and consider an atlas
$\bpsi=(\psi_i)_{1 \le i \le s}$ for the triple $(V,Q,S)$. We are
interested in the polar variety $\polar(e,\dalgo,V)$, for an index
$\dalgo$ in $\{1,\dots,d\}$. Locally, in the chart
$\psi_i=(m_i,\h_i)$, this polar variety can be defined by the
cancellation of {\em all} minors in the Jacobian matrix $\jac(\h_i,
e+\dalgo)$, but all these minors give us too many polynomials for them
to define a chart for $\polar(e,\dalgo,V)$. To resolve this issue, we
localize further, using in a critical manner the so-called {\em
  exchange lemma} of~\cite[Lemma~4]{BaGiHeMb01}.  This idea is best
seen on an example.

\begin{example}\label{example:Lagexchange}
  We will use the following less straightforward example several
  times.  Take $n=6$, $c=2$ and $\f=(f_1, f_2)$, with
  $$f_1=X_1^2 + X_4^2 + X_5^2 - 1$$
  and 
  $$f_2=X_2X_3 + X_1X_6 + X_3X_5 - 1.$$ We take $e=0$, so
  $Q=\{\bullet\}$; one then easily checks that the algebraic set $V$
  defined by $f_1=f_2=0$ is smooth and has dimension $d=4$ in $\Co^6$;
  the polynomials $(m=1,\f)$ form a chart, and actually an atlas, of
  $(V,Q,S)$, with $S=\emptyset$. This example was chosen to have
  rather simple defining equations, while displaying the ``generic''
  behavior.

  Choose $\dalgo=\lfloor \frac{d+3}{2}\rfloor=3$, as we will do in our
  main algorithm; the corresponding truncated Jacobian matrix for the
  two polynomials $(f_1,f_2)$ is
  $$
  \jac(\f,3)=
\left [\begin{matrix}
    2X_{4} & 2X_{5} & 0\\ 
    0&X_3&X_1
  \end{matrix}\right ]. 
$$
The set of all $\x$ in $V$ where $\jac_\x(\f,3)$ has rank less
than two is defined by $(f_1,f_2)$, together with three minors:
$$2 X_1 X_5,\quad 2X_1 X_4,\quad 2 X_3 X_4.$$ While none of these
equations can be omitted in this definition, in the open set
$\Open(X_1)$ defined by $X_1 \ne 0$, only two of them suffice, namely
$2 X_1 X_5$ and $2X_1 X_4$. Factoring out the monomial $X_1$, we see
that in $\Open(X_1)$, the polar variety $W(0,3,V)$ is defined by the
equations $(f_1,f_2,X_4,X_5)$.

The polynomial $X_1$ was chosen as a non-zero 1-minor of $\jac(\f,3)$.
The other such minors are (up to a constant)
$X_3, X_4, X_5$. One can verify that the open sets $\Open(X_1)$,
$\Open(X_3)$, $\Open(X_4)$ and $\Open(X_5)$ cover the polar variety
$W(0,3,V)$, and that in each of these open sets, we can define
$W(0,3,V)$ using only $f_1,f_2$ and two further equations.
\end{example}

The following definition generalizes the construction in the example
above, starting from a $(c-1)$-minor of $\jac(\h,\dalgo)$. 

\begin{definition}\label{def:mH}
  For $\h=(h_1,\dots,h_c)$ in $\C[X_1,\dots,X_n]$, for any integers
  ${\dalgo}$ in $\{1,\dots, n-c\}$, and any
  $(c-1)$-minor $\pminor$ of $\jac(\h,{\dalgo})$, we denote by
  $\sfH(\h,{\dalgo},\pminor)$ the vector of $c$-minors of
  $\jac(\h,{\dalgo})$ obtained by successively adding the missing row
  and the missing columns of $\jac(\h,{\dalgo})$ to $\pminor$.  There
  are $n-c-{\dalgo}+1$ such minors.
\end{definition}

We can then state the basic construction of charts for polar
varieties, which will be immediately followed by the corresponding
construction for atlases. In addition to the choice of a $(c-1)$-minor
of the truncated Jacobian matrix of $\jac(\h,\dalgo)$, the
construction involves the choice of a $c$-minor of $\jac(\h)$ as well
(as the non-vanishing of such a minor allows us to guarantee that
$\jac(\h)$ has full rank). Taking into account arbitrary values of
$e$, and not only $e=0$ as in the example, we arrive at the following
definition.

\begin{definition}\label{sec:chart:notation:polar}
  Let $Q \subset \C^e$ be a finite set and let $V \subset \C^n$ and
  $S\subset \C^n$ be algebraic sets lying over $Q$. Let $\psi=(m,\h)$
  be a chart of $(V,Q,S)$ and let ${\dalgo}$ be an integer in
  $\{1,\dots,d\}$.
  Suppose that $\h=(h_1,\dots,h_c)$. For every $c$-minor $m'$ of
  $\jac(\h,e)$ and every $(c-1)$-minor $m''$ of $\jac(\h,e+{\dalgo})$,
  we define $\chartpolar(\psi,m',m'')$ as the polynomials
  $$\chartpolar(\psi,m',m'')=(m m' m'',\ (\h, \sfH(\h,e+{\dalgo},m''))).$$
\end{definition}

Once we have made explicit the construction of charts, the
construction of the whole atlas follows readily.

\begin{definition}\label{sec:atlas:notation:polar}
  Let $Q \subset \C^e$ be a finite set and let $V \subset \C^n$ and
  $S\subset \C^n$ be algebraic sets lying over $Q$.  Suppose that $V$
  is $d$-equidimensional, let $\bpsi=(\psi_i)_{1\le i\le s}$ be an
  atlas of $(V, Q,S)$ and let ${\dalgo}$ be an integer in
  $\{1,\dots,d\}$. Write  $W=\polar(e,\dalgo,V)$ and 
  for $i$ in $\{1,\dots,s\}$, write $\psi_i=(m_i,\h_i)$. 

  We define $\atlaspolar(\bpsi,V,Q,S,{\dalgo})$ as the sequence of all
  those $\chartpolar(\psi_i,m',m'')$, for $i$ in $\{1,\dots,s\}$ and
  for $m',m''$ respectively a $c$-minor of $\jac(\h_i,e)$ and a
  $(c-1)$-minor $\jac(\h_i,e+\dalgo)$, for which $\Open(m_i m'm'')
  \cap W-S$ is not empty.
\end{definition}

The following result is important in several aspects: it establishes
dimension properties of polar varieties, and does so in a constructive
manner, by relating the atlas of $V$ to that of the polar
variety. This proposition is proved in Section~\ref{sec:proof3.6} of
the electronic appendix.

\begin{proposition}\label{prop:ch4}\label{PROP:CH4}
  Let $Q \subset \C^e$ be a finite set and let $V \subset \C^n$ and
  $S\subset \C^n$ be algebraic sets lying over $Q$, with $S$ finite.
  Suppose that $V$ is equidimensional of dimension~$d$.
  Let $\bpsi$ be an atlas of $(V,Q,S)$, and let ${\dalgo}$ be an
  integer in $\{1,\dots,d\}$. If $2 \le {\dalgo}\leq (d+3)/2$, there
  exists a non-empty Zariski open subset
  $\scrGpolar(\bpsi,V,Q,S,{\dalgo})$ of $\GL(n,e)$ such that for $\mA$
  in $\scrGpolar(\bpsi,V,Q,S,{\dalgo})$, the following holds:
  \begin{itemize}
  \item either $\polar(e,{\dalgo},V^\mA)$ is empty, or
    \smallskip
  \item $\atlaspolar(\bpsi^\mA,V^\mA,Q,S^\mA,{\dalgo})$ is an atlas of
    $(\polar(e,{\dalgo},V^\mA),Q,S^\mA)$, and
    $\polar(e,{\dalgo},V^\mA)$ is equidimensional of dimension
    $\dalgo-1$, with $\sing(\polar(e,{\dalgo},V^\mA))$ contained in the finite set $S^\mA$.
  \end{itemize}
\end{proposition}
The bound $(d+3)/2$ for $\dalgo$ is sharp: for
higher values of $\dalgo$, polar varieties develop high-dimensional
singularities~\cite{BGHSS}.

For $e=0$, these claims were previously established by Bank, Giusti
{\em et al.}~\cite{BaGiHePa05,BGHSS} in the particular case where $V$
is smooth and a complete intersection.  Without these properties, the
proof becomes more involved, but in the end relies on a local version
of those in the above references, working locally using the charts
defined by $\bpsi$.
Let us also point out here the results in~\cite{BaGiHeLeMaSo15}, that
deal with other situations: using arguments in the same vein as
the above references, that paper proves in particular
equidimensionality of polar varieties, in generic coordinates, when we
work over a smooth quasi-affine algebraic set.

The value $\dalgo=1$ is excluded from the above proposition,
essentially because the proof for that case would require a slight
change in the arguments we use. We now show that a stronger statement
actually holds.

Our algorithm will compute the set $\Kpolar(e,1,W)$, with
$W=\polar(e,\dalgo,V^\mA)$ and $\dalgo$ as the proposition above, and
will require this set to be finite. Even if we had stated the previous
proposition with $\dalgo=1$, we would not be able to apply it to $W$,
since $\Kpolar(e,1,W)=\Kpolar(e,1,\polar(e,\dalgo,V^\mA))$ is in
general different from $\Kpolar(e,1,\polar(e,\dalgo,V)^\mA)$. However,
this finiteness result holds as well; for a proof of the following
proposition, see Section~\ref{chap:finitenessproperties} of the
electronic appendix.

\begin{proposition}\label{prop:ch5}\label{PROP:CH5}
  Let $Q \subset \C^e$ be a finite set and let $V \subset \C^n$ be an
  algebraic set lying over~$Q$.  Suppose that $V$ is equidimensional
  of dimension $d$, with finitely many singular points, and let
  $\dalgo$ be an integer such that $2 \le \dalgo \le (d+3)/2$. 

Then,
  there exists a non-empty Zariski open set $\ZOfinite(V,Q,\dalgo)
  \subset \GL(n,e)$ such that, for $\mA$ in $\ZOfinite(V,Q,\dalgo)$,
  writing $W=\polar(e,\dalgo,V^\mA)$, either $W$ is empty, or $W$ is
  equidimensional of dimension $\dalgo-1$, with finitely many singular
  points, and $\Kpolar(e, 1, W)$ is finite.
\end{proposition}
As claimed above, this implies in particular that $\polar(e,1,V^\mA)$
is finite, as one can prove that $\polar(e,1,V^\mA)$ is a subset of
$\Kpolar(e,1, W)$ (this is proved as
Lemma~\ref{prelim:lemma:polarpolar} in the electronic appendix).

The proposition above was proved in~\cite{SaSc11} in the case where
$V$ is a hypersurface, that is, defined by a single equation. In
general, the basic idea of the proof remains the same (study a
suitable incidence variety and relate the choices of $\mA$ that do not
satisfy our constraint to this incidence variety), but the proof
requires significant adaptations, as polar varieties cannot be
described as simply as in the hypersurface case.

%%%%%%%%%%%%%%%%%%%%%%%%%%%%%%%%%%%%%%%%%%%%%%%%%%%%%%%%%%%%

\subsection{Fibers of a projection}\label{ssec:abstract:fibers}

In our algorithm, $V \subset \C^n$ is an algebraic set lying over a
finite set $Q \subset \C^e$, equidimensional of dimension $d$ and with
finitely many singular points. The following result shows that if we
are in generic coordinates, these properties carry over to fibers of
the projection $\pi_{e+\dalgo-1}$.

Precisely, starting from an atlas for $(V,Q,S)$, with $Q$ in $\C^e$,
and given a finite set $\fiber2\subset \C^{e+{\dalgo}-1}$ lying over
$Q$, we show how to get an atlas of $(\Vfiber,\fiber2,\fibersing2)$,
where $\Vfiber$ is the fiber $\fbr(V,\fiber2)$, for a suitable choice
of $\fibersing2$ (the notation used below is the one we will use in
the algorithm). The construction is straightforward: we mainly replace
$Q$ by the new set $\fiber2$ and remove some useless charts from the
collection. The only subtle point lies in the definition of the set
$\fibersing2$: we take $\fibersing2=\fbr(S \cup \polar(e,{\dalgo},V),
\fiber2)$, as this set can be proved to contain all singularities of
the fiber $\Vfiber$.

\begin{definition}\label{sec:atlas:notation:fiber}
  Let $Q \subset \C^e$ be a finite set and let $V \subset \C^n$ and
  $S\subset \C^n$ be algebraic sets lying over $Q$. Suppose that $V$
  is $d$-equidimensional, let $\bpsi=(\psi_i)_{1\le i\le s}$ be an
  atlas of $(V, Q,S)$ and let ${\dalgo}$ be an integer in
  $\{1,\dots,d\}$.

  For $i$ in $\{1,\dots,s\}$, write $\psi_i=(m_i,\h_i)$. Given a
  finite set $\fiber2\subset \C^{e+{\dalgo}-1}$ lying over $Q$, we
  define $\atlasfiber(\bpsi,V,Q,S,\fiber2)$ as the sequence of all
  $\psi_i$ for which $\Open(m_i) \cap \Vfiber-\fibersing2$ is not
  empty, with $\Vfiber=\fbr(V, \fiber2)$ and
  $\fibersing2=\fbr(S \cup \polar(e,{\dalgo},V), \fiber2)$.
\end{definition}

The following statement is a counterpart of Proposition~\ref{prop:ch4}
in the context of fibers. For a proof of this statement, see
Section~\ref{sec:proof3.9} of the electronic appendix.

\begin{proposition}\label{sec:atlas:prop:summary1}\label{SEC:ATLAS:PROP:SUMMARY1}
  Let $Q \subset \C^e$ be a finite set and let $V \subset \C^n$ and
  $S\subset \C^n$ be algebraic sets lying over $Q$, with $S$ finite.
  Suppose that $V$ is equidimensional of dimension~$d$.
  Let $\bpsi$ be an atlas of $(V,Q,S)$, and let ${\dalgo}$ be an
  integer in $\{1,\dots,d\}$.  If $2 \le {\dalgo}\leq (d+3)/2$, there
  exists a non-empty Zariski open subset
  $\scrGfiber(\bpsi,V,Q,S,{\dalgo})$ of $\GL(n,e)$ such that for $\mA$
  in $\scrGfiber(\bpsi,V,Q,S,{\dalgo})$, the following holds.

  Define $W=\polar(e,{\dalgo},V^\mA)$ and let $\fiber2 \subset \C^{e+{\dalgo}-1}$ be a finite set lying
  over $Q$; define $\Vfiber=\fbr(V^\mA,\fiber2)$.  Let further
  $\fibersing2 = \fbr(S^\mA \cup \polar(e,{\dalgo},V^\mA),\fiber2)$.
  Then:
  \begin{itemize}
  \item $\fibersing2$ is finite,
\smallskip
\item either $\Vfiber$ is empty or
  $\atlasfiber(\bpsi^\mA,V^\mA,Q,S^\mA,\fiber2)$ is an atlas of
  $(\Vfiber,\fiber2,\fibersing2)$, and $\Vfiber$ is equidimensional of
  dimension $d-({\dalgo}-1)$, with $\sing(\Vfiber)$ contained in the
  finite set~$\fibersing2$.
  \end{itemize}
\end{proposition}
The dimension claim is natural: imposing that $\Vfiber$ lies over a
finite subset $\fiber2$ of $\C^{e+{\dalgo}-1}$, we expect to reduce
the number of degrees of freedom by ${{\dalgo}-1}$.

Similar statements were proved for instance in~\cite{SaSc03} in the
case $e=0$, for $V$ a complete intersection; the proof of the
proposition above reduces to this situation by working locally on $V$,
using the charts provided by the atlas $\bpsi$.

%%%%%%%%%%%%%%%%%%%%%%%%%%%%%%%%%%%%%%%%%%%%%%%%%%%%%%%%%%%%
%%%%%%%%%%%%%%%%%%%%%%%%%%%%%%%%%%%%%%%%%%%%%%%%%%%%%%%%%%%%
%%%%%%%%%%%%%%%%%%%%%%%%%%%%%%%%%%%%%%%%%%%%%%%%%%%%%%%%%%%%

\section{A family of algorithms}\label{ssec:abstractalgo}

In this section, we describe in a high-level manner a family of
algorithms to compute roadmaps, that are inspired by Canny's original
design. While all geometric constructions are specified, we do not
discuss data representation yet. Correctness, and in particular the
dimension equalities written as comments in the pseudo-code, are
subject to genericity properties; the main contribution of this
section is to make these requirements entirely explicit.

%%%%%%%%%%%%%%%%%%%%%%%%%%%%%%%%%%%%%%%%%%%%%%%%%%%%%%%%%%%%

\subsection{Description}\label{ssec:descriptionabstract}

The family of algorithms described hereafter is based on a
connectivity result which is the combination of Theorem~14 and
Proposition~2 in~\cite{SaSc11}; roughly speaking, this result says
that if we are in generic coordinates, to compute a roadmap of an
algebraic set $V$, it is enough to compute the union of {\it (i)} a
roadmap of a well-chosen polar variety of $V$ and {\it (ii)} a roadmap
of fibers of a corresponding projection. 

In the resulting algorithm, we take as input an integer $e \le n$, an
algebraic set $V \subset \C^n$ that lies over a finite set
$Q \subset \C^e$, and a finite set $C$ of control points. We make the
following assumptions:
\begin{itemize}
\item $V$ is $d$-equidimensional, for some $d > 0$,
\smallskip
\item $V$ has  finitely many singular points,
\smallskip
\item $V \cap \R^n$ is bounded. 
\end{itemize}
As output, we return a roadmap of $(V,C)$. The algorithm is recursive,
the top-level call being with $e=0$ and thus $Q=\{\bullet\} \subset \C^0$.

When $V$ is a curve, we simply return $V$. Else, we first choose a
random change of variables $\mA$ and an index $\dalgo$ denoted by
$\dalgo={\sf Choose}(d)$. The choice of $\dalgo$ is 
the subject of Subsection~\ref{ssec:discussion}; our only constraints are that $\dalgo$ is in
$\{2,\dots, \lfloor (d+3)/2\rfloor\}$ (the lower bounds ensures that 
the corresponding polar variety has dimension at least one; the upper bound
allows us to apply the results of the previous section).

After applying $\mA$, we determine a finite set of points in
$\C^{\dalgo-1}$ written $\fiber2$ in the pseudo-code; explicitly, they
are obtained as a projection of $\Kpolar(e,1,W)\, \cup\,C^\mA$, with
$W = \polar(e,\dalgo,V^\mA)$. We recursively compute roadmaps of the
polar variety $W$ and of the fiber $V''=\fbr(V,\fiber2)$, updating the
control points, and we return the union of these roadmaps.

In the recursive call for the polar variety, the index $e$ does not
change; when we deal with $V''$, we increase the value of $e$ to
$e+\dalgo-1$.

The following pseudo-code describes this recursive algorithm. The
dimension statements on the right border are the expected dimensions
of the corresponding objects; genericity conditions on the change of
coordinates $\mA$ will ensure that these claims are indeed valid
(except when said objects turn out to be empty).

\medskip\noindent \hypertarget{abstractalgorecinit}{${\sf RoadmapRec}$}$(V,Q,C,d,e)$ 
   \hfill $d=\dim(V)$
\begin{enumerate}
\item \label{step:h:empty} if $V$ is empty, return $V$
\item \label{step:h:checkdim} if $d=1$, return $V$
\item \label{step:h:chvar} let $\mA$ be a random change of variables in $\GL(n, e, \QQ)$
\item \label{step:h:choosedalgo} let $\dalgo={\sf Choose}(d)$ \hfill $\dalgo \ge 2$
\item \label{step:h:W} let $W = \polar(e,\dalgo,V^\mA)$ \hfill $\dim(W)=\dalgo-1$
\item \label{step:h:C} let $B = \Kpolar(e,1,W)\, \cup\,C^\mA$ \hfill $\dim(B)= 0$
\item \label{step:h:CP} let $\fiber2 = \pi_{e+\dalgo-1}(B)$ \hfill $\dim(\fiber2)= 0$
\item \label{step:h:PQ} let $C' = C^\mA\ \cup\ \fbr ( W,\fiber2)$ \hfill new control points; $\dim(C')= 0$
\item \label{step:h:PP} let $C'' = \fbr(C', \fiber2)$ \hfill new control points; $\dim(C'')= 0$
\item \label{step:h:VQ} let $V''=\fbr(V^\mA, \fiber2)$  \hfill $\dim(V'')=\dim(V)-(\dalgo-1)$
\item\label{step:h:rec2} let ${R}'={\sf RoadmapRec}(W, Q, C',\dalgo-1,e)$ %% same $Q$
\item\label{step:h:rec1} let ${R}''={\sf RoadmapRec}(V'',\fiber2,C'',d-(\dalgo-1),e+\dalgo-1)$ %% $\fiber2$
\item\label{step:h:final} return ${{R}'}^{\mA^{-1}} \cup {{R}''}^{\mA^{-1}}$
\end{enumerate}

The main algorithm performs an initial call to ${\sf RoadmapRec}$ with
$V$ satisfying the same assumptions as above, $e=0$,
$Q=\{\bullet\} \subset \C^0$, and $C_0$ an arbitrary finite set of control
points. We add $\sing(V)$ to $C_0$ at the top-level call, resulting in
the following main algorithm.

\medskip\noindent ${\sf MainRoadmap}(V,\,C_0)$
\begin{enumerate}
\item return ${\sf RoadmapRec}(V,\{\bullet\},C_0 \cup \sing(V), \dim(V),0)$
\end{enumerate}

%%%%%%%%%%%%%%%%%%%%%%%%%%%%%%%%%%%%%%%%%%%%%%%%%%%%%%%%%%%%

\subsection{Correctness}\label{ssec:correctness}

The nature of Algorithm ${\sf RoadmapRec}$ implies that the recursive
calls can be organized into a binary tree $\scrT$, whose structure
depends only on the dimension $d$ of the top-level input $V$ and our
choice function ${\sf Choose}$. Describing this tree explicitly will
be useful for the proof of the theorem below.

Given a positive integer $d$, the tree $\scrT$ is defined as
follows. Each node $\nodetau$ is labelled with a pair $(d_\nodetau,e_\nodetau)$ of
integers:
\begin{itemize}
\item the root $\rho$ of $\scrT$ is labelled with
  $(d_\rho,e_\rho)=(d,0)$.
\smallskip
\item a node $\nodetau$ is a leaf if and only if $d_\nodetau= 1$. Otherwise,
  it has two children $\nodetau'$ (on the left) and $\nodetau''$ (on the
  right).  Define $\dalgo_\nodetau={\sf Choose}(d_\nodetau)$. Then, $\nodetau'$
  and $\nodetau''$ have respective labels $(d_{\nodetau'},e_{\nodetau'})$ and
  $(d_{\nodetau''},e_{\nodetau''})$, with
  $$d_{\nodetau'}=\dalgo_\nodetau-1,\ e_{\nodetau'}=e_\nodetau
  \quad\text{and}\quad d_{\nodetau''}= d_\nodetau-(\dalgo_\nodetau-1),\
  e_{\nodetau''}=e_\nodetau+\dalgo_\nodetau-1.$$ 
\end{itemize}
In other words, $(d_\nodetau, e_\nodetau)$ are the last two arguments given to
${\sf RoadmapRec}$ at the recursive call considered at node $\nodetau$,
so that the recursive calls of the main algorithm correspond to the nodes 
of $\scrT$. The total number of nodes in $\scrT$ is $2d-1$.

The following theorem proves correctness of Algorithm ${\sf
  MainRoadmap}$ using this formalism. In the statement of the theorem,
we mention in particular {\em internal nodes} of $\scrT$; these are
the nodes that are not leaves, and they correspond to recursive calls
where the dimension is greater than one. We also refer to {\em proper
  ancestors} of a node $\nodetau$: they consist of the parent of
$\nodetau$, the parent of its parent, \dots, all the way to the root.

\begin{theorem}  \label{THEO:MAINABSTRACT}
  Assume that $V$ is a $d$-equidimensional algebraic set with finitely
  many singular points and that $V\cap \R^n$ is bounded.  Let
  $C_0\subset \C^n$ be a finite set of points and let
  $(\mA_\nodetau)_{\nodetau \text{~internal node of~} \scrT}$ be a
  family of matrices, with $\mA$ in $\GL(n, e_\nodetau, \QQ)$ for all
  $\nodetau$.

There
  exists a family of non-empty Zariski open sets
  $(\scrOpen_\nodetau)_{\nodetau\text{~internal node of~} \scrT}$, where for all $\nodetau$,
  $\scrOpen_\nodetau$ is in $\GL(n, e_\nodetau)$ and depends on the
  matrices $(\mA_{\nodeoherletter})_{\nodeoherletter \text{~proper ancestor of~}
    \nodetau}$, such that the following holds: if, for all internal nodes $\nodetau$ of $\scrT$,
   $\mA_\nodetau$ is in
  $\scrOpen_\nodetau$ and if it is used as the change of variables in 
the corresponding recursive call of ${\sf RoadmapRec}$, ${\sf
    MainRoadmap}(V, C_0)$ returns a roadmap of $(V, C_0)$.
\end{theorem}

This theorem is proved in Section~\ref{sec:proof4.1} of the electronic
appendix. Here, we discuss briefly the ingredients involved in the
proof.

Consider the algebraic set $V$ given as top-level input to ${\sf
  MainRoadmap}$, together with an atlas $\bpsi$ of
$(V,\{\bullet\},\sing(V))$. First, we show that the algorithm runs its
course.  To each node $\nodetau$ of $\scrT$, we associate the geometric
objects $V_\nodetau, Q_\nodetau, C_\nodetau$ that are given as input in the
corresponding recursive call, as well as all objects defined there,
such as the curve $R_\nodetau$ (if $\dim(V_\nodetau)=1$), and $W_\nodetau, B_\nodetau, Q''_\nodetau,
C'_\nodetau, C''_\nodetau, \ldots$ otherwise, together with an atlas
$\bpsi_\nodetau$ of $V_\nodetau$.

This is done in a recursive manner. Assuming we have reached a node
$\nodetau$, we define the Zariski open set $\scrOpen_\nodetau$ as the
intersection of those sets obtained by applying
Propositions~\ref{prop:ch4},~\ref{prop:ch5}
and~\ref{sec:atlas:prop:summary1} to $V_\nodetau, Q_\nodetau,
S_\nodetau$ and the atlas $\bpsi_\nodetau$. This allows us to ensure
that the dimension claims on the right border of the description of
the algorithm are valid (unless the corresponding object is empty) and
to define atlases for the children of $\nodetau$, so that we can
continue the construction.

Correctness itself then follows from connectivity results proved
in~\cite{SaSc11}. Propositions~\ref{prop:ch4} and~\ref{prop:ch5} imply
that at each node $\nodetau$, $V_\nodetau$ satisfies the assumptions of
Theorem~14 in~\cite{SaSc11}; this result establishes that $W_\nodetau\cup
V_\nodetau''$ has a non-empty and connected intersection with all
connected components of $V_\nodetau\cap \R^n$. Knowing this, Proposition 2
in that same reference then shows that given roadmaps $R_\nodetau'$ and
$R_\nodetau''$ for $(W_\nodetau, C_\nodetau')$ and $(V_\nodetau'', C_\nodetau'')$, for
$C_\nodetau'$ and $C_\nodetau''$ as defined in Steps~\ref{step:h:PQ}
and~\ref{step:h:PP}, ${{R_\nodetau}'} \cup {{R_\nodetau}''}$ is a roadmap of
$(V_\nodetau^{\mA_\nodetau},C_\nodetau^{\mA_\nodetau})$. Restoring the initial coordinates 
proves our claim.

%%%%%%%%%%%%%%%%%%%%%%%%%%%%%%%%%%%%%%%%%%%%%%%%%%%%%%%%%%%%

\subsection{Discussion}\label{ssec:discussion}

Let us now suggest what kind of complexity one should expect in an
idealized model. As we will see, the function ${\sf Choose}$ which
selects the integer $\dalgo$ is the key factor to determine the
efficiency of the algorithm.

Assume that the input $V$ is described by polynomials of degree $D$ in
$n$ variables; the B\'ezout bound~\cite{Heintz83} implies that it has
degree at most $D^n$; initially, the set $Q$ is empty, and we may assume 
for simplicity that the set $C$ of control points has cardinality $2$. 

If we suppose that we enter ${\sf RoadmapRec}$ with $V$ of degree at
most $\delta$ and $Q$ and $C$ of cardinality at most $\delta$, a
reasonable rule of thumb is that the polar variety $W$ (used in one
recursive call) and the set $V''$ (used in the other recursive call)
will have degree at most $\delta D^n$, and that the same would hold in
terms of cardinality for the new points $Q$ and $C$.  Under the
further assumption that all computations at a given recursive call can be done
in time polynomial in $\delta D^n$ , we deduce that the overall
running time is polynomial in $\delta D^{n \depth}$, where $\depth$ is the
depth of the recursion.

Canny's algorithm corresponds to defining $\dalgo={\sf Choose}(d)=2$
at every step, so that $\depth$ is at most $d=\dim(V)\leq n-1$. For this choice, one can
implement all required operations within the complexity estimates
claimed above without much difficulty, since all polar varieties we
consider are curves (so there is no further recursion on their side);
this leads to a cost polynomial in $D^{nd}$.

Decreasing the depth $\depth$ means increasing $\dalgo$, so that we have
to deal with higher-dimensional polar varieties; this in turn raises
the question of how to efficiently represent them. In the baby-steps /
giant-steps algorithm of~\cite{SaSc11}, we assume that $V$ is defined
by a single polynomial, and we let
$\dalgo ={\sf Choose}(d)\simeq \sqrt{n}$. In that case, the polar
variety has dimension close to $\sqrt{n}$, and we use Canny's
algorithm to process it, since polar varieties of hypersurfaces can be
described easily.

One expects to do better by choosing
$\dalgo={\sf Choose}(d)=\lfloor (d+3)/2 \rfloor \simeq \dim(V)/2$,
yielding a genuine divide-and-conquer algorithm, with a recursion
depth of $\log_2(d)$. We illustrate this in the next subsection.

However, in the context of such divide-and-conquer algorithms, given
algebraic sets $V,Q$ passed as input to ${\sf RoadmapRec}$, it does
not seem manageable from the complexity viewpoint to use generators of
the defining ideal of $V$ to define $\polar(e,\dalgo,V^\mA)$: we
already mentioned that polar varieties can be defined by the
cancellation of minors of a Jacobian matrix, but that there are too
many of them for us to control the complexity in a reasonable
manner. Our solution will be to represent $V$ in $\C^n$ as the Zariski
closure of the projection of an open subset of an algebraic set lying
in a higher-dimensional space.

In Section~\ref{sec:GLS}, we introduce this main technical
contribution, the use of a data structure that we call {\em
  generalized Lagrange systems}, for which we can describe all objects
arising throughout the algorithm and perform all required operations
in a cost matching the rough description above.

%%%%%%%%%%%%%%%%%%%%%%%%%%%%%%%%%%%%%%%%%%%%%%%%%%%%%%%%%%%%

\subsection{Examples}\label{ex:abstract}

For an algebraic set $V$ of dimension $d=2$ in $3$-dimensional space,
there is only one possible behavior for the algorithm, which is to
choose $\dalgo=2$; in this case, we recover Canny's algorithm. The
polar variety $W$ and the fiber $V''$ are then both curves, so there
is no need to work further in the recursive
calls. Figure~\ref{fig:algo1} illustrates this process on the familiar
example of a torus (see also~\cite{Lavalle,BaPoRo06}). The main
features of the algorithm appear on this example: because they are
critical loci, polar varieties intersect each connected component of
$V\cap\Re^n$, but the intersection may not be connected; taking fibers
allows us to re-establish connectivity.

\begin{figure}[]
\centerline{
\includegraphics[width=6.5cm]{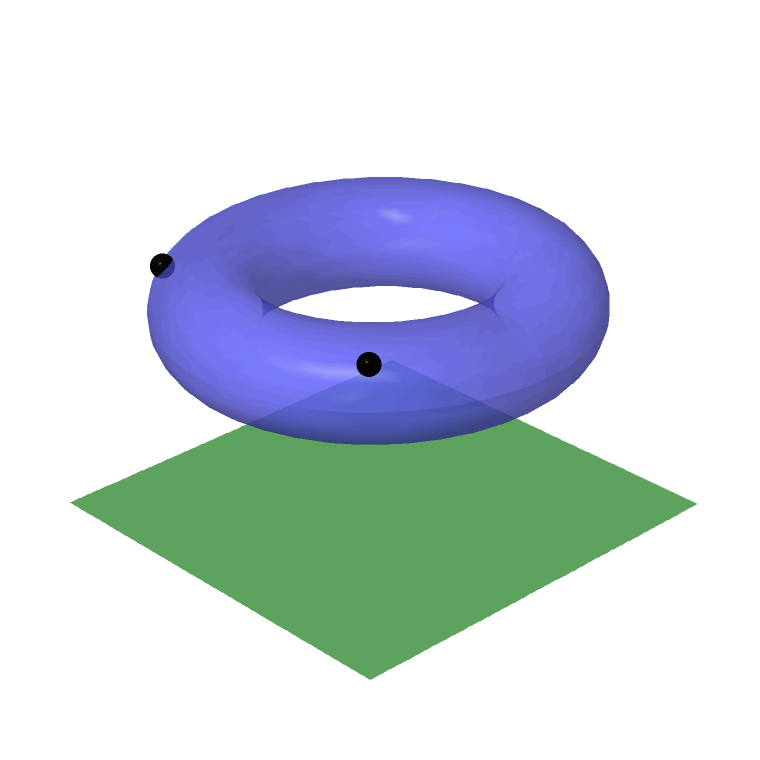}}
\centerline{\includegraphics[width=6.5cm]{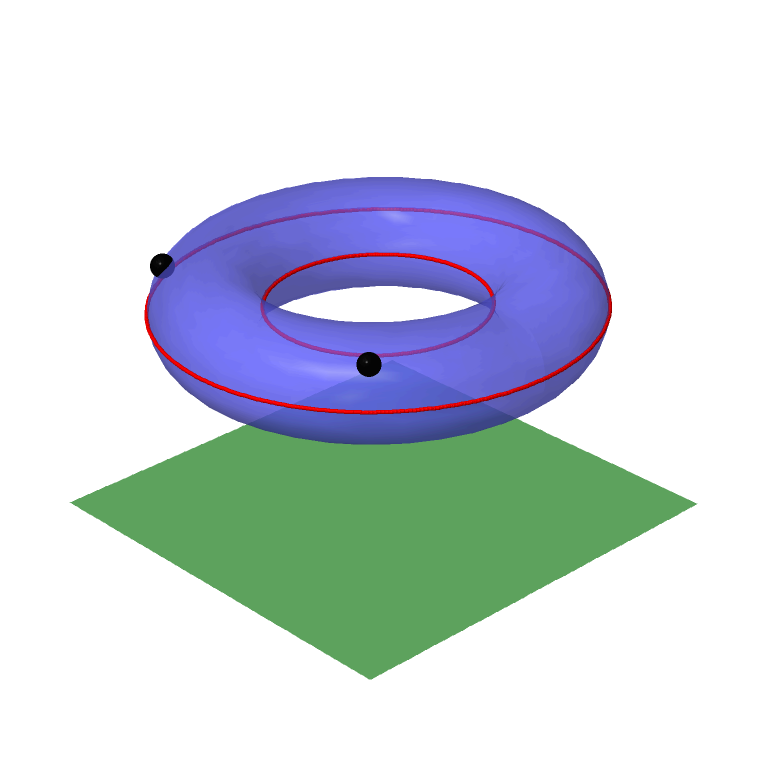}}
\centerline{\includegraphics[width=6.5cm]{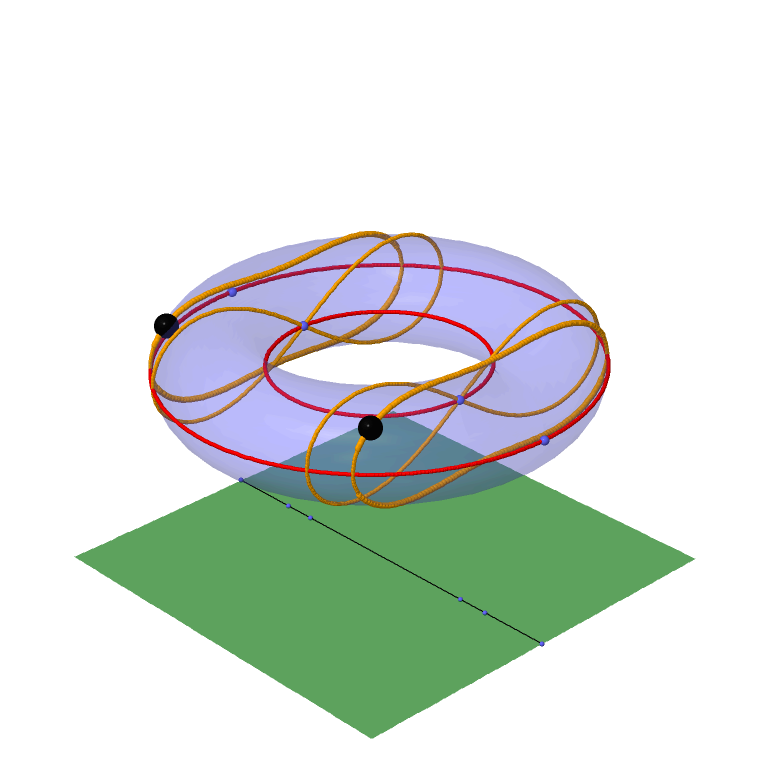}
}
\caption{The torus $V$ with $2$ control points (top), with its polar variety $\polar(0,2,V)$ (middle) and the whole roadmap (bottom)}
\label{fig:algo1}
\end{figure}

As mentioned in the previous subsection, we will be interested in the
divide-and-conquer approach where one takes
$\dalgo=\lfloor\frac{d+3}{2}\rfloor$ at every step.  In order to
illustrate the difference between this and Canny's original design, we
consider the algebraic set $V\subset \Co^6$ defined by the polynomials
$\f=(f_1, f_2)$ introduced in Example~\ref{example:Lagexchange}. The
algebraic set $V$ is smooth, equidimensional of dimension $4$ and
$V \cap \Re^6$ is compact. We take $C_0=\emptyset$; thus, on input
$(V, C_0)$, ${\sf MainRoadmap}$ simply performs a call to
${\sf RoadmapRec}$ with input $V$, $Q=\{\bullet\}$, $C=\emptyset$,
$d=4$ (we are in dimension $4$) and $e=0$ (we have fixed the value of
no variable).

Below, we describe the behaviour of ${\sf RoadmapRec}$ with the
function ${\sf Choose}(d)=\lfloor\frac{d+3}{2}\rfloor$, assuming that
all changes of variables satisfy the assumptions of Theorem~\ref{THEO:MAINABSTRACT}.
  \begin{description}
  \item[{\bf Steps~\ref{step:h:empty}--\ref{step:h:choosedalgo}}]
    We choose a matrix
    $\mA\in \GL(6, 0, \Q)$ and
    we take $\dalgo=\lfloor (4+3)/2 \rfloor=3$. 
\smallskip
  \item[{\bf Step~\ref{step:h:W}}] We compute a representation of the
    polar variety $W=\polar(0, 3, V^\mA)$.  By
    Proposition~\ref{prop:ch4}, if $W$ is not empty, it is equidimensional of dimension $2$.
\smallskip
  \item[{\bf Steps~\ref{step:h:C}--\ref{step:h:PP}}]
    Propositions~\ref{prop:ch5} and~\ref{sec:atlas:prop:summary1}
    imply that the sets $B, Q'', C', C''$ considered at these steps
    are finite, with $Q''\subset \Co^2$.
\smallskip
  \item[{\bf Step~\ref{step:h:rec2}}] We do a recursive call to ${\sf
    RoadmapRec}$ with input $W$, $Q=\{\bullet\}$ $C'$, $d=2$ (we are
    in dimension $2$) and $e=0$ (we have not fixed the value of any
    coordinate).

\smallskip

    A new matrix $\mA'\in \GL(6,0,\Q)$ is chosen at
    Step~\ref{step:h:chvar}, and we set $\dalgo=\lfloor (2+3)/2
    \rfloor=2$.  The finite sets computed at
    Steps~\ref{step:h:C}--\ref{step:h:PP} are denoted by $B_1, Q''_1,
    C'_1, C''_1$ and we have $Q''_1\subset \Co$.

\smallskip

    \begin{itemize}
    \item Proposition~\ref{prop:ch4} implies that the algebraic set
      $R'_1=\polar(0, 2, W^{\mA'})$ considered at Step~\ref{step:h:W}
      has dimension $1$ or is empty; it is returned by the recursive
      call of Step~\ref{step:h:rec2}.

\smallskip

    \item Proposition~\ref{sec:atlas:prop:summary1} implies that the
      algebraic set $R'_2=\fbr(W^{\mA'}, Q''_1)$ considered at
      Step~\ref{step:h:VQ} has dimension $1$ or is empty; it is
      returned by the recursive call of Step~\ref{step:h:rec1}.
    \end{itemize}
\smallskip
  \item[{\bf Step~\ref{step:h:VQ}}] We compute a representation of the
    fiber $V''=\fbr(V^\mA, Q'')$.
    Proposition~\ref{sec:atlas:prop:summary1} implies that $V''$ is
    either empty or equidimensional of dimension~$2$.
\smallskip
  \item[{\bf Step~\ref{step:h:rec1}}] We do a recursive call to ${\sf
    RoadmapRec}$ with input $V''$, $Q''$, $C''$, $d=2$ (we are in dimension $2$)
    and $e=2$ (since $V''$ lies over the finite set $Q''\subset \Co^2$).

\smallskip

    Since $\dim(V'')=2$, a new matrix $\mA''\in \GL(6,2,\Q)$ is chosen
    at Step~\ref{step:h:chvar}, and we set $\dalgo=\lfloor (2+3)/2
    \rfloor=2$.  The finite sets computed at
    Steps~\ref{step:h:C}--\ref{step:h:PP} are denoted by $B_2, Q''_2,
    C'_2, C''_2$, and we have $Q''_2\subset \Co$.

\smallskip

    \begin{itemize} 
    \item Proposition~\ref{prop:ch4} implies that the algebraic set
      $R''_1=\polar(2, 2, {V''}^{\mA''})$ considered at Step~\ref{step:h:W} has
      dimension $1$ or is empty. It is
      returned by the recursive call of Step~\ref{step:h:rec2}.

\smallskip
 
    \item Proposition~\ref{sec:atlas:prop:summary1} implies that the
      algebraic set $R''_2=\fbr({V''}^{\mA''}, Q''_2)$ considered at
      Step~\ref{step:h:VQ} has dimension $1$ or is empty. It is
      returned by the recursive call of Step~\ref{step:h:rec1}.
    \end{itemize}
\smallskip
\item[{\bf Step~\ref{step:h:final}}] We take the union of the
  algebraic sets returned by the recursive calls of
  Steps~\ref{step:h:rec2} and~\ref{step:h:rec1} and undo the linear
  change of variables induced by $\mA$.
  \end{description}

  Hence, the binary tree $\scrT$ defined in Subsection~\ref{ssec:correctness} has
  the following structure. 

\begin{center}
\begin{tikzpicture}[node distance=6cm, level distance=1.5cm,
level 1/.style={sibling distance=8cm, circle},
level 2/.style={sibling distance=2cm, circle}]
%\tikzstyle{every node}=[circle,draw]
\node [arn_n] (Root) {$V$, $Q=\{\bullet\}\subset \Co^0$, $C=\emptyset$, $d=4$, $e=0$}
    child {
    node [arn_n] {$W, Q=\{\bullet\}, C', d=2, e=0$} 
    child { node[leaf] {$ d=1$}  }
    child { node[leaf] {$ d=1$} }
  }
  child {
    node [arn_n] {$V'', Q''\subset \Co^2, C'', d=2, e=2$}
    child { node[leaf] {$ d=1$} }
    child { node[leaf] {$ d=1$} }
  } ;
%\node (dim0) {$d=4$} [right of = Root];

\end{tikzpicture}
\end{center}
The depth of the recursion is only $2$, while it would
be $3$ using Canny's algorithm.

%%%%%%%%%%%%%%%%%%%%%%%%%%%%%%%%%%%%%%%%%%%%%%%%%%%%%%%%%%%%
%%%%%%%%%%%%%%%%%%%%%%%%%%%%%%%%%%%%%%%%%%%%%%%%%%%%%%%%%%%%
%%%%%%%%%%%%%%%%%%%%%%%%%%%%%%%%%%%%%%%%%%%%%%%%%%%%%%%%%%%%

\section{Generalized Lagrange systems}\label{sec:GLS}

%%%%%%%%%%%%%%%%%%%%%%%%%%%%%%%%%%%%%%%%%%%%%%%%%%%%%%%%%%%%

\subsection{Overview}

The core of our construction is the following definition. When we use
this definition, the indeterminates will be $\X=X_1, \ldots, X_n$
together with some pre-existing blocks of Lagrange multipliers. In the
definition, we write these indeterminates as $\Y$.

\begin{definition}\label{def:LS}
  Let $\h=(h_1,\dots,h_c)$ be polynomials in $\KK[\Y]$, where $\KK$ is
  a field and $\Y$ a sequence of $N$ indeterminates, let
  $(L_1,\dots,L_c)$ be new indeterminates and let $\dalgo$ be an
  integer in $\{1,\dots,N-c\}$. Then $\lag(\h,\dalgo,(L_1,\dots,L_c))$ denotes the
  entries of the vector
  $$\left [ \begin{matrix} L_{1} & \cdots & L_{c} \end{matrix}
    \right ] \cdot \jac(\h,\dalgo).$$
\end{definition}

Because our assumption on $\dalgo$ implies that $c \le N-\dalgo$, the
existence of a non-zero vector $\bell=(\ell_1,\dots,\ell_c)$ that
cancels the new equations $\lag(\h,\dalgo,(L_1,\dots,L_c))$ characterizes the set
where the $c \times (N-\dalgo)$ matrix $\jac(\h,\dalgo)$ does not have
full rank $c$; this will allow us to describe polar varieties as {\em
  projections} of zeros of such systems.  The following example
illustrates this idea.

\begin{example}\label{example:Lag}
We continue with the polynomials $\f=(f_1,f_2)$ defined in
Example~\ref{example:Lagexchange}. We let $[L_1, L_2]$ be a row
vector of two new indeterminates, and we choose again
$\dalgo=3$. Then, $\lag(\f, 3, (L_1,L_2))$ denotes the entries of the vector
$$
\left [ \begin{matrix} 
L_{1} & L_{2} 
\end{matrix} \right ] 
\cdot 
\left [\begin{matrix}
2X_{4} & 2X_{5} & 0\\ 
0 & X_{3} & X_{1} \\
\end{matrix}\right ] = 
\left [\begin{matrix}
2L_1X_4~ & ~2L_1 X_5+L_2 X_3~ & ~L_2 X_1
\end{matrix}\right ]. 
$$
If we assume that $X_1$ is non-zero, the last equation becomes
$L_2=0$, and the second and third ones give
$L_1 X_4=L_1 X_5 =0$. If we furthermore introduce a dehomogenization
equation, such as for instance $2L_1 - L_2 =1$, we obtain $L_1=1/2$, 
$L_2=0$, together with $X_4=X_5=0$.
\end{example}
In this example, we can see the main feature of such Lagrange systems:
locally, one can solve for the unknowns $L_1,\dots,L_c$. The projection of the
solution set on the $\X$-space gives us equations $X_4=X_5=0$;
together with the original polynomials $f_1, f_2$, this yields the
equations that locally define the polar variety seen in
Example~\ref{example:Lagexchange}. The following proposition shows
that this is the case in general (in this proposition, we do not
discuss yet the dehomogenization we applied above, so all equations
remain homogeneous with respect to the Lagrange multipliers).

In what follows, given a non-zero polynomial $m$ in $\KK[\Y,L_1,\dots,L_c]$, for
some sequences of indeterminates $\Y$ and $(L_1,\dots,L_c)$ and a field $\KK$,
$\KK[\Y,L_1,\dots,L_c]_m$ denotes the ring of rational functions of the form
$P/m^r$, for $P$ in $\KK[\Y,L_1,\dots,L_c]$ and $r$ in $\N$.

\begin{proposition}\label{lemma:linearsolve}
  Let all notation be as in Definition~\ref{def:LS}, let $\pminor$ be
  a $(c-1)$-minor of $\jac(\h,\dalgo)$ and let $\iota$ be the index of
  the row of $\jac(\h,{\dalgo})$ not in $\pminor$.  If $\pminor \ne
  0$, there exist $(\rho_j)_{j=1,\dots,c,j\ne \iota}$ in
  $\KK[\Y]_{\pminor}$ such that the ideal $I$ generated in
  $\KK[\Y,L_1,\dots,L_c]_{\pminor}$ by $\h$ and $\lag(\h,\dalgo,(L_1,\dots,L_c))$ is the ideal
  generated by
  $$\h,\quad  L_\iota\,\sfH(\h,{\dalgo},\pminor),\quad (L_j  -\rho_j L_\iota)_{j=1,\dots,c, j \ne \iota},$$
  with $\sfH$ as in Definition~\ref{def:mH}.
\end{proposition}
\begin{proof} Without loss of generality, we write the proof in the case
  where $\pminor$ is the upper-left minor of $\jac(\h,{\dalgo})$. In
  particular, $\iota=c$ and the minors in $\sfH(\h,{\dalgo},\pminor)$ are
  built by successively adding to $\pminor$ the last row and columns
  $c,\dots,n-{\dalgo}$ of $\jac(\h,\dalgo)$; below, we denote these minors by
  $M_1,\dots,M_{d-{\dalgo}+1}$.  Write $\jac(\h,{\dalgo})$ as the matrix
$$\jac(\h,{\dalgo}) = \left [ \begin{matrix} {\bf m}_{c-1,c-1}& {\bf
        v}_{c-1,d-\dalgo+1} \\ {\bf u}_{1,c-1} & {\bf
        w}_{1,d-\dalgo+1} \end{matrix}\right ],$$ where subscripts
  denote dimensions. Since $\pminor=\det({\bf m})$ is a unit in
  $\KK[\Y,L_1,\dots,L_c]_{\pminor}$, the ideal considered in the proposition is generated in
  $\KK[\Y,L_1,\dots,L_c]_{\pminor}$ by the entries of
$$[L_{1}~\cdots~L_{c}] \ \jac(\h,{\dalgo})\ \left [\begin{matrix}
    {\bf m}^{-1}&{\bf 0}\\{\bf 0}&{\bf 1}\end{matrix}\right] \left
  [\begin{matrix} {\bf 1}&-{\bf v}\\{\bf 0}&{\bf 1}\end{matrix}\right]
=[L_{1}~\cdots~L_{c}] \ \left [\begin{matrix} {\bf 1}&{\bf
      0}\\{\bf u}{\bf m}^{-1}&{\bf w}-{\bf u}{\bf m}^{-1}{\bf
      v}\end{matrix}\right].
$$
The first $c-1$ entries are of the form $L_{j} + [{\bf u}{\bf m}^{-1}]_j
L_{c},$ so they are as prescribed, and the latter are checked to
be $M_1 L_{c}/\pminor,\dots,M_{d-{\dalgo}+1}L_{c}/\pminor$, by computing
minors of both sides the equality.
\end{proof}

The construction presented so far would be sufficient if only one
polar variety was needed. However, our abstract algorithm computes
polar varieties of polar varieties~\dots; as a result, we will have to
introduce several blocks of Lagrange multipliers.
Our starting point will be the $n$-dimensional space, endowed with
variables $\X=X_1,\dots,X_{n}$. To construct polar varieties in an
iterated manner, our blocks of Lagrange multipliers will be written
$\L_1,\dots,\L_k$, where each block $\L_i$ has the form
$\L_i=L_{i,1},\dots,L_{i,n_i}$, for some integers $n_1,\dots,n_k$. The
systems thus obtained will be called {\em generalized Lagrange
  systems}.

The purpose of this section is to give the precise definitions of
these objects and describe their main properties. Of particular
importance will be the notion of {\em normal form}, which expresses
the fact that one can solve for the Lagrange multipliers, as we did 
above in the case of a single block of multipliers.

%%%%%%%%%%%%%%%%%%%%%%%%%%%%%%%%%%%%%%%%%%%%%%%%%%%%%%%%%%%%

\subsection{Definition of generalized Lagrange systems}\label{ssec:def:GLS}

The definition of generalized Lagrange systems is simple: it involves
straight-line programs and zero-dimensional parametrizations as
defined in Subsection~\ref{ssec:data}.

\begin{definition}\label{def:GLS}
A {\em generalized Lagrange system} is a triple $L=(\Gamma,\scrQ,\scrS)$,
where
\begin{itemize}
\item $\Gamma$ is a straight-line program evaluating a sequence $\F$
  of polynomials in $\QQ[\X,\L]$ of the form
  $\F=(\f,\f_1,\dots,\f_k)$, with $\L=(\L_1,\dots,\L_k)$ and where
\smallskip
  \begin{itemize}
  \item $\X=(X_1,\dots,X_n)$
\smallskip
  \item $\f=(f_1,\dots,f_p)$ is in $\QQ[\X]$ of cardinality $p$;
\smallskip
  \item for $i=1,\dots,k$, $\L_i=(L_{i,1},\dots,L_{i,n_i})$ is a block of $n_i$ variables;
\smallskip
  \item for $i=1,\dots,k$, $\f_i=(f_{i,1},\dots,f_{i,p_i})$ is in
    $\QQ[\X,\L_1,\dots,\L_i]$ of cardinality $p_i$ and $f_{i, j}$ has
    total degree at most $1$ in $\L_s$ for $1\leq j\leq p_i$ and $1 \le s \le i$;
  \end{itemize}
\smallskip
\item $\scrQ$ is a zero-dimensional parametrization with coefficients
  in $\QQ$, defining a finite set $Q=\Zeroes(\scrQ) \subset \C^e$;
\smallskip
\item $\scrS$ is a zero-dimensional parametrization with coefficients in
  $\QQ$, defining a finite set $S=\Zeroes(\scrS) \subset \C^n$ lying over
  $Q$;
\smallskip
\item for $i=0,\dots,k$, $(n + n_1 + \cdots+n_i) - (p+p_1 + \cdots+p_i) \ge e$.
\end{itemize}
We will also write $\F=(F_1,\dots,F_P)$ for the whole set of
equations, and let $N$ be the total number of variables, so that
$$N=n+n_1+\cdots+n_k \quad\text{and}\quad P=p+p_1+\cdots+p_k.$$
Finally, we will write $d=N-e-P$, so that by the last item above 
we have $d \ge 0$.
\end{definition}

We also attach to a generalized Lagrange system a combinatorial
information, its {\em type}, which allows us to easily derive
complexity estimates.

\begin{definition}\label{sec:lagrange:notation:typedegree}
  Let $L=(\Gamma,\scrQ,\scrS)$ be a generalized Lagrange system. Its {\em
    type} is the 4-uple $T=(k,\n,\p,e)$, where $k$,
  $\n=(n,n_1,\dots,n_k)$, $\p=(p,p_1,\dots,p_k)$ and $e$ are as in
  Definition~\ref{def:GLS}.
\end{definition}

In geometric terms, we will consider the set of zeros of $\F$ that lie
over $Q$ and avoid $S$, and most importantly the projection of this
set on the $\X$-space. In all that follows, this particular projection
will be denoted by $\pi_\X: \C^N \to \C^n$; the canonical projection
$(x_1, \ldots, x_n)\mapsto (x_1, \ldots, x_e)$ is still denoted by
$\pi_e$.
\begin{definition}\label{def:CZPC}
  Let $L=(\Gamma,\scrQ,\scrS)$ be a generalized Lagrange system, let
  $\F$ in $\QQ[\X, \L]$ be the sequence evaluated by $\Gamma$, and let
  $Q,S$ and $N$ be as in Definition~\ref{def:GLS}. We define the following
  objects:
\begin{itemize}
\item  $\Cons(L)=\fbr(V(\F),Q)-\pi_\X^{-1}(S)$; this is thus the set of all $(\x,\bell)$ in $\C^N$ that cancel $\F$,
  such that $\pi_e(\x)$ belongs to $Q$ and $\x$ is not in $S$;
\smallskip
\item ${\Proj}(L)=\pi_\X({\Cons(L)}) \subset \C^{n}$;
\smallskip
\item $\Clos{(L)}\subset \C^{n}$ is the Zariski closure of $\Proj(L)$.
\end{itemize}
Since $\Clos{(L)}$ is the object we will be most interested in, we
  will say that $L$ {\em defines} $\Clos{(L)}$.
\end{definition}

A few remarks are in order. First, note that the integer $d$ in
Definition~\ref{def:GLS} is the dimension one would expect for
$\Cons(L)$, if for instance the equations $\F$ define a reduced regular
sequence. Second, while we have
$\Proj(L) \subset \Clos{(L)}-S$, the inclusion may be strict,
as the following example shows (with $S=\emptyset$).

\begin{example}\label{example:Lag2}
  We illustrate these notions with the polynomials $\f=(f_1, f_2)$ of
  Examples~\ref{example:Lagexchange} and~\ref{example:Lag}; the only
  mild difference with the previous example is that Lagrange
  multipliers will now be denoted by $\L_1=[L_{1,1},L_{1,2}]$ instead
  of $[L_1,L_2]$. In this example, and its extensions below, we
  denote by $\Gamma$ any given straight-line program that
  evaluates~$\f$.

  Since $V(\f)$ is smooth, $L=(\Gamma, \scrQ, \scrS)$ is a generalized
  Lagrange system that defines $V(\f)$, where the zero-dimensional
  parametrizations $\scrQ=(\,)$ and $\scrS=(1)$ respectively define
  $\{\bullet\} \subset \Co^0$ and the empty set. There is nothing else
  to say for $L$, since there are actually no Lagrange multipliers in it.

  We saw that $\lag(\f,3,\L_1)$ is the sequence of polynomials
  $$2L_{1,1} X_{4}, \quad 2L_{1,1} X_{5} + L_{1,2} X_{3}, \quad
  L_{1,2} X_{1}.$$ Consider then the linear form
  $\ell=2L_{1,1}-L_{1,2}-1$ already used in Example~\ref{example:Lag};
  from this, we can derive a straight line program $\Gamma'$ that
  evaluates $(\f, \lag(\f, 3, \L_1), \ell)$. The triple $L'=(\Gamma',
  \scrQ, \scrS)$ is then a generalized Lagrange system of type $T=(1,
  (6,2), (2,4), 0)$.

  Example~\ref{example:Lag} implies that in the open set $\Open(X_1)$
  defined by $X_1 \ne 0$, $\Proj(L')$ is defined by
  $f_1=f_2=X_4=X_5=0$, so that $\Proj(L')$ coincides locally with the
  polar variety $\polar(0, 3, V)$. Globally, a calculation shows that
  the set $\Proj(L')$ consists of the polar variety $\polar(0, 3, V)$,
  minus the lines $(0,2,1,0,-1,u)_{u \in \Co}$ and $(0,-2,-1,0,1,u)_{u
    \in \Co}$.  The Zariski closure of $\Proj(L')$ is exactly $\polar(0,
  3, V)$.
\end{example}

%%%%%%%%%%%%%%%%%%%%%%%%%%%%%%%%%%%%%%%%%%%%%%%%%%%%%%%%%%%%

\subsection{Definition of local and global normal form properties}\label{subsection:normalform}

We now introduce some properties, called {\em local} and {\em global
  normal forms}, which will be satisfied by the generalized Lagrange
systems that we consider to compute roadmaps. Given a generalized
Lagrange system $L=(\Gamma, \scrQ, \scrS)$ that defines
$V=\Clos{(L)}$, these properties will in particular allow us to define
charts and atlases related to $V$, establish dimension and smoothness
properties, and assert correctness of our algorithms.

We start with a definition of systems where the variables $\L$
are ``solved'' in terms of the variables $\X$. In all that follows, we
still write $\L=(\L_1,\dots,\L_k)$, with
$\L_i=(L_{i,1},\dots,L_{i,n_i})$ and $N=n+n_1+\cdots+n_k$. 

\begin{definition}\label{chap:lagrange:def:normalform}
  Let $M$ be non-zero in $\QQ[\X]$ and consider polynomials $\H$ in
  $\QQ[\X,\L]_M$, with $\X$ and $\L=(\L_1,\dots,\L_k)$ as above. We
  say that $\H$ is in {\em normal form} in $\QQ[\X,\L]_M$ if these
  polynomials have the form
  $$\H=\left ( h_1,\dots,h_c,
  (L_{1,j}-\rho_{1,j})_{j=1,\dots,n_1},\dots,
  (L_{k,j}-\rho_{k,j})_{j=1,\dots,n_k}\right ),$$ where all $h_i$ are
  in $\QQ[\X]$ and all $\rho_{\ell,j}$ are in~$\QQ[\X]_M$.  We call
  $\h=(h_1,\dots,h_c)$ and $\brho=(L_{i,j}-\rho_{i,j})_{1 \le i \le
    k, 1 \le j \le n_i}$ respectively the {\em $\X$-component} and the
  {\em $\L$-component} of $\H$.
\end{definition}
Remark that in this case, the total number of polynomials in $\H$ is
$c+N-n$.

\medskip

We can now define {\em local normal forms} for generalized Lagrange
systems; the existence of such local normal forms expresses the fact
that we can locally solve for the variables $\L$ over $V=\Clos{(L)}$,
while having a convenient local description of $V$.

\begin{definition}\label{chap:lagrange:def:nf}
  Let $L=(\Gamma,\scrQ,\scrS)$ be a generalized Lagrange system that defines
  a set $V$, and let 
all notation $Q,S,\dots$ be as in Definition~\ref{def:GLS}. A
  {\em local normal form} for $L$ is the data of
  $\phi=(\polmu,\poldelta,\h,\H)$ that satisfies the following conditions:
\begin{enumerate}
\item[${\lnf_1}$.] $\polmu$ and $\poldelta$ are in $\QQ[\X]-\{0\}$ and $\H$ is in
  normal form in $\QQ[\X,\L]_{\polmu \poldelta}$, with $\X$-component $\h=(h_1,\dots,h_c)$;
\smallskip
\item[${\lnf_2}.$] $|\H| = |\F|$, or equivalently $n-c=N-P$;
\smallskip
\item[${\lnf_3}.$] $\langle \F,I \rangle = \langle \H,I \rangle $ in
  $\QQ[\X,\L]_{\polmu \poldelta}$, where $I \subset \QQ[\X]$ is the defining
  ideal of $Q$;
\smallskip
\item[${\lnf_4}.$] $(\polmu,\h)$ is a chart of $(V,Q,S)$;
\smallskip
\item[${\lnf_5}.$] $\poldelta$ does not vanish on $\Open(\polmu) \cap \Proj(L)$.
\end{enumerate}
\end{definition}
The idea behind this definition is that the systems $\F$ and $\H$
define the same solutions $(\x,\bell)$, at least for those $\x$ that
lie above $Q$ and do not cancel $\polmu \poldelta$ (this is
${\lnf_3}$).  We ask that $\polmu$ defines the open set corresponding
to a chart of $V$ (this is
${\lnf_4}$), but we need more: expressing the variables $\L$ in
terms of $\X$ necessarily introduces a denominator, which is the
polynomial $\poldelta$; we authorize that it may vanish somewhere on
$V$, but not on $\Open(\polmu) \cap \Proj(L)$; this is ${\lnf_5}$.
Given a local normal form $\phi$ as above, we will call $\psi$ the
chart {\em associated} with $\phi$.

\begin{example}
  Continuing with the same example as above, $$\phi_1=\Big
  (X_1,\ 1,\ (f_1,f_2,X_4,X_5),\ (f_1,f_2,X_4,X_5,L_{1,1}-\frac
  12,L_{1,2}) \Big)$$ is a local normal form for $L'$, corresponding to
  the open set $\Open(X_1)$, giving us the chart
$(X_1,\ (f_1,f_2,X_4,X_5))$
 of $(W(0,3,V),
  \{\bullet\}, \emptyset)$; here, we have $\poldelta=1$ since solving
  for $L_{1,1}$ and $L_{1,2}$ introduces no further denominator.
  Corresponding to the open set $\Open(X_3)$, a calculation gives the local
  normal form
$$\phi_2 =
(X_3,\ X_3+X_5,\ (f_1,f_2,X_4,X_1X_5),\ (f_1,f_2,X_4,X_1X_5, L_{1,1}-
\frac 12 \frac{X_3}{X_3+X_5}, L_{1,2} + \frac{X_5}{X_3+X_5}),$$
with in particular the chart $(X_3, \ (f_1,f_2,X_4,X_1X_5))$ of
$(W(0,3,V), \{\bullet\}, \emptyset)$. Here, we have
$\poldelta=X_3+X_5$, which is the denominator we introduce in order to
solve the linear equations for $L_{1,1}$ and $L_{1,2}$ over
$\Open(X_3)$.  The locus where this denominator vanishes on
$\Open(X_3)\cap \polar(0,3,V)$ is precisely the two lines mentioned in
Example~\ref{example:Lag2}.
\end{example}

We can finally introduce the global version of the previous
property. Starting from a family of local normal forms
$\phi_i=(\polmu_i,\poldelta_i,\h_i,\H_i)$, we will cover $V-S$ using
the open sets $\Open(\polmu_i)$, in effect obtaining an atlas of
$(V,Q,S)$. However, we may not be able to ensure that the smaller open
sets $\Open(\polmu_i\poldelta_i)$ cover $V-S$ as well (since
$\Proj(L)$ may be smaller than $V-S$, as in the previous
example). Instead, given an ``interesting'' irreducible set $Y$
contained in $V$, but not in $S$, we add the condition that as soon as
some $\polmu_i$ does not vanish identically on $Y$, $\polmu_i
\poldelta_i$ itself does not vanish identically on $Y$, so we can make
sense of the corresponding description by the polynomials $\H_i$
almost everywhere on $Y$.  For instance, if $Y=\{\y\}$ is a single
point, and $\polmu_i$ does not vanish at $\y$, $\poldelta_i \polmu_i$
does not vanish there either.

Taking into account several such $Y$'s, not necessarily irreducible,
we are led to the following definition.

\begin{definition}\label{chap:lagrange:def:gnf}
  Let $L=(\Gamma,\scrQ,\scrS)$ be a generalized Lagrange system that defines
  a set $V$, and let 
all notation  $Q,S,\dots$ be as in Definition~\ref{def:GLS}. A {\em
    global normal form} of $L$ is the data of $\bphi=(\phi_i)_{1 \le i
    \le s}$ such that:
  \begin{enumerate}
  \item[${\gnf_1},$] each $\phi_i$ has the form
    $\phi_i=(\polmu_i,\poldelta_i,\h_i,\H_i)$ and is a local normal form of
    $L$;
\smallskip
  \item[${\gnf_2}.$] $\bpsi=(\polmu_i,\h_i)_{1 \le i \le s}$ is an atlas 
    of $(V,Q,S)$.
  \end{enumerate}
  Let further $\scrY=(Y_1,\dots,Y_r)$ be algebraic sets in $\C^n$.  A
  {\em global normal form} of $(L;\scrY)$ is the data of
  $\bphi=(\phi_i)_{1 \le i \le s}$ such that ${\gnf_1}$ and ${\gnf_2}$
  hold, and such that we also have, for $i$ in $\{1,\dots,s\}$ and $j$
  in $\{1,\dots,r\}$:
  \begin{enumerate}
  \item[${\gnf_3}.$] for any irreducible component $Y$ of $Y_j$
    contained in $V$ and such that $\Open(\polmu_i)\cap Y -S$ is not
    empty, $\Open(\polmu_i \poldelta_i)\cap Y -S$ is not empty.
  \end{enumerate}
\end{definition}
We say that $L$, resp.\ $(L,\scrY)$, has the {\em global normal form
  property} when there exists $\bphi$ as above satisfying $({\gnf_1},
{\gnf_2})$, resp.  $({\gnf_1}, {\gnf_2}, {\gnf_3})$.  Given a
global normal form $\bphi$ as above, we will call $\bpsi$ the atlas
{\em associated} with $\bphi$.

\begin{example}\label{example:globalnormalform}
  We already saw two charts for $L'$ above, built in the open sets
  $\Open(X_1)$ and $\Open(X_3)$, where $X_1$ and $X_3$ are two
  non-zero 1-minors of the truncated Jacobian of the original system
  $\f$. Two more minors can be considered, namely $X_4$ and $X_5$. The
  Lagrange system we consider has no solution over $\Open(X_4) \cap V$,
  so we need not consider $X_4$. For $X_5$, we obtain the local normal 
  form
  $$\phi_3 =
  (X_5,\ X_3+X_5,\ (f_1,f_2,X_1,X_3X_4),\ (f_1,f_2,X_1,X_3X_4, L_{1,1}-
  \frac 12 \frac{X_3}{X_3+X_5}, L_{1,2} + \frac{X_5}{X_3+X_5}).$$ One then
  checks that $\bphi=(\phi_1,\phi_2,\phi_3)$ is a global normal form
  for $L'$.
\end{example}

When $L$ possesses the global normal form property, one can deduce
several useful properties on the sets $\Cons(L)$ and $\Proj(L)$. For
instance, the following proposition is proved in
Section~\ref{chap:lagrange} of the electronic appendix.

\begin{proposition}\label{sec:lagrange:propertyP}\label{SEC:LAGRANGE:PROPERTYP}
  Let $L=(\Gamma,\scrQ,\scrS)$ be a generalized Lagrange system and
  let $\F=(F_1,\dots,F_P)$ in $\QQ[\X, \L]$ and $e \ge 0$ be as in
  Definition \ref{def:GLS}. If $L$ has the global normal form
  property, the following holds:
  \begin{itemize}
  \item the Jacobian matrix $\jac(\F,e)$ has full rank $P$ at every
    point $(\x,\bell)$ in $\Cons(L)$;
  \smallskip
  \item the restriction $\pi_\X: \Cons(L) \to \Proj(L)$ is a bijection.
\end{itemize}
\end{proposition}

It is important to note that just as for charts and atlases, while
local and global normal forms are a useful tool to establish
properties such as the ones above, the algorithms will not explicitly
compute any global normal form.

%%%%%%%%%%%%%%%%%%%%%%%%%%%%%%%%%%%%%%%%%%%%%%%%%%%%%%%%%%%%

\subsection{Initialization and changes of variables}\label{ssec:initchgvar}

The simplest generalized Lagrange systems involve no Lagrange
multipliers at all: they essentially consist in a straight-line
program $\Gamma$ that computes a reduced re\-gular sequence
$\f=(f_1,\dots,f_p)$ in $\QQ[X_1,\dots,X_n]$, together with a
zero-dimensional parametrization of the singular locus of $V(\f)$.
Here, we take $e=0$ and thus $Q=\{\bullet\}$; in this case, recall
that we make the convention that the empty sequence $(\,)$ is seen as
a zero-dimensional parametrization of such a $Q$.

Because there is no canonical choice for a zero-dimensional
parametrization of the singular locus, we will take it as input.
When $V(\f)$ is smooth, so that $\sing(V(\f))$ is empty, we 
represent it using  the sequence $(1)$.

% \begin{definition}\label{def:initialize:lagrange}
%   Let $\Gamma$ be a straight-line program that evaluates polynomials
%   $\f=(f_1,\dots,f_p)$ in $\QQ[\X]$ that define a reduced regular
%   sequence and such that $\sing(V(\f))$ is finite, and let $\scrS$ be
%   a zero-dimensional parametrization of $\sing(V(\f))$.
%   We denote by ${\rm Init}(\Gamma,\scrS)$ the triple
%   $(\Gamma,(\,),\scrS)$.
% \end{definition}

\begin{proposition}\label{sec:lagrange:prop:initnormalform}
  Let $\Gamma$ be a straight-line program that evaluates polynomials
  $\f=(f_1,\dots,f_p)$ in $\QQ[\X]$ that define a reduced regular
  sequence and such that $\sing(V(\f))$ is finite, and let $\scrS$ be
  a zero-dimensional parametrization of $\sing(V(\f))$.

  If $p < n$, then $L=(\Gamma,(\,),\scrS)$ is a generalized Lagrange
  system of type $(0,(n),(p),0)$ such that $\Clos{(L)}=V(\f)$. If
  $\scrY=(Y_1,\dots,Y_r)$ are algebraic sets contained in $\C^n$, then
  $(L;\scrY)$ has the global normal form property, with
  $\bphi=((1,1,\f,\f))$ as a global normal form.
\end{proposition}
The proof is an immediate consequence of the definitions.

\medskip

Our abstract algorithm in Section~\ref{ssec:abstractalgo} 
uses several changes of variables. In all 
cases, they are chosen in $\GL(n,e, \QQ)$, for some integers $e \le n$. 
Suppose then that $L=(\Gamma,\scrQ,\scrS)$ is a generalized Lagrange
system of type $(k,\n,\p,e)$, and recall that $\Gamma$ is a
straight-line program which evaluates a sequence of polynomials $\F$
in $\QQ[\X, \L]$ as in Definition~\ref{def:GLS}. For $\mA$ in $\GL(n,
e)$, we define $L^\mA$ as $L^\mA=(\Gamma^\mA, \scrQ, \scrS^\mA)$,
where $\Gamma^\mA$ is obtained from $\Gamma$ by applying the change of
variable $\Gamma$ to the $\X$-variables $X_1,\dots,X_n$ only; it
computes polynomials $\F^\mA$.  It is immediate that $L^\mA$ is a
generalized Lagrange system, of the same type as $L$.  Note also the
following straightforward equalities:
$$ \Proj(L^\mA)=\Proj(L)^\mA \quad \text{and} \quad
\Clos{(L^\mA)}=\Clos{(L)}^\mA.$$ We can apply the same construction to
systems in normal form. Given a local normal form
$\phi=(\polmu,\poldelta,\h,\H)$ of $\L$, we define $\phi^\mA$ in the natural
manner, as the 4-uple $(\polmu^\mA,\poldelta^\mA,\h^\mA,\H^\mA)$. Here as
well, for the last entry, we let $\mA$ act on the $\X$ variables of
the polynomials $\H$; thus, if $\H$ has the form
$$\H=\left ( h_1,\dots,h_c,
(L_{1,j}-\rho_{1,j})_{j=1,\dots,n_1},\dots,
(L_{k,j}-\rho_{k,j})_{j=1,\dots,n_k}\right ),$$ 
then $\H^\mA$ is 
$$\H=\left ( h^\mA_1,\dots,h^\mA_c,
(L_{1,j}-\rho^\mA_{1,j})_{j=1,\dots,n_1},\dots,
(L_{k,j}-\rho^\mA_{k,j})_{j=1,\dots,n_k}\right ).$$ 
Naturally, $\phi^\mA$ is a local normal form of $L^\mA$.

Finally, if $\bphi=(\phi_i)_{1 \le i \le s}$ is a global normal form
of $L$, resp.\ of $(L,(Y_1,\dots,Y_r))$, then
$\bphi^\mA=(\phi^\mA_i)_{1 \le i \le s}$ is a global normal form of
$L^\mA$, resp.\ of $(L^\mA,(Y^\mA_1,\dots,Y^\mA_r))$.

%%%%%%%%%%%%%%%%%%%%%%%%%%%%%%%%%%%%%%%%%%%%%%%%%%%%%%%%%%%%

\subsection{Generalized Lagrange systems and polar varieties}\label{ssec:abstract:GLS:polar}

Starting from a generalized Lagrange system $L$ that defines an
algebraic set $V=\Clos{(L)}$, we are now interested in constructing
generalized Lagrange systems for polar varieties of $V$. The following
definition associates to $L$ a new generalized Lagrange system
$\Polarlag(L,\u,\dalgo)$, where $\dalgo$ will denote the index of the
polar variety we consider, and $\u$ is a vector of constants.  This
definition generalizes the process described in
Example~\ref{example:Lag}.

\begin{definition}\label{sec:lagrange:notation:polar}
  Let $L=(\Gamma,\scrQ,\scrS)$ be a generalized Lagrange system of
  type $(k,\n,\p,e)$, with $\n=(n,n_1,\dots,n_k)$,
  $\p=(p,p_1,\dots,p_k)$ and $\L=\L_1,\dots,\L_k$, and let $\F \subset
  \QQ[\X,\L]$ be the polynomials computed by $\Gamma$. Let
  $N=n+n_1+\cdots+n_k$, $P=p+p_1+\cdots+p_k$, and let $\dalgo$ be an
  integer in $\{1,\dots,N-e-P\}$.

  Let $\L_{k+1}=L_{k+1,1},\dots,L_{k+1,P}$ be new indeterminates. For $\u=(u_1,\dots,u_{P})$ in $\QQ^{P}$, define
  $$\F'_{\u}=\Big (\F, \  \lag(\F,e+\dalgo, \L_{k+1}), \ 
  u_{1} L_{k+1,1} + \cdots + u_{P} L_{k+1,P} - 1\Big ),$$ where
  $\lag(\F,e+\dalgo,\L_{k+1})$ denotes the entries of the vector
  $$\left [ \begin{matrix} L_{k+1,1} & \cdots & L_{k+1,P} \end{matrix}
    \right ] \cdot \jac(\F,e+\dalgo).$$ We define $\Polarlag(L,\u,\dalgo)$
  as the triple $(\Gamma'_{\u},\scrQ,\scrS)$, where $\Gamma'_\u$ 
  is a straight-line program that evaluates $\F'_\u$. 
\end{definition}
In other words, we take our input equations and we add the linear
equations involving Lagrange multipliers that describe that
$\jac(\F,e+\dalgo)$ is rank-deficient. The affine form involving the
coefficients $\u$ allows us to dehomogenize the equations involving 
the new Lagrange multipliers.

In order to make this definition unambiguous, let us explain how to
construct $\Gamma'_\u$: take the straight-line program $\Gamma$,
together with the straight-line program obtained by applying
Baur-Strassen's differentiation algorithm \cite{BaurStrassen} (to
compute the Jacobian of $\F'_\u$), and do the matrix-product vector
and the dot product in the direct manner.

\begin{lemma}\label{lemma:GLS:typeW}
  With notation as above, $\Polarlag(L,\u,\dalgo)$ is a generalized
  Lagrange system of type $(k+1,\n',\p',e)$, with
  $\n'=(n,n_1,\dots,n_k,P)$ and $\p'=(p,p_1,\dots,p_k,N-e-\dalgo+1)$.  In
  particular, the total numbers of indeterminates and equations
  involved in $\Polarlag(L,\u,\dalgo)$ are respectively
  $$N'=N+P \quad\text{and}\quad P'=N+P-e-(\dalgo-1),$$ so that
  $N'-e-P'=\dalgo-1$.
\end{lemma}
\begin{proof}
  The only point that deserves mention is that $N'-P' \ge e$, which is
  true because $N'-P'=e+(\dalgo-1)$ and $\dalgo \ge 1$.
\end{proof}

In the following proposition, we state how normal form properties are
transferred from $L$ to $\Polarlag(L,\u,\dalgo)$.  The statement of
the proposition is technical; here is what it says in essence. If $L$
defines $V \subset \C^n$ and has the global normal form property, we
expect $\Polarlag(L,\u,\dalgo)$ to possess it as well, and we expect
this generalized Lagrange system to define the polar variety
$\polar(e,\dalgo,V)$, at least in generic coordinates. However, this
may not be the case: the global normal form of $L$ involves
denominators, and if these denominators vanish identically on some
irreducible component of $\polar(e,\dalgo,V)$, we are not able to
derive a meaningful description at these points.

This proposition shows why we introduced the notion of global normal
form attached to $(L;Y_1,\dots,Y_r)$, for some algebraic sets
$Y_1,\dots,Y_r$. Indeed, we will prove that for a generic choice of
$\u$ and of a change of coordinates $\mA$, if we assume that
$(L^\mA,\polar(e,\dalgo,V^\mA))$, or equivalently
$(L,\polar(e,\dalgo,V^\mA)^{\mA^{-1}})$, have the global normal form
property, then our claim indeed holds. Since we may have to prove the
same property for further constructions of polar varieties (or fibers,
where the same issue will arise), we are led to the general kind of
statement made here, involving some extra algebraic sets $Y_i$. The
proof of the following proposition is in
Section~\ref{chapter:constructionspolar} of the electronic appendix.

\begin{proposition}\label{sec:lagrange:prop:transfer:polar}\label{SEC:LAGRANGE:PROP:TRANSFER:POLAR}
  Let $Q \subset \C^e$ be a finite set and let $V \subset \C^n$ and
  $S\subset \C^n$ be algebraic sets lying over $Q$, with $S$ finite.
  Suppose that $V$ is equidimensional of dimension $d$, with finitely
  many singular points.

  Let $\bpsi$ be an atlas of $(V,Q,S)$, let $\dalgo$ be an integer in
  $\{2,\dots,d\}$ such that $\dalgo \le (d+3)/2$, and let
  $\mA\in\GL(n,e)$ be in the open set
  $\scrGpolar(\bpsi,V,Q,S,\dalgo)$ defined in
  Proposition~\ref{prop:ch4}; write
  $W=\polar(e,\dalgo,V^\mA)$.

  Let $L=(\Gamma,\scrQ,\scrS)$ be a generalized Lagrange system such
  that $V=\Clos{(L)}$, $Q=\Zeroes(\scrQ)$ and $S=\Zeroes(\scrS)$.  Let
  $\scrY=(Y_1,\dots,Y_r)$ be algebraic sets in $\C^n$ and let finally
  $\bphi$ be a global normal form for $(L; (W^{\mA^{-1}}, \scrY))$ such
  that $\bpsi$ is the associated atlas of $(V,Q,S)$.

  There exists a non-empty Zariski open set
  $\mathscr{I}(L,\bphi,\mA,\scrY)\subset \C^{P}$ such that for all
  $\u$ in $\mathscr{I}(L,\bphi,\mA,\scrY) \cap \QQ^P$, the following holds:
  \begin{itemize}
  \item $\Polarlag(L^\mA, \u, \dalgo)$ is a generalized Lagrange system
    that defines $W$;
\smallskip
  \item If $W$ is not empty, then $(\Polarlag(L^\mA, \u, \dalgo);
    \scrY^\mA)$ admits a global normal form whose atlas is $\atlaspolar(\bpsi^\mA,V^\mA,Q,S^\mA,\dalgo)$ (Definition
    \ref{sec:atlas:notation:polar}).
 \end{itemize}
\end{proposition}

\begin{example}\label{example:Lag3}
  We illustrate Definition~\ref{sec:lagrange:notation:polar} starting
  from the polynomials $\f=(f_1,f_2)$ we have been using since
  Example~\ref{example:Lagexchange}, namely
$$\left \{\begin{array}{l}
f_1=X_1^2 + X_4^2 + X_5^2 - 1\\[2mm]
f_2=X_2X_3 + X_1X_6 + X_3X_5 - 1.
  \end{array} \right .$$
  Recall that $\Gamma$ denotes a straight-line program that evaluates
  $\f$.  Since $V(\f)$ is smooth, we saw that $L=(\Gamma, (\,), (1))$ is a
  generalized Lagrange system that defines $V(\f)$.
  
  Next, take the generalized Lagrange system $L'=(\Gamma', (\,), (1))$
  of Example~\ref{example:Lag2}, where $\Gamma'$ is a straight-line
  program that evaluates $\F=(\f, \lag(\f, 3, \L_1), \ell)$ and $\ell$
  is the linear form $\ell=2L_{1,1}-L_{1,2}-1$.  This generalized
  Lagrange system was  built according to
  Definition~\ref{sec:lagrange:notation:polar}, starting from
  polynomials $\f=(f_1,f_2)$; we saw in Example~\ref{example:Lag2}
  that $\Clos{(L')}$ is none other than $\polar(0,3, V(\f))$, which has
  dimension~$2$ (in this case, no change of variables was necessary).

  To do one more step, we now consider a $(6\times 6)$ invertible
  matrix $\mA$ with entries in $\Q$.  Taking $\L_2=[L_{2,1}, \ldots,
  L_{2, 6}]$ and a random vector $\u=(u_1,u_2,u_3,u_4,u_5,u_6)\in
  \Q^6$, we build now a new generalized Lagrange system
  $\Polarlag({L'}^\mA,2,\u) = (\Gamma'', \scrQ, \scrS)$ as in
  Definition~\ref{sec:lagrange:notation:polar}, where $\Gamma''$ is a
  straight-line program that evaluates $\F^\mA$, $\lag(\F^\mA, 2)$ and
  the linear form
$$
\ell'=u_1L_{2, 1}+u_2L_{2,2}+u_3L_{2,3}+u_4L_{2,4}+u_5L_{2,5}+u_6L_{2,6}-1. 
$$
Proposition~\ref{sec:lagrange:prop:transfer:polar} shows 
that for a generic choice of $\mA$ and $\u$,
$\Clos{(\Polarlag({L'}^\mA, 2, \u))}$ coincides with $\polar(0, 2, \polar(0, 3, V(\f))^\mA)$.

Since the type of $L'$ was $(1,(6,2),(2,4),0)$, and since we add $7$
equations and $6$ variables, the type of the new generalized Lagrange
system is $(2,(6,2,6),(2,4,7),0)$.
\end{example}

%%%%%%%%%%%%%%%%%%%%%%%%%%%%%%%%%%%%%%%%%%%%%%%%%%%%%%%%%%%%

\subsection{Generalized Lagrange systems and fibers}\label{ssec:abstract:GLS:fibers}

Suppose that $L=(\Gamma,\scrQ,\scrS)$ is a generalized Lagrange system
which defines an algebraic set $V=\Clos{(L)}\subset \C^n$; let
$Q=\Zeroes(\scrQ)$. We now build a generalized Lagrange system that
defines a fiber of the form $\fbr(V,\fiber2)$, for some $\fiber2
\subset \C^{e+\dalgo-1}$ lying over $Q$, and we study its properties
(remark that the notation $\fiber2$ or $\dalgo$ are those that were
used in our abstract algorithm).

\begin{definition}\label{sec:lagrange:notation:fiber}
  Let $L=(\Gamma,\scrQ,\scrS)$ be a generalized Lagrange system of
  type $(k,\n,\p,e)$.  Let $N=n+n_1+\cdots+n_k$ and
  $P=p+p_1+\cdots+p_k$ and let $\dalgo$ be an integer in
  $\{1,\dots,N-e-P\}$.
  
  Let $\param2$ be a zero-dimensional parametrization that encodes a
  finite set $\fiber2 \subset \C^{e+\dalgo-1}$ and let finally
  $\paramsing2$ be a zero-dimensional parametrization that encodes a
  finite set $\fibersing2 \subset \C^n$ lying over $\fiber2$.  We
  define $\Fiberlag(L,\param2,\paramsing2)$ as the triple
  $(\Gamma,\param2,\paramsing2)$.
\end{definition}
In all cases where we use this construction, $L$ will have the global
normal form property; then, the quantity $N-e-P$ that appears above is
none other than the dimension $d$ of $\Clos{(L)}$. 

\begin{lemma}\label{lemma:GLS:typeF}
  With notation as above, $\Fiberlag(L,\param2,\paramsing2)$ is a
  generalized Lagrange system of type $(k,\n,\p,e+\dalgo-1)$.  In
  particular, the total numbers of indeterminates and equations
  involved in $\Fiberlag(L,\param2,\paramsing2)$ are respectively
  $N'=N$ and $P'=P,$ so that $N'-(e+\dalgo-1)-P'=d-(\dalgo-1)$.
\end{lemma}
\begin{proof}
  The only point that deserves a verification is that
  $(n+n_1+\cdots+n_k)-(p+p_1+\cdots+p_k) \ge e+\dalgo-1$, or
  equivalently that $N-e- P\ge \dalgo-1$; this inequality actually
  holds by definition of $\dalgo$.
\end{proof}

We can finally show how global normal form properties are inherited
through this construction. The following statement is a close analogue
for fibers of the one we obtained previously for polar varieties; its
proof is in Section~\ref{chapter:constructionsfiber}.

\begin{proposition}\label{sec:lagrange:prop:transfer:fiber}\label{SEC:LAGRANGE:PROP:TRANSFER:FIBER}
  Let $Q \subset \C^e$ be a finite set and let $V \subset \C^n$ and
  $S\subset \C^n$ be algebraic sets lying over $Q$, with $S$ finite.
  Suppose that $V$ is equidimensional of dimension $d$, with finitely
  many singular points.
  
  Let $\bpsi$ be an atlas of $(V,Q,S)$, let $\dalgo$ be an integer in
  $\{2,\dots,d\}$ such that $\dalgo \le (d+3)/2$, and let $\mA\in
  \GL(n,e)$ be in the open set $\scrGfiber(\bpsi,V,Q,S,\dalgo)$
  defined in Proposition~\ref{sec:atlas:prop:summary1}; write
  $W=\polar(e,\dalgo,V^\mA)$.

  Let $\param2$ and $\paramsing2$ be zero-dimensional parametrizations
  with coefficients in $\QQ$ that respectively define a finite set
  $\fiber2 \subset \C^{e+\dalgo-1}$ lying over $Q$ and the set
  $\fibersing2 = \fbr(S^\mA \cup W,\fiber2)$, and let
  $\Vfiber=\fbr(V^\mA,\fiber2)$.

  Let $L=(\Gamma,\scrQ,\scrS)$ be a generalized Lagrange system such
  that $V=\Clos{(L)}$, $Q=\Zeroes(\scrQ)$ and $S=\Zeroes(\scrS)$. Let
  $\scrY=(Y_1,\dots,Y_r)$ be algebraic sets in $\C^n$ and let finally
  $\bphi$ be a global normal form for $(L;({\Vfiber}^{\mA^{-1}},
  \scrY))$ such that $\bpsi$ is the associated atlas of $(V,Q,S)$.
  Then the following holds:
  \begin{itemize}
  \item $\Fiberlag(L^\mA, \param2, \paramsing2)$ is a generalized
    Lagrange system which defines $\Vfiber;$
\smallskip
  \item if $\Vfiber$ is not empty,
    $(\Fiberlag(L^\mA, \param2, \paramsing2); \scrY^\mA)$ admits a
    global normal form whose atlas is
    $\atlasfiber(\bpsi^\mA,V^\mA,Q,S^\mA,\fiber2)$ (Definition
    \ref{sec:atlas:notation:fiber}).
 \end{itemize}
\end{proposition}

%%%%%%%%%%%%%%%%%%%%%%%%%%%%%%%%%%%%%%%%%%%%%%%%%%%%%%%%%%%%
%%%%%%%%%%%%%%%%%%%%%%%%%%%%%%%%%%%%%%%%%%%%%%%%%%%%%%%%%%%%
%%%%%%%%%%%%%%%%%%%%%%%%%%%%%%%%%%%%%%%%%%%%%%%%%%%%%%%%%%%%

\section{Solving generalized Lagrange systems}\label{sec:solvingglag}

We now describe the routines used in our main algorithm
for ``solving'' generalized Lagrange systems --- for instance, to
compute a one-dimensional parametrization of a set of the form
$\Clos{(L)}$, when it is known to have dimension one, or compute
critical points on this set.

These routines rely on variants of algorithms in \cite{GiLeSa01}, and
as such, their running time depends on degree bounds for the varieties
defined by the systems we have to solve (see Section~\ref{sec:prelim}
for preliminaries on degrees of algebraic sets). Ge\-neralized Lagrange
systems possess a multi-homogeneous structure which will allow us to
give strong degree bounds for these varieties. We start by stating
these bounds; they are variants of the classical one (see
e.g.~\cite{vanderWaerden29,vanderWaerden78}) adapted to our
setting. Next we state our complexity results for various
computational problems as mentioned above.

%%%%%%%%%%%%%%%%%%%%%%%%%%%%%%%%%%%%%%%%%%%%%%%%%%%%%%%%%%%%

\subsection{Degree bounds}\label{ssec:degreebounds}

Let $e$ be a non-negative integer. In this section, we consider
polynomials $\F=(F_1,\dots,F_P)$ in $\C[\X,\L_1,\dots,\L_k]$, with
$n-e,n_1,\dots,n_k$ variables in the respective blocks
$\X,\L_1,\dots,\L_k$, and having degrees in $\X,\L_1,\dots,\L_k$ respectively bounded by
$$
\begin{array}{cl}
(D_1,0, 0,\dots,0) & \text{for  $F_1,\dots,F_{p}$} \\
(D_2,1,0,\dots,0) & \text{for $F_{p+1},\dots,F_{p+p_1}$} \\
\vdots & \vdots \\
(D_2,1,1,\dots,1) & \text{for $F_{p+\cdots+p_{k-1}+1},\dots,F_{p+\cdots+p_k}$},
\end{array}
$$ so that $P=p+p_1+ \cdots +p_k$; the total number of variables is
$N-e$, with $N=n+n_1+\cdots+n_k$. We assume that all $p_i$'s and $n_i$'s
are positive (including $n$ and $p$).

The structure of these systems is essentially that of the generalized
Lagrange systems our algorithm will construct by repeating the
constructions defined in Section~\ref{sec:GLS}, except that we only
have $n-e$ variables in the first block: this accounts for the fact
that in generalized Lagrange systems, we will ensure that the first
$e$ variables can assume finitely many values (so we may essentially
see them as being constant for such degree calculations). As for
generalized Lagrange systems, we assume that the following properties
are satisfied for $0 \le i \le k$:
\begin{equation}\label{eq:N-P}
  N_i-e \ge P_i, \quad\text{with}\quad N_i = n+\cdots+n_i\quad\text{and}\quad P_i = p+\cdots+p_i.
\end{equation}
Remark in particular that $N=N_k$ and $P=P_k$.

\begin{definition}\label{sec:solvelagrange:notationsNPdelta}
  Given integers $k,e,D_1,D_2$ and sequences of integers $\n=(n,n_1,\dots,n_k)$ and
  $\p=(p,p_1,\dots,p_k)$ as above, we define $\DF(k, e,\n,\p,D_1, D_2)$, as
  $$
  \DF(k, e, \n,\p,D_1, D_2)=(P_k+1)^k D_1^{p} D_2^{n-e-p}  \prod_{i=0}^{k-1} N_{i+1}^{N_i-e-P_i}.
  $$
\end{definition}
The quantity $\DF(k, e,\n,\p,D_1, D_2)$ is derived from calculations
that are in essence intersection products in the Chow ring of the
multi-projective space $\P^{n-e} \times \P^{n_1} \times \cdots \times
\P^{n_k}$. Concretely, this means that it is an upper bound on the sum
of coefficients of a truncated product of the form
$$ (D_1 \zeta_0)^{p} (D_2 \zeta_0 + \zeta_1)^{p_1} \cdots  
(D_2 \zeta_0 + \zeta_1 + \cdots + \zeta_k)^{p_k} \bmod\langle
\zeta_0^{n-e+1},\zeta_1^{n_1+1}\dots,\zeta_k^{n_k+1}\rangle$$
in $\Z[\zeta_0,\zeta_1,\dots,\zeta_k]$. 

Let $\Delta$ be the ideal generated by all $P$-minors of
$\jac(\F)$. We consider the Zariski closure $V$ of $V(\F)-V(\Delta)$:
the irreducible components of $V$ are thus those irreducible
components of $V(\F)$ where $\jac(\F)$ has generically full rank $P$.
For $i \le P$, let $V_i$ be the Zariski closure of
$V(F_1,\dots,F_i)-V(\Delta)$; thus, $V_P=V$. Our main result in this
subsection is the following degree bound.

\begin{proposition}\label{sec:posso:prop2}\label{SEC:POSSO:PROP2}
  Suppose that all inequalities in~\eqref{eq:N-P} hold. Then, 
  for $i$ in $\{1,\dots,P\}$, $V_i$ has degree at most $\DF(k, e,\n,\p,D_1,
  D_2)$.
\end{proposition}
This proposition is proved in Section~\ref{chap:degreebounds} of
the electronic appendix of this paper. The key feature in this bound
is that even though we have many equations of degree $D_1$ or $D_2$
(later on, we will have $P\simeq n^2$ such equations), these degrees
only appear with exponent $O(n)$; the other terms in the product are
of a combinatorial nature. This is to be compared with a direct
application of B\'ezout's theorem, which would lead to bounds of the
form $D_1^{p} D_2^{P-p}$ and would be unsuitable for our purposes.

We will use this result in the following context. If
$L=(\Gamma,\scrQ,\scrS)$ is a generalized Lagrange system with the
global normal form property, Proposition~\ref{SEC:LAGRANGE:PROPERTYP}
will allow us to apply the previous proposition; it will imply that
the algebraic set $\Clos{(L)}$ has degree at most $\kappa \delta$,
with $\kappa = \deg(\scrQ)$ and $\delta=\DF(k, e,\n,\p,D, D-1)$:
indeed, there are $\kappa$ points in $\Zeroes(\scrQ)$, and we
apply the proposition above each of these points. 

%%%%%%%%%%%%%%%%%%%%%%%%%%%%%%%%%%%%%%%%%%%%%%%%%%%%%%%%%%%%

\subsection{Algorithms for generalized Lagrange systems}\label{solvinglagrange:algorithms}

Let $L=(\Gamma, \scrQ, \scrS)$ be a generalized Lagrange system of
type $(k, \n, \p, e)$, where $\Gamma$ is a straight-line program of
length $E$ that computes polynomials $\F=(\f, \f_1, \ldots, \f_k)$,
with $\f\subset \QQ[\X]$ and $\f_i\subset \QQ[\X, \L_1, \ldots, \L_i]$
for $1 \le i \le k$.  Below, the integer $D$ denotes the maximum
degree of the polynomials in $\f$; then, by Definition~\ref{def:GLS},
for $1 \le i \le k$, the maximum of the degrees in $\X$
(resp. $\L_1,\dots,\L_i$) of the polynomials in $\f_i$ is at most
$D-1$ (resp.\ $1$). We write $d=N-e-P$, still using the notation of 
Definition~\ref{def:GLS}.

The goal of this paragraph is to state complexity estimates for
routines which take as input $L$, assuming that $L$ has the global 
normal form property, and do the following: 
\begin{itemize}
\item return a one-dimensional parametrization of $\Clos{(L)}$, when
  this set has dimension  $d=1$;
\smallskip
\item return a zero-dimensional parametrization of $\polar(e,1,
  \Clos{(L)})-S$, with $S=\Zeroes(\scrS)$, assuming that this set is
  well-defined and finite;
\smallskip
\item take a zero-dimensional parametrization $\scrQ''$ as an
  additional input and return a zero-dimensional parametrization of
  the fiber
  $\fbr(\Clos{(L)}, \Zeroes(\scrQ''))$, assuming that this set is finite.
\end{itemize}
Whenever the algorithms below return paramet\-rizations, these
para\-met\-rizations will have coefficients in $\QQ$.

These algorithms are based on the geometric resolution algorithm
of~\cite{GiLeSa01,Lecerf2000} (that itself follows previous work
of~\cite{GiHeMoPa95,GiHeMoPa97,GiHeMoMoPa98}), with a slight
modification. Indeed, since the generalized Lagrange system
$L=(\Gamma,\scrQ,\scrS)$ defines an algebraic set $V=\Clos{(L)}$ lying
over $Q=\Zeroes(\scrQ)$, our algorithms need to ``solve'' equations with
coefficients in $\QQ[T]/\langle q \rangle$, where $q$ is the
squarefree polynomial appearing in $\scrQ$. If $q$ was irreducible, we
could directly apply the techniques in~\cite{GiLeSa01,Lecerf2000}, but
in general, we have to rely on {\em dynamic evaluation}
techniques~\cite{D5}. Details are given in Section~\ref{chap:posso}.

The quantity $\delta=\DF(k, e,\n,\p,D, D-1)$ introduced in
Definition~\ref{sec:solvelagrange:notationsNPdelta} will play a
crucial role in the cost analysis of our algorithms, as will the
degrees $\degQ=\deg(\scrQ)$ and $\degS = \deg(\scrS)$.  The main
feature of the geometric resolution algorithm
of~\cite{GiLeSa01,Lecerf2000}, which will be crucial for our main
result, is that its running time is {\em polynomial} in these
quantities.

We recall that our algorithms are randomized, in a sense that was
described in the introduction: failure can occur only if one of our
randomly chosen values happens to belong to some hypersurface of the
corresponding parameter space.  

We start with the routine
${\sf SolveLagrange}$ that computes a one-dimensional parametrization
of $\Clos{(L)}$ when it has dimension $d=1$; the proof is in
Section~\ref{proof:solvelagrange:prop:basicsolve} of the electronic
appendix.

\begin{proposition}\label{chap:solvelagrange:prop:basicsolve}\label{CHAP:SOLVELAGRANGE:PROP:BASICSOLVE}
  There exists a probabilistic algorithm ${\sf SolveLagrange}$ which
  takes as input a generalized Lagrange system $L$ of type $(k, \n,
  \p, e)$ such that $N-e-P=1$, and returns either a one-dimensional
  parametrization with coefficients in $\QQ$ or ${\sf fail}$ using
  $$\softO(N^3(E+N^3) (D+k) \degQ^3\delta^3 + N\degQ
  \delta\degS^2)$$ operations in~$\QQ$, using the notation introduced
  above. If either
  \begin{itemize}
  \item $\Clos{(L)}$ is empty,
\smallskip
  \item or $L$ has a global normal form,
  \end{itemize}
  then in case of success, the output of ${\sf SolveLagrange}$
  describes $\Clos{(L)}$.  In addition, $\Clos{(L)}$ has degree at most
  $\degQ \delta$.
\end{proposition}

Next, we state complexity estimates for computing $\polar(e, 1,
\Clos{(L)})-S$, with $S=\Zeroes(\scrS)$, whenever this set is
well-defined and zero-dimensional.  For a proof of the following
proposition, see Section~\ref{proof:main:critical} of the electronic
appendix.

\begin{proposition}\label{sec:main:critical}\label{SEC:MAIN:CRITICAL}
  There exists a probabilistic algorithm ${\sf W}_1$ which takes as
  input a generalized Lagrange system $L$ of type $(k, \n, \p, e)$ and
  returns either a zero-dimensional parametrization with coefficients
  in $\QQ$ or ${\sf fail}$ using
  $$\softO( (k+1)^{2d+1}  N^{4d+8}E D^{2d+1} \degQ^2 \delta^2 +N\degS^2)$$
 operations in~$\QQ$, using the notation introduced
  above. If either  $\Clos{(L)}$ is empty, or
  \begin{itemize}
    \item $\Clos{(L)}$ is $d$-equidimensional (so that $\polar(e, 1,
    \Clos{(L)})$ is well-defined),
\smallskip
  \item $\polar(e, 1, \Clos{(L)})$ is finite,
\smallskip
  \item and $(L; \polar(e, 1, \Clos{(L)}))$ has a global normal form,
  \end{itemize}
  then in case of success, the output of ${\sf W}_1$ describes
  $\polar(e,1,\Clos{(L)})-S$, with $S=\Zeroes(\scrS)$. In addition, the
  finite set $\polar(e,1,\Clos{(L)})-S$ has degree at most
  $\degQ\delta N^d (D-1+k)^d$.
\end{proposition}

Finally, we give complexity estimates for the computation of fibers.
The following proposition is proved in Section~\ref{proof:prop:fiber}
of the electronic appendix.

\begin{proposition}\label{sec:main:prop:fiber}\label{SEC:MAIN:PROP:FIBER}
  There exists a probabilistic algorithm ${\sf Fiber}$ which takes as
  input a generalized Lagrange system $L=(\Gamma,\scrQ,\scrS)$ of type
  $(k, \n, \p, e)$ and a zero-dimensional parametrization $\scrQ''$ of
  degree $\degQ''$, defining a finite set of points
  $Q''\subset \C^{e+d}$ lying over $Q=\Zeroes(\scrQ)$, and which returns
  either a zero-dimensional parametrization with coefficients in $\QQ$
  or ${\sf fail}$ using
  $$\softO( N^3(N E+N^3)D {\degQ''}^2 \delta^2 + N\degS^2 )$$
  operations in~$\QQ$, using the notation introduced
  above. If either
  \begin{itemize}
  \item $\Clos{(L)}$ is empty,
\smallskip
  \item or $\fbr(\Clos{(L)}, Q'')$ is finite and $(L; \fbr(\Clos{(L)}, Q''))$
    has a global normal form,
  \end{itemize}
  then in case of success, the output of ${\sf Fiber}$ describes
  $\fbr(\Clos{(L)}, Q'')-S$, with $S=\Zeroes(\scrS)$. In addition, $\fbr(\Clos{(L)}, Q'')-S$ has
  degree at most $\degQ''\delta$.
\end{proposition}

%%%%%%%%%%%%%%%%%%%%%%%%%%%%%%%%%%%%%%%%%%%%%%%%%%%%%%%%%%%%
%%%%%%%%%%%%%%%%%%%%%%%%%%%%%%%%%%%%%%%%%%%%%%%%%%%%%%%%%%%%
%%%%%%%%%%%%%%%%%%%%%%%%%%%%%%%%%%%%%%%%%%%%%%%%%%%%%%%%%%%%

\section{Main algorithms}\label{sec:mainalgo}

We finally describe and prove the correctness of our main algorithms;
they are the concrete version of the abstract algorithms
${\sf RoadmapRec}$ and ${\sf MainRoadmap}$ given in
Section~\ref{ssec:abstractalgo}. Whereas we had maintained some
flexibility in the choice of the parameter $\dalgo$ in these abstract
algorithms, we now choose the value $\dalgo=\lfloor (d+3)/2\rfloor$,
as we saw that it leads to a recursion tree of logarithmic depth.  

The geometric objects taken as input or constructed in the algorithms
of Section~\ref{ssec:abstractalgo} will be encoded by the generalized
Lagrange systems introduced in Section~\ref{sec:GLS} and (for finite
sets) by zero-dimensional parametrizations; the output is encoded by a
one-dimensional parametrization. 

%%%%%%%%%%%%%%%%%%%%%%%%%%%%%%%%%%%%%%%%%%%%%%%%%%%%%%%%%%%%

\subsection{Description}

We start with the description of our recursive algorithm
${\sf RoadmapRecLagrange}$, which is the concrete counterpart of
algorithm ${\sf RoadmapRec}$ of Section~\ref{ssec:abstractalgo}. It
takes as input
\begin{itemize}
\item a generalized Lagrange system $L=(\Gamma, \scrQ, \scrS)$ 
which has the global normal form property;
\smallskip
\item a zero-dimensional parametrization $\scrC$ that describes
  control points. 
\end{itemize}

In order to implement all operations, we use basic subroutines
manipulating zero-dimensional or one-dimensional parametrizations such
as ${\sf Union}$ (of zero-dimensional or one-dimensional
parametrizations), ${\sf Projection}$ (of zero-dimensional
parametrizations) and ${\sf Lift}$ (that computes
$\fbr(\Zeroes(\scrC),\Zeroes(\scrQ))$ where $\scrC$ and $\scrQ$ are
zero-dimensional parametrizations). These routines are described in
Section~\ref{chap:posso} of the electronic appendix; here, we will
simply mention that they run in time polynomial in $N$ and all
involved degrees. We also use the routines ${\sf SolveLagrange}$,
${\sf W_1}$ and ${\sf Fiber}$ which were described in the previous
section.

Some of these routines may return ${\sf fail}$; in that
case, by convention, the algorithm ${\sf RoadmapRecLagrange}$ and the
upcoming top-level algorithm ${\sf MainRoadmapLagrange}$ return
${\sf fail}$ as well. Finally, in the algorithm, for
$\mA\in \GL(n, e, \QQ)$, we use notation such as $\scrC^\mA$ for
readability; more precisely, this should be read as
${\sf ChangeVariables}(\scrC,\mA)$ where ${\sf ChangeVariables}$ is a
routine that takes as input $\scrC$ and $\mA$ and returns a
zero-dimensional parametrization that encodes $\Zeroes(\scrC)^\mA$.

\medskip\noindent ${\sf RoadmapRecLagrange}$$(L,\,\scrC)$ \hfill $L=(\Gamma,\scrQ,\scrS)$
\begin{enumerate}\itemsep0em
\item\label{rmplag:step:1} if $d=N-e-P\leq 1$, return ${\sf SolveLagrange}(L)$
\item\label{rmplag:step:2} let $\mA$ be a random change of variables
  in $\GL(n,e,\QQ)$ and $\u$ be a random vector in $\QQ^P$
\item\label{rmplag:step:3} let $\dalgo=\lfloor (d+3)/2 \rfloor$ \hfill
  $\dalgo \ge 2;\ \dalgo\simeq d/2$
\item\label{rmplag:step:4} let $L' = \Polarlag(L^\mA,\u,\dalgo)$ \hfill
  $d_{L'}=\dalgo-1\simeq d/2$
\item\label{rmplag:step:5} let $\scrB = {\sf Union}({\sf W}_1(L'), \scrC^\mA)$ \hfill $\dim(\Zeroes(\scrB))= 0$
\item\label{rmplag:step:6} let $\scrQ'' = {\sf
    Projection}(\scrB,e+\dalgo-1)$ \hfill $\dim(\Zeroes(\scrQ''))= 0$
\item\label{rmplag:step:7} let $\scrC' = {\sf Union}(\scrC^\mA, {\sf
    Fiber}(L', \scrQ''))$ \hfill new control points; $\dim(\Zeroes(\scrC'))=
  0$
\item\label{rmplag:step:9} let $\scrC'' = {\sf Lift}(\scrC',\scrQ'')$
  \hfill new control points; $\dim(\Zeroes(\scrC''))= 0$
\item\label{rmplag:step:7-a} let $\scrS' = {\sf Union}(\scrS^\mA, {\sf Fiber}(L', \scrQ''))$ \hfill  $\dim(\Zeroes(\scrS'))= 0$
\item\label{rmplag:step:9-a} let $\scrS'' = {\sf Lift}(\scrS',\scrQ'')$  \hfill  $\dim(\Zeroes(\scrS''))= 0$
\item\label{rmplag:step:8}\label{step:l:pp} let $\mathscr{R}'={\sf RoadmapRecLagrange}(L',\, \scrC')$
\item\label{rmplag:step:10} let $L''=
  \Fiberlag(L^\mA,
  \scrQ'',\scrS'')$ \hfill $d_{L''}=d-(\dalgo-1)\simeq d/2$
\item\label{rmplag:step:11} let $\mathscr{R}''={\sf RoadmapRecLagrange}(L'',\, \scrC'')$
\item\label{rmplag:step:12} return ${\sf  Union}({\mathscr{R}'}^{\mA^{-1}}, {\mathscr{R}''}^{\mA^{-1}})$
\end{enumerate}

Our main algorithm takes the following input:
\begin{itemize}
\item a straight-line program $\Gamma$ that computes a reduced regular
  sequence $\f=(f_1, \ldots, f_p)$ in $\QQ[\X]=\QQ[X_1, \ldots, X_n]$,
  such that $V(\f)$ satisfies the assumptions of our main theorem,
\smallskip
\item a zero-dimensional parametrization $\scrC$ encoding a finite set
  of points in $V$.
\end{itemize}
It starts by constructing a zero-dimensional parametrization $\scrS$
which encodes $\sing(V(\f))$ using a routine ${\sf SingularPoints}$,
then calls ${\sf RoadmapRecLagrange}$, taking as input the generalized
Lagrange system $(\Gamma,(\,),\scrS)$. The routine ${\sf
  SingularPoints}$ is described in Section~\ref{sec:posso:singularpoints} of the
electronic appendix.

\medskip\noindent
${\sf MainRoadmapLagrange}(\Gamma, \scrC_0)$
\begin{enumerate}
\item\label{step:main:1} $\scrS={\sf SingularPoints}(\Gamma)$
\smallskip
\item return ${\sf RoadmapRecLagrange}((\Gamma,(\,),\scrS), {\sf Union}(\scrC_0, \scrS))$
\end{enumerate}

%%%%%%%%%%%%%%%%%%%%%%%%%%%%%%%%%%%%%%%%%%%%%%%%%%%%%%%%%%%%

\subsection{Correctness}\label{ssec:lag:correctness}

To prove the correctness of ${\sf MainRoadmapLagrange}$ on input
$(\Gamma, \scrC_0)$, it is sufficient to prove the correctness of
${\sf RoadmapRecLagrange}$ with input $(\Gamma, (\,),\scrS)$ and ${\sf
  Union}(\scrC_0, \scrS)$.

The strategy of our proof is to establish that this algorithm computes
the same objects as ${\sf RoadmapRec}$ when taking
$\dalgo=\lfloor(d+3)/2\rfloor$.  As in
Subsection~\ref{ssec:correctness}, we consider the binary tree $\scrT$
recording the recursive calls to ${\sf RoadmapRec}$.

To each node of the tree $\scrT$, one can now associate integers
$(k_\nodetau, \n_\nodetau, \p_\nodetau, e_\nodetau)$, that will be the
type of the generalized Lagrange system $L_\nodetau$ given as input to
${\sf RoadmapRecLagrange}$ in the corresponding recursive call.  We
can then denote by $P_\nodetau$ the sum of the entries of $\p_\nodetau$
(that is, the total number of equations in $L_\nodetau$). With this
notation, our correctness statement can be formulated as follows.

\begin{proposition}\label{prop:correctnessglobal}\label{PROP:CORRECTNESSGLOBAL}
  Consider polynomials $\f=f_1,\dots,f_p$ in $\QQ[X_1,\dots,X_n]$,
  given by a straight-line program $\Gamma$, that define a reduced
  regular sequence.

  Suppose that $V=V(\f) \subset \C^n$ has finitely many singular
  points and that $V(\f)\cap \R^n$ is bounded. Consider also a
  zero-dimensional parametrization $\scrC_0$ that describes a finite
  set $C_0 \subset \C^n$.
  
  Suppose that the matrices $(\mA_\nodetau)_{\nodetau  \text{~internal node of~} \scrT}$
  satisfy the assumptions of Theorem~\ref{THEO:MAINABSTRACT}. Then,
  there exists a family of non-empty Zariski open sets
  $\scrIopen_\nodetau\subset \C^{P_\nodetau}$, for $\nodetau$ 
  an internal node of $\scrT$, such that the following holds.

  Consider vectors $(\u_\nodetau)_{\nodetau \text{~internal node of~}
    \scrT}$, with $\u_\nodetau$ in $\QQ^{P_\nodetau}$ for all
  $\nodetau$.  If, for all internal nodes $\nodetau$ of $\scrT$,
  $\u_\nodetau$ is in $\scrIopen_\nodetau$, $\mA_\nodetau$
and $\u_\nodetau$ are used in 
  the corresponding recursive call of ${\sf RoadmapRecLagrange}$, 
 and if all calls to
  subroutines such as ${\sf Union}$, ${\sf Projection}$, ${\sf W_1}$,
  ${\sf Lift}$ are successful, then ${\sf MainRoadmapLagrange}(\Gamma,
  \scrC_0)$ returns a roadmap of $(V, C_0)$.
\end{proposition}
The proof is given in Section~\ref{chap:correctness} of the electronic
appendix; we briefly discuss its main points here.

As in Subsection~\ref{ssec:correctness}, to each node $\nodetau$ of
$\scrT$ are associated the algebraic sets
$V_\nodetau,\ B_\nodetau,\dots$ that are used by our abstract
algorithm at the corresponding recursive call. In addition, we now
also have a generalized Lagrange system $L_\nodetau$, together with
zero-dimensional parametrizations $\scrB_\nodetau, \scrQ''_\nodetau$,
etc.  The gist of the proof is to establish that at each such node
$\nodetau$, $L_\nodetau$ defines $V_\nodetau$, and similarly
$\Zeroes(\scrB_\nodetau)=B_\nodetau$, etc.  

In order to prove this by induction, we rely on
Propositions~\ref{sec:lagrange:prop:transfer:polar}
and~\ref{sec:lagrange:prop:transfer:fiber}. They show the existence of
a Zariski open $\scrIopen_\nodetau\subset \C^{P_\nodetau}$ such 
that if $\u_\nodetau$ belongs to $\scrIopen_\nodetau$, then the
generalized Lagrange systems $L'_\nodetau$ and $L''_\nodetau$ defined
at Steps~\ref{rmplag:step:4} and~\ref{rmplag:step:10} respectively
define the polar variety $W_\nodetau$ and the fiber $V''_\nodetau$.

 In order to apply these propositions, we need to assume that
 $L_\nodetau$ has the global normal form property; then, we know that
 this property is transferred to the descendants $L'_\nodetau$ and
 $L''_\nodetau$. However, we pointed out while stating the two
 propositions above that we need slightly stronger assumptions: when
 for instance we build polar varieties, we actually need $(L_\nodetau,
 W_\nodetau^{\mA_\nodetau^{-1}})$ to have the global normal form
 property in order to deduce that it is still the case for
 $L'_\nodetau$.  Having in mind to apply this property recursively
 means that at the top-level, the initial generalized Lagrange system
 must have the global normal form property in conjunction with a host
 of algebraic sets, corresponding in essence to all objects built
 throughout the algorithm. This is however precisely guaranteed by
 Proposition~\ref{sec:lagrange:prop:initnormalform}.

%%%%%%%%%%%%%%%%%%%%%%%%%%%%%%%%%%%%%%%%%%%%%%%%%%%%%%%%%%%%

\subsection{Complexity analysis}\label{ssec:complexity}

This final paragraph is devoted to the complexity analysis of
Algorithm ${\sf MainRoadmapLagrange}$.  In the last 
section of the electronic appendix, we
prove the following result. Taken with Proposition~\ref{prop:correctnessglobal},
it establishes the main result stated in the introduction.

\begin{proposition}\label{prop:complexity:mainlagrange}\label{PROP:COMPLEXITY:MAINLAGRANGE}
  Consider polynomials $\f=f_1,\dots,f_p$ in $\QQ[X_1,\dots,X_n]$ of
  degrees bounded by $D$, given by a straight-line program $\Gamma$ of
  length $E$, that define a reduced regular sequence.  

  Suppose that $V=V(\f) \subset \C^n$ has finitely many singular
  points and that $V(\f)\cap \R^n$ is bounded. Consider also a
  zero-dimensional parametrization $\scrC_0$ of degree $\degC$ that describes a finite
  set $C_0 \subset \C^n$.

  Suppose that all matrices $\mA_\nodetau$ and all vectors $\u_\nodetau$
  satisfy the assumptions of Proposition~\ref{prop:correctnessglobal},
  and that all calls to subroutines such as ${\sf Union}$, ${\sf
    Projection}$, ${\sf W_1}$, ${\sf Lift}$ are successful. Then, ${\sf MainRoadmapLagrange}(\Gamma, \scrC_0)$ either returns
  ${\sf fail}$ or returns a one-dimensional parametrization of degree
  bounded by
\[
\softO \left (
\degC 16^{3d}  (n \log_2(n))^{2(2d+12\log_2(d))(\log_2(d)+6)}D^{(2n+1)(\log_2(d)+4)}
\right )\]
using 
\[
\softO \left (
\degC^3 16^{9d} E (n \log_2(n))^{6(2d+12\log_2(d))(\log_2(d)+7)}D^{3(2n+1)(\log_2(d)+5)}
\right ) 
\]
arithmetic operations in $\QQ$, with $d=n-p$.
\end{proposition}
We refer the reader to Section~\ref{chap:main:sec:complexite} of the
electronic appendix for the detailed cost analysis of this
proposition. Instead, we give here the main lines of an argument that
shows that the running time is polynomial in $\degC (nD)^{n \log(d)}$.

The divide-and-conquer nature of the algorithm implies that at all
stages, the total number of variables $N$ in the generalized Lagrange
systems we handle is $O(n^2)$; as a result, the quantity $\DF(k,
e,\n,\p,D, D-1)$ associated to any of these generalized Lagrange
systems is seen to become $(nD)^{O(n)}$.  

In geometric terms, all the inputs to our algorithms are pairs of the
form $(V,Q)$, with $V$ lying over a finite set $Q$, together with
control points $C$. Using the upper bound above,
Proposition~\ref{sec:posso:prop2} implies that the degree of the fiber
of $V$ above each point of $Q$ is $(nD)^{O(n)}$.

We also need to control the growth of the sets $Q$. Using the degree
bound in Proposition~\ref{sec:main:critical}, one can deduce that the
degree of $Q$ (as well as that of all finite sets computed in the
algorithm, and in particular the set of control points) grows by a
factor $(nD)^{O(n)}$ through each recursive call. Hence, all these
sets admit an overall degree bound of the form $\degC (nD)^{O(n
  \log(d))}$. The running time of the algorithm can be analyzed along
the same lines, once we notice that for the subroutines we use, the
running time is essentially polynomial in the input and output
degrees.

%%%%%%%%%%%%%%%%%%%%%%%%%%%%%%%%%%%%%%%%%%%%%%%%%%%%%%%%%%%%

\subsection{Example}\label{example:final}

  We illustrate the execution of ${\sf MainRoadmapLagrange}$ when
  $\Gamma$ is a straight-line program evaluating the polynomials
  $\f=(f_1, f_2)\subset \Q[X_1, \ldots, X_6]$ given in
  Example~\ref{example:Lagexchange} and $\scrC_0=(1)$ is the parametrization
  encoding the empty set (so we have no control points).

  In the example of Subsection~\ref{ex:abstract}, we showed the
  execution of the divide-and-conquer version of the abstract
  algorithm ${\sf RoadmapRec}$ on the variety $V=V(\f)$.  In what
  follows, we focus on data representation by means of generalized
  Lagrange systems, and in particular on their {\em types}; recall that
  they take the form $(k,\n,\p,e)$, where $k$ is the number of blocks
  of Lagrange multipliers that were introduced, $\n$ gives the number
  of unknowns in each block, $\p$ gives the number of equations in each block, and
  $e$ indicates how many variables are fixed. The set $\Clos{(L)}$
  defined by a generalized Lagrange system $L$ is expected to have
  dimension $|\n| -e - |\p|$, where $|\ |$ denotes the sum of the
  entries of a vector. The reader can then verify how the description
  below matches that in  Subsection~\ref{ex:abstract}.

  Since $V$ is smooth, the parametrization $\scrS$ computed by ${\sf
    SingularPoints}$ defines the empty set and ${\sf
    RoadmapRecLagrange}$ is called with inputs the generalized
  Lagrange system $L=(\Gamma, (\, ), (1))$ and the parametrization
  $(1)$; the generalized Lagrange system $L$ has type $(0,(6),(2),0)$.

  In what follows, we assume that we are under the assumptions of
  Proposition~\ref{PROP:CORRECTNESSGLOBAL}, so that correctness is
  guaranteed.
  \begin{description}
  \item[{\bf Steps~\ref{rmplag:step:1}--\ref{rmplag:step:3}}] We have $d=6-2=4$; a matrix
    $\mA\in \GL(6, 0, \Q)$ and a vector $\u=(u_1, u_2)\in \Q^2$ are
    randomly chosen (Step~\ref{rmplag:step:2}) and we have $\dalgo=3$
    (Step~\ref{rmplag:step:3}).  \smallskip
  \item[{\bf Step~\ref{rmplag:step:4}}] We construct the generalized
    Lagrange system $L'=\Polarlag(L^\mA, \u, 3)$; it is the triple
    $(\Gamma', (), (1))$ where $\Gamma'$ is a straight-line program
    evaluating
$$
f^\mA_1, f^\mA_2, \quad [L_{1,1}, L_{1,2}]\cdot
\left [\begin{matrix}
\frac{\partial f^\mA_1}{\partial X_4} & \frac{\partial f^\mA_1}{\partial X_5} & \frac{\partial f^\mA_1}{\partial X_6} \\
\frac{\partial f^\mA_2}{\partial X_4} & \frac{\partial f^\mA_2}{\partial X_5} & \frac{\partial f^\mA_2}{\partial X_6} \\
\end{matrix}\right ], \quad 
u_1L_{1,1}+u_2L_{1,2}-1.
$$ This
construction is essentially what was described in
Examples~\ref{example:Lag} and~\ref{example:Lag2} (except that we did not apply a change of
variables in those examples).

\smallskip

The type of $L'$ is $(1, (6, 2), (2, 4), 0)$. By
Proposition~\ref{sec:lagrange:prop:transfer:polar}, $\Clos{(L')}$ is
the polar variety $W=\polar(0, 3, V(\f^\mA))$; by
Proposition~\ref{prop:ch4}, if it is not empty, it has dimension $2$,
as confirmed by the type of $L'$, since $2=(6+2)-0-(2+4)$.

\smallskip

  \item[{\bf Step~\ref{rmplag:step:8}}] This step consists in a
    recursive call to ${\sf RoadmapRecLagrange}$ with inputs $L'$ and
    $\scrC'$, where $\scrC'$ is constructed in
    Steps~\ref{rmplag:step:5}--\ref{rmplag:step:9-a}.  In this
    recursive call, we have $d=2$ and $\dalgo=2$.  Denoting by
    $\mA'\in \GL(6, 0, \Q)$ the matrix chosen at
    Step~\ref{rmplag:step:2} of that recursive call, the behavior is
    as follows:

\smallskip

    \begin{itemize}
    \item the generalized Lagrange system constructed at
      Step~\ref{rmplag:step:4} has type  $(2, (6,
    2, 6), (2, 4, 7), 0)$; it encodes  $\polar(0, 2, {W}^{\mA'})$, which either is empty
      or has dimension $1= (6+2+6)-0-(2+4+7)$. Its construction
      was illustrated in Example~\ref{example:Lag3}.
\smallskip
    \item the generalized Lagrange system constructed at
      Step~\ref{rmplag:step:10} has type $(1, (6, 2), (2, 4), 1)$ and
      encodes the fiber $\fbr(W^{\mA'}, \Zeroes(\scrQ''_1))$, where
      $\scrQ''_1$ is built at Step~\ref{rmplag:step:6} of that
      recursive call; it is either empty or has dimension
      $1=(6+2)-1-(2+4)$.
    \end{itemize}
\smallskip
    The recursive calls at Steps~\ref{rmplag:step:8}
    and~\ref{rmplag:step:11} will consist in executing
    ${\sf SolveLagrange}$ on their respective inputs and return
    one-dimensional parametrizations. The last step takes the union of
    the curves encoded by these parametrizations.
\smallskip

\item[{\bf Step~\ref{rmplag:step:10}}] At this step the generalized
  Lagrange system $L''=\Fiberlag(L^\mA, \scrQ'', \scrS'')$ is
  constructed. It has type $(0,(6),(2),2)$. 
  Proposition~\ref{sec:lagrange:prop:transfer:fiber} ensures that
  $L''$ defines the fiber
  $V''=\fbr(V^\mA, \Zeroes(\scrQ''))$, which is equidimensional of dimension $2=6-2-2$,
  if it is not empty.  \smallskip
  \item[{\bf Step~\ref{rmplag:step:11}}] This step consists in a
    recursive call to ${\sf RoadmapRecLagrange}$ with inputs $L''$ and
    $\scrC''$. Here, we have $d=2$ and $\dalgo=2$. Denoting by
    $\mA''\in \GL(6, 2, \Q)$ the matrix  chosen at Step~\ref{rmplag:step:2}, 
    we now have the following behavior:
\smallskip
    \begin{itemize}
    \item the generalized Lagrange system constructed at
      Step~\ref{rmplag:step:4} has type $(1, (6, 2), (2, 3), 2)$; it
       encodes $\polar(2,2,{V''}^{\mA''})$, which either is empty
      or has dimension $1$;
\smallskip    
\item the generalized Lagrange system constructed at
      Step~\ref{rmplag:step:10} has type $(0, (6), (2), 3)$ and will
      encode the fiber $\fbr({V''}^{\mA''}, \Zeroes(\scrQ''_2))$, which is either
      empty or has dimension $1$.
    \end{itemize}
\smallskip
The output of this step is a one-dimensional parametrization of the union of these curves.
\smallskip
  \item[{\bf Step~\ref{rmplag:step:12}}] This last step takes the
    union of the one-dimensional parametrizations computed through the
    recursive calls of Steps~\ref{rmplag:step:8}
    and~\ref{rmplag:step:11}, and restores the initial coordinates.
  \end{description}
  In the figure below, we show how the recursive calls are organized
  into a binary tree. The labels of the internal nodes of the tree indicate the
  input of ${\sf RoadmapRecLagrange}$ and the dimension of the set
  it defines; at, the leaves, the input defines a curve.
  \begin{center}
\begin{tikzpicture}[node distance=4cm, level distance=1.5cm,
level 1/.style={sibling distance=3.5cm, circle},
level 2/.style={sibling distance=2cm, circle}]
%\tikzstyle{every node}=[circle,draw]
\node [arn_n] (Root) {$L=(\Gamma,(\,),(1))$, $d=4$}
    child {
    node [arn_n] {$(L', \scrC'), \quad d=2$} 
    child { node[leaf] {$ d=1$}  }
    child { node[leaf] {$ d=1$} }
  }
  child {
    node [arn_n] {$(L'', \scrC''), \quad d=2$}
    child { node[leaf] {$ d=1$} }
    child { node[leaf] {$ d=1$} }
  } ;
%\node (dim0) {$d=4$} [right of = Root];

\end{tikzpicture}
\end{center}

\paragraph*{Acknowledgments}
  This research was supported by Institut Universitaire de France, the
  GeoLMI grant (ANR 2011 BS03 011 06) of the French National Research
  Agency, NSERC and the Canada Research Chairs program.

  We thank Saugata Basu and Marie-Fran\c{c}oise Roy for useful
  discussions during the preparation of this article. We also wish to 
  thank the referees of a previous version of this article for their
  very helpful comments.

\newpage
\tableofcontents
\newpage
\bibliographystyle{plain}
\bibliography{theo1}

\appendix
%\setcounter{section}{1}

%%%%%%%%%%%%%%%%%%%%%%%%%%%%%%%%%%%%%%%%%%%%
% Preliminaries
%\input{preliminaries-new}
\section{Preliminaries}\label{sec:preliminaries}

In Section~\ref{sec:prelim}, we introduced basic material on algebraic
sets. In this section, we further discuss {\em locally closed sets},
basic properties of polar varieties and, in the last section, of
charts and atlases that are used further.

%%%%%%%%%%%%%%%%%%%%%%%%%%%%%%%%%%%%%%%%%%%%%%%%%%%%%%%%%%%%

\subsection{Locally closed sets}\label{chap:prelim:sec:definitions}

We say that a subset $\lcs$ of $\C^n$ is {\em locally closed} if it
can be written $\lcs=\BasicOpen \cap \algZ$, with $\BasicOpen$ Zariski
open and $\algZ$ Zariski closed. For $\x$ in such a $\lcs$, we define
$\T_\x \lcs$ as $\T_\x \algZ$ (this is independent of the choice of
$\algZ$ or $\BasicOpen$).

The {\em dimension} of $\lcs$ is defined as that of its Zariski closure
$V$, and we say that $\lcs$ is equidimensional if $V$ is. When it is the
case, we define $\reg(\lcs)=\reg(V) \cap \lcs$ and $\sing(\lcs)=\sing(V)\cap
\lcs$; we say that $\lcs$ is non-singular if $\reg(\lcs)=\lcs$.

A first example of a locally closed set is the set $\reg(V)$, for $V$
an equidimensional algebraic set.  The following construction shows
some other locally closed sets that will arise naturally in the
sequel. Let $\f=(f_1,\dots,f_p)$ be polynomials in
$\C[X_1,\dots,X_n]$, with $p \le n$. We define $\oreg(\f)$ as the set
of all $\x$ in $V$ such that $\jac(\f)$ has full rank $p$ at $\x$.
Since $\jac(\f)$ having rank less than $p$ is a closed condition,
$\oreg(\f)$ is locally closed.

We also define $\freg(\f)$ as the Zariski closure of $\oreg(\f)$.  It
is the union of the irreducible components $V_i$ of $V(\f)$ such that
$\jac(\f)$ has generically full rank $p$ on $V_i$; if $\freg(\f)$ is
not empty, it is $(n-p)$-equidimensional by the Jacobian criterion
\cite[Theorem 16.19]{Eisenbud95}. Besides, if $\jac(\f)$ has full rank
$p$ at some point $\x\in \freg(\f)$, $\x$ is in $\reg(\freg(\f))$, so
we have $\oreg(\f) \subset \reg(\freg(\f))$. The converse may not be
true, so that the inclusion may be strict in general.

Slightly more generally, let $Q$ be a finite subset of $\C^e$ and let
$\f=(f_1,\dots,f_p)$ be in $\C[X_1,\dots,X_n]$, with now $p \le n -e
$. Just as we defined $\oreg(\f)$ and $\freg(\f)$ when $e=0$, we can
define $\oreg(\f,Q)$ and $\freg(\f,Q)$: the former is the set of all
$\x$ in $\fbr(V(\f),Q)$ such that $\jac(\f,e)$ has full rank $p$ at
$\x$, and $\freg(\f,Q)$ is the Zariski closure of $\oreg(\f,Q)$.  By
the Jacobian criterion, $\freg(\f,Q)$ is either empty or
$(n-e-p)$-equidimensional.

%% Let $V\subset \C^n$ be an algebraic set and $f\in \C[X_1, \ldots,
%%   X_n]$. Another example of a locally closed set is the intersection
%% of $V$ with the complement $\Open(f)$ of the hypersurface
%% defined by $f=0$, for some non-zero polynomial $f$. In this context,
%% $\C[X_1, \ldots, X_n]_f$ denotes the localization of $\C[X_1, \ldots,
%%   X_n]$ by $f$, that is, the ring $\C[X_1, \ldots, X_n,1/f]$. The
%% localization of an ideal $I\subset \C[X_1, \ldots, X_n]$ will be
%% denoted by $I_f$; such ideals define closed subsets of
%% $\Open(f)$.

The following lemma will help us to give local descriptions of
algebraic sets.
\begin{lemma}\label{lemma:prelim:locallyclosed}
  Let $V \subset \C^n$ be an algebraic set and let $\BasicOpen \subset
  \C^n$ be a Zariski open set. Suppose that there exists an integer
  $c$, and that for all $\x$ in $\BasicOpen \cap V$ there exist
  \begin{itemize}
  \item an open set $\BasicOpen'_\x \subset \BasicOpen$ that contains $\x$,
\smallskip
  \item polynomials $\h_\x=(h_{\x,1},\dots,h_{\x,c})$ in $\C[X_1,\dots,X_n]$,
    with $c \le n$,
  \end{itemize}
  such that 
  \begin{itemize}
  \item $\BasicOpen'_\x \cap V = \BasicOpen'_\x \cap V(\h_\x)$
\smallskip
  \item  $\jac(\h_\x)$ has full rank $c$ at $\x$.
  \end{itemize}
  Then, $\lcs=\BasicOpen \cap V$ is either empty or a non-singular
  $d$-equidimensional locally closed set, with $d=n-c$, and for all $\x$
  in $\BasicOpen \cap V$, $\T_\x \lcs = \T_\x V= \ker(\jac_\x(\h_\x))$.
\end{lemma}
\begin{proof}
  If $\BasicOpen \cap V$ is empty, there is nothing to prove, so we will
  assume it is not the case. Take $\x$ in $\BasicOpen \cap V$ and let
  $\BasicOpen'_\x$ and $\h_\x$ be as above. By the Jacobian
  criterion~\cite[Theorem 16.19]{Eisenbud95}, we know that there
  exists a unique irreducible component $\algZ$ of $V(\h_\x)$ containing
  $\x$, that $\algZ$ has dimension $d=n-c$, that $\algZ$ is non-singular at
  $\x$ and that $\T_\x \algZ$ is the nullspace of the Jacobian of $\h_\x$
  at~$\x$.

  In the next few paragraphs, we prove that $\algZ$ is actually an
  irreducible component of $V$, and that it is the only irreducible
  component of $V$ containing $\x$.

  We restrict $\BasicOpen'_\x$ to an open set $\BasicOpen''_\x$, still
  containing $\x$, so as to be able to assume that $\BasicOpen''_\x
  \cap V(\h_\x)=\BasicOpen''_\x \cap \algZ$. On the other hand, by
  restriction to $\BasicOpen''_\x$, we also deduce that
  $\BasicOpen''_\x \cap V = \BasicOpen''_\x \cap V(\h_\x)$, so that
  $\BasicOpen''_\x \cap V = \BasicOpen''_\x \cap \algZ$.  The Zariski
  closure of $\BasicOpen''_\x \cap \algZ$ is equal to $\algZ$ (since
  the former is a non-empty open subset of $\algZ$), so upon taking
  Zariski closure, the former equality implies that $\algZ$ is
  contained in $V$.

  Next, we prove that $\algZ$ is actually an irreducible component of
  $V$. Let indeed $\algZ'$ be an irreducible component of $V$ containing
  $\algZ$, so that we have $\algZ \subset \algZ' \subset V$. Taking the
  intersection with $\BasicOpen''_\x$, we deduce that $\BasicOpen''_\x
  \cap \algZ \subset \BasicOpen''_\x \cap \algZ' \subset \BasicOpen''_\x \cap V$.
  Since the right-hand side is equal to $\BasicOpen''_\x \cap \algZ$, we
  deduce that $\BasicOpen''_\x \cap \algZ=\BasicOpen''_\x \cap \algZ'$, which implies
  that $\algZ=\algZ'$.

  Similarly, we prove that $\algZ$ is the only irreducible component
  of $V$ containing $\x$. Let indeed $\algZ''$ be any other
  irreducible component of $V$. The inclusion $\algZ'' \subset V$
  yields $\BasicOpen''_\x \cap \algZ'' \subset \BasicOpen''_\x \cap
  \algZ$. This implies that $\BasicOpen''_\x \cap \algZ''$ is empty,
  since otherwise taking the Zariski closure would yield $\algZ''
  \subset \algZ$.  Thus, we have proved our claim on $\algZ$; it
  implies in particular that $\T_\x V = \T_\x \algZ$, that is,
  $\ker(\jac_\x(\h_\x))$.

  We can now conclude the proof of the lemma. We know that
  $\BasicOpen\cap V$ is a locally closed set, and we assumed that it
  is non-empty. Besides, its Zariski closure $V'$ is the union of the
  irreducible components of $V$ that intersect $\BasicOpen$.  Let
  $V''$ be one of them and let $\x$ be in $\BasicOpen\cap V''$.
  Because $\x$ is in $\BasicOpen\cap V$, the construction of the
  previous paragraphs shows that $V''$ coincides with the irreducible
  variety $\algZ$ defined previously, so $\dim(V'')=n-c$. This proves
  that $V'$ is $d$-equidimensional, with $d=n-c$.

  Finally, we have to prove that for all $\x$ in $\BasicOpen\cap V$,
  $\x$ is in $\reg(V')$. We know that there exists a unique
  irreducible component $\algZ$ of $V$ that contains $\x$, that $\algZ$ is
  non-singular at $\x$ and that $\T_\x \algZ=\ker(\jac_\x(\h_\x))$. But
  then, $\algZ$ is also the unique irreducible component of $V'$ that
  contains $\x$, so $\x$ is indeed in $\reg(V')$.
\end{proof}

%%%%%%%%%%%%%%%%%%%%%%%%%%%%%%%%%%%%%%%%%%%%%%%%%%%%%%%%%%%%
%%%%%%%%%%%%%%%%%%%%%%%%%%%%%%%%%%%%%%%%%%%%%%%%%%%%%%%%%%%%
%%%%%%%%%%%%%%%%%%%%%%%%%%%%%%%%%%%%%%%%%%%%%%%%%%%%%%%%%%%%

\subsection{Critical points and polar varieties}\label{ssec:A12}

Let $V \subset \C^n$ be an equidimensional algebraic set (possibly
empty) and let $\varphi:V \to \C^m$ be a polynomial mapping.  A point
$\x \in \reg(V)$ is a {\em critical point} of $\varphi$ if
$d_\x\varphi(\T_\x V) \ne \C^m$, \new{where $d_\x\varphi$ is the
  differential of $\varphi$ at $\x$}. We denote by
$\openpolar(\varphi,V) \subset \reg(V)$ the set of all critical points
of $\varphi$; this is a locally closed set. A {\em critical value} of
$\varphi$ is the image by $\varphi$ of a critical point; a {\em
  regular value} is a point of $\C^m$ which is not a critical value.

We also define $\Kpolar(\varphi, V)$%% \gls{Kpolar4}
as the union of
$\openpolar(\varphi,V)$ and $\sing(V)$. The following lemma shows in
particular that this is an algebraic set.
\begin{lemma} \label{lemma:critequi} 
  Suppose that $V$ is $d$-equidimensional. Given generators $\f$ of
  $\IdealV$, the following holds:
  $$\openpolar(\varphi,V) = \left \{ \x \in V \ | \ \rank(\jac_\x(\f))=n-d
  \ \text{and}\ \rank \left [ \begin{matrix}\jac_\x(\f)
      \\ \jac_\x(\varphi) \end{matrix} \right ] < n-d+m \right \}$$ 
and
  $$\Kpolar(\varphi,V) = \left \{ \x \in V \ | \ \rank \left
  [ \begin{matrix}\jac_\x(\f) \\ \jac_\x(\varphi) \end{matrix} \right ] <
  n-d+m\right \}.$$ 
  In particular, $\Kpolar(\varphi,V)$ is Zariski closed,
  and we have $\Kpolar(\varphi,V) = \polar(\varphi,V) \cup \sing(V)$,
  where $\polar(\varphi,V)$ is the Zariski closure of $\openpolar(\varphi,V)$.
\end{lemma}
\begin{proof}
  For $\x$ in $V$, $\x$ is in $\openpolar(\varphi,V)$ if and only if we
  have $\x\in\reg(V)$ and $\dim(d_\x\varphi(\T_\x V)) < m$. By
  Lemma~\ref{sec:prelim:lemma:sing}, the first condition amounts to
  the rank condition $ \rank(\jac_\x(\f))=n-d$. When this is
  satisfied, since $\T_\x V$ is the nullspace of $\jac_\x(\f)$, the
  second condition amounts to $$\rank \left
    [ \begin{matrix}\jac_\x(\f) \\ \jac_\x(\varphi) \end{matrix}
  \right ] < n-d+m,$$ which proves the formula for
  $\openpolar(\varphi,V)$. To prove the one for $\Kpolar(\varphi,V)$, observe
  that $\sing(V)$ is the subset of $V$ where $\jac(\f)$ has rank less
  than $n-d$, so that $\Kpolar(\varphi,V)$ is the subset of all $\x$ in $V$
  such that
  $$\left ( \rank(\jac_\x(\f))=n-d
  \ \text{and}\ \rank \left [ \begin{matrix}\jac_\x(\f)
      \\ \jac_\x(\varphi) \end{matrix} \right ] < n-d+m  \right )$$ 
  or  $$\rank(\jac_\x(\f))<n-d.$$
  Now, if $\jac_\x(\f)$ has rank less than $n-d$, then $\left [\begin{smallmatrix}\jac_\x(\f)
      \\ \jac_\x(\varphi) \end{smallmatrix} \right ]$ has rank less
  than $n-d+m$, so the condition above is equivalent to the one
  given in the statement of the lemma. The last property follows immediately,
  since the above expression of $\Kpolar(\varphi,V)$ shows that it is Zariski closed.
\end{proof}

Polar varieties are a particular case of the previous definition: if
$V$ is a $d$-equidimensional algebraic subset of $\C^n$ lying over a
finite subset $Q$ of $\C^e$, then we have
$\openpolar(e,d,V)=\openpolar(\pi_{e,d},V)$,
$\polar(e,d,V)=\polar(\pi_{e,d},V)$ and
$\Kpolar(e,d,V)=\Kpolar(\pi_{e,d},V)$. In particular, we obtain
that $$\Kpolar(e,d,V)=\polar(e,d,V) \cup \sing(V).$$ The
following lemma, which handles the simple case $e=0$, is similarly a
direct consequence of Lemma~\ref{lemma:critequi}.
\begin{lemma}\label{sec:prelim:lemma:Ksing}
  If $V\subset \C^n$ is a $d$-equidimensional algebraic set,
  $\IdealV=\langle \f\rangle$ and $\dalgo$ is in $\{1, \ldots, d\}$,
  then $\Kpolar(0,\dalgo,V)$ is the zero-set of $\f$ and of all
  $c$-minors of $\jac(\f,\dalgo)$, where $c=n-d$ is the codimension of
  $V$.  
\end{lemma}
\begin{lemma}\label{prelim:lemma:polarinclusions}
  Let $Q$ be a finite subset of $\C^e$, and let $V$ be an algebraic
  subset of $\C^n$ lying over $Q$. If $V$ is $d$-equidimensional,
  the following inclusions hold:
  $$\openpolar(e,1, V)\subset \openpolar(e,2, V)\subset \cdots \subset \openpolar(e,d, V).$$    
\end{lemma}
\begin{proof}
  Lemma~\ref{lemma:critequi} shows that for $1\le i \le i' \le d$, 
  $\openpolar(e,i, V)$ and $\openpolar(e,i', V)$ are defined by 
  rank conditions on matrices 
  $$\bf{M}_{i,\x}=\left [ \begin{matrix}\jac_\x(\f)
      \\ \jac_\x(\pi_{e,i}) \end{matrix} \right ]
  \quad\text{and}\quad
  \bf{M}_{i',\x=}\left [ \begin{matrix}\jac_\x(\f)
      \\ \jac_\x(\pi_{e,i'}) \end{matrix} \right ],$$
  where $\f$ is a finite set of generators of the ideal of $V$.
  The latter matrix is obtained by adding $i'-i$ rows to the former one;
  hence, if $\jac_\x(\f)$ has rank $n-d$ and
 $\bf{M}_{i,\x}$ has rank less than $n-d+i$, 
 $\bf{M}_{i',\x}$ has rank less than $n-d+i'$.
\end{proof}

Also, one of the constructions which are used in our roadmap algorithm
consists in considering polar varieties of polar varieties (see
Section~\ref{ssec:abstractalgo}). In this context, the following lemma
will be useful.

\begin{lemma}\label{prelim:lemma:polarpolar}
  Let $Q$ be a finite subset of $\C^e$, and let $V$ be an algebraic
  subset of $\C^n$ lying over $Q$. Suppose that $V$ is
  $d$-equidimensional, and let $\dalgo$ be an integer in
  $\{1,\dots,d\}$.  Suppose further that $W=\polar(e,\dalgo, V)$ is
  equidimensional. Then $\openpolar(e,1,V)$, and thus $\polar(e,1,V)$,
  are subsets of $\Kpolar(e,1, W)$.
\end{lemma}
\begin{proof}
  When $\openpolar(e, 1, V)$ is empty, we are done. Hence, assume it
  is not empty and let $\x$ be in $\openpolar(e, 1,
  V)$. Lemma~\ref{prelim:lemma:polarinclusions} implies that $\x$ is
  in $W$. Since we have assumed $W$ to be equidimensional, it makes
  sense to consider its singular and regular loci. If $\x$ is in
  $\sing(W)$, then $\x$ is in $\Kpolar(e, 1, W)$, by definition, so we
  are done. Assume now that $\x$ is in $\reg(W)$, and denote by $\T_\x
  W$ the tangent space to $W$ at~$\x$.

  By definition of $\openpolar(e, 1, V)$, $\x$ is in $\reg(V)$ and
  $d_\x \pi_{e,1} (\T_\x V)\neq \C$. Moreover, since $W\subset V$,
  $\T_\x W\subset \T_\x V$. We deduce that $d_\x \pi_{e,1} (\T_\x
  W)\neq \C$; hence $\x$ is in $\openpolar(e,1, W)$, and we are done.
\end{proof}

An essential ingredient for our algorithms is the control of the
dimension of polar varieties of an algebraic set $V\subset \C^n$,
together with the dimension of fibers taken on these polar varieties,
under the assumption that $V$ is equidimensional with finitely many
singular points. We mention the following result in this direction,
which holds in generic coordinates; it is sufficient for us to 
state it for $e=0$.

\begin{lemma}\label{sec:prelim:lemma:finitefiberpolar}
  Let $V$ be an algebraic subset of $\C^n$, and suppose that $V$ is
  $d$-equidimensional, with finitely many singular points. Then, for
  $\dalgo$ in $\{1,\dots,d\}$, there exists a non-empty Zariski open
  set $\ZOffp(V,\dalgo) \subset \GL(n)$ such that, for $\mA$ in
  $\ZOffp(V,\dalgo)$, for any $\x\in \C^{\dalgo-1}$,
  $\fbr(\polar(0,\dalgo,V^\mA), \x)$ and
  $\fbr(\Kpolar(0,\dalgo,V^\mA), \x)$ are finite.
\end{lemma}
This result is proved in~\cite[Theorem~1]{SaSc03}. Note that the
assumptions of that theorem require that $V$ be non-singular, but this
result extends to our setting where $\sing(V)$ is finite. Indeed, that
assumption was only used to ensure another property, that the
dimension of $\Kpolar(0,\dalgo,V^\mA)$ be at most $\dalgo-1$; the claim we are
making here still holds as soon as $\sing(V)$ is finite.

\medskip

Finally, we will have to consider the case of locally closed sets
instead of algebraic sets. Suppose thus that $\lcs\subset \C^n$ is a
locally closed set with Zariski closure $V$ and that $\lcs$ is
$d$-equidimensional; let further $\varphi$ be a polynomial mapping $V
\to \C^m$. Then, we define $\openpolar(\varphi,\lcs)$ as
$\openpolar(\varphi,\lcs)=\openpolar(\varphi,V) \cap \lcs$. 
In this context, we say that $\y \in \C^m$ is a {\em regular value of
  $\varphi$ on $\lcs$} if $\varphi^{-1}(\y)\cap \lcs$ and $\openpolar(\varphi,\lcs)$
do not intersect, and a {\em critical value of $\varphi$ on $\lcs$} if
they do.

In particular, if $V$ lies over a finite set $Q \subset \C^e$, for all
$\dalgo\in\{1,\dots,n\}$, $\openpolar(e,\dalgo,\lcs)$ is defined as
$\openpolar(e,\dalgo,\lcs)=\openpolar(e,\dalgo,V) \cap \lcs$.

%%%%%%%%%%%%%%%%%%%%%%%%%%%%%%%%%%%%%%%%%%%%%%%%%%%%%%%%%%%%
%%%%%%%%%%%%%%%%%%%%%%%%%%%%%%%%%%%%%%%%%%%%%%%%%%%%%%%%%%%%
%%%%%%%%%%%%%%%%%%%%%%%%%%%%%%%%%%%%%%%%%%%%%%%%%%%%%%%%%%%%

\subsection{Properties of charts and atlases}\label{ssec:chartsatlasesprelim}

\subsubsection{Charts}

In this paragraph, we state a few of properties of charts, as defined
in Definition~\ref{def:chart}.

\begin{lemma}\label{sec:atlas:lemma:singS}
  Let $Q \subset \C^e$ be a finite set and let $V \subset \C^n$ and
  $S\subset \C^n$ be algebraic sets lying over $Q$.
  
  Let $\psi=(m,\h)$ be a chart of $(V,Q,S)$, with
  $\h=(h_1,\dots,h_c)$.  Then, $\Open(m)\cap V-S$ is a non-singular
  $d$-equidimensional locally closed set, with $d=n-e-c$. Besides, for
  all $\x$ in $\Open(m)\cap V-S$, $\T_\x V =
  \underbrace{(0,\dots,0)}_{e} \times \ker(\jac_\x(\h,e))$.
\end{lemma}
\begin{proof}
  Let $\singSBasicOpen\subset \C^n$ be the non-empty Zariski open set
  $\Open(m)-S$. For all $\x=(x_1,\dots,x_n)$ in $\singSBasicOpen\cap V$,
  let $\h_\x$ be the polynomials $ (X_1-x_1, \ldots,
  X_e-x_e,\h)$.
  Letting $\singSBasicOpen'_\x\subset \singSBasicOpen$ be an open set
  containing $\x$ such that $\fbr(V(\h),Q)$ and $\fbr(V(\h),\y)$
  coincide in $\singSBasicOpen'_\x$, where $\y=(x_1,\dots,x_e)$, we are in
  a position to apply Lemma~\ref{lemma:prelim:locallyclosed} to $V$,
  $\singSBasicOpen'_\x$ and $\h_\x$.  The lemma proves that
  $\singSBasicOpen \cap V$ is either empty or a non-singular
  $d$-equidimensional locally closed set, with $d=n-e-c$, and that for
  all $\x$ in $\singSBasicOpen \cap V$, $\T_\x V =
  \ker(\jac_\x(\h_\x))$.
  This is exactly the claimed result (since we know that
  $\singSBasicOpen \cap V$ is not empty).
\end{proof}

\begin{lemma}\label{sec:lemma:singS}
  Let $Q \subset \C^e$ be a finite set and let $V \subset \C^n$ and
  $S\subset \C^n$ be algebraic sets lying over $Q$.

  Suppose that $V$ is $d$-equidimensional and let $\psi=(m,\h)$ be a
  chart of $(V,Q,S)$. Then $\Open(m)\cap V-S$ is contained in
  $\reg(V)$, and $\h$ has cardinality $c=n-e-d$.
\end{lemma}
\begin{proof}
  The previous lemma implies that for all $\x$ in $\Open(m)\cap
  V-S$, $\T_\x V$ has dimension $n-e-c$, and also proves that the
  Zariski closure of $\Open(m)\cap V-S$ has the same dimension.
  Since this Zariski closure is the union of some irreducible
  components of $V$, it has dimension $d=\dim(V)$, so $d=n-e-c$, and
  every $\x$ as above is in $\reg(V)$.
\end{proof}

Conversely, provided that $V$ is equidimensional, the following lemma shows
that charts always exist at regular points.
\begin{lemma}\label{sec:atlas:lemma:chartregularpoint}
  Let $Q \subset \C^e$ be a finite set and let $V \subset \C^n$ and
  $S\subset \C^n$ be algebraic sets lying over $Q$.

  Suppose that $V$ is $d$-equidimensional.  For $\x$ in $\reg(V)-S$,
  there exists a chart $\psi=(m,\h)$ of $(V, Q,S)$ such that $\x\in
  \Open(m)$.
\end{lemma}
\begin{proof}
  Let $\x=(x_1,\dots,x_n)$ be in $\reg(V)-S$, let $\y=(x_1,\dots,
  x_e)\in Q$ and let $\H=( X_1-x_1, \ldots, X_e-x_e,h_1, \ldots, h_s)$
  be generators of the ideal of $V_\y=\fbr(V,\y)$. Without loss of
  generality, we assume that the polynomials $h_1, \ldots, h_s$ lie in
  $\C[X_{e+1}, \ldots, X_n]$, by evaluating the variables $X_1,
  \ldots, X_e$ at $x_1, \ldots, x_e$. We also consider a polynomial
  $q\in \C[X_1, \ldots, X_e]$ such that $q$ vanishes at all points of
  $Q$ except $\y$; note that this implies that $\Open(q)\cap
  V=V_\y$.

  Since $\x$ is in $\reg(V)$, and thus in $\reg(V_\y)$, the rank of
  $\jac(\H)$ at $\x$ is the codimension $c'=n-d$ of $V_\y$;
  equivalently, due to the shape of the polynomials $\H$, $\jac(\H,e)$
  has rank $c=c'-e$ at $\x$.  Up to renumbering the polynomials in
  $\H$, one can suppose that $\h=(h_1,\dots,h_{c})$ is such that
  $\jac_\x(\h,e)$ has full rank~$c$, or equivalently, that
  $\h'=(X_1-x_1, \ldots, X_e-x_e,h_1,\dots,h_{c})$ is such that
  $\jac_\x(\h')$ has full rank~$c'$.

  We let $\polmu$ be a $c$-minor of $\jac(\h,e)$ such that
  $\polmu(\x)\neq 0$ and let $\algZ$ be the Zariski closure of
  $\Open(q\polmu)\cap V(\h')$. Since
  $\x\in \Open(q\polmu)\cap V(\h')$, $\algZ$ is not empty.  Also, at
  all points of $\Open(q\polmu)\cap V(\h')$, $\jac(\h, e)$ has full
  rank $c$, or equivalently $\jac(\h')$ has full rank $c'$.  We deduce
  by Lemma~\ref{lemma:prelim:locallyclosed} that
  $\Open(q\polmu)\cap V(\h')$ is a non-singular
  $d$-equidimensional locally closed set, lying over $\y$ and
  containing $\x$; in particular, there is a unique irreducible
  component $\algZ'$ of $\algZ$ which contains~$\x$, and it has dimension
  $d$~\cite[Chapter~9, Theorem 9]{CLO}.

  We claim that $\algZ'$ is contained in $V_\y$. Indeed, since $\x$
  belongs to $\reg(V_\y)$, and $V_\y$ is $d$-equidimensional, there is
  a unique $d$-dimensional irreducible component $Y$ of $V_\y$ that
  passes through $\x$. Since all polynomials $\H$, and thus $\h'$,
  vanish on $Y$, we deduce that $\Open(q\polmu)\cap Y$ is
  contained in $\Open(q\polmu)\cap V(\h')$; taking the Zariski
  closure, we deduce that $Y$ is contained in $\algZ$ (since
  $\Open(q\polmu)\cap Y$ is a non-empty open subset of $Y$, its
  Zariski closure is $Y$). Thus, $Y$ is $d$-dimensional, irreducible,
  and contained in $\algZ$; this implies that $Y=\algZ'$, proving our claim.

  Let now $U$ be the Zariski closure of $\algZ-V$: it is the union of all
  irreducible components of $\algZ$ that are not contained in $V$. We
  proved before that there is a unique irreducible component $\algZ'$ of
  $\algZ$ which contains $\x$, and that $\algZ'$ is contained in $V_\y$, and
  thus in $V$; as a consequence, $\x$ is not in $U$. Then, there
  exists a polynomial $\polmu'$ in the ideal of $U$ such that
  $\polmu'(\x)\neq 0$. Define $m = q \polmu \polmu'$; we claim that $\psi=(m,
  \h)$ is a chart of $(V, Q, S)$.
  \begin{itemize}
  \item[$\sfC_1.$] Since by construction $\x\in \Open(q\polmu \polmu')\cap
    V-S$, this set is not empty.
\smallskip
  \item[$\sfC_2.$] We have to prove that
    $\Open(q\polmu \polmu')\cap V-S= \Open(q\polmu
    \polmu')\cap \fbr(V(\h),Q)-S$.
    Observe that due to our choice of $q$, this amounts to proving
    that
    $\Open(q\polmu \polmu')\cap V_\y-S= \Open(q\polmu
    \polmu')\cap V(\h')-S$.
    
    One inclusion is straightforward: if $\x'$ is in
    $\Open(q\polmu \polmu')\cap V_\y-S$, all polynomials $\H$
    vanish at $\x'$, and so do all polynomials $\h'$. Conversely, take
    $\x'$ in $\Open(q\polmu \polmu')\cap V(\h')-S$. This implies that
    $\x'$ is in $V'$, but it cannot be in $U$, since $\polmu'(\x')\ne 0$;
    thus, $\x'$ must be in $V$, or equivalently in $V_\y$, and we are
    done.

\smallskip
  \item[$\sfC_3.$] By construction, $c=n-d-e$, so $c+e=n-d$ satisfies
    $c+e \le n$.
\smallskip
  \item[$\sfC_4.$] Finally, take $\x'$ in
    $\Open(q\polmu \polmu')\cap V-S$. We have to prove that
    $\jac(\h,e)$ has full rank $c$ at $\x'$; this is immediate from
    the fact that $\polmu(\x')\ne 0$, and that $\polmu$ is a $c$-minor of
    that same matrix.
  \end{itemize}
  Since by construction $\x$ is in $\Open(q\polmu \polmu')$, the proof is
  complete.
\end{proof}

We finish this paragraph with a straightforward result: we can read
off the polar varieties as those points where the rank of a submatrix
of the Jacobian of $\h$ drops.

\begin{lemma}\label{sec:atlas:chartpolar}
  Let $Q \subset \C^e$ be a finite set and let $V \subset \C^n$ and
  $S\subset \C^n$ be algebraic sets lying over $Q$.

  Suppose that $V$ is $d$-equidimensional, let $\psi=(m,\h)$, with
  $\h=(h_1,\dots,h_c)$, be a chart of $(V,Q,S)$, and let ${\dalgo}$ be
  an integer in $\{1,\dots,d\}$. Then, for $\x$ in $\Open(m) \cap
  V-S$, $\x$ belongs to $\polar(e,{\dalgo}, V)$ if and only if $\jac_\x(\h,
  e+{\dalgo})$ does not have full rank $c$.
\end{lemma}
\begin{proof}
  Let $\x$ be in $\Open(m)\cap V-S$. By
  Lemma~\ref{sec:atlas:lemma:singS}, $\T_\x V$ coincides with
  $(0,\dots,0) \times \ker(\jac_\x(\h,e))$. Since $\x$ is in $\reg(V)$
  (Lemma~\ref{sec:lemma:singS}), it belongs to $\polar(e,{\dalgo},V)$ if
  and only if it belongs to $\openpolar(e, {\dalgo}, V)$. This is the case if and only
  if the projection $\ker(\jac_\x(\h,e))\to \C^{n-e-{\dalgo}}$ is not onto,
  and elementary linear algebra, as in Lemma~\ref{lemma:critequi},
  implies that this is equivalent to the submatrix $\jac_\x(\h, e+{\dalgo})$
  having rank less than $c$.
\end{proof}

%%%%%%%%%%%%%%%%%%%%%%%%%%%%%%%%%%%%%%%%%%%%%%%%%%%%%%%%%%%%

\subsubsection{Atlases}\label{ssec:atlas:def}

In this section, we investigate properties of atlases
(Definition~\ref{def:atlas}), as a way to describe coverings of an
algebraic set $V$ by means of charts.

Let $V \subset \C^n$ and $S\subset \C^n$ be algebraic sets lying over
a finite set $Q \subset \C^e$. Consider an atlas $\bpsi=(\psi_i)_{1
  \le i \le s}$ of $(V, Q,S)$, with $\psi_i=(m_i,\h_i)$ for all $i$.
When the vectors of polynomials $\h_i$ in charts $\psi_i$ do not have the same
cardinality, one may not expect that $V$ be equidimensional. Even when
they all have the same cardinality, there may still be the possibility
that $V$ has isolated points in $S$, so the following lemma is the
best we can hope for in this direction.

\begin{lemma}\label{sec:atlas:lemma:singSX}
  Let $Q \subset \C^e$ be a finite set and let $V \subset \C^n$ and
  $S\subset \C^n$ be algebraic sets lying over $Q$.

  Let $\bpsi=(\psi_i)_{1 \le i \le s}$ be an atlas of $(V,Q,S)$, with
  each $\psi_i$ of the form $(m_i,\h_i)$. If all $\h_i$ have common
  cardinality $c$, then $V-S$ is a non-singular $d$-equidimensional
  locally closed set, with $d=n-e-c$.
\end{lemma}
\begin{proof}
  Lemma~\ref{sec:atlas:lemma:singS} shows that for all $i \le s$,
  $\Open(m_i)\cap V-S$ is a non-singular $d$-equidimen\-sional locally
  closed set. Properties ${\sfA_2}$ and ${\sfA_3}$ in
  Definition~\ref{def:atlas} conclude the proof of the lemma.
\end{proof}

When we know that $V$ is equidimensional, better can be said.

\begin{lemma}\label{sec:coro:lemma:singSX}
  Let $Q \subset \C^e$ be a finite set and let $V \subset \C^n$ and
  $S\subset \C^n$ be algebraic sets lying over $Q$.

  Suppose that $V$ is $d$-equidimensional and let $\bpsi=(m_i,\h_i)_{1
    \le i \le s}$ be an atlas of $(V,Q,S)$. Then $\sing(V)$ is
  contained in $S$, and all $\h_i$ have common cardinality $c=n-e-d$.
\end{lemma}
\begin{proof}
  Lemma~\ref{sec:lemma:singS} proves that each $\Open(m_i)\cap V-S$ is contained in $\reg(V)$, so their union is. By
  assumption, the union of the sets $\Open(m_i)\cap V-S$ contains $V-S$,
  so that $V-S$ is contained in $\reg(V)$.  The same corollary also
  proves that all $\h_i$ have cardinality $c=n-e-d$.
\end{proof}

Slightly less elementary, the following lemma shows that
atlases always exist.

\begin{lemma}\label{sec:atlas:lemma:glob}
  Let $Q \subset \C^e$ be a finite set and let $V \subset \C^n$ be an
  algebraic set lying over $Q$.  Suppose that $V$ is
  $d$-equidimensional. Then, there exists an atlas of
  $(V, Q, \sing(V))$.
\end{lemma}
\begin{proof}
  Applying Lemma~\ref{sec:atlas:lemma:chartregularpoint} with
  $S=\sing(V)$, we deduce that for all $\x$ in $\reg(V)$, there exists
  a chart $\psi_\x=(m_\x,\h_\x)$ of $(V,Q,\sing(V))$, such that
  $m_\x(\x) \ne 0$. The open subsets $\Open(m_\x)$ cover $V-S=\reg(V)$;
  the following compactness argument shows that we can extract a
  finite cover from it.
  
  Let $I$ be the defining ideal of $V$. Then, the zero-set of $I +
  \langle (m_\x)_{\x \in \reg(V)} \rangle$ is contained in
  $\sing(V)$. Let $J=\langle f_1,\dots,f_r\rangle$ be the defining
  ideal of $\sing(V)$; then, every $f_i$ belongs to the radical of $I
  + \langle (m_\x)_{\x \in \reg(V)} \rangle$. Thus, there exists for
  all $i$ an expression of the form
  \begin{equation}\label{eq:fiK}
  f_i^{e_i} = \sum_{\x \in K} c_{i,\x} m_\x + I,    
  \end{equation}
 for some finite subset $F$ of $\reg(V)$. This implies that the
 finitely many $\Open(m_\x)$, for $\x$ in $F$, cover $\reg(V)$,
 which proves ${\sfA_3}$ by taking $\bpsi=(\psi_\x)_{\x \in F}$.

 It remains to prove that ${\sfA_2}$ holds, or in other words that
 $F$ is not empty. If that were not the case, Eq.~\eqref{eq:fiK}
 would imply that $V \subset \sing(V)$, a contradiction.
\end{proof}

%%%%%%%%%%%%%%%%%%%%%%%%%%%%%%%%%%%%%%%%%%%%%%%%%%%%%%%%%%%%

%%%%%%%%%%%%%%%%%%%%%%%%%%%%%%%%%%%%%%%%%%%%

%%%%%%%%%%%%%%%%%%%%%%%%%%%%%%%%%%%%%%%%%%%%
%\input{proof3.6}
\section{Proof of Proposition~\ref{prop:ch4}}\label{sec:proof3.6}\label{sec:chartsatlas}

The goal of this section is to prove Proposition~\ref{prop:ch4} which
we recall now: {\em Let $Q \subset \C^e$ be a finite set and let $V
  \subset \C^n$ and $S\subset \C^n$ be algebraic sets lying over $Q$,
  with $S$ finite.  Suppose that $V$ is equidimensional of
  dimension~$d$.  Let $\bpsi$ be an atlas of $(V,Q,S)$, and let
  ${\dalgo}$ be an integer in $\{1,\dots,d\}$. If $2 \le {\dalgo}\leq
  (d+3)/2$, there exists a non-empty Zariski open subset
  $\scrGpolar(\bpsi,V,Q,S,{\dalgo})$ of $\GL(n,e)$ such that for $\mA$
  in $\scrGpolar(\bpsi,V,Q,S,{\dalgo})$, the following holds:
  \begin{itemize}
  \item either $\polar(e,{\dalgo},V^\mA)$ is empty, or
    \smallskip
  \item $\atlaspolar(\bpsi^\mA,V^\mA,Q,S^\mA,{\dalgo})$ is an atlas of
    $(\polar(e,{\dalgo},V^\mA),Q,S^\mA)$, and
    $\polar(e,{\dalgo},V^\mA)$ is equidimensional of dimension
    $\dalgo-1$, with $\sing(\polar(e,{\dalgo},V^\mA))$ contained in the finite set $S^\mA$.
  \end{itemize}
}

%%%%%%%%%%%%%%%%%%%%%%%%%%%%%%%%%%%%%%%%%%%%%%%%%%%%%%%%%%%%
%%%%%%%%%%%%%%%%%%%%%%%%%%%%%%%%%%%%%%%%%%%%%%%%%%%%%%%%%%%%
%%%%%%%%%%%%%%%%%%%%%%%%%%%%%%%%%%%%%%%%%%%%%%%%%%%%%%%%%%%%

\subsection{Geometry of polar varieties}\label{sec:geometry_critical}

We start with preliminary material. As was mentioned when we stated
this proposition, we need a local variant of results
from~\cite[Section~3]{BGHSS}, which were proved for smooth complete
intersections. Since the proofs are somewhat subtle, we prefer to give
them here {\em in extenso}, in order to avoid overlooking any
difficulties.

Throughout this subsection, we use the definitions and notation
introduced in Sections~\ref{sec:prelim},~\ref{sec:dim:smooth:finite}
and~\ref{chap:prelim:sec:definitions}.  Let $\h=(h_1, \ldots, h_c)$ be
polynomials in $\C[X_1, \ldots, X_n]$.  We are going to prove a few
results about polar varieties associated to the locally closed set
$\oreg(\h)$, provided we are in generic coordinates. These results are
summarized in the following proposition.

\begin{proposition}\label{sec:prelim:prop:main1}\label{SEC:PRELIM:PROP:MAIN1}
  Let $\h=(h_1,\dots,h_c)$ in $\C[X_1,\dots,X_n]$, with $1 \le c \le
  n$. Let ${{\dalgo}}$ be an integer satisfying $1 \le {{\dalgo}}\le
  d$, with $d=n-c$.

  Then, there exists a non-empty Zariski open set $\ZOlp(\h,\dalgo) \subset
  \GL(n)$ such that, for $\mA$ in $\ZOlp(\h,\dalgo)$, the
  following properties hold:
  \begin{itemize}
\item[(1)] for all $\x$ in $\oreg(\h^\mA)$, there exists a $c$-minor $\minor$
  of $\jac(\h^\mA)$ such that $\minor(\x)\neq 0$;
\smallskip
\item[(2)] all irreducible components of the Zariski closure
  of the set $\openpolar(0,{{\dalgo}}, \oreg(\h^\mA))$ have dimension ${\dalgo}-1$; 
\smallskip
\item[(3)] if ${\dalgo}\leq (d+3)/2$ then for all $\x\in \oreg(\h^\mA)$, there
  exists a $(c-1)$-minor $\pminor$ of $\jac(\h^\mA,{{\dalgo}})$ such that
  $\pminor(\x)\neq 0$;
\smallskip
\item[(4)] for every $c$-minor $\minor$ of the Jacobian matrix
  $\jac(\h^\mA)$ and for every $(c-1)$-minor $\pminor$ of the
  truncated Jacobian matrix $\jac(\h^\mA,{{\dalgo}})$, the polynomials
  $(\h^\mA,\sfH(\h^\mA,{{\dalgo}},\pminor))$ (see
  Definition~\ref{def:mH}) define $\openpolar(0,{{\dalgo}}, \oreg(\h^\mA))$ in
  $\Open(\minor \pminor)$, and their Jacobian matrix has full
  rank $n-({\dalgo}-1)$ at all points of $\Open(\minor \pminor)
  \cap \openpolar(0,{{\dalgo}}, \oreg(\h^\mA))$.
  \end{itemize}
\end{proposition}

\noindent The rest of Section~\ref{sec:geometry_critical} is devoted to the proof of this proposition.

%%%%%%%%%%%%%%%%%%%%%%%%%%%%%%%%%%%%%%%%%%%%%%%%%%%%%%%%%%%%
%%%%%%%%%%%%%%%%%%%%%%%%%%%%%%%%%%%%%%%%%%%%%%%%%%%%%%%%%%%%
%%%%%%%%%%%%%%%%%%%%%%%%%%%%%%%%%%%%%%%%%%%%%%%%%%%%%%%%%%%%

\subsubsection{Sard's lemma and weak transversality}

In this paragraph, we re-prove two well-known transversality results
(Sard's lemma and Thom's weak transversality) in the context of
algebraic sets. These claims are folklore, but we did not find a
suitable reference for them.

The cornerstone of transversality is Sard's lemma; here, we give a
version for (possibly singular) algebraic sets. Note
that~\cite[Proposition~3.7]{Mumford76} establishes this claim when $V$
is irreducible and $\varphi$ is dominant. We will show that the same
arguments apply, up to minor modifications.

\begin{proposition}\label{prop:Sard}
  Let $V\subset \C^n$ be an equidimensional algebraic set and let
  $\varphi:V \to \C^m$ be a polynomial mapping. Then
  $\varphi(\openpolar(\varphi,V))$ is contained in a hypersurface of~$\C^m$.
\end{proposition}
\begin{proof}
Let us write the irreducible decomposition of the Zariski closure of
$\openpolar(\varphi,V)$ as
$$\overline{\openpolar(\varphi,V)} = \cup_{1\le i \le r} \algZ_i,$$ where the $\algZ_i$
are irreducible algebraic subsets of $V$. We suppose, by
contradiction, that $\varphi(\openpolar(\varphi,V))$ is dense in $\C^m$.
Then, $\varphi(\algZ_1 \cup \cdots \cup \algZ_r)$ is dense as well, which
implies that (up to renumbering) $\varphi(\algZ_1)$ is dense in $\C^m$.

By~\cite[Proposition~3.6]{Mumford76} (which applies to dominant
mappings between irreducible varieties), there exists a non-empty
open subset $\algZ'_1$ of $\algZ_1$ where all points are regular and
non-critical for $\varphi$. 

To continue, we prove that the equality $\openpolar(\varphi,V) =
\overline{\openpolar(\varphi,V)} \cap \reg(V)$ holds. Indeed, since
$\openpolar(\varphi,V)$ is contained in both $\overline{\openpolar(\varphi,V)}$
and $\reg(V)$, it is contained in $\overline{\openpolar(\varphi,V)} \cap
\reg(V)$. Conversely, Lemma~\ref{lemma:critequi} implies that
$\openpolar(\varphi,V)=\Kpolar(\varphi,V) \cap \reg(V)$, and that $\Kpolar(\varphi,V)$
is an algebraic set. Since $\openpolar(\varphi,V)$ is contained in
$\Kpolar(\varphi,V)$, its Zariski closure is contained in $\Kpolar(\varphi,V)$
too, so $\overline{\openpolar(\varphi,V)} \cap \reg(V)$ is contained in
$\Kpolar(\varphi,V) \cap \reg(V)$, that is, in $\openpolar(\varphi,V)$.

Taking the intersection with $\algZ_1$, the previous claim implies that $
\openpolar(\varphi,V) \cap \algZ_1 = \reg(V) \cap \algZ_1$; in particular, this is
an open subset of $\algZ_1$. More precisely, this is a {\em non-empty}
open subset of $\algZ_1$: if $\openpolar(\varphi,V) \cap \algZ_1$ were empty, we
would have $\openpolar(\varphi,V) = \openpolar(\varphi,V)-\algZ_1$, and thus
$\openpolar(\varphi,V) \subset \overline{\openpolar(\varphi,V)}-\algZ_1 \subset \algZ_2
\cup \cdots \cup \algZ_r$; taking the Zariski closure would yield
$\overline{\openpolar(\varphi,V)} \subset \algZ_2 \cup \cdots \cup \algZ_r$, a
contradiction.

Hence, both $\algZ'_1$ and $\openpolar(\varphi,V) \cap \algZ_1$ are non-empty
open subsets of $\algZ_1$. Since $\algZ_1$ is irreducible, they must intersect
at some point $\x$. Since $\x$ is in $\algZ'_1$, $\x$ is regular on $\algZ_1$
and $d_\x \varphi(\T_\x \algZ_1) =\C^m$ (recall that $d_\x \varphi$ denotes
the differential of $\varphi$ at $\x$). Since $\x$ is in
$\openpolar(\varphi,V)$, $\x$ is regular on $V$ and
$d_\x \varphi(\T_\x V) \ne \C^m$.  However, $d_\x \varphi(\T_\x \algZ_1)$ is
contained in $d_\x \varphi(\T_\x V)$, a contradiction.
\end{proof}

We continue with Thom's weak transversality theorem, specialized to
the particular case of transversality to a point; this can be
rephrased in terms of critical / regular values only. Our setup is the
following. Let $n,{{\dalgo}},m$ be positive integers and let
$\Phi(\X,\Theta) :\C^n \times \C^{{\dalgo}} \to \C^m$ be a polynomial
mapping. For $\vartheta$ in $\C^{{\dalgo}}$, $\Phi_{\vartheta}:\C^n\to
\C^m$ denotes the induced mapping $\x \mapsto \Phi(\x,\vartheta)$.

\begin{proposition}\label{thm:transversality}
  Let $\new{\BasicOpen} \subset \C^n$ be a Zariski open set and
  suppose that $0$ is a regular value of $\Phi$ on $\new{\BasicOpen}
  \times \C^{{\dalgo}}$. Then there exists a non-empty Zariski open
  subset $\UOpen\subset \C^{{\dalgo}}$ such that for all $\vartheta\in
  \UOpen$, $0$ is a regular value of $\Phi_{\vartheta}$ on
  $\new{\BasicOpen}$.
\end{proposition}

Before proving this proposition, let us establish a basic lemma.

\begin{lemma}\label{lemma:matrank1}
  Let $\mM$ be a matrix of the form $\left [\begin{smallmatrix} \mM_1 \\
      \mM_2 \end{smallmatrix} \right ]$. Then the equality $$\rank(\mM) = \rank(\mM_1) +
  \rank(\mM_2 | \ker(\mM_1) )$$ holds,
where $\mM_2 | \ker(\mM_1)$ denotes the restriction of the linear map defined by 
$\mM_2$ to the kernel of $\mM_1$.
\end{lemma}
\begin{proof}
  Let $(a,b)$ be the dimensions of $\mM$.
  From the equalities
$$\begin{array}{ccl}
\dim \ker(\mM_1) &=& \rank(\mM_2 | \ker(\mM_1) ) + \dim \ker(\mM_2 | \ker(\mM_1) ) \\[1mm]
&=& \rank(\mM_2 | \ker(\mM_1) ) + \dim (\ker(\mM_2) \cap  \ker(\mM_1)) \\[1mm]
&=& \rank(\mM_2 | \ker(\mM_1) ) + \dim\ker(\mM),
\end{array}$$
we deduce
$$b-\rank(\mM_1) = \rank(\mM_2 | \ker(\mM_1) ) + b-\rank(\mM)$$
and thus
$\rank(\mM) = \rank(\mM_1) + \rank(\mM_2 | \ker(\mM_1) ).$
\end{proof}

\begin{proof}[of Proposition~\ref{thm:transversality}]
  Let
  $\closedX'=\Phi^{-1}(0) \cap (\new{\BasicOpen} \times
  \C^{{\dalgo}})$
  and let $\closedX\subset \C^n\times \C^{{\dalgo}}$ be the Zariski
  closure of $\closedX'$. We will first prove: {\em if
    $\closedX' \ne \emptyset$, $\closedX$ is
    $(n+{{\dalgo}}-m)$-equidimensional, and $\closedX'$ is contained
    in $\reg(\closedX)$}.

  Assume that $\closedX' \ne \emptyset$, and take $(\x,\vartheta)$ in $\closedX'$;
  then, by assumption, $\jac_{(\x,\vartheta)}(\Phi)$ has full rank
  $m$. Since in a neighborhood of $(\x,\vartheta)$, $\closedX$ coincides with
  $\Phi^{-1}(0)$, the Jacobian criterion \cite[Theorem
  16.19]{Eisenbud95} implies that there is a unique irreducible
  component $\closedX_{(\x,\vartheta)}$ of $\closedX$ that contains $(\x,\vartheta)$, that
  $(\x,\vartheta)$ is regular on this component, that
  $\dim(\closedX_{(\x,\vartheta)})=n+{{\dalgo}}-m$ and that $\T_{(\x,\vartheta)} \closedX_{(\x,\vartheta)}$ is
  the nullspace of $\jac_{(\x,\vartheta)} (\Phi)$. 
 
  Since every irreducible component of $\closedX$ intersects $\closedX'$, this
  implies that $\closedX$ itself is equidimensional of dimension
  $n+{{\dalgo}}-m$, and thus that $\closedX'$ is contained in $\reg(\closedX)$. We
  are thus done with our claims on $\closedX$; note that we have also proved
  that for $(\x,\vartheta)$ in $\closedX'$, $\T_{(\x,\vartheta)} \closedX$ is the
  nullspace of $\jac_{(\x,\vartheta)} (\Phi)$ in $\C^n \times
  \C^{{\dalgo}}$.

  Denote by $\pi: \C^n\times \C^{{\dalgo}} \to \C^{{\dalgo}}$ the projection
  $(\x,\vartheta) \mapsto \vartheta$. We now prove: {\em if $\vartheta
    \in \C^{{\dalgo}}$ is such that $0$ is a critical value of
    $\Phi_{\vartheta}$ on $\new{\BasicOpen}$, then $\vartheta$ is a critical value of
    the restriction of 
    $\pi$ to $\closedX$}.
  
  Let $\vartheta \in \C^{{\dalgo}}$ be such that $0$ is a critical
  value of $\Phi_{\vartheta}$ on $\new{\BasicOpen}$. Thus, there exists $\x$ in
  $\openpolar(\Phi_{\vartheta}, \new{\BasicOpen})$ such that
  $\Phi(\x,\vartheta)=\Phi_{\vartheta}(\x)=0$. Since $\x$ lies in
  $\openpolar(\Phi_{\vartheta},\new{\BasicOpen})$, the matrix
  $\jac_\x(\Phi_{\vartheta})= \jac_{(\x,\vartheta)}(\Phi;\closedX) $ has
  rank less than $m$.

  On the other hand, our construction shows that $(\x,\vartheta)$ is
  in $\closedX'$ (so $\closedX'$ is not empty), and thus, using the above claim, in
  $\reg(\closedX)$. To conclude, we prove that $(\x,\vartheta)$ is in
  $\openpolar(\pi,\closedX)$; this is enough since by construction
  $\vartheta=\pi(\x,\vartheta)$. Let us consider the matrices
$$\jac_{(\x,\vartheta)}(\Phi)
= \left [ \begin{matrix} 
    \jac_{(\x,\vartheta)}(\Phi;\closedX) & \jac_{(\x,\vartheta)}(\Phi;\Theta)
\end{matrix} \right ]$$
and 
$$\mM = \left [ \begin{matrix}
    \jac_{(\x,\vartheta)}(\Phi;\closedX)  & \jac_{(\x,\vartheta)}(\Phi;\Theta)\\
      {\bf 0}_{{{\dalgo}}\times n} & {\bf 1}_{{{\dalgo}} \times {{\dalgo}}} 
    \end{matrix} \right ].$$
  By Lemma~\ref{lemma:matrank1}, we have the equality
  $\rank(\mM) = \rank(\jac_{(\x,\vartheta)}(\Phi)) + \rank( \pi\ |\
  \ker(\jac_{(\x,\vartheta)}(\Phi)))$.
  Since, as we saw above, the nullspace of
  $\jac_{(\x,\vartheta)}(\Phi)$ is the tangent space to $\closedX$ at
  $(\x,\vartheta)$, we get
$$\rank(\mM) = \rank(\jac_{(\x,\vartheta)}(\Phi)) + \rank(
\pi\ |\ \T_{(\x,\vartheta)}\closedX).$$
Recall that by assumption, $ \rank(\jac_{(\x,\vartheta)}(\Phi))=m$, so
that $$\rank(\mM) = m+ \rank( \pi\ |\ \T_{(\x,\vartheta)}\closedX).$$
On the other hand, one sees that
$\rank(\mM)=\rank(\jac_{(\x,\vartheta)}(\Phi;\closedX))+{{\dalgo}}$. Since
we have noted that $\rank(\jac_{(\x,\vartheta)}(\Phi;\closedX)) <m$,
we deduce that
$\rank( \pi\ |\ \T_{(\x,\vartheta)}\closedX) < {{\dalgo}}$, as
requested.

We can now conclude the proof of the proposition.
Proposition~\ref{prop:Sard} shows that the critical values of $\pi$ on
$\closedX$ are contained in a hypersurface of $ \C^{{\dalgo}}$, say $\Delta$. Let
$\UOpen=\C^{{\dalgo}}-\Delta$; this is a non-empty Zariski open subset
of $\C^{{\dalgo}}$. The former assertion shows that for all $\vartheta\in
\UOpen$, $0$ is a regular value of $\Phi_{\vartheta}$ on
$\new{\BasicOpen}$, as claimed.
\end{proof}

%%%%%%%%%%%%%%%%%%%%%%%%%%%%%%%%%%%%%%%%%%%%%%%%%%%%%%%%%%%%
%%%%%%%%%%%%%%%%%%%%%%%%%%%%%%%%%%%%%%%%%%%%%%%%%%%%%%%%%%%%
%%%%%%%%%%%%%%%%%%%%%%%%%%%%%%%%%%%%%%%%%%%%%%%%%%%%%%%%%%%%

\subsubsection{Rank estimates}

In this paragraph, we prove a key result towards
Proposition~\ref{sec:prelim:prop:main1}, following a construction
from~\cite{BaGiHePa05,BGHSS}.

We consider polynomials $\h=(h_1,\dots,h_c)$ in $\C[X_1,\dots,X_n]$,
with $1 \le c \le n$, and we let $d=n-c$. We further denote by
$\mA=A_{1,1},\dots,A_{1,n},\dots,A_{d,1},\dots,A_{d,n}$ a family of
$dn$ new indeterminates. For ${{\dalgo}} \le d$, $\mA_{\le {\dalgo}}$
denotes the $ {{\dalgo}} n$ indeterminates
$A_{1,1},\dots,A_{1,n},\dots,A_{{{\dalgo}},1},\dots,A_{{{\dalgo}},n}$
and the $(c+{{\dalgo}})\times n$ polynomial matrix $J_{{\dalgo}}$ is
defined as
$$J_{{\dalgo}}=\begin{bmatrix}
  & \jac(\h) & \\
A_{1,1} & \cdots & A_{1,n} \\
\vdots & & \vdots \\
A_{{{\dalgo}},1} & \cdots & A_{{{\dalgo}},n} 
\end{bmatrix}.$$ We will often view elements $\a \in \C^{{{\dalgo}} n}$ as
vectors of length ${{\dalgo}}$ of the form
$\a=(\a_1,\dots,\a_{{\dalgo}})$ with all $\a_i$ in $\C^n$; for such an
$\a$, the matrix $J_{{\dalgo}}(\X,\a)$ (where the indeterminates $\mA$
are evaluated at $\a$) is then naturally defined. When $\a$ is a
sequence of linearly independent vectors, we say that $\a$ has rank
${\dalgo}$. We start with a result that is a slight generalization of
\cite[Lemma 3]{BaGiHePa05}.

\begin{lemma}\label{lemma:bound:codim}
  Let $\a\in \C^{{{\dalgo}}n}$,
  $\lcsw=\{\x\in \oreg(\h)\mid {\rm rank}(J_{{\dalgo}}(\x, \a))\leq
  c+{\dalgo}-1\}$
  and $\algZ$ be an irreducible component of the Zariski closure of
  $\lcsw$. Then, $\algZ$ has dimension at least~${\dalgo}-1$.
\end{lemma}
\begin{proof}
  Let $\mathfrak{a}$ be the ideal generated by all $(c+{\dalgo})$-minors
  of the $(c+{\dalgo})\times n$ matrix $J_{{\dalgo}}(\X, \a)$.
  One can rewrite $\lcsw$ as $\lcsw = \oreg(\h) \cap V(\mathfrak{a})
  \subset \freg(\h) \cap V(\mathfrak{a})$.  Thus, if the extended ideal
  $\mathfrak{a}\cdot \C[\freg(\h)]$ is not a proper ideal of $\C[\freg(\h)]$, $\freg(\h) \cap
  V(\mathfrak{a})$, and thus $\lcsw$, are empty, and we are done; we
  suppose it is not the case.

  Since $\oreg(\h)$ is an open subset of $\freg(\h)$, $\lcsw$ is an open subset of
  $\freg(\h) \cap V(\mathfrak{a})$, and its Zariski closure is the union of
  some irreducible components of $\freg(\h) \cap V(\mathfrak{a})$.  Let us
  take one of these irreducible components; call it $\algZ$. If we let
  $\mathfrak{p}$ be the ideal of definition of $\algZ$ in $\C[\freg(\h)]$, then,
  by definition, $\mathfrak{p}$ is an isolated prime component of the
  determinantal ideal $\mathfrak{a} \cdot \C[\freg(\h)]$. By \cite[Theorem
    3]{EN62}, the height of $\mathfrak{p}$ is at most
  $n-c-({\dalgo}-1)$. This implies that the codimension of $\algZ$ in $\freg(\h)$
  is at most $n-c-({\dalgo}-1)$. Since $\freg(\h)$ has dimension $n-c$, $\algZ$
  has dimension at least~$\dalgo-1$.
\end{proof}

Our key result in this paragraph is the following claim on the rank of
$J_{{\dalgo}}$, which says that for suitable values of ${\dalgo}$, and
for a generic $\a$, the matrix $J_{{\dalgo}}(\x,\a)$ has rank defect
at most one for any $\x$ in $\oreg(\h)$. Surprisingly, it does not use
transversality; only dimension considerations.
\begin{proposition}\label{prop:rankJ}
  For ${{\dalgo}}$ in $\{1,\dots,\lfloor (d+3)/2\rfloor\}$, there
  exists a non-empty Zariski open subset $\ZOA12_{{\dalgo}} \subset
  \C^{{{\dalgo}} n}$ such that for all $(\x,\a) \in \oreg(\h) \times
  \ZOA12_{{\dalgo}}$, the matrix $J_{{\dalgo}}(\x,\a)$ has rank at
  least $c+{{\dalgo}}-1$.
\end{proposition}
For ${{\dalgo}}$ as above, let us denote by ${\sf a}_{{\dalgo}}$ the
property in the proposition, so that proving the proposition amounts
to proving that ${\sf a}_{{\dalgo}}$ holds for
${{\dalgo}}=1,\dots,\lfloor (d+3)/2\rfloor$.  Obviously, ${\sf a}_1$
holds, since for all $\x$ in $\oreg(\h)$, $\jac_\x(\h)$ has rank $c
=c+1-1$ (so we can take $\ZOA12_1 = \C^n$). Thus, we can now focus on
the case ${{\dalgo}} \ge 2$.

For such a ${{\dalgo}}$, we will consider pairs of the form
$\mathsf{m}=(\mathsf{m}_{\text{row}},\mathsf{m}_{\text{col}})$ where
$\mathsf{m}_{\text{row}} \subset \{1,\dots,c+{{\dalgo}}-1\}$ and
$\mathsf{m}_{\text{col}} \subset \{1,\dots,n\}$ are sets of
cardinality $c+{{\dalgo}}-2$, and such that $\{1,\dots,c\} \subset
\mathsf{m}_{\text{row}}$. To one such $\mathsf{m}$, one can associate
the square submatrix $J_\mathsf{m}$ of size $c+{{\dalgo}}-2$ of $J_{{\dalgo}}$
whose rows and columns are indexed by the entries of
$\mathsf{m}_{\text{row}}$ and $\mathsf{m}_{\text{col}}$.  Thus,
$J_\mathsf{m}$ contains all rows coming from $\jac(\h)$ and excludes
two rows depending on the variables $\mA_{\le {{\dalgo}}}$, one of them
being the last row of $J_{{\dalgo}}$. We denote by $g_\mathsf{m}$ the
determinant of $J_\mathsf{m}$; this is a polynomial in $\C[\X,\mA_{\le
    {{\dalgo}}-1}]$, which we will see in $\C[\X,\mA_{\le {{\dalgo}}}]$ as well
when needed.

We denote by $\mathsf{Sub}_{{\dalgo}}$ the set of all pairs
$\mathsf{m}=(\mathsf{m}_{\text{row}},\mathsf{m}_{\text{col}})$ as
above such that, additionally, there exists $(\x,\a)\in
\oreg(\h)\times \C^{{{\dalgo}} n}$ such that $g_\mathsf{m}(\x,\a) \ne
0$.  Then, for $\mathsf{m}\in \mathsf{Sub}_{{\dalgo}}$, we introduce the
following condition:
\begin{enumerate}
\item [${\sf R}_{\mathsf{m}}:$] There exists a non-empty Zariski open
  subset $\ZOA12_\mathsf{m} \subset \C^{{{\dalgo}} n}$ such that
  for all $(\x,\a)$ in $\oreg(\h) \times \ZOA12_\mathsf{m}$, if
  $g_\mathsf{m}(\x,\a) \ne 0$, the matrix $J_{{\dalgo}}(\x,\a)$ has rank
  at least $c+{{\dalgo}}-1$.
\end{enumerate}

\begin{lemma}\label{lemma:rankJ}
  Let ${{\dalgo}}$ be in $\{2,\dots,d\}$; suppose that
  ${\sf a}_{{{\dalgo}}-1}$ holds, and that ${\sf R}_\mathsf{m}$ holds
  for all $\mathsf{m} \in \mathsf{Sub}_{{\dalgo}}$.  Then
  ${\sf a}_{{{\dalgo}}}$ holds.
\end{lemma}
\begin{proof}
  Under the assumptions of the lemma, we define $\ZOA12_{{\dalgo}}$ as
  the intersection of $\ZOA12_{{{\dalgo}}-1} \times \C^n \subset
  \C^{{{\dalgo}} n}$ (which is well-defined, since ${\sf
    a}_{{{\dalgo}}-1}$ holds)  with all $\ZOA12_\mathsf{m}$, for
  $\mathsf{m} \in \mathsf{Sub}_{{\dalgo}}$; this is still a non-empty
  Zariski open subset of $\C^{{{\dalgo}} n}$.

  Let us prove that this choice satisfies our constraints. We take
  $(\x,\a)$ in $\oreg(\h) \times \ZOA12_{{\dalgo}}$, and we prove that the
  matrix $J_{{\dalgo}}(\x,\a)$ has rank at least $c+{{\dalgo}}-1$.

  Let $\a'$ be the projection of $\a$ in $\C^{({{\dalgo}}-1) n}$.
  Because $\x$ is in $\oreg(\h)$, and because by construction $\a'$ is
  in $\ZOA12_{{{\dalgo}}-1}$, we know by the induction assumption that the
  matrix $J_{{{\dalgo}}-1}(\x,\a')$ has rank at least $c+{{\dalgo}}-2$.  Since (by
  assumption) $\jac_\x(\h)$ has full rank $c$, this implies that there
  exists a non-zero minor of size $c+{{\dalgo}}-2$ of $J_{{{\dalgo}}-1}(\x,\a')$,
  that contains the first $c$ rows. In other words, there exists
  $\mathsf{m}$ in $\mathsf{Sub}_{{{\dalgo}}}$ such that
  $g_\mathsf{m}(\x,\a)\ne 0$.

  Because $\a$ is in $\ZOA12_\mathsf{m}$, we deduce that
  $J_{{\dalgo}}(\x,\a)$ has rank at least $c+{{\dalgo}}-1$, concluding the proof.
\end{proof}

Recall that we already established that the statement ${\sf a}_1$ of
Proposition~\ref{prop:rankJ} holds for ${\dalgo}=1$. Thus, in order to
prove Proposition~\ref{prop:rankJ} (by induction on ${\dalgo}$), it
suffices to establish the following lemma.

\begin{lemma}
  For ${{\dalgo}}$ in $\{2,\dots,\lfloor (d+3)/2\rfloor\}$ and $\mathsf{m}$
  in $\mathsf{Sub}_{{\dalgo}}$, ${\sf R}_\mathsf{m}$ holds.
\end{lemma}
\begin{proof}
  Let ${{\dalgo}}$ and
  $\mathsf{m}=(\mathsf{m}_\text{row},\mathsf{m}_\text{col}) \in
  \mathsf{Sub}_{{\dalgo}}$
  be fixed. We let $i_1,i_2$ in $\{c+1,\dots,c+{{\dalgo}}\}$ be the
  two row indices not in $\mathsf{m}_\text{row}$ and
  $j_1,\dots,j_{d-{{\dalgo}}+2}$ be the column indices not in
  $\mathsf{m}_\text{col}$.

  Let us split the indeterminates $\mA_{\le {{\dalgo}}}$ into $\mA'$ and
  $\mA''$, where $\mA''$ contains the $2(d-{{\dalgo}}+2)$ variables
  $$A_{i_1,j_1},\dots,A_{i_1,j_{d-{{\dalgo}}+2}} \quad\text{and}\quad
  A_{i_2,j_1},\dots,A_{i_2,j_{d-{{\dalgo}}+2}}$$ and $\mA'$ contains all
  other ones, arranged in any order. Note in particular that the
  determinant $g_\mathsf{m}$ belongs to $\C[\X,\mA']$.
  Accordingly, any $\a \in \C^{{{\dalgo}} n}$ will be written as
  $\a=(\a',\a'')$, with $\a' \in \C^{{{\dalgo}} n - 2(d-{{\dalgo}}+2)}$ and $\a''
  \in \C^{2(d-{{\dalgo}}+2)}$.
  
  For $u \in \{1,2\}$ and $v \in \{1,\dots,d-{{\dalgo}}+2\}$, let us
  consider the $(c+{\dalgo}-1)$-minor $g_{u,v} \in \C[\X,\mA_{\le {{\dalgo}}}]$ of
  $J_{{\dalgo}}$ obtained by selecting all rows / columns from $\mathsf{m}$,
  as well as the one indexed by $(i_u,j_v)$, which corresponds to the
  position of the variable $A_{i_u,j_v}$ in $J_{{\dalgo}}$. There are
  $2(d-{{\dalgo}}+2)$ such minors, one for each variable in $\mA''$, and
  they can be written as $g_{u,v} = A_{i_u,j_v} g_{\mathsf{m}} +
  h_{u,v}$, with $h_{u,v} \in \C[\X,\mA']$.

  Introduce a new variable $T$ and consider the algebraic set $\alg2Z
  \subset \C^{n+ {{\dalgo}} n +1}$ defined by
  $$\alg2Z=V(h_1,\dots,h_c,\ g_{1,1},\dots, g_{2,d-{{\dalgo}}+2},\ g_\mathsf{m}T-1).$$
  The Jacobian matrix of these equations with respect to the variables
  $\X,\mA',\mA'',T$ is
  $$ \begin{bmatrix} 
    ~\jac(\h)~ & ~0~ & ~0~ & ~0~ \\
      \star & \star & \mathbf{D} & 0 \\
      \star & \star & \star & g_{\mathsf{m}}
  \end{bmatrix},$$
  where $\mathbf{D}$ is a diagonal matrix of size $2(d-{{\dalgo}}+2)$ having
  $g_\mathsf{m}$ on the diagonal. Thus, this Jacobian matrix has
  full rank $c+2(d-{{\dalgo}}+2)+1$ at every point of $\alg2Z$ (note that
  $g_\mathsf{m}(\x,\a) \ne 0$ implies that $\jac_\x(\h)$ has full rank
  $c$). 

  Next, we prove that $\alg2Z$ is not empty. Indeed, since we assume that
  $\mathsf{m}$ is in $\mathsf{Sub}_{{\dalgo}}$, there exists $(\x,\a)\in
  \oreg(\h)\times \C^{{{\dalgo}} n}$ such that $g_\mathsf{m}(\x,\a) \ne 0$.
  Write $\a=(\a',\a'')$.  Because $g_\mathsf{m}$ belongs to
  $\C[\X,\mA']$, we can change the values of $\a''$ without affecting
  the fact that $g_\mathsf{m}(\x,\a) \ne 0$. Since we have seen that
  the polynomials $g_{u,v}$ have the form $g_{u,v} = A_{i_u,j_v}
  g_{\mathsf{m}} + h_{u,v}$, with $h_{u,v} \in \C[\X,\mA']$, it is
  thus always possible to find suitable values for the variables
  $\mA''$ that ensure that $g_{u,v}(\a)=0$ for all $u,v$.  To
  summarize, $\alg2Z$ is not empty, and thus by the Jacobian
  criterion, it is equidimensional of
  dimension $d+ {{\dalgo}} n-2(d-{{\dalgo}}+2)$.

  Let $\alg2Z'$ be the Zariski closure of the projection of $\alg2Z$ on $\C^{n+
    {{\dalgo}} n}$ obtained by forgetting the coordinate $T$. Note that the
  restriction of the projection $\alg2Z\to \alg2Z'$ is birational; we deduce
  that $\alg2Z'$ is still equidimensional of dimension $d+ {{\dalgo}} n
  -2(d-{{\dalgo}}+2)$. Finally, let $\alg2Z''$ be the Zariski closure of the
  projection of $\alg2Z'$ on $\C^{{{\dalgo}} n}$ obtained by forgetting the
  coordinates $\X$; thus, $\alg2Z''$ has dimension at most $d+ {{\dalgo}}
  n-2(d-{{\dalgo}}+2)$. This implies that $\alg2Z''$ is a strict Zariski closed
  subset of $\C^{{{\dalgo}} n}$. Indeed, our assumption $2{{\dalgo}} \le d +3$
  implies that $d+ {{\dalgo}} n -2(d-{{\dalgo}}+2) < {{\dalgo}} n$.

  Let us take $\ZOA12_\mathsf{m}$ as the complementary of $\alg2Z''$ in
  $\C^{{{\dalgo}} n}$. To conclude, we prove that for all $(\x,\a)$ in
  $\oreg(\h) \times \ZOA12_\mathsf{m}$, if $g_\mathsf{m}(\x,\a) \ne
  0$, the matrix $J_{{\dalgo}}(\x,\a)$ has rank at least $c+{{\dalgo}}-1$.
  Indeed, for $(\x,\a)$ in $\oreg(\h) \times \ZOA12_\mathsf{m}$, such
  that $g_\mathsf{m}(\x,\a) \ne 0$, we can define
  $t=1/g_\mathsf{m}(\x,\a)$. The point $(\x,\a,t)$ does not belong to
  $\alg2Z$ (otherwise $\a$ would be in $\alg2Z''$), which implies that
  $g_{u,v}(\x,\a) \ne 0$ for some index $(u,v)$. The claim follows.
\end{proof}

%%%%%%%%%%%%%%%%%%%%%%%%%%%%%%%%%%%%%%%%%%%%%%%%%%%%%%%%%%%%

\subsubsection{Proof of Proposition~\ref{sec:prelim:prop:main1}}

As above, we consider polynomials $\h=(h_1,\dots,h_c)$
in $\C[X_1,\dots,X_n]$, with $1 \le c \le n$ and we let $d=n-c$.
Recall what we have to prove: for ${\dalgo} \in \{1,\dots,d\}$, there
exists a non-empty Zariski open subset
$\ZOlp(\h,{{\dalgo}})\subset \GL(n)$, such that for $\mA$ in
$\ZOlp(\h,{{\dalgo}})$, the following holds:
\begin{itemize}
\item[(1)] for all $\x$ in $\oreg(\h^\mA)$, there exists a $c$-minor 
  $\minor$ of $\jac(\h^\mA)$ such that $\minor(\x)\ne 0$;
\smallskip
\item[(2)] every irreducible component of the Zariski closure
  of $\openpolar(0,{{\dalgo}}, \oreg(\h^\mA))$ has dimension ${\dalgo}-1$; 
\smallskip
\item[(3)] if ${\dalgo}\leq (d+3)/2$ then for all $\x$ in $\oreg(\h^\mA)$, there exists a
  $(c-1)$-minor $\pminor$ of $\jac(\h^\mA,{{\dalgo}})$ such that $\pminor(\x)\ne 0$;
\smallskip
\item[(4)] for every $c$-minor $\minor$ of the Jacobian matrix
  $\jac(\h^\mA)$ and for every $(c-1)$-minor $\pminor$ of the
  truncated Jacobian matrix $\jac(\h^\mA,{{\dalgo}})$, the polynomials
  $(\h^\mA,\sfH(\h^\mA,{{\dalgo}},\pminor))$ (see
  Definition~\ref{def:mH}) define $\openpolar(0,{{\dalgo}}, \oreg(\h^\mA))$ in
  $\Open(\minor \pminor)$, and their Jacobian matrix has full
  rank $n-({\dalgo}-1)$ at all points of $\Open(\minor \pminor)
  \cap \openpolar(0,{{\dalgo}}, \oreg(\h^\mA))$.
\end{itemize}
For ${\dalgo}$ as above, consider the polynomial mapping
$$\begin{array}{cccc}
\Phi:&\C^{n+c+{\dalgo}+{\dalgo}n} & \to & \C^{c+n} \\
& (\x,\lambda,\vartheta,\a) & \mapsto & \left (\h(\x),
[\lambda_1\ \cdots\ \lambda_c\ \vartheta_1\ \cdots\ \vartheta_{{\dalgo}}]\cdot \begin{bmatrix}
  & \jac_\x(\h) & \\
a_{1,1} & \cdots & a_{1,n} \\
\vdots & & \vdots \\
a_{{{\dalgo}},1} & \cdots & a_{{{\dalgo}},n} 
\end{bmatrix}\right );
\end{array}$$
note that the matrix involved is none other than $J_{{\dalgo}}$.
For $\a$ in $\C^{{{\dalgo}} n}$, we denote by $\Phi_\a$ the induced mapping
$\C^{n+c+{{\dalgo}}} \to \C^{c+n}$ defined by
$\Phi_\a(\x,\lambda,\vartheta)=\Phi(\x,\lambda,\vartheta,\a)$.

\begin{lemma}\label{prop:RK}
  Let $\ZOpropRK \subset \C^{n+c+{{\dalgo}}}$ be the open set defined by the
  rank conditions $\rank(\jac_\x(\h))=c$ and $\lambda\ne (0,\dots,0)$.
  There exists a non-empty Zariski open subset $\ZOdeltapropRK_{{\dalgo}}$ of
  $\C^{{{\dalgo}} n}$ such that for all $\a$ in $\ZOdeltapropRK_{{\dalgo}}$,
  $\a$ has rank ${{\dalgo}}$ and for $(\x,\lambda,\vartheta)$ in $\ZOpropRK
  \cap \Phi_\a^{-1}(0)$, the Jacobian matrix
  $\jac_{(\x,\lambda,\vartheta)} \Phi_\a$ has full rank $c+n$.
\end{lemma}
\begin{proof}
 In Section~3.2 of~\cite{BGHSS}, the following fact is proved: for
 {\em any} $(\x,\lambda,\vartheta,\a)$ in $\ZOpropRK$, the Jacobian matrix
 $\jac_{(\x,\lambda,\vartheta,\a)} \Phi$ has full rank $c+n$. This is
 in particular true for $(\x,\lambda,\vartheta,\a)$ in $\Phi^{-1}(0)$,
 so applying the weak transversality theorem
 (Proposition~\ref{thm:transversality}) to $\Phi$ on $\ZOpropRK \times
 \C^{{{\dalgo}} n}$ shows the existence of a non-empty Zariski open
 subset $\ZOdeltapropRK_{{\dalgo}}$ of $\C^{{{\dalgo}} n}$ such that for all
 $\a$ in $\ZOdeltapropRK_{{\dalgo}}$, and for $(\x,\lambda,\vartheta)$ in $\ZOpropRK
 \cap \Phi_\a^{-1}(0)$, the Jacobian matrix
 $\jac_{(\x,\lambda,\vartheta)}(\Phi_\a)$ has full rank $c+n$. 
 Upon restricting  $\ZOdeltapropRK_{{\dalgo}}$, we may in addition assume that 
 for all such $\a$, $\rank(\a)=\dalgo$.
\end{proof}

Let $\ZOdeltapropRK_{{\dalgo}} \subset \C^{{{\dalgo}} n}$ be as in
Lemma~\ref{prop:RK}.  When ${\dalgo}\leq (d+3)/2$, we let
$\ZOA12_{{\dalgo}} \subset \C^{{{\dalgo}} n}$ be as in
Proposition~\ref{prop:rankJ} else we set $\ZOA12_{{\dalgo}}\subset
\C^{{{\dalgo}} n}$ as the set of $\a$'s such that $\a$ has rank
${\dalgo}$. We consider the subset $\ZOlp(\h,{{\dalgo}}) \subset
\GL(n)$ of all invertible matrices $\mA$ such that the first
${{\dalgo}}$ rows of $\mA^{-1}$ are in $\ZOA12_{{\dalgo}} \cap
\ZOdeltapropRK_{{\dalgo}}$. This is a non-empty Zariski open subset of
$\GL(n)$. In what follows, we take $\mA$ in $\ZOlp(\h,{{\dalgo}})$,
and we prove that the conclusions of the proposition hold.  We will in
particular let $\b \in \C^{{{\dalgo}} n}$ be defined by taking the
first ${{\dalgo}}$ rows of $\mA^{-1}$; thus, $\b$ is in
$\ZOA12_{{\dalgo}}$ and~$\ZOdeltapropRK_{{\dalgo}}$.

Take first $\x$ in $\oreg(\h^\mA)$. The first point is clear, by 
definition of $\oreg(\h^\mA)$.
%%   The matrix identity $\jac(\h^\mA) = \jac(\h)^\mA \mA$ implies
%% that $\oreg(\h)^\mA=\oreg(\h^\mA)$, so that $\rank(\jac_\x(\h^\mA))
%% = c$. This proves the first point.  We prove now the second point.
Consider next the matrix identity
 $\jac(\h^\mA) = \jac(\h)^\mA \mA$. A first consequence of it 
is that  $\oreg(\h)^\mA=\oreg(\h^\mA)$. It implies further that
\begin{equation}\label{eq:hA}
 \left [ \begin{array}{cc} \multicolumn{2}{c}{~\jac(\h^\mA)~} \\  
{\bf 1}_{{{\dalgo}}} & {\bf 0}
  \end{array} \right ]=
\left [ \begin{array}{c} \jac(\h)^\mA \\ \b \end{array}\right ] \mA =
J_{{\dalgo}}(\mA \X,\b) \mA.
\end{equation}
Let
$\lcsw=\{\x\in \oreg(\h)\mid {\rm rank}(J_{{\dalgo}}(\x, \b))\leq
c+{\dalgo}-1\}$.
By Lemma~\ref{lemma:critequi} and the above identity, we deduce that
$\openpolar(0,{\dalgo}, \oreg(\h^\mA))=\lcsw^\mA$. The following lemma will
allow us to estimate the dimension of $\lcsw$, and thus of
$\openpolar(0,{\dalgo}, \oreg(\h^\mA))$.

\begin{lemma}
  Let $\ZOpropRK$ be as in Lemma~\ref{prop:RK}. Then $\lcsw$ is the projection
  of $\ZOpropRK \cap \Phi_\b^{-1}(0)$ on the $\X$-space.
\end{lemma}
\begin{proof}
  A point $\x \in \oreg(\h)$ belongs to $\lcsw$ if and only if
  $J_{{\dalgo}}(\x, \b)$ has rank less than $c+\dalgo$, that is, if
  and only if there exists a nonzero vector
  $[\lambda_1\ \cdots\ \lambda_c\ \vartheta_1\ \cdots\
  \vartheta_{{\dalgo}}]$
  in the right nullspace of $J_{{\dalgo}}(\x, \b)$ (recall that this
  matrix has more columns than rows).
  For any such
  $[\lambda_1\ \cdots\ \lambda_c\ \vartheta_1\ \cdots\ \vartheta_{{\dalgo}}]$,
  $\lambda_1,\dots,\lambda_c$ cannot be all zero, since then this would imply 
  that $\b$ has rank less than $\dalgo$.
\end{proof}

Using the Jacobian criterion in the form of
Lemma~\ref{lemma:prelim:locallyclosed}, together with
Lemma~\ref{prop:RK}, we deduce that $\ZOpropRK \cap \Phi_\b^{-1}(0)$
is either empty or a non-singular $\dalgo$-equidimensional locally
closed set.

We can now prove the second point of
Proposition~\ref{sec:prelim:prop:main1}.  If
$\ZOpropRK \cap \Phi_\b^{-1}(0)$ is empty, its projection $\lcsw$ is
empty as well, and so is $\openpolar(0,{\dalgo},
\oreg(\h^\mA))$.
Otherwise, we saw in Lemma~\ref{lemma:bound:codim} that each
irreducible component of $\lcsw$ has dimension at least $\dalgo-1$, so
the following lemma is sufficient to conclude. In this lemma, we
denote by $\pi_\X$ the projection on the $\X$-space.

\begin{lemma}
  The locally closed set $\lcsw$ has dimension at most $\dalgo-1$.
\end{lemma}
\begin{proof}
  We saw that the Zariski closure $C$ of $\ZOpropRK \cap \Phi_\b^{-1}(0)$ is a
  $\dalgo$-equidimensional algebraic set. Let us write
  $C=\cup_{i\in I} C_i$, with all $C_i$ irreducible of dimension
  $\dalgo$.

  For $i$ in $I$, let $T_i$ be the Zariski closure of $\pi_\X(C_i)$,
  so that the projection $C_i \to T_i$ is a dominant mapping between
  irreducible varieties. The set $\lcsw$ is contained in the union of the
  $T_i$'s, so it is enough to prove that $\dim(T_i) \le \dalgo-1$
  holds for all $i$.

  Remark first that for all $i$, $\lcsw \cap T_i$ is dense in $T_i$.
  Indeed, define $C'_i = \ZOpropRK \cap \Phi_\b^{-1}(0) \cap C_i$; by
  construction, this is a dense subset of $C_i$, so that $T_i$ is also
  the Zariski closure of $\pi_\X(C'_i)$. On the other hand,
  $\pi_\X(C'_i)$ is contained in $\lcsw$, and thus in
  $\lcsw \cap T_i$, and we just saw that it is dense in $T_i$. Thus
  $\lcsw \cap T_i$ itself is dense in $T_i$.

  Fix $i$ such that $\dim(T_i)$ is maximal, and let $J \subset I$ be
  the set of all indices $j \in I$ such that $T_i=T_j$; thus, for $j$
  not in $J$, $T_i \cap T_j$ is a proper subvariety of $T_i$. This
  allows us to define a non-empty open set $\Omega \subset T_i$ such
  that for $y$ in $\Omega$, the following properties are satisfied:
  \begin{itemize}
  \item for all $j$ in $J$, for any irreducible component $F$ of
    $\pi_\X^{-1}(y) \cap C_j$, $F$ has dimension $\dalgo-\dim(T_i)$
    (this is by the theorem on the dimension of fibers for the
    projection $C_j \to T_j=T_i$);
\smallskip
  \item for all $j$ not in $J$, $\pi_\X^{-1}(y) \cap C_j$ is empty;
\smallskip
  \item $y$ is in $\lcsw$.
  \end{itemize}
  Take such a $y$. Then, $\pi_\X^{-1}(y) \cap C$ is the union of the
  sets $\pi_\X^{-1}(y) \cap C_j$, for $j$ in $J$, so it is an
  equidimensional algebraic set of dimension $\dalgo-\dim(T_i)$.

  On the other hand, $\pi_\X^{-1}(y) \cap \ZOpropRK \cap \Phi_\b^{-1}(0)$ has
  positive dimension, since it is defined by a homogeneous system (and
  does not consist only on the trivial solution $[0 \cdots 0]$). Since
  this set is contained in $\pi_\X^{-1}(y) \cap C$, the latter must
  have dimension at least one. Altogether, this implies that
  $\dim(T_i) \le \dalgo -1$, which implies that
  $\dim(\lcsw) \le \dalgo -1$.
\end{proof}

We prove now the third point, taking $\x$ in $\oreg(\h^\mA)$ and $\y=\mA\x$,
so that $\y \in \oreg(\h)$.  Because we assume that ${\dalgo}\le
(d+3)/2$ and that $\b$ is in $\ZOA12_{{\dalgo}}$, we deduce from
Proposition~\ref{prop:rankJ} that $J_{{\dalgo}}(\y,\b)$ has rank at
least $c+{{\dalgo}}-1$. Because $\mA$ is a unit, the matrix
equality~\eqref{eq:hA} implies that $\jac(\h^\mA,{{\dalgo}})$ has rank
at least $c-1$ at $\x$, and the third claim follows.

Only the last point is left to prove. Take $\minor$ and $\pminor$ as in the
proposition, respectively a $c$-minor of $\jac(\h^\mA)$ and a
$(c-1)$-minor of $\jac(\h^\mA,{{\dalgo}})$; without loss of generality, we
can assume that $\pminor \ne 0$. Let further $\iota$ be the index of the
row of $\jac(\h^\mA,{{\dalgo}})$ not in $\pminor$.

By Lemma~\ref{lemma:critequi}, we know that
  $$\openpolar(0,{{\dalgo}},\oreg(\h^\mA)) = \{ \x \in \oreg(\h^\mA) \ | \ \rank (\jac_\x(\h^\mA))=c
  \ \text{and}\ \rank (\jac_\x(\h^\mA,{{\dalgo}})) < c \}.$$
  Inside $\Open(\minor)$, $\oreg(\h^\mA)$ coincides with
  $V(\h^\mA)$. As a consequence, inside $\Open(\minor)$,
  $\openpolar(0,{{\dalgo}},\oreg(\h^\mA))$ coincides with the set of
  all $\x$ in $V(\h^\mA)$ such that all $c$-minors of
  $\jac(\h^\mA,{{\dalgo}})$ vanish at $\x$. Restricting further, we
  deduce from the exchange lemma of e.g.~\cite[Lemma~4]{BaGiHeMb01}
  that inside $\Open(\minor \pminor)$,
  $\openpolar(0,{{\dalgo}},\oreg(\h^\mA))$ coincides with
  $V(\h^\mA, \sfH(\h^\mA,{{\dalgo}}, \pminor))$, for the polynomials
  $\sfH(\h^\mA,{{\dalgo}}, \pminor)$ introduced in
  Definition~\ref{def:mH}.  Thus, it remains to prove that for all
  $\x$ in
  $V(\h^\mA, \sfH(\h^\mA,{{\dalgo}}, \pminor)) \cap \Open(\minor
  \pminor)$,
  the Jacobian matrix of $(\h^\mA, \sfH(\h^\mA,{{\dalgo}}, \pminor))$
  has full rank, equal to $n-{{\dalgo}}+1$.  (This will in particular
  reprove the second item in our
  proposition~\ref{sec:prelim:prop:main1}, but only in the open set
  $\Open(\minor \pminor)$.)

Let $L_1,\cdots,L_c$ and $T_1,\dots,T_{{\dalgo}}$ be new variables. We
deduce from~\eqref{eq:hA} that the ideal generated by the entries of
the vector
\begin{equation*}
[L_1 ~\cdots~ L_c\ T_1 ~\cdots~ T_{{\dalgo}}]\cdot
\left [ \begin{array}{cc} \multicolumn{2}{c}{~\jac(\h^\mA)~} \\  
{\bf 1}_{{{\dalgo}}} & 0
  \end{array} \right ]  
\end{equation*}
also admits for generators the entries of 
\begin{equation*}
[L_1 ~\cdots~ L_c\ T_1 ~\cdots~ T_{{\dalgo}}]\cdot
\left [ \begin{array}{c} \jac(\h) \\ \b \end{array}\right ]^\mA.  
\end{equation*}

Looking at the first equation above, and using
Proposition~\ref{lemma:linearsolve}, we deduce that there exist
$(\rho_j)_{j=1,\dots,c,j \ne \iota}$ and
$(\tau_i)_{i=1,\dots,{{\dalgo}}}$ in $\C[\X]_{\pminor}$ such that in
$\C[\X,\L,\mathbf{T}]_{\pminor}$, the ideal generated by the entries
of
\begin{equation*}
\h^\mA,\ 
[L_1 ~\cdots~ L_c\ T_1 ~\cdots~ T_{{\dalgo}}]\cdot
\left [ \begin{array}{cc} \multicolumn{2}{c}{~\jac(\h^\mA)~} \\  
{\bf 1}_{{{\dalgo}}} & 0
  \end{array} \right ]  
\end{equation*}
admits for generators polynomials of the form
\begin{equation}\label{eq:hHrhotau}
\h^\mA,\ L_\iota \sfH(\h^\mA,{{\dalgo}}, \pminor),\ (L_j  -\rho_j L_\iota)_{j=1,\dots,c, j \ne \iota},\ (T_i-\tau_i L_\iota)_{i=1,\dots,{{\dalgo}}}.  
\end{equation}
On the other hand, we also observe that 
\begin{equation*}
\h^\mA,\ [L_1~\cdots~L_c~T_1~\cdots~T_{{\dalgo}}]\cdot
\left [ \begin{array}{c} \jac(\h) \\ \b \end{array}\right ]^\mA
\end{equation*}
coincide with the entries of the polynomial vector $\Phi_\b^{\mA}$,
where $\Phi:\C^{n+c+{{\dalgo}}+{{\dalgo}} n} \to \C^{c+n}$ is the
polynomial mapping defined at the beginning of this paragraph, and
where the superscript ${}^\mA$ indicates that $\mA$ acts on the
variables $\X$.

Now, let $\x$ be in
$V(\h^\mA, \sfH(\h^\mA,{{\dalgo}}, \pminor)) \cap \Open(\minor
\pminor)$.
Define first $\lambda_\iota=1$, then $\lambda_j = \rho_j(\x)$ for
$j=1,\dots,c,j\ne \iota$ and $\vartheta_i=\tau_i(\x)$ for
$i=1,\dots,{{\dalgo}}$; these are all well-defined, since
$\pminor(\x)\ne 0$. It follows that $(\x,\lambda,\vartheta)$ cancels
all equations in~\eqref{eq:hHrhotau}.  Let $\y=\mA \x$. The previous
statements show that $(\y,\lambda,\vartheta)$ is in
$\Phi_\b^{-1}(0)$. Now, recall that $\b$ is in
$\ZOdeltapropRK_{{\dalgo}}$; besides, since $\minor(\x)\ne 0$, $\x$ is
in $\oreg(\h^\mA)$ and thus $\y$ is in $\oreg(\h)$. Since also
$\lambda \ne 0$, Lemma~\ref{prop:RK} implies that
$\jac_{\y,\lambda,\vartheta}(\Phi_\b)$ has full rank $c+n$ at
$(\y,\lambda,\vartheta)$.

Through the change of variables $\mA$, this implies that the Jacobian
of $\Phi_\b^{\mA}$ has full rank $c+n$ at $(\x,\lambda,\vartheta)$,
and this in turn implies the same property for the Jacobian of
$$\h^\mA,\ L_\iota \sfH(\h^\mA,{{\dalgo}}, \pminor),\ (L_j -\rho_j
L_\iota)_{j=1,\dots,c, j \ne \iota},\ (T_i-\tau_i
L_\iota)_{i=1,\dots,{{\dalgo}}}.$$ This finally implies that the Jacobian
matrix of $(\h^\mA, \sfH(\h^\mA,{{\dalgo}}, \pminor))$ has full rank $n-{{\dalgo}}+1$ at
$\x$, so the proof is complete.

%%%%%%%%%%%%%%%%%%%%%%%%%%%%%%%%%%%%%%%%%%%%%%%%%%%%%%%%%%%%
%%%%%%%%%%%%%%%%%%%%%%%%%%%%%%%%%%%%%%%%%%%%%%%%%%%%%%%%%%%%
%%%%%%%%%%%%%%%%%%%%%%%%%%%%%%%%%%%%%%%%%%%%%%%%%%%%%%%%%%%%

\subsection{Charts and atlases for polar varieties} \label{ssec:app:atlases:polar}

We can now prove that if $\psi$ is a chart for a triple $(V,Q,S)$, the
construction $\chartpolar(\psi,m',m'')$ of
Definition~\ref{sec:chart:notation:polar} does indeed define a chart
for $\polar(e,{\dalgo},V)$, at least in generic coordinates and for
some suitable values of ${\dalgo}$.

\begin{lemma}\label{sec:chart:lemma:polarchart}
  Let $Q \subset \C^e$ be a finite set and let $V \subset \C^n$ and
  $S\subset \C^n$ be algebraic sets lying over $Q$.
  Suppose that $V$ is $d$-equidimensional, let $\psi=(m,\h)$ be a
  chart of $(V,Q,S)$, and let ${\dalgo}$ be an integer in
  $\{1,\dots,d\}$.

  There exists a non-empty Zariski open $\scrGpolarchart(\psi,V,Q,S,{\dalgo})
  \subset \GL(n,e)$ such that, for $\mA$ in
  $\scrGpolarchart(\psi,V,Q,S,{\dalgo})$, the following holds, where we write
  $W=\polar(e,{\dalgo},V^\mA)$.
  \begin{itemize}
  \item For any minors $m'$ and $m''$ of $\jac(\h^\mA)$ as in
    Definition~\ref{sec:chart:notation:polar}, writing
    $\chartpolar(\psi^\mA,m',m'')=(m^\mA m' m'',\h')$, the set
    $\Open(m^\mA m' m'') \cap W-S^\mA$ coincides with
    $\Open(m^\mA m' m'') \cap \fbr(V(\h'),Q)-S^\mA$.
\smallskip
  \item For $m',m''$ as above, if $\Open(m^\mA m' m'') \cap
    W-S^\mA\neq \emptyset$, then $\chartpolar(\psi^\mA,m',m'')$ is a chart of
    $(W,Q,S^\mA)$.
\end{itemize}
Moreover, when we additionally assume that ${\dalgo}\le (d+3)/2$, the
following holds for $\mA$ in $\scrGpolarchart(\psi,V,Q,S,{\dalgo})$.
  \begin{itemize}
  \item The sets $\Open(m^\mA m' m'') -S^\mA$, taken for all
    $m',m''$, cover $\Open(m^\mA)\cap V^\mA-S^\mA$.
\smallskip
  \item The sets $\Open(m^\mA m' m'')-S^\mA$, taken for all
    $m',m''$ such that $\Open(m^\mA m' m'') \cap W-S^\mA$ is not
    empty, cover $\Open(m^\mA)\cap W-S^\mA$.
  \end{itemize}
\end{lemma}
\begin{proof}
  For $\y=(x_1,\dots,x_e)$ in $Q$, let $\h_\y$ be the polynomials
  $$\h(x_1,\dots,x_e,X_{e+1},\dots,X_n),$$ which are in
  $\C[X_{e+1},\dots,X_n]$; more generally, for any
  $f \in \C[X_1,\dots,X_n]$, $f_\y$ will be defined in this manner.
  Let further $\ZOlp_\y$ be the non-empty Zariski
  open subset of $\GL(n-e)$ obtained by applying
  Proposition~\ref{sec:prelim:prop:main1} to $\h_\y$: this is valid,
  since, by assumption ${\dalgo}\leq d$ and, by
  Lemma~\ref{sec:lemma:singS}, $\h_\y$ involves $n-e-d$ equations in
  $n-e$ variables, so the assumptions of that proposition are
  satisfied.

  Let ${\mathscr{G}_\y} \subset \GL(n, e)$ be obtained by taking the
  direct sum of the identity matrix of size $e$ with the elements of
  $\ZOlp_\y$, and let finally $\scrGpolarchart(\psi,V,Q,S,{\dalgo})$
  be the intersection of the finitely many ${\mathscr{G}_\y}$'s.  This
  is a non-empty Zariski open subset of $\GL(n,e)$.  We now take $\mA$
  in $\scrGpolarchart(\psi,V,Q,S,{\dalgo})$, we let $\mA'\in \GL(n-e)$ be
  its second summand, and we prove that the claims of the proposition
  hold.

  Because $\mA$ is block-diagonal and leaves the first $e$ variables
  invariant, for any polynomial $h$ and for any $\y$ in $Q$, we have
  $(h_\y)^{\mA'}=(h^\mA)_\y$; we simply write it
  $h^\mA_\y$. Geometrically, we define the algebraic sets $V^\mA_\y
  \subset \C^n$ (by restricting the points in $V^\mA$ to those lying
  over $\y$) and ${V'_\y}^\mA \subset \C^{n-e}$ (by forgetting the
  first $e$ coordinates from $V^\mA_\y$), and similarly the sets
  $S^\mA_\y \subset \C^n$ and ${S'_\y}^\mA \subset \C^{n-e}$.

  Let now $m',m''$ be minors of respectively $\jac(\h,e)$ and
  $\jac(\h,e+{\dalgo})$, and let $\h'=(\h,
  \sfH(\h,e+{\dalgo},m'))$.
  We first prove the following claim: {\em in the open set
    $\Open(m^\mA m' m'')-S^\mA$, $\fbr(V(\h'),Q)$ coincides with
    $\openpolar(e,{\dalgo},V^\mA)$ and at any of these points,
    $\jac(\h',e)$ has full rank $n-e-({\dalgo}-1)$.}

  Fix $\y$ in $Q$, so that $m'_\y$ and $m''_\y$ are minors of
  respectively the matrices $\jac(\h^\mA_\y)$ and
  $\jac(\h^\mA_\y,{\dalgo})$. The polynomials $\h'_\y$ are precisely
  the polynomials considered in point (4) of
  Proposition~\ref{sec:prelim:prop:main1}. Because $\mA'$ is in
  $\ZOlp_\y$, that proposition implies that the
  polynomials $\h'_\y$ define $\openpolar({\dalgo}, \oreg(\h^\mA_\y))$ in
  $\Open(m'_\y m''_\y)$, and that their Jacobian matrix has full
  rank $n-e-({\dalgo}-1)$ everywhere on $\Open(m'_\y m''_\y)\cap
  \openpolar(0, {\dalgo}, \oreg(\h^\mA_\y))$.

  Using ${\sfC_2}$ and ${\sfC_4}$ for $\psi^\mA$ and restricting to
  the fiber above $\y$, we deduce that in
  $\Open(m^\mA_\y)-{S'_\y}^\mA$, ${V'_\y}^\mA$ coincides with
  $\oreg(\h^\mA_\y)$, so in
  $\Open(m^\mA_\y m'_\y m''_\y)-{S'_\y}^\mA$, the polynomials $\h'_\y$
  define $\openpolar(0,{\dalgo},{V'_\y}^\mA)$ as well.  Transporting all
  objects back to $\C^n$, and taking the union over all $\y \in Q$, we
  obtain that in $\Open(m^\mA m' m'')-S^\mA$, $\fbr(V(\h'),Q)$
  is the disjoint union of all $\openpolar(e,{\dalgo},V^\mA_\y)$,
  which is none other than $\openpolar(e,{\dalgo},V^\mA)$. Besides, at any
  of these points, $\jac(\h',e)$ has full rank $n-e-({\dalgo}-1)$, so
  our claim is proved.

  We can now prove the first two items. As a preliminary, remark that
  the number of polynomials in $\chartpolar(\psi^\mA,m',m'')$ is
  $c'=n-e-(\dalgo-1)$; then, $c'+e = n-({\dalgo}-1)$, so the
  assumption ${\dalgo} \ge 1$ implies $c'+e \le n$, which will
  establish ${\sfC_3}$ below.

  Writing $W=\polar(e,{\dalgo},V^\mA)$, we saw in
  Subsection~\ref{ssec:A12} the inclusions
  $$\openpolar(e,{\dalgo},V^\mA)\ \subset\ W\ \subset\ \Kpolar(e,\dalgo,V^\mA)=
  \openpolar(e,{\dalgo},V^\mA) \cup \sing(V^\mA).$$
  Let us take the intersection with $\Open(m^\mA m'm'')-S^\mA$.
  Lemma~\ref{sec:lemma:singS} shows that $\Open(m^\mA)-S^\mA$ does not
  intersect $\sing(V^\mA)$, so we deduce that
  $\Open(m^\mA m'm'')\cap W-S^\mA= \Open(m^\mA m'm'')\cap
  \openpolar(e,{\dalgo},V^\mA)-S^\mA$,
  which is equal to $\Open(m^\mA m'm'')\cap \fbr(V(\h'),Q)-S^\mA$ in
  view of the claim above. This remark, and the rank property for
  $\jac(\h',e)$ mentioned just above, prove properties ${\sfC_2}$ and
  ${\sfC_4}$ for $\chartpolar(\psi^\mA,m',m'')$; if
  $\Open(m^\mA m' m'') \cap W-S^\mA$ is not empty, we also have
  ${\sfC_1}$, and ${\sfC_3}$ was proved above. Thus, we are done with
  the first two items in the lemma.
  
  The third point is easier. Take $\x=(x_1,\dots,x_n)$ in
  $\Open(m^\mA) \cap V^\mA-S^\mA$, so that $\y=(x_1,\dots,x_e)$ is in
  $Q$, and let $\z=(x_{e+1},\dots,x_n)$. Since $\x$ is in
  $\Open(m^\mA) \cap V^\mA-S^\mA$, by ${\sfC_4}$ for
  $\psi^\mA$, the matrix $\jac(\h^\mA,e)$ has full rank $c$ at $\x$;
  equivalently, the matrix $\jac_\z(\h^\mA_\y)$ has full rank $c$ at
  $\z$, so $\z$ is in $\oreg(\h^\mA_\y)$.

  Now, we assume additionally that ${\dalgo}\leq (d+3)/2$. Due to our
  choice of $\mA$, we can apply
  Proposition~\ref{sec:prelim:prop:main1}; we deduce from points (1)
  and (3) of that proposition that there exist minors
  $\polmu',\polmu''$ of $\jac(\h^\mA_\y)$ and
  $\jac(\h^\mA_\y,{\dalgo})$ that do not vanish at $\z$. Now, there
  exist minors $m'$ and $m''$ of $\jac(\h^\mA,e)$ and
  $\jac(\h^\mA,e+{\dalgo})$ such that $\polmu'=m'_\y$ and
  $\polmu''=m''_\y$.  In particular, we deduce that $m'(\x)$ and
  $m''(\x)$ are both non-zero, so $\x$ is actually in
  $\Open(m^\mA m' m'')-S^\mA$. The third item is proved.

  The fourth point is obvious. Take $\x=(x_1,\dots,x_n)$ in
  $\Open(m^\mA) \cap W-S^\mA$. Then, $\x$ is in
  $\Open(m^\mA) \cap V^\mA-S^\mA$, so, since ${\dalgo}\le (d+3)/2$ by
  assumption, there exists $m'$ and $m''$ as before such that $\x$ is
  in $\Open(m^\mA m' m'')-S^\mA$. In particular,
  $\Open(m^\mA m' m'')\cap W-S^\mA$ is not empty.
\end{proof}

\begin{lemma}\label{lemma:lowerbounddimW}
  Let $Q \subset \C^e$ be a finite set and let $V \subset \C^n$ and
  $S\subset \C^n$ be algebraic sets lying over $Q$, with $S$ finite.
  Suppose that $V$ is $d$-equidimensional and let ${\dalgo}$ be an
  integer in $\{1,\dots,d\}$.  Then all irreducible components of
  $\polar(e,\dalgo,V)$ have dimension at least $\dalgo-1$.
\end{lemma}
\begin{proof}
  Up to replacing $n$ by $n-e$ and $\polar(e,\dalgo,V)$ by
  $\polar(0,\dalgo,V)$, and to working over all points of $Q$
  independently, we can assume that $e=0$ (so as to allow us to use
  Lemma~\ref{lemma:bound:codim}, which was written in this
  context). Then, it is enough to prove that for any $\x$ in
  $\openpolar(0,\dalgo,V)$, any irreducible component of
  $\polar(0,\dalgo,V)$ passing through $\x$ has dimension at least
  $\dalgo-1$.

  Consider the atlas $\bpsi=(\psi_i)_{1 \le i \le s}$ of
  $(V,\{\bullet\},\sing(V))$ introduced in
  Lemma~\ref{sec:atlas:lemma:glob}, and write $\psi_i=(m_i,\h_i)$ for
  all $i$. We know from Lemma~\ref{sec:coro:lemma:singSX} that all
  $\h_i$ have cardinality $c=n-d$. Besides, there exists an index $i$
  such that $\x$ is in $\Open(m_i)-\sing(V)$, and in this open set, Lemma~\ref{sec:atlas:chartpolar} shows that
  $$\polar(0,{{\dalgo}},V) \quad\text{and}\quad \left \{ \x \in \oreg(\h_i)
  \ | \ \rank (\jac_\x(\h_i,\dalgo)) < \dalgo \right
  \}
$$ coincide. In particular, the irreducible components of
  $\polar(0,{{\dalgo}},V)$ containing $\x$ are also the irreducible
  components of the Zariski closure of the locally closed set on the
  right-hand side. Now, for $\x$ in $\Open(m_i)-\sing(V)$, the matrix $\jac_\x(\h_i,\dalgo)$ satisfies
  the following equality:
$$\rank\, \left [ \begin{array}{cc} \multicolumn{2}{c}{~\jac_\x(\h_i)~} \\  
{\bf 1}_{{{\dalgo}}} & {\bf 0}
  \end{array} \right ] = c+ \rank(\jac_\x(\h_i,\dalgo)) $$
so applying Lemma~\ref{lemma:bound:codim} finishes proof.
\end{proof}

\subsection{Proof of the proposition}
We can now prove Proposition~\ref{prop:ch4}.
  Write $\bpsi=(\psi_i)_{1 \le i \le s}$. To each $\psi_i$, we
  associate the non-empty Zariski open subset
  $\scrGpolarchart(\psi_i,V,Q,S,{\dalgo})$ of
  Lemma~\ref{sec:chart:lemma:polarchart}, and we let
  $\scrGpolar(\bpsi,V,Q,S,{\dalgo})$ be their intersection; it is
  still non-empty and Zariski open.

  Take $\mA$ in $\scrGpolar(\bpsi,V,Q,S,{\dalgo})$ and write
  $W=\polar(e,{\dalgo},V^\mA)$; assume that $W$ is not empty
  (otherwise, there is nothing to do). Then, by
  Lemma~\ref{lemma:lowerbounddimW}, all irreducible components of $W$
  have dimension at least $\dalgo-1 \ge 1$.
  Let us prove that $\atlaspolar(\bpsi^\mA,V^\mA,Q,S^\mA,{\dalgo})$ is
  an atlas of $W$.
  \begin{itemize}
  \item   For all minors $m'$ and $m''$ of $\jac(\h_i^\mA)$
  as in Definitions~\ref{sec:chart:notation:polar}
  and~\ref{sec:atlas:notation:polar}, the second item in
  Lemma~\ref{sec:chart:lemma:polarchart} shows that if $\Open(m_i^\mA
  m' m'') \cap W-S^\mA$ is not empty, $\chartpolar(\psi_i^\mA,m',m'')$
  is a chart of $(W,Q,S^\mA)$. Thus, we have proved ${\sfA_1}$.
\smallskip
\item We next prove ${\sfA_3}$, that is, that all corresponding $\Open(m_i^\mA
  m' m'')$ cover $W-S^\mA$. For any fixed $i$, the last item in
  Lemma~\ref{sec:chart:lemma:polarchart} shows that the sets $\Open(
  m_i^\mA m' m'')\cap W-S^\mA$ cover $\Open(m_i^\mA)\cap
  W-S^\mA$. Since the open sets $\Open(m_i^\mA)$ cover $V-S^\mA$, and
  thus $W-S^\mA$, our claim is proved.
\smallskip
\item  ${\sfA_2}$ follows from the fact that
  $W$ is not contained in $S^\mA$ (since $W$ have dimension at least $1$,
  and $S$ is finite).
  \end{itemize}
  Hence, $\atlaspolar(\bpsi^\mA,V^\mA,Q,S^\mA,{\dalgo})$ is an atlas
  of $W$. Lemma~\ref{sec:coro:lemma:singSX} shows that all sequences
  of polynomials appearing in the atlas $\bpsi$ have the same
  cardinality; this implies that all polynomial sequences appearing in
  $\atlaspolar(\bpsi^\mA,V^\mA,Q,S^\mA,{\dalgo})$ have the same
  cardinality as well. As a result, Lemma~\ref{sec:atlas:lemma:singSX}
  implies that $W-S^\mA$ is a non-singular
  $({\dalgo}-1)$-equidimen\-sional locally closed set.  Since all
  irreducible components of $W$ have dimension at least~$1$, $W$ is
  the Zariski closure of $W-S^\mA$. Thus, $W$ itself is
  $({\dalgo}-1)$-equidimensional, and singular points of
  $\polar(e,{\dalgo},V^\mA)$ are contained in $S^\mA$; in particular, they
  are in finite number. 
%\end{proof}

%%%%%%%%%%%%%%%%%%%%%%%%%%%%%%%%%%%%%%%%%%%%

%%%%%%%%%%%%%%%%%%%%%%%%%%%%%%%%%%%%%%%%%%%%
%\input{proof3.9}
\section{Proof of Proposition~\ref{sec:atlas:prop:summary1}}\label{sec:proof3.9}

The proof of Proposition~\ref{prop:ch5} uses
Proposition~\ref{sec:atlas:prop:summary1}; hence, we prove the latter
first. Its statement is as follows:  {\em Let $Q \subset \C^e$ be a
  finite set and let $V \subset \C^n$ and $S\subset \C^n$ be algebraic
  sets lying over $Q$, with $S$ finite.  Suppose that $V$ is
  equidimensional of dimension~$d$.  Let $\bpsi$ be an atlas of
  $(V,Q,S)$, and let ${\dalgo}$ be an integer in $\{1,\dots,d\}$.  If
  $2 \le {\dalgo}\leq (d+3)/2$, there exists a non-empty Zariski open
  subset $\scrGfiber(\bpsi,V,Q,S,{\dalgo})$ of $\GL(n,e)$ such that
  for $\mA$ in $\scrGfiber(\bpsi,V,Q,S,{\dalgo})$, the following
  holds.

  Define $W=\polar(e,{\dalgo},V^\mA)$ and let
  $\fiber2 \subset \C^{e+{\dalgo}-1}$ be a finite set lying over $Q$;
  define $\Vfiber=\fbr(V^\mA,\fiber2)$.  Let further
  $\fibersing2 = \fbr(S^\mA \cup \polar(e,{\dalgo},V^\mA),\fiber2)$.
  Then:
  \begin{itemize}
  \item $\fibersing2$ is finite,
\smallskip
\item either $\Vfiber$ is empty or
  $\atlasfiber(\bpsi^\mA,V^\mA,Q,S^\mA,\fiber2)$ is an atlas of
  $(\Vfiber,\fiber2,\fibersing2)$, and $\Vfiber$ is equidimensional of
  dimension $d-({\dalgo}-1)$, with $\sing(\Vfiber)$ contained in the
  finite set~$\fibersing2$.
  \end{itemize}
}

The outline of this section is similar to that of
Section~\ref{sec:chartsatlas}: we first work locally, showing how to
construct a chart for the set above, then handle global properties.

\begin{lemma}\label{lemma:chart:fiber}
  Let $Q \subset \C^e$ be a finite set and let $V \subset \C^n$ and
  $S\subset \C^n$ be algebraic sets lying over $Q$.  Suppose that
  $(V,Q)$ is equidimensional of dimension $d$, with finitely many
  singular points, let $\psi=(m,\h)$ be a chart of $(V,Q,S)$ and let
  ${\dalgo}$ be an integer in $\{1,\dots,d\}$.

  There exists a non-empty Zariski open
  $\scrGfiberchart(\psi,V,Q,S,{\dalgo}) \subset \GL(n,e)$ such that, for
  $\mA$ in $\scrGfiberchart(\psi,V,Q,S,{\dalgo})$, the following holds.

  Let $\fiber2 \subset \C^{e+{\dalgo}-1}$ be a finite set lying over
  $Q$ and define $\Vfiber=\fbr(V^\mA,\fiber2)$.  Let further
  $\fibersing2 = \fbr(S^\mA \cup \polar(e,{\dalgo},V^\mA),\fiber2)$. Then
  either $\Open(m^\mA)\cap \Vfiber-\fibersing2$ is empty or
  $\psi^\mA$ is a chart of $(\Vfiber,\fiber2,\fibersing2)$, and
  $\fibersing2$ is finite if $S$ is.
\end{lemma}
\begin{proof}
  For $\y$ in $Q$, let $V'_\y \subset \C^{n-e}$ be the algebraic set
  obtained by forgetting the first $e$ coordinates in $V_\y=\fbr(V,
  \y)$, let $\ZOffp_{\y}$ be the Zariski open set associated to $V'_\y$
  and ${\dalgo}$ by Lemma~\ref{sec:prelim:lemma:finitefiberpolar} and
  let $\scrGfiber'_{,\y}\subset \GL(n,e)$ be obtained as the direct
  sum of the size-$e$ identity matrix and $\ZOffp_\y$. Finally, we
  take for $\scrGfiberchart(\psi,V,Q,S,{\dalgo})$ the intersection of all
  $\scrGfiber'_{,\y}$, for $\y$ in $Q$.

  Take $\mA$ in $\scrGfiberchart(\psi,V,Q,S,{\dalgo})$, and let
  $\mA' \in \GL(n-e)$ be its second
  summand. Lemma~\ref{sec:prelim:lemma:finitefiberpolar} shows that
  for any $\y$ in $Q$ and $\x$ in $\C^{{\dalgo}-1}$,
  $\fbr(\polar(0,{\dalgo},{V'_\y}^{\mA}),\x)$ is finite. Transporting
  back to $\C^n$, this shows that for $\y$ in $Q$ and $\x$ in
  $\C^{e+{\dalgo}-1}$ lying over $\y$,
  $\fbr(\polar(e,{\dalgo},V_\y^\mA),\x)$ is finite. Considering all
  $\y\in Q$ at once, this implies that for any finite $\fiber2$ in
  $\C^{e+{\dalgo}-1}$ lying over $Q$,
  $\fbr(\polar(e,{\dalgo},V^\mA),Q)$ is finite.  So if we assume that
  $S$ is finite,
  $\fibersing2= \fbr(S^\mA \cup \polar(e,{\dalgo},V^\mA), \fiber2)$ is
  finite as well.

  We have thus proved the last claim. Let then
  $\Vfiber=\fbr(V^\mA,\fiber2)$ and assume that $\Open(m^\mA)\cap
  \Vfiber-\fibersing2$ is not empty; we can now establish the defining
  properties of a chart.
  \begin{itemize}
\item[${\sfC_1.}$] By assumption, $\Open(m^\mA)\cap
  \Vfiber-\fibersing2$ is not empty.
\smallskip
\item[${\sfC_2.}$] By construction, $\Open(m^\mA)\cap
  \Vfiber-\fibersing2 = \Open(m^\mA) \cap
  \fbr(V^\mA,\fiber2)-\fibersing2$, which is equal to $\Open(m^\mA)
  \cap V^\mA \cap
  \pi_{e+{\dalgo}-1}^{-1}(\fiber2)-\fibersing2$. Because $\psi^\mA$ is
  a chart of $(V^\mA,Q,S^\mA)$, and because $\fibersing2$ contains
  $S^\mA$, we can rewrite this as $\Open(m^\mA) \cap
  \fbr(V(\h^\mA),Q) \cap
  \pi_{e+{\dalgo}-1}^{-1}(\fiber2)-\fibersing2$, or equivalently as
  $\Open(m^\mA) \cap \fbr(V(\h^\mA),\fiber2)-\fibersing2$, since
  $\fiber2$ lies over $Q$. Thus, ${\sfC}_2$ is proved.
\smallskip
\item[${\sfC_3.}$] We have to prove that $c+e+{\dalgo}-1 \le n$. By
  assumption on ${\dalgo}$, we have $c+e+{\dalgo}-1 \le c+e+d-1$, and by
  Lemma~\ref{sec:lemma:singS}, $d=n-e-c$, so that $c+e+{\dalgo}-1
  \le n-1$, which is stronger than what we need.
\smallskip
\item[${\sfC_4.}$] Finally, we have to prove that for all $\x$ in
  $\Open(m^\mA)\cap \Vfiber-\fibersing2$, the Jacobian matrix
  $\jac(\h^\mA,e+{\dalgo}-1)$ has full rank $c$ at $\x$. Any such $\x$
  does not belong to $\fibersing2$, and thus does not belong to
  $\fbr(\polar(e,{\dalgo},V^\mA),\fiber2)$. Since $\x$ lies over $\fiber2$,
  we deduce that $\x$ is not in $\polar(e,{\dalgo},V^\mA)$.  Because $\x$
  is in $\Open(m^\mA)$, Lemma~\ref{sec:atlas:chartpolar} implies
  that $\jac(\h^\mA,e+{\dalgo})$, and thus
  $\jac(\h^\mA,e+{\dalgo}-1)$, have full rank at $\x$.  
  \end{itemize}
The lemma is proved.
\end{proof}

\begin{proof}[of Proposition~\ref{sec:atlas:prop:summary1}]
  Write $\bpsi=(\psi_i)_{1 \le i \le s}$; for $i$ in $\{1,\dots,s\}$,
  we write $\psi_i=(m_i,\h_i)$.  To each $\psi_i$, we associate the
  non-empty Zariski open subset
  $\scrGfiberchart(\psi_i,V,Q,S,{\dalgo})$ of
  Lemma~\ref{lemma:chart:fiber}, and we let
  $\scrGfiber(\bpsi,V,Q,S,{\dalgo})$ be their intersection; it is
  still non-empty and Zariski open.  Take $\mA$ in
  $\scrGfiber(\bpsi,V,Q,S,{\dalgo})$ and write
  \[
  \Vfiber=\fbr(V^\mA,\fiber2) \qquad\text{ and }\qquad \fibersing2 =
  \fbr(S^\mA \cup \polar(e,{\dalgo},V^\mA),\fiber2).\]
  Because $\mA$ is in $\scrGfiber(\bpsi,V,Q,S,{\dalgo})$, it is in
  particular in $\scrGfiberchart(\psi_i,V,Q,S,{\dalgo})$ \new{for
    some $1\leq i \leq s$}. Then Lemma~\ref{lemma:chart:fiber} proves
  that since $S$ is finite, $\fibersing2$ is finite.
  
  Let us further assume that $\Vfiber$ is not empty; Krull's principal
  ideal theorem then implies that every irreducible component of
  $\Vfiber$ has dimension at least $d-({\dalgo}-1)>0$. We now prove
  that $\atlasfiber(\bpsi^\mA,V^\mA,Q,S^\mA,\fiber2)$ is an atlas of
  $(\Vfiber,\fiber2,\fibersing2)$.
  \begin{itemize}
  \item Up to reordering the $\psi_i$, we can write
    $\atlasfiber(\bpsi^\mA,V^\mA,Q,S^\mA,\fiber2)=((\psi^\mA_i)_{1 \le
    i\le s'})$.  In Lemma~\ref{lemma:chart:fiber}, we proved that each
    such $\psi_i^\mA$ is a chart of $(\Vfiber,\fiber2,\fibersing2)$,
    so we have proved that ${\sfA_1}$ holds.
    \smallskip
  \item By assumption, the open sets $\Open(m_i)$, $i=1,\dots,s$,
    cover $V-S$, which implies that the sets $\Open(m_i^\mA)$, for the
    same values of $i$, cover $V^\mA-S^\mA$.  This implies that the
    open sets $\Open(m^\mA_i)$, $i=1,\dots,s$, cover
    $\Vfiber-\fibersing2$, since $\Vfiber\subset V$ and $S\subset
    \fibersing2$. Since we kept only those $\psi^\mA_i$ for which
    $\Open(m_i^\mA) \cap \Vfiber-\fibersing2$ is not empty, this
    establishes ${\sfA_3}$.
    \smallskip
  \item In order to prove ${\sfA_2}$ it suffices to verify that
    $\Vfiber$ is not a subset of $\fibersing2$; this is case, since 
    we saw that $\Vfiber$ has positive dimension, and $\fibersing2$ is finite.
  \end{itemize}
  Hence, we have proved that
  $\atlasfiber(\bpsi^\mA,V^\mA,Q,S^\mA,\fiber2)$ is an atlas of
  $(\Vfiber,\fiber2,\fibersing2)$.

  Lemma~\ref{sec:coro:lemma:singSX} shows that all $\h_i$ have the
  same cardinality.  As a result, Lemma~\ref{sec:atlas:lemma:singSX}
  implies that $\Vfiber-\fibersing2$ is a non-singular
  $(d-({\dalgo}-1))$-equidimensional locally closed set. Since 
  all irreducible components of $\Vfiber$ have dimension at least
  $d-({\dalgo}-1)> 0$, we deduce that $\Vfiber$ itself is
  $({\dalgo}-1)$-equidimensional and has all its singular points in
  $\fibersing2$.
\end{proof}

%%%%%%%%%%%%%%%%%%%%%%%%%%%%%%%%%%%%%%%%%%%%

%%%%%%%%%%%%%%%%%%%%%%%%%%%%%%%%%%%%%%%%%%%%
%\input{BSCS-new}
\section{Proof of Proposition~\ref{prop:ch5}}\label{chap:finitenessproperties}

%%%%%%%%%%%%%%%%%%%%%%%%%%%%%%%%%%%%%%%%%%%%%%%%%%%%%%%%%%%%
%%%%%%%%%%%%%%%%%%%%%%%%%%%%%%%%%%%%%%%%%%%%%%%%%%%%%%%%%%%%
%%%%%%%%%%%%%%%%%%%%%%%%%%%%%%%%%%%%%%%%%%%%%%%%%%%%%%%%%%%%

The goal of this section is to prove the finiteness properties of
polar varieties stated as Proposition~\ref{prop:ch5}; they read as
follows: {\em Let $Q \subset \C^e$ be a finite set and let $V \subset
  \C^n$ be an algebraic set lying over~$Q$.  Suppose that $V$ is
  equidimensional of dimension $d$, with finitely many singular
  points, and let $\dalgo$ be an integer such that $2 \le \dalgo \le
  (d+3)/2$. Then, there exists a non-empty Zariski open set
  $\ZOfinite(V,Q,\dalgo) \subset \GL(n,e)$ such that, for $\mA$ in
  $\ZOfinite(V,Q,\dalgo)$, writing $W=\polar(e,\dalgo,V^\mA)$, either
  $W$ is empty, or $W$ is equidimensional of dimension $\dalgo-1$,
  with finitely many singular points, and $\Kpolar(e, 1, W)$ is
  finite.}

This claim extends to an arbitrary equidimensional algebraic set $V$
results that were already proved in~\cite{SaSc11} in the hypersurface
case. The proof techniques are similar, but slightly simpler for some
aspects (we do not rely anymore on some deep results of Mather's on
generic projections~\cite{Mather73}), and more involved in some others
(polar varieties are easier to define for hypersurfaces).

To prove this result, one can assume without loss of generality that
$e=0$. Assume indeed that we have proved our claim in that case. For
an arbitrary value of $e$, consider the finitely many points $\y \in
Q$ one after the other; for any such $\y$, define $V_\y \subset
\C^{n-e}$ as the set obtained from $\fbr(V,\y)\subset \C^n$ by
projection on the last $n-e$ coordinates: applying the case $e=0$ of
our proposition to the sets $V_\y$, it is enough to take the
intersection of the finitely many open sets $\ZOfinite(V_\y,\dalgo)
\subset \GL(n-e)$, and embed this intersection into $\GL(n,e)$ by
taking the direct sum with the identity matrix of size $e$.

%% Proposition~\ref{prop:ch4} implies that for a generic $\mA$,
%% $\polar(0, \dalgo, V^\mA)$ is either empty or
%% $(\dalgo-1)$-equidimensional, in which case the second polar variety
%% $\polar(0, 1, \polar(0, \dalgo, V^\mA))$ is well-defined (and possibly
%% empty, if $\polar(0, \dalgo, V^\mA)$ is).  Thus, we can focus on
%% proving that $\Kpolar(0, 1, \polar(0, \dalgo, V^\mA))$ is finite for a
%% generic $\mA$. Since $\sing(\polar(0,\dalgo, V^\mA))$ is finite
%% (Proposition~\ref{prop:ch4}), we will prove equivalently that
%% $\openpolar(0,1, \polar(0,\dalgo, V^\mA))$ is a finite set.

%%%%%%%%%%%%%%%%%%%%%%%%%%%%%%%%%%%%%%%%%%%%%%%%%%%%%%%%%%%%
%%%%%%%%%%%%%%%%%%%%%%%%%%%%%%%%%%%%%%%%%%%%%%%%%%%%%%%%%%%%

\subsection{The locally closed set \texorpdfstring{$\algX$}{$\algX$}}

In all that follows, we use the notation of
Proposition~\ref{prop:ch5}.  For $\g=(g_1, \ldots, g_{\dalgo})\in
\C^{\dalgo}$, let $\rho_\g$ be the mapping $(x_1, \ldots,
x_{\dalgo})\mapsto g_1 x_1+\cdots+g_{\dalgo} x_{\dalgo}$; we will
denote by $\g_0\in \C^{\dalgo}$ the row vector $(1, 0, \ldots, 0)$, so
that $\rho_{\g_0}\circ \pi_{\dalgo}$ is simply the projection
$\pi_1$. With this notation, our goal is thus to prove that for a
generic choice of $\mA$, $$\openpolar(0,1,
\polar(0,\dalgo,V^\mA))=\openpolar(0,\rho_{\g_0}\circ \pi_{\dalgo},
\polar(0,\dalgo,V^\mA))$$ is finite.

In this paragraph, we define a set $\algX \subset \C^{n^2} \times
\C^n\times \C^{\dalgo} $ consisting of triples $(\mA,\x,\g)$ such that
$\x$ is in $\openpolar(0,\dalgo,V^\mA)$ and $\rho_\g \circ \pi_{\dalgo}$ vanishes on
$\T_\x \polar(0,\dalgo,V^\mA)$.  In order to ensure that this set is locally
closed, we will restrict $\mA$ to a suitable open set of $\GL(n)$, on
which a ``uniform'' description of the polar varieties will be
available.

The construction is slightly technical, but simple in essence: we
construct a family of polynomials (written \new{$\mathbf{P}$} below) in an
algorithmic manner, which will ensure that it defines the polar
variety $\polar(0,\dalgo,V^\mA)$ for a generic $\mA$.

Let $\F=(F_1,\dots,F_s) \subset \C[X_1,\dots,X_n]$ be generators of
the ideal of $V$ and let ${\mathfrak{A}}= ({\mathfrak{A}}_{i,j})_{1\leq i,j\leq
  n}$ be a matrix of new indeterminates. We define $\F^{\mathfrak{A}}$ as
usual, as the set of polynomial $(F_1({\mathfrak{A}}
\X),\dots,F_s({\mathfrak{A}} \X))$, and we define the polynomials
$\GG$ and $\JJ$ in $\C[{\mathfrak{A}}][X_1,\dots,X_n]$ as the sets of
$(n-d)$-minors of respectively $\jac(\F^{\mathfrak{A}})$ and
$\jac(\F^{\mathfrak{A}},\dalgo)$, where the derivatives are taken with
respect to $X_1,\dots,X_n$ only. For $\mA$ in $\GL(n)$, the
polynomials $\GG(\mA,\X)\subset \C[X_1,\dots,X_n]$ are defined by
evaluating the variables ${\mathfrak{A}}$ at $\mA$.

\begin{lemma}\label{lemma:5.2.1}
  For $\mA$ in $\GL(n)$, the zero-set of $(\F^\mA, \GG(\mA,\X))$ is
  $\sing(V^\mA)$ and the zero-set of $(\F^\mA, \JJ(\mA,\X))$ is
  $\Kpolar(0,\dalgo,V^\mA)$.
\end{lemma}
\begin{proof}
  For $\mA$ in $\GL(n)$, the ideal $\langle \F^\mA \rangle$ is the
  defining ideal of $V^\mA$, and the polynomials $\GG(\mA,\X)$ and
  $\JJ(\mA,\X)$ are simply the corresponding minors of the matrix
  $\jac(\F^\mA)$; our claim for $\sing(V^\mA)$ is then
  straightforward, and that for $\Kpolar(0,\dalgo,V^\mA)$ follows from
  Lemma~\ref{sec:prelim:lemma:Ksing}.
\end{proof}

Applying a radical ideal computation algorithm, say for definiteness
that in \cite[Theorem~8.99]{BeWe93}, we obtain a finite set of
polynomials $\H \subset \C({\mathfrak{A}})[X_1,\dots,X_n]$ that
generate the radical of the ideal $\langle \F^{\mathfrak{A}}, \JJ
\rangle$ in $\C({\mathfrak{A}})[X_1,\dots,X_n]$. For $\mA$ in
$\GL(n)$, the polynomials $\H(\mA,\X)$ are defined similarly to the
polynomials $\GG(\mA,\X)$ above (provided no denominator vanishes),
and the following lemma shows that they have the expected
specialization properties.

\begin{lemma}\label{lemma:5.2.2}
  There exists a non-empty Zariski open subset $\ZOK_1 \subset
  \GL(n)$ such that for $\mA$ in $\ZOK_1$, the polynomials
  $\H(\mA,\X)$ are well-defined and the ideal $\langle \H(\mA,\X)
  \rangle$ is radical, with zero-set $\Kpolar(0,\dalgo,V^\mA)$.
\end{lemma}
\begin{proof}
  Because we are in characteristic zero, it is possible to compute the
  radical of an ideal, over either $\C({\mathfrak{A}})[X_1,\dots,X_n]$
  or $\C[X_1,\dots,X_n]$, using an algorithm that does only arithmetic
  operations in $(+,-,\times,\div)$ and zero-tests; this is the case
  for the algorithm of~\cite[Theorem~8.99]{BeWe93} that we mentioned
  above (and would not be the case in positive characteristic).

  We choose for $\ZOK_1$ a non-empty Zariski open set where all
  steps performed to compute the radical of $\langle \F^{\mA},
  \JJ(\mA,\X) \rangle$ over $\C[X_1,\dots,X_n]$ are the mirror of
  those done to compute $\H$ over $\C({\mathfrak{A}})[X_1,\dots,X_n]$. For
  instance, $\ZOK_1$ can be taken as the locus where none of
  the (finitely many) non-zero rational functions in $\C({\mathfrak{A}})$
  that appear during the computation is undefined or vanishes. For
  $\mA$ in $\ZOK_1$, the ideal $\langle \H(\mA,\X) \rangle$ is
  then radical, and its zero-set is $\Kpolar(0,\dalgo,V^\mA)$, in view of the
  previous lemma.
\end{proof}

Doing similarly for colon ideal computation, using for instance the
algorithm in~\cite[Corollary~6.34]{BeWe93}, we obtain a finite set of
polynomials $$\mathbf{P}\subset \C({\mathfrak{A}})[X_1,\dots,X_n]$$ that generate
the colon ideal $\langle \H \rangle : \langle \F^{\mathfrak{A}}, \GG
\rangle$. 

\begin{lemma}\label{lemma:5.2.3}
  There exists a non-empty Zariski open subset $\ZOK_2 \subset
  \ZOK_1$ such that for $\mA$ in $\ZOK_2$, the polynomials 
  $\mathbf{P}(\mA,\X)$ are well-defined and the ideal $\langle
  \mathbf{P}(\mA,\X) \rangle$ is radical, with zero-set $\polar(0,\dalgo,V^\mA)$.
\end{lemma}
\begin{proof}
  The first point is proved as in the previous lemma, by choosing an
  open set $\ZOK_2 \subset \ZOK_1$ where all algorithmic steps in
  colon ideal computation specialize well. Then, because
  $\langle \H(\mA,\X) \rangle$ is radical (by the previous lemma), we
  know that $\langle \mathbf{P}(\mA,\X) \rangle$ is radical as well. To prove
  the second point, we use the fact that for any $\mA$ in
  $\ZOK_2$, the zero-set of $\langle \mathbf{P}(\mA,\X) \rangle$ is
  the Zariski closure of $\Kpolar(0,\dalgo,V^\mA)-\sing(V^\mA)$ since
  $\langle \H(\mA,\X) \rangle$ is radical and defines
  $\Kpolar(0,\dalgo,V^\mA)$ (by the previous lemma). The latter set is
  simply $\openpolar(0,\dalgo,V^\mA)$, so we are done.
\end{proof}

We are going to restrict further the Zariski open set $\ZOK_2$ by
taking its intersection with the following subsets of $\GL(n)$:
\begin{itemize}
\item the non-empty open set
  $\scrGpolar(\bpsi,V,\{\bullet\},\sing(V),\dalgo) \subset \GL(n)$
  defined defined by applying Proposition~\ref{prop:ch4} to the atlas
  $\bpsi$ of $(V,\{\bullet\},\sing(V))$ given in
  Lemma~\ref{sec:atlas:lemma:glob}; it ensures that
  $\polar(0,\dalgo,V^\mA)$ is either empty or
  $(\dalgo-1)$-equidimensional and that
  $\sing(\polar(0,\dalgo,V^\mA))$ is contained in $\sing(V^\mA)$.
  \smallskip
\item the non-empty open set
  $\scrGfiber(\bpsi,V,\{\bullet\},\sing(V),\dalgo) \subset \GL(n)$ defined by
  applying Proposition~\ref{sec:atlas:prop:summary1} to the same
  atlas; it has the property that for $\mA$ in this set, the
  restriction of $\pi_{\dalgo-1}$ to $\Kpolar(0,\dalgo,V^\mA)$, or
  equivalently to $\polar(0,\dalgo,V^\mA)$, has finite fibers;
\end{itemize}
Let us then call $\ZOK_\Klastindex$ the intersection of the non-empty Zariski open sets $\ZOK_2$,
$\scrGpolar(\bpsi,V,\{\bullet\},\sing(V),\dalgo)$ and
$\scrGfiber(\bpsi,V,\{\bullet\},\sing(V),\dalgo)$ in $\GL(n)$; this is
a non-empty Zariski open subset of $\GL(n)$. Having defined
$\ZOK_\Klastindex$ allows us to define $\algX \subset \C^{n^2}\times
\C^n \times \C^{\dalgo}$ as the set of triples $(\mA,\x,\g)$ such that
the following holds:
\begin{itemize}
\item $\mA$ is in $\ZOK_\Klastindex$,
\smallskip
\item $\x$ is in $\openpolar(0,\dalgo,V^\mA)$,
\smallskip
\item $\rho_\g \circ \pi_{\dalgo}$ vanishes on $\T_\x \polar(0,\dalgo,V^\mA)$.
\end{itemize}

\begin{lemma}\label{lemma:5.2.4}
  The set $\algX$ is locally closed.
\end{lemma}
\begin{proof}
  Let $\mathfrak{g}_1,\dots,\mathfrak{g}_{\dalgo}$ be new indeterminates that
  stand for the entries of $\g=(g_1,\dots,g_{\dalgo})$, and consider the
  set $\algX' \subset \C^{n^2}\times \C^n \times \C^{\dalgo}$
  defined through the following properties:
  \begin{itemize}
  \item $\mA$ is in $\ZOK_\Klastindex$,
\smallskip
  \item $(\mA,\x)$ is in $V(\mathbf{P})-V(\F^{\mathfrak{A}},\GG)$,
\smallskip
  \item the matrix obtained by adjoining to $\jac(\mathbf{P},\X)$ the row with
    entries $$[\mathfrak{g}_1,\dots,{\mathfrak{g}}_{\dalgo},0,\dots,0]$$ has rank
    $n-(\dalgo-1)$ at $(\mA,\x,\g)$.
  \end{itemize}
  By construction, $\algX'$ is locally closed, since it is the
  intersection of three locally closed sets (note that $\ZOK_\Klastindex$
  is an open subset of $\GL(n)$, which is itself open in
  $\C^{n^2}$). We conclude by proving that
  $\algX=\algX'$. The defining conditions on $\mA$ are
  identical on both sides; we then inspect those on $(\mA,\x)$ and
  finally on $(\mA,\x,\g)$.

  Lemmas~\ref{lemma:5.2.1} and~\ref{lemma:5.2.3} show that since $\mA$
  is in $\ZOK_\Klastindex$, $(\mA,\x)$ belongs to
  $V(\mathbf{P})-V(\F^{\mathfrak{A}},\JJ)$ if and only if $\x$ belongs to
  $\polar(0,\dalgo,V^\mA)-\sing(V^\mA)$, that is, to
  $\openpolar(0,\dalgo,V^\mA)$, so the defining conditions on $(\mA,\x)$
  are the same for $\algX$ and $\algX'$.

  Finally, we deal with the last conditions. In view of the above, we
  can assume that $\mA$ is in $\ZOK_\Klastindex$ and that $\x$ is in
  $\openpolar(0,\dalgo,V^\mA)$. Remark in particular that in this case,
  $\x$ is in $\reg(\polar(0,\dalgo,V^\mA))$, since $\mA \in \ZOK_\Klastindex$
  implies that $\sing(\polar(0,\dalgo,V^\mA))$ is contained in
  $\sing(V^\mA)$, whereas $\x$ is in
  $\openpolar(0,\dalgo,V^\mA) \subset \reg(V^\mA)$. Remember as well
  that $\polar(0,\dalgo,V^\mA)$ is $(\dalgo-1)$-equidimensional. This,
  together with Lemma~\ref{lemma:5.2.3}, implies that
  $\jac(\mathbf{P},\X)$ has rank $n-(\dalgo-1)$ at $(\mA,\x)$ and that
  its nullspace is $\T_\x \polar(0,\dalgo,V^\mA)$. The rank condition on
  the augmented matrix is then equivalent to
  $\rho_\g \circ \pi_{\dalgo}$ vanishing on
  $\T_\x \polar(0,\dalgo,V^\mA)$.
\end{proof}

%%%%%%%%%%%%%%%%%%%%%%%%%%%%%%%%%%%%%%%%%%%%%%%%%%%%%%%%%%%%
%%%%%%%%%%%%%%%%%%%%%%%%%%%%%%%%%%%%%%%%%%%%%%%%%%%%%%%%%%%%
%%%%%%%%%%%%%%%%%%%%%%%%%%%%%%%%%%%%%%%%%%%%%%%%%%%%%%%%%%%%

\subsection{The dimension of \texorpdfstring{$\algX$}{$\algX$}}

In this paragraph, we prove that $\algX$ has dimension at most
$\dalgo+n^2$. This is done by applying the theorem on the dimension of
fibers twice. We define the projection
$$\begin{array}{cccc}
  \pi_{\mathfrak{A}} : & \C^{n^2} \times
 \C^n \times \C^{\dalgo} & \rightarrow & \C^{n^2} \\
&(\mA, \x,\g) & \mapsto & \mA;
\end{array} $$
and 
$$\begin{array}{cccc}
  \pi_{\X} : & \C^{n^2} \times
 \C^n \times \C^{\dalgo} & \rightarrow & \C^{n} \\
&(\mA, \x,\g) & \mapsto & \x.
\end{array} $$
Then, for $\mA$ in $\ZOK_\Klastindex$, $\algAX$ denotes the
fiber $\pi_{\mathfrak{A}}^{-1}(\mA) \cap \algX \subset \C^{n^2}
\times \C^n \times \C^{\dalgo}$.  In order to prove the bound on
$\dim(\algX)$, we will first prove that $\algAX$ has
dimension at most~$\dalgo$ and apply a form of the theorem on the
dimension of fibers to $\pi_\mathfrak{A}$. To prove the dimension
bound on $\algAX$, we will apply the same theorem, but to the
restriction of $\pi_\X$ to $\algAX$.

The definition of $\algX$ implies that $(\mA,\x,\g)$ is in
$\algAX$ if and only if $\x$ is in $\openpolar(0,\dalgo,V^\mA)$
and $\rho_\g \circ \pi_{\dalgo}$ vanishes on $\T_\x \polar(0,\dalgo,V^\mA)$,
and Lemma~\ref{lemma:5.2.4} implies that $\algX$ and thus
$\algAX$ are locally closed subsets of
$\C^{n^2} \times \C^n \times \C^{\dalgo}$.

As a useful preliminary, we prove the following lemma on the dimension of
fibers on locally closed sets.

\begin{lemma}\label{lemma:5.2.5}
  Let $\Slocallyclosed \subset \C^n$ be a locally closed set and let $r\in \N$ be
  such that the Zariski closure of $\pi_r(\Slocallyclosed)$ has dimension
  $s$. Assume that for all $\x$ in $\pi_r(\Slocallyclosed)$, the fiber
  $\pi_r^{-1}(\x) \cap \Slocallyclosed$ has dimension at most $t$. Then $\Slocallyclosed$ has
  dimension at most $s+t$.
\end{lemma}
\begin{proof}
%% S=X-Y, with X-Y bar=X->T in X
%% T'=T cap S = T cap (X-Y) = T cap X - Y = T -Y
  Let $T$ be an irreducible component of the Zariski closure of $\Slocallyclosed$
  and let $T'=\Slocallyclosed \cap T$; because $\Slocallyclosed$ is locally closed, one deduces
  that $T'$ is an open dense subset of $T$.

  Let further $C$ be the Zariski closure of $\pi_r(T)$. We claim that
  $\dim(C)\le s$. Indeed, because $T'$ is dense in $T$, we infer that
  $C$ is also the Zariski closure of $\pi_r(T')$. Since $\pi_r(T')$ is
  contained in $\pi_r(\Slocallyclosed)$, we conclude that its Zariski closure has
  dimension at most $s$.

  Since $T'$ is open dense in $T$, we can write $T'=T-Y$, where $Y$ is
  a strict algebraic subset of $T$; in particular, $\dim(Y)<\dim(T)$.
  Let us then consider the restriction of $\pi_r$ to a projection $T
  \to C$ and let $m$ be the dimension of its generic fiber, so that we
  have $m=\dim(T)-\dim(C)$. We claim that for a generic $\x$ in $C$,
  the fiber $\pi_r^{-1}(\x) \cap Y$ has dimension less than $m$.

  To prove this claim, we decompose $Y$ into its irreducible
  components, and distinguish those whose projection is dense in $C$
  from the others. Let us thus write $Y=Y_1 \cup \cdots \cup Y_u \cup
  Z_1 \cup \cdots \cup Z_v$, with all $Y_i,Z_j$ irreducible, and such
  that for all $i,j$, $\pi_r(Y_i)$ is not dense in $C$ and
  $\pi_r(Z_j)$ is dense in $C$. We can then consider fibers 
  of the form $\pi_r^{-1}(\x) \cap Y_i$ and $\pi_r^{-1}(\x) \cap Z_j$
  separately.
  \begin{itemize}
  \item For $1 \le i \le u$, there exists an open dense subset $O_i$
    of $C$ such that for $\x$ in $O_i$, the fiber $\pi_r^{-1}(\x) \cap
    Y_i$ is empty.
\smallskip
  \item For $1 \le j \le v$, let $m'_j$ be the dimension of the
    generic fiber of the restriction of $\pi_r$ to $Z_j$. This implies
    that $m'_j=\dim(Z_j)-\dim(C) < m$ (since
    $\dim(Z_j)<\dim(T)$). Thus, there exists an open dense subset $U_j$
    of $C$ such that for $\x$ in $U_j$, the fiber $\pi_r^{-1}(\x) \cap
    Z_j$ has dimension $m'_j$, which is less than $m$.
  \end{itemize}
  Our claim on the fibers $\pi_r^{-1}(\x) \cap Y$ is thus proved.
  Now, for $\x$ in $C$, the fiber $\pi_r^{-1}(\x) \cap T'$ is the
  set-theoretic difference of the Zariski closed sets $\pi_r^{-1}(\x)
  \cap T$ and $\pi_r^{-1}(\x) \cap Y$. For a generic $\x$ in $C$,
  $\pi_r^{-1}(\x) \cap T$ has dimension $m$, so in view of the previous
  discussion, we deduce that for a generic $\x$ in $C$, the fiber
  $\pi_r^{-1}(\x) \cap T'$ is a locally closed set of dimension $m$ as well.

  On the other hand, for any $\x$ in $C$, our assumption says that
  this fiber has dimension at most $t$, so that $t \ge m$. Since
  $m=\dim(T)-\dim(C) \ge \dim(T)-s$, we get $\dim(T) \le s+t$.  Doing
  so for all $T$, we get $\dim(\Slocallyclosed) \le s+t$.
\end{proof}

Let $\mA$ be in $\ZOK_\Klastindex$. In order to bound the dimension of
$\algAX$, we will apply the previous lemma to the restriction
of the projection $\pi_\X$  to $\algAX$.

Note that the image of $\algAX$ by $\pi_\X$ is contained in
$\openpolar(0,\dalgo,V^\mA)$.  For all $\x$ in
$\openpolar(0,\dalgo,V^\mA)$, let thus $\algxAX$ be the fiber
$\pi_\X^{-1}(\x) \cap \algAX$. Remark that set of all $\g$ such
that $(\mA,\x,\g)$ belongs to $\algX$ is a vector space, say
$\new{E}_{\x,\mA} \subset \C^{\dalgo}$, since $\rho_{a \g + a' \g'}= a
\rho_\g + a' \rho_{\g'}$ for all $a,a' \in \C$ and $\g,\g' \in
\C^{\dalgo}$; then, $\algxAX$ takes the form
$\{\mA\}\times\{\x\}\times \new{E}_{\x,\mA}$.

First, we need a lemma estimating the dimension of the vector space
$\new{E}_{\x,\mA}$, or equivalently of $\algxAX$.

\begin{lemma}\label{lemma:dimbeta}
  For $\mA \in \GL(n)$ and $\x\in \openpolar(0,\dalgo,V^\mA)$, the
  following equality holds:
  $$
  \dim(\pi_{\dalgo}(\T_\x \polar(0,\dalgo,V^\mA)))+\dim(\algxAX)=\dalgo.
  $$
\end{lemma}
\begin{proof}
  For a given $\mA$ and $\x$, $\g$ belongs to $\new{E}_{\x,\mA}$ if
  and only if the linear form $\rho_\g$ vanishes on
  $\pi_{\dalgo}(\T_\x \polar(0,\dalgo,V^\mA))$.  Thus $\new{E}_{\x,\mA}$
  is isomorphic to the dual of the cokernel of
  $\pi_{\dalgo}: \T_\x \polar(0,\dalgo,V^\mA) \to \C^{\dalgo}$, and the
  dimension equality follows.
\end{proof}

Thus, in order to control $\dim(\new{\algxAX})$, we need
to discuss the possible dimensions of
$\pi_{\dalgo}(\T_\x \polar(0,\dalgo,V^\mA))$, for
$\x \in \openpolar(0,\dalgo,V^\mA)$. It is then natural to introduce the
sets
$$\SlocallyclosediA=\{\x\in \openpolar(0,\dalgo,V^\mA)\mid \dim(\pi_{\dalgo}(\T_\x
\polar(0,\dalgo,V^\mA)))=\dalgo-i\}\text{ for } 1\leq i \leq \dalgo.$$
The following lemma relates the dimension of $\pi_r(\T_\x \Slocallyclosed)$ and
$\pi_r(\Slocallyclosed)$, for $\pi_r$ a projection and $\Slocallyclosed$ a locally closed set.

\begin{lemma}\label{lemma:dimpiTxS}
  Let $\Slocallyclosed \subset \C^n$ be a locally closed set and let $r,s \in
  \mathbb{N}$ be such that for all $\x$ in $\Slocallyclosed$, $\pi_r(\T_\x \Slocallyclosed)$ has
  dimension at most $s$. Then the Zariski closure of $\pi_r(\Slocallyclosed)$ has
  dimension at most $s$ as well.
\end{lemma}
\begin{proof}
  Let $\algZ \subset \C^n$ be the Zariski closure of $\Slocallyclosed$, and let
  $\algZ_1,\dots,\algZ_k$ be its irreducible components. We will prove
  that the Zariski closure $C_i$ of $\pi_r(\algZ_i)$ has dimension at
  most $s$ for all $i$. This will be enough to conclude, since the
  union of the sets $C_i$ contains $\pi_r(\Slocallyclosed)$.

  Fix $i \le k$. Let
  $\openB_i=\Slocallyclosed \cap \algZ_i-\cup_{i'\ne i}\algZ_{i'}$. Remark that
  $\openB_i$ is an open dense subset of $\algZ_i$, and that for $\x$ in
  $\openB_i$, $\T_\x \Slocallyclosed = \T_\x \algZ_i$, so that $\pi_r(\T_\x \algZ_i)$ has
  dimension at most $s$.

  On the other hand, applying Sard's lemma in the form
  of~\cite[Theorem 3.7]{Mumford76} to the restriction of $\pi_r$ to
  $\algZ_i$, we know that there exists a non-empty Zariski open subset
  $\BasicOpen_i$ of $C_i$ such that for $\x$ in
  $\pi_r^{-1}(\BasicOpen_i) \cap \reg(\algZ_i)$, $\dim(\pi_r(\T_\x
  \algZ_i))=\dim(C_i)$.  Intersecting $\pi_r^{-1}(\BasicOpen_i) \cap
  \reg(\algZ_i)$ with $\openB_i$, we obtain a non-empty open subset
  $\UOpen_i$ of $\algZ_i$ such that for $\x$ in $\UOpen_i$, we have
  simultaneously $\dim(\pi_r(\T_\x \algZ_i))=\dim(C_i)$ and
  $\dim(\pi_r(\T_\x \algZ_i)) \le s$.
\end{proof}

\begin{lemma}\label{lemma:Sdi}
  For all $\mA\in \ZOK_\Klastindex$ and for all
  $i \in \{1,\dots,\dalgo\}$, $\SlocallyclosediA$ is a locally closed subset
  of $\C^n$ of dimension at most $\dalgo-i$, and
  $\cup_{i=1}^{\dalgo} \SlocallyclosediA$ is a partition of
  $\openpolar(0,\dalgo,V^\mA)$.
\end{lemma}
\begin{proof}
  Since $\mA$ is in $\ZOK_\Klastindex$, $\polar(0,\dalgo,V^\mA)$ is either empty or
  $(\dalgo-1)$-equidimensional, and in that case its singular locus is
  contained in that of $V^\mA$. 
  
  We can of course suppose that $\polar(0,\dalgo,V^\mA)$ is not empty. Then,
  for all
  $\x\in \openpolar(0,\dalgo,V^\mA) \subset \reg(\polar(0,\dalgo,V^\mA))$,
  $\T_\x \polar(0,\dalgo,V^\mA)$ has dimension $\dalgo-1$, which implies that
  its image by $\pi_{\dalgo}$ has dimension at most $\dalgo-1$.  This
  implies in turn that $\cup_{i=1}^{\dalgo} \SlocallyclosediA$ is a partition
  of $\openpolar(0,\dalgo,V^\mA)$.

  Next, we prove that each $\SlocallyclosediA$ is a locally closed
  set. Indeed, $\openpolar(0,\dalgo,V^\mA)$ is locally closed, and for
  $\x$ in $\openpolar(0,\dalgo,V^\mA)\subset \reg(\polar(0,\dalgo,V^\mA))$,
  $\pi_{\dalgo}(\T_\x \polar(0,\dalgo,V^\mA))$ having dimension $\dalgo-i$
  amounts to $\jac(\mathbf{P}(\mA,\X),\dalgo)$ having rank $n-\dalgo-i+1$ at
  $\x$, which is a locally closed condition.

%% rk (MM) = rk jac() + rk pi|ker jac
%%         = n-(\dalgo-1) + dim(pi(T_x))
%%         = n-\dalgo+1+\dalgo-i
%%         = n-i+1
%% but also = \dalgo + rk jac(\dalgo)

  We can now fix $i \in \{1,\dots,\dalgo\}$. Since $\SlocallyclosediA$
  is a subset of $\Kpolar(0,\dalgo,V^\mA)$, and since $\mA$ has been
  chosen in the Zariski open set
  $\ZOK_\Klastindex \subset \scrGfiber(\bpsi,V,\{\bullet\},\sing(V),
  \dalgo)$,
  we conclude from the defining property of
  $\scrGfiber(\bpsi,V,\{\bullet\},\sing(V),\dalgo)$ given in
  Proposition~\ref{sec:atlas:prop:summary1} that for all
  $\y\in \C^{\dalgo}$, the fiber
  $\pi_{\dalgo}^{-1}(\y)\cap \SlocallyclosediA$ is finite (precisely,
  the defining property of $\scrGfiber(\bpsi,V,\{\bullet\},\dalgo)$
  applies to the fibers of $\pi_{\dalgo-1}$, which is stronger than
  what we use here).
  
  Next, we prove that the Zariski closure of $\pi_{\dalgo}(\SlocallyclosediA)$
  has dimension at most $\dalgo-i$.  Take $\x$ in $\SlocallyclosediA$, so that
  in particular $\x$ is in $\reg(\polar(0,\dalgo,V^\mA))$. We know that
  $\SlocallyclosediA$ is contained in $\openpolar(0,\dalgo,V^\mA)$, so upon taking
  Zariski closure and tangent spaces, we deduce that $\T_\x \SlocallyclosediA$
  is contained in $\T_\x \polar(0,\dalgo,V^\mA)$.  This implies that
  $\pi_{\dalgo}(\T_\x \SlocallyclosediA)$ is contained in $\pi_{\dalgo}(\T_\x
  \polar(0,\dalgo,V^\mA))$. Because $\x$ is in $\SlocallyclosediA$, we deduce that
  $\pi_{\dalgo}(\T_\x \SlocallyclosediA)$ has dimension at most
  $\dalgo-i$. Lemma~\ref{lemma:dimpiTxS} then implies that the Zariski
  closure of $\pi_{\dalgo}(\SlocallyclosediA)$ has dimension at most
  $\dalgo-i$, as claimed. Using the finiteness property for the fibers
  of $\pi_{\dalgo}$ (previous paragraph), Lemma~\ref{lemma:5.2.5} then
  implies that $\dim(\SlocallyclosediA) \le \dalgo-i$ as well.
 \end{proof}

We can then deduce an upper bound on the dimension of $\algAX$.
\begin{corollary}\label{lemma:DIMXmA}
  The set $\algAX$ has dimension at most $\dalgo$.
\end{corollary}
\begin{proof}
  By Lemma \ref{lemma:Sdi}, $\openpolar(0,\dalgo,V^\mA)$ is the disjoint
  union of the locally closed sets
  $$\SlocallyclosediA=\{\x\in \openpolar(0,\dalgo,V^\mA)\mid \dim(\pi_{\dalgo}(\T_\x
  \polar(0,\dalgo,V^\mA)))=\dalgo-i\}\text{ for } 1\leq i \leq \dalgo,$$
  with in addition $\dim(\SlocallyclosediA)\leq \dalgo-i$ for all $i$.

  For $i$ as above, let us further define $\algiAX =
  \algAX \cap \pi_{\X}^{-1}(\SlocallyclosediA)$; this is still a
  locally closed set in $\C^{n^2} \times \C^n \times \C^{\dalgo}$.  By
  construction, $\pi_{\X}(\algiAX)$ is contained in
  $\SlocallyclosediA$, so its Zariski closure has dimension at most $\dalgo-i$
  (Lemma~\ref{lemma:Sdi}).  On the other hand, because
  $\pi_{\X}(\algiAX)$ is contained in $\SlocallyclosediA$, we also
  know that for every $\x$ in $\pi_{\X}(\algiAX)$, the
  fiber $\pi_\X^{-1}(\x)\cap \algiAX$, which is equal to
  $\algxAX$, has dimension $i$
  (Lemma~\ref{lemma:dimbeta}).

  Applying Lemma~\ref{lemma:5.2.5}, we deduce that $\algiAX$
  has dimension at most $\dalgo$. Since $\algAX$ is the union 
  of the finitely many subsets $\algiAX$, its Zariski closure
  is contained in the union of the Zariski closures of those sets,
  so it has dimension at most $\dalgo$ as well.
\end{proof}

We now come to the main result of this paragraph.

\begin{corollary}\label{coro:5.2.10}
  The set $\algX$ has dimension at most $\dalgo+n^2$.
\end{corollary}
\begin{proof}
  This follows from applying Lemma~\ref{lemma:5.2.5} to the
  restriction of the projection $\pi_\mathfrak{A}: \C^{n^2}\times \C^n
  \times \C^{\dalgo} \to \C^{n^2}$ to $\algX$ and using the previous
  lemma to bound the dimension of the fibers.
\end{proof}

%%%%%%%%%%%%%%%%%%%%%%%%%%%%%%%%%%%%%%%%%%%%%%%%%%%%%%%%%%%%
%%%%%%%%%%%%%%%%%%%%%%%%%%%%%%%%%%%%%%%%%%%%%%%%%%%%%%%%%%%%
%%%%%%%%%%%%%%%%%%%%%%%%%%%%%%%%%%%%%%%%%%%%%%%%%%%%%%%%%%%%

\subsection{Proof of Proposition~\ref{prop:ch5}}

We can now complete the proof of Proposition~\ref{prop:ch5}. We start
by turning the situation around and considering the projection
$$
\begin{array}{cccc}
  \varsigma : &  \C^{n^2}  \times \C^n \times \C^{\dalgo}& \rightarrow & \C^{n^2} \times \C^{\dalgo}\\
  &(\mA, \x, \g) & \mapsto & (\mA, \g).
\end{array}
$$ We claim that most fibers of this projection are finite.
Precisely, let $Y\subset \C^{n^2}\times \C^{\dalgo}$ be the Zariski
closure of the set of all $(\mA,\g)\in \C^{n^2}\times \C^{\dalgo}$ such
that the fiber $\varsigma^{-1}(\mA,\g)\cap \algX$ is infinite.

\begin{lemma}\label{lemma:algebraicY}
  The set $Y$ is a strict Zariski closed subset of $\C^{n^2}\times\C^{\dalgo}$. 
\end{lemma}
\begin{proof}
  By definition, $Y$ is Zariski closed, so it remains to prove that it
  does not cover $\C^{n^2}\times\C^{\dalgo}$.  Let $\closedZ$ be an
  irreducible component of the Zariski closure of
  $\algX$. Corollary~\ref{coro:5.2.10} shows that $\closedZ$ has
  dimension at most $\dalgo+n^2$, so either $\varsigma(\closedZ)$ is not
  dense in $\C^{n^2} \times \C^{\dalgo}$, in which case for a generic
  $(\mA,\g)\in \C^{n^2} \times \C^{\dalgo}$ the fiber
  $\varsigma^{-1}(\mA,\g) \cap \closedZ$ is empty, or it is dense in
  $\C^{n^2} \times \C^{\dalgo}$, in which case that fiber is
  generically finite.
\end{proof}

Because $Y$ is a strict Zariski closed set of
$\C^{n^2}\times\C^{\dalgo}$, we claim that there exists a non-zero
$\g_1 \in \C^{\dalgo}$ and a non-empty Zariski open set
$\ZOK_\Klastlastindex \subset \ZOK_\Klastindex$ in $\C^{n^2}$ such
that for $\mA$ in $\ZOK_\Klastlastindex$, $(\mA,\g_1)$ is not in $Y$.
Indeed, consider the projection
$\C^{n^2}\times\C^{\dalgo}\to \C^{\dalgo}$ and its restriction to an
irreducible component $Y'$ of $Y$. Either this restriction is
dominant, in which case its generic fiber has dimension less than
$n^2$, or the image is contained in a strict Zariski closed subset of
$\C^{\dalgo}$.

Let us take $\g_1$ and $\ZOK_\Klastlastindex$ as above, with in addition
$\g_1$ non-zero. For $\mA$ in $\ZOK_\Klastlastindex$, the fiber
$\varsigma^{-1}(\mA, \g_1)$ is finite.  In other words, there exist
finitely many $\x$ in $\openpolar(0,\dalgo,V^\mA)$ such that
$\rho_{\g_1} \circ \pi_{\dalgo}$ vanishes on $\T_\x \polar(0,\dalgo,V^\mA)$.
The following lemma shows how we will obtain a similar result for
$\g_0=(1, 0, \ldots, 0)^t$ instead of $\g_1$.

\begin{lemma}
  Let $\mB$ be in $\GL(n)$ of the form
  $$\mB = \begin{bmatrix} \mB' & \mzero \\ \mzero & {\bf 1}_{n-\dalgo}  
  \end{bmatrix},$$
  with $\mB'$ in $\GL(\dalgo)$. 
  Then, for $\mA$ in $\GL(n)$, the following equalities hold:
  $$ V^{\mA \mB} = (V^\mA)^\mB, \quad \openpolar(0,\dalgo,V^{\mA \mB}) =
  \openpolar(0,\dalgo,V^\mA)^\mB \quad\text{and}\quad
  \polar(0,\dalgo,V^{\mA \mB}) = \polar(0,\dalgo,V^\mA)^\mB.$$ Besides,
  for $\x$ in $\polar(0,\dalgo,V^{\mA\mB})$, we have
  $$ \T_\x \polar(0,\dalgo,V^{\mA\mB}) =(\T_{\x^{\mB^{-1}}} \polar(0,\dalgo,V^\mA))^\mB$$
  and for $\u$ in $\T_\x \polar(0,\dalgo,V^{\mA\mB})$ and $\g$ in $\C^{\dalgo}$, we have
  $$ (\rho_\g \circ \pi_{\dalgo}) (\u) = (\rho_{{\mB'}^{-t} \g} \circ \pi_{\dalgo})(\u^{\mB^{-1}}).$$
\end{lemma}
\begin{proof}
  The first equality is a direct consequence of the definition of
  $V^\mA$; it implies in particular that
  $\sing(V^{\mA
    \mB})=\sing(V^\mA)^\mB$.
  In~\cite[Section~2.3]{SaSc03}, we prove that
  $\Kpolar(0,\dalgo,V^{\mA \mB})=\Kpolar(0,\dalgo,V^\mA)^\mB$; in view of the
  previously noted equality of $\sing(V^{\mA \mB})$ and
  $\sing(V^\mA)^\mB$, we deduce that
  $\openpolar(0,\dalgo,V^{\mA \mB})=\openpolar(0,\dalgo,V^\mA)^\mB$, and similarly
  for their Zariski closures, that
  $\polar(0,\dalgo,V^{\mA \mB})=\polar(0,\dalgo,V^\mA)^\mB$.  The fourth equality
  follows immediately.

  To prove the last equality, take $\u$ in $\T_\x \polar(0,\dalgo,V^{\mA\mB})$ and
  $\g$ in $\C^{\dalgo}$. The third equality implies that $\u$ is of the form
  $\v^\mB$, for some $\v$ in $\T_{\x^{\mB^{-1}}} \polar(0,\dalgo,V^\mA)$.  Due to
  the form of $\mB$, we can write
  $\pi_{\dalgo}(\u)=\pi_{\dalgo}(\v^\mB)=\pi_{\dalgo}(\v)^{\mB'}$, which implies
  that $\rho_\g(\pi_{\dalgo}(\u)) = \rho_{\g'}(\pi_{\dalgo}(\v))$, with
  $\g'={\mB'}^{-t} \g$.
\end{proof}

Let us choose any $\mB$ and $\mB'$ as in the lemma, with additionally
${\mB'}^{-1} \g_0=\g_1$ (such a $\mB'$ exists, because $\g_1$ is
non-zero). We then let $\ZOfinite \subset \C^{n^2}$ be the non-empty
Zariski open set defined by
$\ZOfinite=\{\mA \mB \mid \mA \in \ZOK_\Klastlastindex\}$. We will now prove
that $\ZOfinite$ fulfills the conditions of
Proposition~\ref{prop:ch5}.

Take $\mA$ in $\ZOfinite$ and write $\mA =\mA' \mB$, with $\mA'$ in
$\ZOK_\Klastlastindex$. Because $\mA'$ is in $\ZOK_\Klastlastindex$,
and thus in $\ZOK_\Klastindex$, we know that either
$\polar(0,\dalgo,V^{\mA'})$ is empty, or it is equidimensional of
dimension $\dalgo-1$, with finitely many singular points. If it is not
empty, the previous lemma shows that $\polar(0,\dalgo,V^{\mA})
=\polar(0,\dalgo,V^{\mA'})^\mB$, so that $\polar(0,\dalgo,V^{\mA})$ is
equidimensional of dimension $\dalgo-1$, with finitely many singular
points as well. This proves the second property.

It remains to prove that $\Kpolar(0,1,\polar(0,\dalgo,V^{\mA}))$ is
finite; for this, as said in the introduction of this section, it is
enough to prove that $\openpolar(0,1,\polar(0,\dalgo,V^{\mA}))$ is
finite. By definition, $\x$ is in
$\openpolar(0,1,\polar(0,\dalgo,V^{\mA}))$ if and only if $\x$ is in
$\reg(\polar(0,\dalgo,V^\mA))$ and $\pi_1$ vanishes on $\T_\x
\polar(0,\dalgo,V^\mA)$.

Remark that there are only finitely many $\x$ in
$\reg(\polar(0,\dalgo,V^\mA))$ that are not in $\openpolar(0,\dalgo,V^\mA)$:
indeed, any such $\x$ is in
$\polar(0,\dalgo,V^\mA)-\openpolar(0,\dalgo,V^\mA)$, which is by construction
contained in the finite set $\sing(V^\mA)$. Thus, to conclude, it is
enough to show that there exist finitely many $\x$ in
$\openpolar(0,\dalgo,V^\mA)$ such that $\pi_1$ vanishes on
$\T_\x \polar(0,\dalgo,V^\mA)$.

\begin{lemma}
  For $\x$ in $\openpolar(0,\dalgo,V^\mA)$, $\pi_1$ vanishes on
  $\T_\x \polar(0,\dalgo,V^\mA)$ if and only if $(\mA',\x^{\mB^{-1}},\g_1)$
  belongs to $\varsigma^{-1}(\mA', \g_1)$.
\end{lemma}
\begin{proof}
  Take $\x$ in $\openpolar(0,\dalgo,V^\mA)$ and let $\y=\x^{\mB^{-1}}$.
  The previous lemma shows that $\T_\x \polar(0,\dalgo,V^{\mA}) =(\T_{\y}
  \polar(0,\dalgo,V^{\mA'}))^\mB$, and that for $\v$ in $\T_\y \polar(0,\dalgo,V^{\mA'})$ and
  $\u=\v^\mB$, we have
  $$\pi_1(\u) = (\rho_{\g_0} \circ \pi_{\dalgo}) (\u) = (\rho_{\g_1} \circ
  \pi_{\dalgo})(\v).$$
  Thus, $\pi_1$ vanishes on $\T_\x \polar(0,\dalgo,V^\mA)$ if and only
  $\rho_{\g_1} \circ \pi_{\dalgo}$ vanishes on
  $\T_\y \polar(0,\dalgo,V^{\mA'})$.  Because, by assumption, $\mA'$ is in
  $\ZOK_\Klastindex$ and (by the previous lemma) $\y$ is in
  $\openpolar(0,\dalgo,V^{\mA'})$, this is the case if and only if
  $(\mA', \y, \g_1)$ is in $\algX$. This is equivalent to
  $(\mA',\y,\g_1)$ belonging to $\varsigma^{-1}(\mA', \g_1)$.
\end{proof}

The construction of $\ZOfinite$ implies that $\varsigma^{-1}(\mA', \g_1)$ is
finite, so our finiteness property is proved.

%%%%%%%%%%%%%%%%%%%%%%%%%%%%%%%%%%%%%%%%%%%%

%%%%%%%%%%%%%%%%%%%%%%%%%%%%%%%%%%%%%%%%%%%%
%\input{proof4.1}
\section{Proof of Theorem~\ref{THEO:MAINABSTRACT}}\label{sec:proof4.1}

In this section, we prove the following statement
(Theorem~\ref{THEO:MAINABSTRACT}) on Algorithm ${\sf MainRoadmap}(V,
C_0)$. To state this result, recall that the recursive calls of ${\sf
  RoadmapRec}$ are organized into a binary tree that we denoted by
$\scrT$.

{\em 
  Assume that $V$ is a $d$-equidimensional algebraic set with finitely
  many singular points and that $V\cap \R^n$ is bounded.  Let
  $C_0\subset \C^n$ be a finite set of points and let
  $(\mA_\nodetau)_{\nodetau \text{~internal node of~} \scrT}$ be a
  family of matrices, with $\mA$ in $\GL(n, e_\nodetau, \QQ)$ for all
  $\nodetau$.

There
  exists a family of non-empty Zariski open sets
  $(\scrOpen_\nodetau)_{\nodetau\text{~internal node of~} \scrT}$, where for all $\nodetau$,
  $\scrOpen_\nodetau$ is in $\GL(n, e_\nodetau)$ and depends on the
  matrices $(\mA_{\nodeoherletter})_{\nodeoherletter \text{~proper ancestor of~}
    \nodetau}$, such that the following holds: if, for all internal nodes $\nodetau$ of $\scrT$,
   $\mA_\nodetau$ is in
  $\scrOpen_\nodetau$ and if it is used as the change of variables in 
the corresponding recursive call of ${\sf RoadmapRec}$, ${\sf
    MainRoadmap}(V, C_0)$ returns a roadmap of $(V, C_0)$.
} 

%%%%%%%%%%%%%%%%%%%%%%%%%%%%%%%%%%%%%%%%%%%%%%%%%%%%%%%%%%%%

\subsection{An induction property} \label{chap:abstractalgo:objects}

Let $V \subset \C^n$ be a $d$-equidimensional algebraic set with
finitely many singular points and let $C$ be a finite set in $\C^n$
which contains $\sing(V)$. 

Let $(\mA_\nodetau)_{\nodetau\text{~internal node of~} \scrT}$ be a
family of matrices, with $\mA$ in $\GL(n, e_\nodetau, \QQ)$ for all
$\nodetau$.  We are going to associate to each node of $\scrT$ some
algebraic sets such as $V_\nodetau, Q_\nodetau, C_\nodetau,
S_\nodetau,\dots$, and an atlas $\bpsi_\nodetau$ of $(V_\nodetau,
Q_\nodetau, S_\nodetau)$; if $\tau$ is an internal node, we also
associate to it the subset $\scrOpen_\nodetau$ of $\GL(n, e_\nodetau)$
mentioned in the theorem.  In order to initialize the construction, we
also consider an atlas $\bpsi$ of $(V,\bullet,\sing(V))$ (such an atlas 
always exist; see Lemma~\ref{sec:atlas:lemma:glob}).

The construction is by induction on the nodes $\nodetau$ of $\scrT$; the
induction property will be written as follows:
\begin{enumerate}
\item [${\sfTa:}$] There exists a family of non-empty Zariski open sets
  $(\scrOpen_{\nodeoherletter})_{\nodeoherletter\text{~proper ancestor of~} \nodetau}$,
  with $\scrOpen_{\nodeoherletter}$ in $\GL(n, e_{\nodeoherletter})$ for all $\nodeoherletter$,
  and with the following properties.
  Suppose that $\mA_{\nodeoherletter}$ belongs to $\scrOpen_{\nodeoherletter}$ for
  all proper ancestors  $\nodeoherletter$ of $\nodetau.$ Then, we associate to
  the node $\nodetau$ the objects
  $(V_{\nodetau},Q_\nodetau,S_\nodetau,C_\nodetau,\bpsi_\nodetau)$,
  which satisfy the following:
\begin{enumerate}
\smallskip
\item[${\sf t_{1}.}$] $Q_\nodetau$ is a finite subset of $\C^{e_\nodetau}$ and $S_\nodetau,C_\nodetau$
  are finite subsets of $\C^n$;
\smallskip
\item[${\sf t_{2}.}$] $V_\nodetau,S_\nodetau,C_\nodetau$ lie over $Q_\nodetau$;
\smallskip
\item[${\sf t_{3}.}$] either $V_\nodetau$ is empty, or $V_\nodetau$ lies
  over $Q_\nodetau$ and is $d_\nodetau$-equidimensional with finitely many
  singular points, in which case $\bpsi_\nodetau$ is an atlas of
  $(V_\nodetau,Q_\nodetau,S_\nodetau)$;
\smallskip
\item[${\sf t_{4}.}$] the inclusion $S_\nodetau \subset C_\nodetau$ holds.
\end{enumerate}
\end{enumerate}

The root $\rho$ of ${\scrT}$ (which has no proper ancestor) satisfies
${\sfTa}$, provided we define
$$V_\rho=V, \quad Q_\rho=\bullet,\quad S_\rho = \sing(V_\rho),\quad
C_\rho=C, \quad \bpsi_\rho = \bpsi.$$ Suppose now that an internal
node $\nodetau$ satisfies ${\sfTa}$. We define the subset $\scrOpen_\nodetau$ of
  $\GL(n, e_\nodetau)$ as follows:
\begin{itemize}
\item If $\mA_{\nodeoherletter}$ belongs to $\scrOpen_{\nodeoherletter}$ for all
  proper ancestors $\nodeoherletter$ of $\nodetau$, and if $V_\nodetau$ is
  empty, we take $\scrOpen_\nodetau=\GL(n, e_\nodetau)$.
\smallskip
\item If $\mA_{\nodeoherletter}$ belongs to $\scrOpen_{\nodeoherletter}$ for all
  proper ancestors $\nodeoherletter$ of $\nodetau$, and if $V_\nodetau$ is
  not empty, we define $\scrOpen_\nodetau$ as the intersection of the sets 
  $$\scrGpolar(\bpsi_\nodetau,V_\nodetau,Q_\nodetau,S_\nodetau,{\dalgo_\nodetau}), \qquad 
  \ZOfinite(V_\nodetau,Q_\nodetau,\dalgo_\nodetau)
  \qquad
  \text{ and }\qquad
  \scrGfiber(\bpsi_\nodetau,V_\nodetau,Q_\nodetau,S_\nodetau,{\dalgo_\nodetau})% \mathscr{K}(V_{\nodetau},Q_\nodetau,\dalgo_\nodetau)
  $$
  of Propositions~\ref{prop:ch4},~\ref{prop:ch5}
  and~\ref{sec:atlas:prop:summary1}.
\smallskip
\item Else, we take $\scrOpen_\nodetau=\GL(n, e_\nodetau)$.
\end{itemize}
In the first two cases, we then define $B_\nodetau,\fiber2_\nodetau,C'_\nodetau,C''_\nodetau$,
$W_\nodetau=\polar(e_\nodetau,\dalgo_\nodetau,V_\nodetau^{\mA_\nodetau})$ and
$V''_\nodetau=\fbr(V_\nodetau^{\mA_\nodetau}, \fiber2_\nodetau)$ as in algorithm ${\sf
  RoadmapRec}$.

\begin{lemma}\label{sec:abstractalgo:lemma:correctnessA1}
  If an internal node $\nodetau$ satisfies $\sfTa$, and if
  $\mA_{\nodeoherletter}$ belongs to $\scrOpen_{\nodeoherletter}$ for all ancestors
  $\nodeoherletter$ of $\nodetau$ (including $\tau$ itself), then
  $B_\nodetau,\fiber2_\nodetau,C'_\nodetau,C''_\nodetau$ are finite.
\end{lemma}
\begin{proof}
  We are necessarily in one of the first two cases in the previous 
case discussion.
  When $V_\nodetau$ is empty, all statements are clear. Otherwise, the
  finiteness of $B_\nodetau$, and thus of its projection $\fiber2_\nodetau$,
  are consequences of Proposition~\ref{prop:ch5}.  The first item in
  Proposition~\ref{sec:atlas:prop:summary1} implies that $C'_\nodetau$ is
  finite, and $C''_\nodetau$ is finite because it is a subset of
  $C'_\nodetau$.
\end{proof}

Let $\nodetau',\nodetau''$ be the children of an internal node $\nodetau$. 
If we are under the assumptions of the previous lemma, using in particular
Definitions~\ref{sec:atlas:notation:polar}
and~\ref{sec:atlas:notation:fiber}, we set
$$
V_{\nodetau'}=W_\nodetau,\ \
Q_{\nodetau'}=Q_\nodetau,\ \
S_{\nodetau'}=S_\nodetau^{\mA_\nodetau},\ \
C_{\nodetau'}=C'_\nodetau,\ \
\bpsi_{\nodetau'}=\atlaspolar(\bpsi_\nodetau^{\mA_\nodetau}, V^{\mA_\nodetau}_\nodetau,Q_\nodetau, S^{\mA_\nodetau}_\nodetau,\dalgo_\nodetau)
$$ 
and
$$
V_{\nodetau''}=V''_\nodetau,\ \ 
Q_{\nodetau''}=\fiber2_\nodetau,\ \
S_{\nodetau''}= \fbr(S_\nodetau^{\mA_\nodetau}\cup W_\nodetau, \fiber2_\nodetau),\ \
C_{\nodetau''}=C''_\nodetau,$$
and finally $$
\bpsi_{\nodetau''}=\atlasfiber(\bpsi_\nodetau^{\mA_\nodetau}, V_\nodetau^{\mA_\nodetau},Q_\nodetau,S_\nodetau^{\mA_\nodetau},\fiber2_\nodetau).
$$
Note that, by the previous lemma, $C_{\nodetau'},Q_{\nodetau'}$ and
$C_{\nodetau''},Q_{\nodetau''}$ are finite. 

\begin{lemma}\label{sec:abstractalgo:lemma:correctnessA2}
  If an internal node $\nodetau$ satisfies $\sfTa$, its children
  $\nodetau'$ and $\nodetau''$ satisfy $\sfTa$.
\end{lemma}
\begin{proof}
  This is mostly a routine verification. Property $\sfTa$ at either
  $\nodetau'$ or $\nodetau''$ amounts to assuming that
  $\mA_{\nodeoherletter}$ belongs to $\scrOpen_{\nodeoherletter}$ for
  all ancestors $\nodeoherletter$ of $\nodetau$, including $\nodetau$
  itself. In particular, we are under the assumptions of the previous
  lemma.

  By definition, $Q_{\nodetau'}=Q_\nodetau \subset \C^{e_{\nodetau'}}$
  is finite; as pointed out above, the previous lemma implies that
  this is also the case for $Q_{\nodetau''} \subset
  \C^{e_{\nodetau''}}$. Moreover, $S_{\nodetau'}$ is finite by
  construction, and $S_{\nodetau''}$ is finite by
  Proposition~\ref{sec:atlas:prop:summary1}. Thus, item ${\sf
    t_{1}}$ is proved.

  Then, one easily sees that $V_{\nodetau'},S_{\nodetau'},C_{\nodetau'}$ lie over
  $Q_{\nodetau'}=Q_{\nodetau}$; the same holds for $\nodetau''$ by construction.
  Thus, item ${\sf t_{2}}$ is proved. Next, we have to prove that the
  following holds:
  \begin{itemize}
  \item either $V_{\nodetau'}$ is empty, or $V_{\nodetau'}$ lies over
    $Q_{\nodetau'}$ and is $d_{\nodetau'}$-equidimensional with finitely many
    singular points, in which case $\bpsi_{\nodetau'}$ is an atlas of
    $(V_{\nodetau'}, Q_{\nodetau'}, S_{\nodetau'})$;
\smallskip
  \item either $V_{\nodetau''}$ is empty, or $(V_{\nodetau''}$ lies over
    $Q_{\nodetau''}$ and is $d_{\nodetau''}$-equidimensional with finitely
    many singular points, in which case $\bpsi_{\nodetau''}$ is an atlas
    of $(V_{\nodetau''}, Q_{\nodetau''}, S_{\nodetau''})$.
  \end{itemize}
  When $V_\nodetau$ is empty, both $V_{\nodetau'}$ and $V_{\nodetau''}$ are empty.
  Otherwise, both statements are consequences of
  Propositions~\ref{prop:ch4} and~\ref{sec:atlas:prop:summary1}, so
  ${\sf t_{3}}$ is proved. We finally prove ${\sf t_{4}}$: because
  $C_{\nodetau'}=C'_\nodetau$ contains $C_\nodetau^{\mA_\nodetau}$, which itself
  contains $S_\nodetau^{\mA_\nodetau}=S_{\nodetau'}$ (by induction assumption),
  property ${\sf t_{4}}$ holds for $\nodetau'$. For $\nodetau''$, recall
  that
  $S_{\nodetau''} = \fbr(S_\nodetau^{\mA_\nodetau} \cup W_\nodetau,\fiber2_\nodetau)$,
  whereas
  $C_{\nodetau''} = \fbr(C_\nodetau^{\mA_\nodetau} \cup W_\nodetau,\fiber2_\nodetau)$, so
  the claim follows from the similar property at $\nodetau$.
\end{proof}

Thus, repeated applications of the previous lemma allow us to define a
family of non-empty Zariski open sets $\scrOpen_\nodetau \subset
\GL(n, e_\nodetau)$, for $\nodetau$ internal node of $\scrT$, for which
all nodes of $\scrT$ satisfy property $\sfTa$.

%% Thus, if $\mA_\nodetau$ satisfies ${\sfTb}$, both children of $\nodetau$
%% satisfy the induction assumption. This leads us to the following
%% definition of a ``lucky'' choice for the set of all matrices
%% $\mA_\nodetau$.

%% \begin{definition}\label{def:abstract:H}
%%   Let $V \subset \C^n$ be a $d$-equidimensional algebraic set with
%%   finitely many singular points, let $C\subset \C^n$ be a finite set
%%   that contains $\sing(V)$ and let $\bpsi$ be an atlas of
%%   $(V,\bullet,\sing(V))$. Let further
%%   $\scrA=(\mA_\nodetau)_{\nodetau \in \scrT}$ be a family of matrices, with
%%   $\mA_\nodetau$ in $\GL(n,e_\nodetau, \QQ)$ for all $\nodetau$ in $\scrT$.

%%   We say that $\scrA$ satisfies assumption ${\sfT}(V,C,\bpsi)$ if for
%%   all $\nodetau$ in $\scrT$, $\nodetau$ satisfies $\sfTa$ and $\mA_\nodetau$
%%   satisfies $\sfTb$.
%% \end{definition}

%% When there is no ambiguity on $V, C, \bpsi$ we simply write that
%% $\scrA$ satisfies $\sfT$. It is important to note that, in order to
%% ensure ${\sfT}$, each matrix $\mA_\nodetau$ has to lie in the non-empty
%% Zariski open set $\scrOpen_\nodetau$ defined above, which depends on
%% $V,C,\bpsi$ and all previous changes of variables.

%%%%%%%%%%%%%%%%%%%%%%%%%%%%%%%%%%%%%%%%%%%%%%%%%%%%%%%%%%%%

\subsection{Proof of the theorem} 

In the previous subsection, we showed how to define all objects
attached to $\scrT$; we now prove that the algorithm
${\sf MainRoadmap}$ correctly returns a roadmap of $(V,C_0)$. The proof
is similar to that of our first generalization of Canny's
algorithm~\cite{SaSc11}, adapted to the fact that we handle more
general polar varieties.

The key ingredient is a connectivity result which is part
of~\cite[Theorem 14]{SaSc11}.  As stated, the theorem in that
reference also handles the transfer of some complete intersection
properties to systems defining the polar varieties we were
considering. These complete intersection properties do not hold in our
more general context, but the proof of the connectivity statement
given in~\cite[Section 4.3]{SaSc11} does not use them.

The following statement combines this connectivity result
and~\cite[Proposition 2]{SaSc11}, which ensures that taking the union
of roadmaps of the polar variety $W$ and the fiber
$V''=\fbr(V,\pi_{e+\dalgo-1}(B))$ with
$B=\Kpolar(e,1,V)\cup \Kpolar(e,1,W)\ \cup\ C$, one obtains a roadmap
of $V$. Observe that Lemma \ref{prelim:lemma:polarpolar} implies that
$B=\Kpolar(e,1,W)\cup C$; this yields the following proposition.
 
\begin{proposition}\label{sec:abstractalgo:theo:connect}
  Let $V$ and $Q$ be algebraic sets in $\C^n$ and $\C^e$ such that $V$
  lies over $Q$, is $d$-equidimensional with finitely many singular
  points and $V\cap \R^n$ is bounded. Let $C\subset \C^n$ be a finite
  set of points and let $\dalgo$ be in $\{1,\dots,d\}$. Suppose that
  the following assumptions hold:
  \begin{itemize}
  \item $V\cap \R^n$ is bounded;
\smallskip
  \item either the set $W=\polar(e,\dalgo,V)$ is empty, or $W$ is
    $(\dalgo-1)$-equidimensional with finitely many singular points;
\smallskip
  \item the set $B= \Kpolar(e,1,W)\ \cup\ C$ is finite;
\smallskip
  \item either the set $V''=\fbr(V,\fiber2)$, with
    $\fiber2=\pi_{e+\dalgo-1}(B)$, is empty, or $V''$ is
    $(d-(\dalgo-1))$-equidimensional with finitely many singular
    points;
\smallskip
  \item the set $C'=C \cup \fbr(W,\fiber2)$ is finite.
  \end{itemize}
  Let further $C''=\fbr(C',\fiber2)$. If $R'$ and $R''$ are roadmaps of
  respectively $(W, C')$ and $(V'',C'')$, then $R' \cup R''$ is a
  roadmap of $(V,C)$.
\end{proposition}

This proposition allows us to prove Theorem~\ref{THEO:MAINABSTRACT}.
In the previous section, we defined a family of non-empty Zariski open
sets $\scrOpen_\nodetau \subset \GL(n, e_\nodetau)$, for $\nodetau$
internal node of $\scrT$, for which all nodes of $\scrT$ satisfy
property $\sfTa$. Suppose now, as in the theorem, that $\mA_\nodetau$
is in $\scrOpen_\nodetau$ for all internal nodes $\nodetau$ of
$\scrT$. By property $\sfTa$, we associate to each node $\nodetau$ of
$\scrT$ the objects
$V_{\nodetau},Q_\nodetau,S_\nodetau,C_\nodetau,\bpsi_\nodetau$, which
satisfy properties ${\sf t_{1}},\dots,{\sf t_{4}}$.

To each node $\nodetau$ of the tree $\scrT$, we can then associate an
algebraic set $R_\nodetau$ in the obvious manner:
\begin{itemize}
\item if $\nodetau$ is a leaf, we define ${R}_\nodetau$ as $V_\nodetau$,
\smallskip
\item else, letting $\nodetau'$ and $\nodetau''$ be the children of $\nodetau$, we
  denote by ${R}_\nodetau$ the union of the curves ${R}_{\nodetau'}^{\mA_{\nodetau}^{-1}}$
  and ${R}_{\nodetau''}^{\mA_{\nodetau}^{-1}}$.
\end{itemize}
 
\begin{lemma}\label{sec:abstractalgo:lemma:correctness}
  For any node $\nodetau$ of $\scrT$, ${R}_\nodetau$ is a roadmap of
  $(V_\nodetau, C_\nodetau)$.
\end{lemma}
\begin{proof}
  First, remark that if $V\cap \R^n$ bounded, $V_\nodetau \cap \R^n$ is
  bounded for any $\nodetau$ in $\scrT$: indeed, all these algebraic sets
  are obtained from $V$ by a combination of either taking polar
  varieties or fibers, through changes of variables with coefficients
  in $\QQ$.

  The proof of the lemma is by decreasing induction on the depth of
  $\nodetau$.  If $\nodetau$ is a leaf (i.e. $d_\nodetau=1$), we know from
  ${\sfTa}$ that $V_\nodetau$ is either empty or $1$-equidimensional, so
  our assertion holds.  Thus, we can suppose that $\nodetau$ is not a leaf
  and we let $\nodetau'$ and $\nodetau''$ be the children of~$\nodetau$.

  If $V_\nodetau$ is empty, both $V_{\nodetau'}$ and $V_{\nodetau''}$ are empty,
  so (by the induction assumption) $R_{\nodetau'}$ and $R_{\nodetau''}$ are
  empty; as a result, $R_{\nodetau}$ is empty, and our claim holds. Else,
  assumption $\sfTa$ implies that $V_\nodetau$ is $d_\nodetau$-equidimensional
  with finitely many singular points, so that
  $V_\nodetau^{\mA_\nodetau}$ does too; besides, similar statements
  hold for $V_{\nodetau'}$ and $V_{\nodetau''}$, and
  all sets $B_\nodetau$ and $C'_\nodetau$ are finite.
  
  We are thus in a position to apply
  Proposition~\ref{sec:abstractalgo:theo:connect}. Together with the
  induction assumption, that proposition implies that $R_{\nodetau'} \cup
  R_{\nodetau''}$ is a roadmap of $(V_\nodetau^{\mA_\nodetau},C_\nodetau^{\mA_\nodetau})$.
  We deduce that ${R}_\nodetau={R}_{\nodetau'}^{\mA_\nodetau^{-1}}\cup
  R_{\nodetau''}^{\mA_\nodetau^{-1}}$ is a roadmap of $(V_\nodetau, C_\nodetau)$.
\end{proof}

Applying Lemma \ref{sec:abstractalgo:lemma:correctness} to $V$ and
$C_0\cup \sing(V)$ shows that ${\sf MainRoadmap}(V, C)$ returns a
roadmap of $(V,C_0 \cup \sing(V))$, which is in particular a roadmap
of $(V,C_0)$.  This proves Theorem~\ref{THEO:MAINABSTRACT}.

%%%%%%%%%%%%%%%%%%%%%%%%%%%%%%%%%%%%%%%%%%%%

%%%%%%%%%%%%%%%%%%%%%%%%%%%%%%%%%%%%%%%%%%%%
%\input{proof5.13}
\section{Proof of Proposition~\ref{sec:lagrange:propertyP}}\label{chap:lagrange}

Let us recall the statement of
Proposition~\ref{sec:lagrange:propertyP}.  {\em Let
  $L=(\Gamma,\scrQ,\scrS)$ be a generalized Lagrange system and let
  $\F=(F_1,\dots,F_P)$ in $\QQ[\X, \L]$ and $e \ge 0$ be as in
  Definition \ref{def:GLS}. If $L$ has the global normal form
  property, the following holds:
  \begin{itemize}
  \item the Jacobian matrix $\jac(\F,e)$ has full rank $P$ at every
    point $(\x,\bell)$ in $\Cons(L)$;
\smallskip
  \smallskip
  \item the restriction $\pi_\X: \Cons(L) \to \Proj(L)$ is a bijection.
  \end{itemize}}

We start with two useful lemmas.

\begin{lemma} \label{lemma:conseq0} Let $L=(\Gamma,\scrQ,\scrS)$ be a
  generalized Lagrange system, with $U=\Proj(L)$, $V=\Clos{(L)}$,
  $Q=Z(\mathscr{Q})$ and $S=Z(\mathscr{S})$, and let
  $\F=(F_1,\dots,F_P)$ in $\QQ[\X, \L]$ and $e \ge 0$ be as in
  Definition \ref{def:GLS}.  Suppose that $\phi=(\polmu, \poldelta,
  \h, \H)$ is a local normal form for $L$.  Then, the following
  equalities hold in $\C^n$:
  $$
  \begin{array}{ccl}
    \Open(\polmu \poldelta) \cap U&=&  \Open(\polmu\poldelta) \cap \fbr(V(\h),Q) -S\\[2mm]
                                        &=&\Open(\polmu\poldelta) \cap V -S.
  \end{array}
  $$
\end{lemma}
\begin{proof} For the first equality, note that $U$ is contained in
  $\pi_e^{-1}(Q)$.  Thus, for $\x$ in
  $\Open(\polmu\poldelta) \cap \pi_e^{-1}(Q)$, we have to prove
  that $\x$ is in $U$ if and only if $\h(\x)=0$ and $\x$ is not in
  $S$.  Suppose that $\x$ is in $U$ and let $\F$ be the sequence of
  polynomials evaluated by $\Gamma$ as in Definition
  \ref{def:GLS}. Thus, there exists $\bell\in\C^{N-n}$ such that
  $\F(\x,\bell)=0$.  Because $\pi_e(\x)$ is in $Q$, and
  $\polmu(\x) \poldelta(\x)$ is not zero, ${\lnf_3}$ implies that
  $(\x,\bell)$ cancels $\H$ and so $\x$ cancels $\h$; besides, by
  definition of $U$, $\x$ is not in $S$. We are done for the first
  inclusion.

  Conversely, suppose that $\x$ cancels $\h$ and does not belong to
  $S$. 
  Since $\polmu(\x)\poldelta(\x)\neq 0$, we can determine $\bell\in\C^{N-n}$
  using the $\L$-component of $\H$, as no denominator vanishes. Then,
  $(\x,\bell)$ is a root of $\H$, and thus (by ${\lnf_3}$) of
  $\F$. Finally, we assumed that $\x$ does not belong to $S$, so
  $(\x,\bell)$ is in $\Cons(L)$, and $\x$ is in $U=\Proj(L)$, as
  claimed.

  To prove the second equality, observe that, through property
  ${\sfC_2}$ of charts, ${\lnf_4}$ implies that
  $ \Open(\polmu) \cap V -S= \Open(\polmu) \cap
  \fbr(V(\h),Q) -S$ and intersect with $\Open(\poldelta)$.
\end{proof}

Next, we relate the Jacobian matrix of the polynomials $\F$ in a
generalized Lagrange system $L=(\Gamma,\scrQ,\scrS)$ and that of the
polynomials $\H$ in a local normal form.

\begin{lemma}\label{lemma:mS}
  Let $L=(\Gamma,\scrQ,\scrS)$ be a generalized Lagrange system, with
  $Q=Z(\scrQ) \subset \C^e$, let $\F$ in $\QQ[\X, \L]$ be the sequence
  of polynomials evaluated by $\Gamma$ as in Definition \ref{def:GLS}
  and let $I$ be the defining ideal of $Q$.
    
  Suppose that $\phi=(\polmu, \poldelta, \h, \H)$ is a local normal
  form for $L$, with $\h$ of cardinality $c$. Then, there exists a
  $(P\times P)$ matrix $\mS$ with entries in
  $\QQ[\X]_{\polmu \poldelta}$, such that
  $\jac(\H,e) = \mS\, \jac(\F,e)$ holds over
  $\QQ[\X,\L]_{\polmu\poldelta}/\langle \F , I\rangle$ and such that
  $\det(\mS)$ divides any $c$-minor of $\jac(\h,e)$ in
  $\QQ[\X,\L]_{\polmu\poldelta}/\langle \F , I\rangle$.
\end{lemma}
\begin{proof}
  Since the ideal $I$ is generated by polynomials in
  $\QQ[X_1,\dots,X_e]$, the equality $\langle \H \rangle=\langle \F
  \rangle$ in $\QQ[\X,\L]_{\polmu \poldelta}/I$ implies the existence
  of a $(P\times P)$ matrix $\mS$ with entries in
  $\QQ[\X,\L]_{\polmu\poldelta}/I$ such that $\jac(\H,e)=\mS\,
  \jac(\F,e)$ over $\QQ[\X,\L]_{\polmu\poldelta}/\langle \F ,
  I\rangle$. We can use the $\L$-component of $\H$ to eliminate all
  $\L$ variables appearing in $\mS$, so as to take all entries of
  $\mS$ in $\QQ[\X]_{\polmu\poldelta}$; this maintains equality modulo
  $\langle \F , I\rangle$, so the first point is proved.

  Let then $m'$ be a $c$-minor of $\jac(\h,e)$, and let $\mm'$ be the
  corresponding $(c\times c)$ submatrix of $\jac(\h,e)$. We can embed
  $\mm'$ into a unique $(P \times P)$ submatrix $\mM'$ of
  $\jac(\H,e)$, by adjoining to it all rows corresponding to the
  $\L$-component of $\H$, and all columns corresponding to the $\L$
  variables. Due to \new{the} block structure of $\H$, and thus of
  $\jac(\H,e)$, we have that $\det(\mM')=\det(\mm')=m'$.

  Let finally $\mM''$ \new{be} the $(P\times P)$ submatrix of
  $\jac(\F, e)$ obtained by selecting the same columns as those for
  $\mM'$. From the equality $\jac(\H,e)=\mS\, \jac(\F,e)$, we obtain
  $\mM'=\mS \,\mM''$ over
  $\QQ[\X,\L]_{\polmu\poldelta}/\langle \F, I\rangle$. We deduce that
  the determinant of $\mS$ divides that of $\mM'$, which is $m'$, in
  $\QQ[\X,\L]_{\polmu\poldelta}/\langle \F, I\rangle$.
\end{proof}

\begin{corollary}\label{coro:rkF}
  Let $L=(\Gamma,\scrQ,\scrS)$ be a generalized Lagrange system, with
  $U=\Proj(L)$, $Q=Z(\mathscr{Q})$ and $S=Z(\mathscr{S})$ and $\F$ in
  $\QQ[\X, \L]$ as in Definition \ref{def:GLS}.

  Suppose that $\phi=(\polmu, \poldelta, \h, \H)$ is a local normal
  form for $L$.  For $\x$ in $\Open(\polmu \poldelta) \cap U$, and for
  all $\bell$ such that $(\x,\bell)$ is in $\Cons(L)$, the Jacobian
  matrix $\jac(\F,e)$ has full rank $P$ at $(\x,\bell)$.
\end{corollary}
\begin{proof}
  Let $\x$ and $\bell$ be as in the statement of the corollary and let
  $V=\Clos{(L)}$. Lemma~\ref{lemma:conseq0} implies that
  $\Open(\polmu\poldelta)\cap U$ is contained in
  $\Open(\polmu)\cap V-S$.  Consequently, by \new{property} ${\lnf_4}$
  \new{of local normal forms} and property $\sfC_4$ of charts, the
  Jacobian matrix $\jac(\h,e)$ has full rank $c$ at $\x$; this easily
  implies that the matrix $\jac(\H,e)$ has full rank $P$ at
  $(\x,\bell)$. Because $(\x,\bell)$ is in $V(\F,I)$,
  Lemma~\ref{lemma:mS} above implies that the equality
  $\jac(\H,e) = \mS\, \jac(\F,e)$ holds at $(\x,\bell)$. Thus,
  $\jac(\F,e)$ has full rank $P$ at $(\x,\bell)$.
\end{proof}

\begin{proof}[of Proposition~\ref{sec:lagrange:propertyP}]
  Let $U=\Proj(L)$, $V=\Clos{(L)}$, $Q=Z(\scrQ)$ and $S=Z(\scrS)$; let
  further $\bphi=(\phi_i)_{1\leq i \leq s}$ with $\phi_i=(\polmu_i,
  \poldelta_i, \h_i, \H_i)$ be a global normal form of $L$ and
  $(\x,\bell)$ be in $\Cons(L)$, so that $\x$ is in
  $U=\Proj(L)$. Since $U\subset V-S$, property ${\gnf_2}$ of global
  normal forms implies that there exists $i$ such that $\x$ is in
  $\Open(\polmu_i)$.  By ${\lnf_5}$, $\x$ is in $\Open(\polmu_i
  \poldelta_i) \cap U$, and Corollary~\ref{coro:rkF} implies that
  $\jac(\F, e)$ has full rank $P$ at $(\x, \bell)$. We have proved the
  first point.

  Next, we prove that the restriction $\pi_\X: \Cons(L) \to \Proj(L)$
  is a bijection. By construction, we know that it is onto, so we have
  to prove that it is injective. Let thus $\x$ be in $U$. As we saw
  above, since $\bphi$ is a global normal form, there exists
  $i\in\{1,\dots, s\}$ such that $\x$ is in
  $\Open(\polmu_i\poldelta_i)\cap U$. If $\bell \in \C^{N-n}$ is such
  that $(\x,\bell)$ is in $\Cons(L)$, then $(\x,\bell)$ cancels
  $\langle \F,I\rangle$, so by ${\lnf_3}$, it cancels
  $\langle \H_i,I \rangle$.  As a result, the value of $\bell$ is
  uniquely determined, as it is obtained by evaluating the
  ${\bf L}$-component of $\H_i$ at $\x$.
\end{proof}

Using this result, we exhibit the relationships between the sets
$\Cons(L)$, $\Proj(L)$ and $\Clos{(L)}$ associated to a generalized
Lagrange system $L$, and the set $\freg(\F,Q)$ defined in
Subsection~\ref{chap:prelim:sec:definitions}, where $\F$ and $Q$ are
as in Definition \ref{def:GLS}. These claims will be used in
Section~\ref{proof:solvelagrange:prop:basicsolve}.

\begin{lemma}\label{lemma:degree:clos}
  Let $L=(\Gamma,\scrQ,\scrS)$ be a generalized Lagrange system, with
  $Q=Z(\scrQ)$, $S=Z(\scrS)$, $\F$ in $\QQ[\X, \L]$ and $d=N-e-P$ as
  in Definition \ref{def:GLS}. Let further $Y=\freg(\F,Q) \subset
  \C^N$. If $L$ has the global normal form property, the following
  holds:
  $$\Cons(L)=Y-\pi_\X^{-1}(S),  \quad
  \Proj(L)=\pi_\X(Y-\pi_\X^{-1}(S)).$$ In addition, $Y$,
  $\Cons(L)$ and $\Clos{(L)}$ are $d$-equidimensional.
\end{lemma}
\begin{proof}
  Using Proposition~\ref{sec:lagrange:propertyP}, we know that $\jac(\F,e)$
  has maximal rank at any point of
  $\Cons(L)=\fbr(V(\F),Q)-\pi_\X^{-1}(S)$; this implies that
  $\Cons(L)=Y-\pi_\X^{-1}(S)$.  The last equality
  is straightforward from the fact that $\Proj(L)=\pi_\X(\Cons(L))$. 
%%  and $\Clos{(L)}$ is the Zariski closure of $\Proj(L)$.

  As was mentioned in Subsection~\ref{chap:prelim:sec:definitions}, the
  Jacobian criterion shows that $Y$ is either empty or
  $d$-equidimensional.  By the global normal form property, $\Clos{(L)}$
  is not empty, so neither is $Y$; thus, $\Cons(L)$ is
  $d$-equidimensional as well (in the sense that its Zariski closure
  is) and the only missing part is the fact that $\Clos{(L)}$ is
  $d$-equidimensional.

  This will follow from the second item in
  Proposition~\ref{sec:lagrange:propertyP}, which states that the
  projection $\Cons(L) \to \Proj(L)$ is one-to-one. Let indeed
  $\closedZ$ be the Zariski closure of $\Cons(L)$, and let
  $\closedZ=\cup_{1\leq i \le s}\closedZ_i$ be its decomposition into
  irreducible; we saw above that all $\closedZ_i$ have dimension $d$.

  For $i$ in $\{1,\dots,s\}$, define $\lcswi=\Cons(L) \cap
  \closedZ_i$; each $\lcswi$ is a locally closed set, with Zariski
  closure $\closedZ_i$, and their union is equal to $\Cons(L)$. This
  in turn implies that $\Proj(L)$ is the union of the sets
  $\pi_\X(\lcswi)$. Denoting by $V_i$ the Zariski closure of
  $\pi_\X(\lcswi)$, this also implies that $\Clos{(L)}$ is the union of
  $\cup_{i=1}^s V_i$.

  Because the Zariski closure $V_i$ of $\pi_\X(\lcswi)$ coincides
  with that of $\pi_\X(\closedZ_i)$, it must be irreducible. The
  inequality $\dim(V_i) \le d$ clearly holds for all $i$; on the other
  hand, by Proposition~\ref{sec:lagrange:propertyP}, the fibers of the
  restriction of $\pi_\X$ are all finite, so Lemma~\ref{lemma:5.2.5}
  implies that $d \le \dim(V_i)$ holds as well for all $i$. This
  implies that $\Clos{(L)}$ is $d$-equidimensional, as claimed.
\end{proof}

% Suppose that a pair $(V, Q)$ satisfies $(\AS,d,e)$; hence, $\polar(e,1,V)$ is
% well-defined for such a $V$. The following lemma is crucial in order
% to compute $\polar(e,1,V)$ when $V$ is given as $\Clos{(L)}$, where in
% addition we suppose that $(L;\polar(e,1,V))$ has the global normal form
% property.
% Similarly to what we just did for polar varieties, we next show how we
% can compute fibers of the form $\fbr(\Clos{(L)},\fiber2)$.

%%%%%%%%%%%%%%%%%%%%%%%%%%%%%%%%%%%%%%%%%%%%

%%%%%%%%%%%%%%%%%%%%%%%%%%%%%%%%%%%%%%%%%%%%
%\input{proof5.18}
\section{Proof of Proposition~\ref{sec:lagrange:prop:transfer:polar}}\label{chapter:constructionspolar}

This section is devoted to the proof of
Proposition~\ref{sec:lagrange:prop:transfer:polar}, whose statement is
as follows: {\em Let $Q \subset \C^e$ be a finite set and let $V
  \subset \C^n$ and $S\subset \C^n$ be algebraic sets lying over $Q$,
  with $S$ finite.  Suppose that $V$ is equidimensional of dimension
  $d$, with finitely many singular points.
  
  Let $\bpsi$ be an atlas of $(V,Q,S)$, let $\dalgo$ be an integer in
  $\{2,\dots,d\}$ such that $\dalgo \le (d+3)/2$, and let
  $\mA\in\GL(n,e)$ be in the open set
  $\scrGpolar(\bpsi,V,Q,S,\dalgo)$ defined in
  Proposition~\ref{prop:ch4}; write
  $W=\polar(e,\dalgo,V^\mA)$.

  Let $L=(\Gamma,\scrQ,\scrS)$ be a generalized Lagrange system such
  that $V=\Clos{(L)}$, $Q=\Zeroes(\scrQ)$ and $S=\Zeroes(\scrS)$.  Let
  $\scrY=(Y_1,\dots,Y_r)$ be algebraic sets in $\C^n$ and let finally
  $\bphi$ be a global normal form for $(L; (W^{\mA^{-1}}, \scrY))$ such
  that $\bpsi$ is the associated atlas of $(V,Q,S)$.

  There exists a non-empty Zariski open set
  $\mathscr{I}(L,\bphi,\mA,\scrY)\subset \C^{P}$ such that for all
  $\u$ in $\mathscr{I}(L,\bphi,\mA,\scrY) \cap \QQ^P$, the following holds:
  \begin{itemize}
  \item $\Polarlag(L^\mA, \u, \dalgo)$ is a generalized Lagrange system
    that defines $W$;
    \smallskip
  \item If $W$ is not empty, then $(\Polarlag(L^\mA, \u, \dalgo);
    \scrY^\mA)$ admits a global normal form whose atlas is $\atlaspolar(\bpsi^\mA,V^\mA,Q,S^\mA,\dalgo)$ (Definition
    \ref{sec:atlas:notation:polar}).
  \end{itemize}
}

%%%%%%%%%%%%%%%%%%%%%%%%%%%%%%%%%%%%%%%%%%%%%%%%%%%%%%%%%%%%

\subsection{Local analysis}\label{sssec:lag}

First, we deal with local normal forms. In order to prepare for the
proof of the main proposition in the next subsection, we introduce
here extra statements related to a new set of points $\scrX$.

\begin{proposition}\label{prop:local}
  Let $Q \subset \C^e$ be a finite set and let $V \subset \C^n$ and
  $S\subset \C^n$ be algebraic sets lying over $Q$, with $S$ finite.
  Suppose that $V$ is equidimensional of dimension $d$, with finitely
  many singular points.

  Let $L=(\Gamma,\scrQ,\scrS)$ be a generalized Lagrange system of
  type $(k,\n,\p,e)$ that defines $V$, with $Q=\Zeroes(\scrQ)$ and
  $S=\Zeroes(\scrS)$; write $\n=(n,n_1,\dots,n_k)$. Let
  $\phi=(\polmu,\poldelta,\h,\H)$ be a local normal form for $L$ and let
  $\psi=(\polmu,\h)$ be the associated chart of $(V,Q,S)$; write
  $\h=(h_1,\dots,h_c)$ and
  $$\H=\left ( h_1,\dots,h_c,
    (L_{1,j}-\rho_{1,j})_{j=1,\dots,n_1},\dots,
    (L_{k,j}-\rho_{k,j})_{j=1,\dots,n_k}\right ).$$
  Let $\dalgo$ be an integer in $\{2,\dots,d\}$, such that
  $\dalgo \le (d+3)/2$, let $\mA \in \GL(n,e)$ be in the open set
  $\scrGpolarchart(\psi,V,Q,S,\dalgo)$ defined in
  Lemma~\ref{sec:chart:lemma:polarchart} and let
  $W=\polar(e,\dalgo,V^\mA)$.

  Let $m'$ and $m''$ be respectively a $c$-minor of $\jac(\h^\mA,e)$
  and a $(c-1)$-minor of $\jac(\h^\mA,e+\dalgo)$ and let
  $(\polmu',\h')=\chartpolar(\psi^\mA,m',m'')$ be as in
  Definition~\ref{sec:chart:notation:polar}, with in particular
  $\polmu'=\polmu^\mA m' m''$. Suppose that the following holds:
  \begin{itemize}
  \item for each irreducible component $Z$ of $W^{\mA^{-1}}$ such that
    $\Open(\polmu)\cap Z-S$ is not empty,
    $\Open(\polmu \poldelta)\cap Z-S$ is not empty;
\smallskip
  \item $\Open(\polmu')\cap W-S^\mA$ is not empty.
  \end{itemize}
  Finally, let $\scrX$ be a finite subset of
  $\Open(\polmu'\poldelta^\mA)\cap V^\mA-S^\mA$.  Then, there exists a
  non-empty Zariski open set
  $\mathscr{I}(L,\phi,\mA,m',m'',\scrX) \subset \C^P$ such that for
  $\u$ in $\mathscr{I}(L,\phi,\mA,m',m'',\scrX) \cap \QQ^P$, the
  following holds:
  \begin{itemize}
  \item There exists a non-zero polynomial $\poldelta'_\u$ in $\QQ[\X]$
    and $(\rho_{k+1,j,\u})_{1 \le j \le P}$ in $\QQ[\X]_{\polmu'
      \poldelta'_\u}$, such that, writing
    $$
      \H'_\u =  \left ( \h',(L_{1,j}-\rho^\mA_{1,j})_{1 \le j \le
    n_1},\dots, (L_{k,j}-\rho^\mA_{k,j})_{1 \le j \le
    n_k},(L_{k+1,j}-\rho_{k+1,j,\u})_{1 \le j \le P}\right ) 
$$
    $\phi'_\u= (\polmu', \poldelta'_\u, \h', \H'_\u)$ is a local normal form for
    $\Polarlag(L^\mA,\u,\dalgo)$;
\smallskip
  \item $\poldelta'_\u$ vanishes nowhere on $\scrX$;
\smallskip
  \item the sets $\Open(\polmu')\cap \Clos{(\Polarlag(L^{\mA}, \u, \dalgo))}
    -S^\mA$ and $\Open(\polmu')\cap W-S^\mA$ coincide.
  \end{itemize}
\end{proposition}
\noindent The proof of this proposition will occupy this subsection;
we freely use all notation introduced in the proposition. We start by
proving that the localization $\QQ[\X]_{\polmu'\poldelta^\mA}$ is well-defined.

\begin{lemma}\label{lemma:GLS:nonzero}
  The polynomial $\polmu' \poldelta^\mA$ is non-zero.
\end{lemma}
\begin{proof}
  By ${\lnf_1}$ applied to $L$, the polynomial $\poldelta$ (and thus
  $\poldelta^\mA$) is non-zero. Since we assume that $\Open(\polmu')\cap
  W-S^\mA$ is not empty, $\polmu'$ is non-zero.
\end{proof}

First, we deal with the Lagrange system associated with $\H^\mA$. In
all that follows, we recall that we write $c=|\h|$ and that the
notation $\lag$ is from Definition~\ref{def:LS}.

\begin{lemma}\label{lemma:GLS:lagH}
  Let $\iota$ be the index of the row of $\jac(\h^\mA,e+\dalgo)$ that
  does not belong to $m''$. There exist rational functions
  $(\rho^\star_{k+1,j})_{j=1,\dots,c,j\ne \iota}$ in
  $ \QQ[\X]_{\polmu'\poldelta^\mA}$ such that in
  $\QQ[\X,\L']_{\polmu' \poldelta^\mA}$, the ideal
  $\langle \H^\mA, \lag(\H^\mA,e+\dalgo,\L_{k+1}) \rangle$ coincides
  with the ideal
$$\left \langle
\begin{array}{l}
 \h^\mA,\ (L_{1,j}-\rho^\mA_{1,j})_{1 \le j \le n_1},\ \dots,\ (L_{k,j}-\rho^\mA_{k,j})_{1 \le j \le  n_k},\\[1mm]
 M_1 L_{k+1,\iota},\dots,M_{n-e-c-\dalgo+1} L_{k+1,\iota},\quad\\[1mm]
  (L_{k+1,j}-\rho^\star_{k+1,j} L_{k+1,\iota})_{j \ne \iota},\quad  L_{k+1,c+1},\dots,L_{k+1,P}
\end{array}\right \rangle,$$
   where $M_1,\dots, M_{n-e-c-\dalgo+1}$ are the $c$-minors of
  $\jac(\h^\mA,e+{\dalgo})$ obtained by successively adding the missing
  row and the missing columns of $\jac(\h^\mA,e+\dalgo)$ to $m''$.
\end{lemma}
\begin{proof}
  The proof is in two steps. First, due to the special form of the
  polynomials $\H^\mA$, we show that the Lagrange system associated
  with these polynomials can be rewritten in a very simple manner in
  terms of the Lagrange system of $\h^\mA$.  Recall that $\H^\mA$
  takes the form $\H^\mA=\h^\mA, (L_{i,j}-\rho_{i,j}^\mA)_{1 \le i \le
    k, 1 \le j \le n_i}$. For $i$ in $\{1,\dots,k\}$ and $j$ in
  $\{1,\dots,n_j\}$, let us consider the column of $\jac(\H^\mA,e+\dalgo)$
  corresponding to derivatives with respect to $L_{i,j}$. The gradient
  row of the equation $L_{i,j}-\rho^\mA_{i,j}$ has a $1$ at the entry
  corresponding to this column, and this is the only equation giving a
  non-zero entry in this column. As a result, the equation
  $L_{k+1,u}=0$ appears in the Lagrange system, where $u$ is the index
  in $\{c+1,\dots,P\}$ of the equation $L_{i,j}-\rho^\mA_{i,j}$.  This
  proves that in $\QQ[\X,\L']_{\polmu' \poldelta^\mA}$, the ideal $\langle
  \H^\mA, \lag(\H^\mA,e+\dalgo,\L_{k+1}) \rangle$ is the ideal generated
  by
  $$\left \langle \H^\mA,\ \lag(\h^\mA,e+\dalgo,\ [L_{k+1,1},\dots,L_{k+1,c}]), \ L_{k+1,c+1},\dots,L_{k+1,P} \right \rangle.$$ 
  
  Lemma~\ref{sec:lemma:singS} shows that $d=n-e-c$, so inequality
  $\dalgo \le d$ can be restated as $e+\dalgo \le n-c$.  Thus, since
  we also have $m'' \ne 0$ (since $\polmu'\ne 0$), the assumption of
  Proposition~\ref{lemma:linearsolve} are satisfied. This proposition
  implies that there exist rational functions
  $(\rho^\star_{k+1,j})_{j=1,\dots,c,j\ne \iota}$ in
  $ \QQ[\X]_{\polmu'\poldelta^\mA}$ such that in
  $\QQ[\X,\L']_{\polmu' \poldelta^\mA}$, the ideal
  $\langle \h^\mA, \lag(\h^\mA,e+\dalgo,[L_{k+1,1},\dots,L_{k+1,c}])
  \rangle$ is the ideal generated by
  $$\left \langle \h^\mA,\,M_1 L_{k+1,\iota},\dots,M_{n-e-c-\dalgo+1}
  L_{k+1,\iota},\, (L_{k+1,j}-\rho^\star_{k+1,j} L_{k+1,\iota})_{j \ne
    \iota}\right \rangle,$$ where $M_1,\dots, M_{n-e-c-\dalgo+1}$ are the
  $c$-minors of $\jac(\h^\mA,e+{\dalgo})$ obtained by successively adding
  the missing row and the missing columns of $\jac(\h^\mA,e+\dalgo)$ to
  $m''$. This finishes the proof of the lemma.
\end{proof}

As before, call $\F$ the polynomials computed by $\Gamma$. We can now
use the relationship between $\H^\mA$ and $\F^\mA$ in order to rewrite
the Lagrange system of $\F^\mA$.

Let $I$ be the defining ideal of $Q$. From Lemma~\ref{lemma:mS}, we
know that there exists a $(P\times P)$ matrix $\mS$ with entries in
$\QQ[\X]_{\polmu^\mA \poldelta^\mA}$, such that
$\jac(\H^\mA,e) = \mS\, \jac(\F^\mA,e)$ holds over
$\QQ[\X,\L]_{\polmu^\mA \poldelta^\mA}/\langle \F^\mA, I\rangle$ and
such that $\det(\mS)$ divides $m'$ in
$\QQ[\X,\L]_{\polmu^\mA\poldelta^\mA}/\langle \F^\mA, I\rangle$.
Since $\polmu^\mA$ divides $\polmu'$, all previous equalities carry
over to $\QQ[\X,\L]_{\polmu'\poldelta^\mA}/\langle \F^\mA, I\rangle$.

\begin{lemma}
  There exists a matrix $\mT$ with entries in
  $\QQ[\X]_{\polmu'\poldelta^\mA}$ such that the product $\mT\, \mS$
  computed over
  $\QQ[\X,\L]_{\polmu'\poldelta^\mA}/\langle \F^\mA, I\rangle$ is the
  identity matrix.
\end{lemma}
\begin{proof}
  Because $\det(\mS)$ divides $m'$, and thus $\polmu'$, in
  $\QQ[\X,\L]_{\polmu'\poldelta^\mA}/\langle \F^\mA, I\rangle$, $\mS$
  admits an inverse with entries in
  $\QQ[\X,\L]_{\polmu'\poldelta^\mA}/\langle \F^\mA, I\rangle$. This inverse
  may be rewritten using the $\L$-component of $\H^\mA$, so as to
  involve the $\X$ variables only.
\end{proof}

For $i$ in $\{1,\dots, P\}$, let $L^\star_{k+1,i} \in
\QQ[\X,\L_{k+1}]_{\polmu' \poldelta^\mA}$ be the $i$th entry of the size-$P$
column vector $\mT^t\, \L_{k+1}^t$, where we see $\L_k$ as a row
vector of size $P$, and let $\L_{k+1}^\star$ be the row vector
$[L^\star_{k+1,1}, \ldots, L^\star_{k+1,P}]$.

Let further $\h'$ be the sequence of polynomials
$h^\mA_1,\dots,h^\mA_c,M_1,\dots,M_{n-e-c-\dalgo+1}$.  Recall that for
$\u=(u_1,\dots,u_P)$ in $\QQ^P$, the system we consider in the
generalized Lagrange system $\Polarlag(L^{\mA}, \u, \dalgo)$ is
$$\F'_{\u}=\Big (\F^\mA, \ \lag(\F^\mA,e+\dalgo, \L_{k+1}), \ u_{1}
L_{k+1,1} + \cdots + u_{P} L_{k+1,P} - 1\Big ).$$ Introducing the new
equation $u_{1} L_{k+1,1} + \cdots + u_{P} L_{k+1,P} - 1$ will allow us to
cancel some spurious terms $L_{k+1,\iota}$ appearing in 
Lemma~\ref{lemma:GLS:lagH}.

\begin{lemma}\label{lemma:GLS:lagF}
  Let $\u$ be in $\QQ^P$. In $\QQ[\X,\L']_{\polmu' \poldelta^\mA}$, the ideal
  $\langle \F'_\u, I\rangle$ coincides with the ideal
    $$\left \langle
  \begin{array}{l}
    I,\quad \h',\ (L_{1,j}-\rho^\mA_{1,j})_{1 \le j \le n_1},\ \dots,\ (L_{k,j}-\rho^\mA_{k,j})_{1 \le j \le  n_k},\\[1mm]
    (L^\star_{k+1,j}-\rho^\star_{k+1,j}  L^\star_{k+1,\iota})_{j \ne \iota},\quad  L^\star_{k+1,c+1},\dots,L^\star_{k+1,P},\quad\\[1mm]
    u_{1} L_{k+1,1} + \cdots + u_{P} L_{k+1,P} - 1
  \end{array}\right \rangle.$$
\end{lemma}
\begin{proof}
  The matrix $\mT$ satisfies the equality
  $$\jac(\F^\mA,e)=\mT\,\jac(\H^\mA,e)$$ over $\QQ[\X,\L]_{\polmu'
    \poldelta^\mA}/\langle \F^\mA, I\rangle$.
  Discarding the first $\dalgo$ columns in this equality, we get
  $\jac(\F^\mA,e+\dalgo)=\mT\,\jac(\H^\mA,e+\dalgo)$ over
  $\QQ[\X,\L]_{\polmu' \poldelta^\mA}/ \langle \F^\mA,
  I\rangle$.
  Left-multiplying by the row-vector $\L_{k+1}$, and using the fact
  that $\langle \F^\mA, I \rangle = \langle \H^\mA, I \rangle$ shows
  that the ideal
  $\langle I, \F^\mA, \lag(\F^\mA,e+\dalgo,\L_{k+1})\rangle$ is the
  ideal generated by
  $$\left \langle I,\ \H^\mA,\ \lag(\H^\mA,e+\dalgo,\L^\star_{k+1}) \right
  \rangle.$$ Evaluating the entries of $\L_{k+1}$ at
  $L^\star_{k+1,1},\dots,L^\star_{k+1,P}$ and using
  Lemma~\ref{lemma:GLS:lagH} shows that in $\QQ[\X,\L']_{\polmu'
    \poldelta^\mA}$, the ideal $\langle I, \H^\mA,
  \lag(\H^\mA,e+\dalgo,\L^\star_{k+1})\rangle$ coincides with the ideal
  $$\left \langle
  \begin{array}{l}
    I,\quad \h^\mA,\ (L_{1,j}-\rho^\mA_{1,j})_{1 \le j \le n_1},\ \dots,\ (L_{k,j}-\rho^\mA_{k,j})_{1 \le j \le  n_k},\\[1mm]
    M_1 L^\star_{k+1,\iota},\dots,M_{n-e-c-\dalgo+1} L^\star_{k+1,\iota},\quad
    (L^\star_{k+1,j}-\rho^\star_{k+1,j} L^\star_{k+1,\iota})_{j \ne \iota},\quad  \\[1mm]
L^\star_{k+1,c+1},\dots,L^\star_{k+1,P}
  \end{array}\right \rangle.$$
Let now $\u$ be in $\QQ^P$. We deduce from the previous equality that
in $\QQ[\X,\L']_{\polmu' \poldelta^\mA}$, the ideal
$\langle \F'_\u,I \rangle$ is the ideal generated by
  $$\left \langle
  \begin{array}{l}
    I,\quad \h^\mA,\ (L_{1,j}-\rho^\mA_{1,j})_{1 \le j \le n_1},\ \dots,\ (L_{k,j}-\rho^\mA_{k,j})_{1 \le j \le  n_k},\\[1mm]
    M_1 L^\star_{k+1,\iota},\dots,M_{n-e-c-\dalgo+1} L^\star_{k+1,\iota},\quad\\[1mm]
    (L^\star_{k+1,j}-\rho^\star_{k+1,j} L^\star_{k+1,\iota})_{j \ne \iota},\quad  L^\star_{k+1,c+1},\dots,L^\star_{k+1,P}\\[1mm]
    u_{1} L_{k+1,1} + \cdots + u_{P} L_{k+1,P} - 1
  \end{array}\right \rangle.$$
Let $u^\star_1,\dots,u^\star_P$ be the entries of the size-$P$ vector
$\mS\, \u$, which lie in $\QQ[\X]_{\polmu' \poldelta^\mA}$.  Then, due
to the definition of $L^\star_{k+1,i}$ as the $i$th entry of
$\mT^t\L_{k+1}^t$, the equality
  $$u_{1} L_{k+1,1} + \cdots + u_P L_{k+1,P} = 
  u^\star_{1} L^\star_{k+1,1} + \cdots + u^\star_{P} L^\star_{k+1,P}$$
  holds in
  $\QQ[\X,\L']_{\polmu' \poldelta^\mA}/\langle \F'_\u,I \rangle$.  As
  a consequence,
  $u^\star_{1} L^\star_{k+1,1} + \cdots + u^\star_{P}
  L^\star_{k+1,P}-1$
  is in $\langle \F'_\u,I\rangle$. We deduce further that
  $$(u^\star_{1} \rho_{k+1,1} + \cdots + u^\star_{c-1} \rho_{k+1,c}) L^\star_{k+1,\iota}-1$$
  is in $\langle \F'_\u,I\rangle$, where we write $\rho_{k+1,\iota}=1$.
  This shows that the ideal $\langle \F'_\u,I \rangle$ is the ideal generated by
  $$\left \langle
  \begin{array}{l}
    I,\quad \h',\ (L_{1,j}-\rho^\mA_{1,j})_{1 \le j \le n_1},\ \dots,\ (L_{k,j}-\rho^\mA_{k,j})_{1 \le j \le  n_k},\\[1mm]
    (L^\star_{k+1,j}-\rho^\star_{k+1,j}  L^\star_{k+1,\iota})_{j \ne \iota},\quad  L^\star_{k+1,c+1},\dots,L^\star_{k+1,P},\quad \\[1mm]
    u_{1} L_{k+1,1} + \cdots + u_{P} L_{k+1,P} - 1
  \end{array}\right \rangle,$$
  as claimed.
\end{proof}

To continue, we will rely on genericity properties for $\u$, that we
describe now. Let $\U=(U_1,\dots,U_P)$ be new indeterminates, let
$(t_{i,j})_{1 \le i,j \le P}$ be the entries of $\mT^t$ and let $\mM $
be the $(P\times P)$ matrix with entries in
$\QQ[\U,\X]_{\polmu' \poldelta^\mA}$ defined by
\begin{equation}
  \label{eq:mM}
\mM=\left [ \begin{matrix} 
t_{1,1} - \rho^\star_{k+1,1} t_{\iota,1} & \cdots & t_{1,P} - \rho^\star_{k+1,1} t_{\iota,P}\\
\vdots & & \vdots \\
\text{\sout{$t_{\iota,1} - \rho^\star_{k+1,\iota} t_{\iota,1}$}} & \cdots & \text{\sout{$t_{\iota,P} - \rho^\star_{k+1,\iota} t_{\iota,P}$}}\\
\vdots & & \vdots \\
t_{c,1} - \rho^\star_{k+1,c} t_{\iota,1} & \cdots & t_{c,P} - \rho^\star_{k+1,c} t_{\iota,P}\\
U_1  & \cdots & U_P \\
t_{c+1,1}  & \cdots & t_{c+1,P} \\
\vdots & & \vdots \\
t_{P,1}  & \cdots & t_{P,P} 
\end{matrix} \right ].
\end{equation}
We let $\mM^\star$ be the matrix $\mM$ multiplied by the minimal 
power of $\polmu' \poldelta^\mA$ such that $\mM^\star$ has entries in
$\QQ[\U,\X]$ and let further $\Lambda \in \QQ[\U,\X]$ be the
determinant of $\mM^\star$. Finally, for $\u$ in $\QQ^P$, we denote by
$\poldelta'_\u$ the polynomial $\poldelta^\mA \Lambda(\u,\X) \in \QQ[\X]$.

\begin{lemma}\label{lemma:GLS:Hprime}
  Let $\u$ in $\QQ^P$ be such that $\Lambda(\u,\X)\ne 0$. There exist
  rational functions $(\rho_{k+1,j,\u})_{1 \le j \le P}$ in
  $\QQ[\X,\L']_{\polmu'\poldelta'_\u}$ such that in
  $\QQ[\X,\L']_{\polmu' \poldelta'_\u}$, the ideal
  $\langle \F'_\u,I \rangle$ is equal to the ideal
  $$\langle I,\quad \h',\ (L_{1,j}-\rho^\mA_{1,j})_{1 \le j \le
    n_1},\dots, (L_{k,j}-\rho^\mA_{k,j})_{1 \le j \le
    n_k},\ (L_{k+1,j}-\rho_{k+1,j,\u})_{1 \le j \le P}\rangle.$$
\end{lemma}
\begin{proof}
  Starting from the conclusion of Lemma~\ref{lemma:GLS:lagF}, it
  remains to solve for the variables $L_{k+1,i}$.  Let us consider the
  subsystem
  $$(L^\star_{k+1,j}-\rho^\star_{k+1,j} L^\star_{k+1,\iota})_{j \ne \iota},\quad
  L^\star_{k+1,c+1},\dots,L^\star_{k+1,P}, \quad
  u_{1} L_{k+1,1} + \cdots + u_P L_{k+1,P}-1.
  $$
  This is an affine system in the indeterminates $L_{k+1,1},\dots,L_{k+1,P}$, 
  with matrix $\mM(\u,\X)$. By construction, the determinant of $\mM(\u,\X)$
  is invertible in $\QQ[\X,\L']_{\polmu'\poldelta'_\u}$, and the result follows using
  Cramer's formulas.
\end{proof}
In what follows, we let $\H'_\u$ be the polynomials in
$\QQ[\X]_{\polmu'\poldelta'_\u}$ given by
$$\H'_\u=\left ( \h',\ (L_{1,j}-\rho^\mA_{1,j})_{1 \le j \le
  n_1},\dots, (L_{k,j}-\rho^\mA_{k,j})_{1 \le j \le
  n_k},(L_{k+1,j}-\rho_{k+1,j,\u})_{1 \le j \le P}\right ).$$ Remark
that these polynomials, as well as $\poldelta'_\u$ itself, depend
on the choice of $\u$.

The following results will allow us to ensure the existence of values
of $\u$ that satisfy the assumptions of the former lemma.  Remark that
$\Proj(\Polarlag(L^{\mA}, \u, \dalgo))$ is contained in $\Proj(L)$,
since we add equations and $\scrQ$ and $\scrS$ do not change.

\begin{lemma}\label{lemma9}
  For $\x$ in $\Open(\polmu')\cap \Proj(L)^\mA$, the polynomial
  $\Lambda(\U,\x)$ is not identically zero.
\end{lemma}
\begin{proof}
  It suffices to prove the existence of one value of $\u$ for which
  $\Lambda(\u,\x)\ne 0$. Because $\x$ is in
  $\Open(\polmu^\mA)\cap \Proj(L)^\mA$, the local normal form property
  ${\lnf_5}$ implies that it is in
  $\Open(\polmu^\mA \poldelta^\mA)\cap \Proj(L)^\mA$, and thus in
  $\Open(\polmu' \poldelta^\mA)\cap \Proj(L)^\mA$; in particular, both
  matrices $\mS$ and $\mT$ can be evaluated at $\x$.  Besides, because
  $\x$ is in $\Proj(L)^\mA$, there exists $\bell \in \C^N$ such that
  $(\x,\bell)$ is in $\fbr(V(\F^\mA), Q)$. Since
  $\polmu' \poldelta^\mA$ does not vanish at $\x$, the equality
  $\mT\, \mS={\bf 1}$ that holds over
  $\QQ[\X,\L]_{\polmu'\poldelta^\mA}/\langle \F^\mA, I\rangle$ still
  holds after specialization at $(\x,\bell)$.

  Let then $\u=(u_1,\dots,u_P)$ be the value at $\x$ of the row of
  index $\iota$ in $\mT^t$. Evaluating $U_1,\dots,U_P$ at
  $u_{1},\dots,u_{P}$ in the determinant $\Lambda(\U,\x)$ of
  $\mM(\U,\x)$ gives us the determinant of $\mT^t(\x)$, which is
  non-zero. As a result, $\Lambda(\U,\x)$ itself is non-zero.
\end{proof}

\begin{lemma}\label{sec:lagrange:lemma:L5}
  For $\u$ in $\QQ^P$ and $\x$ in
  $\Open(\polmu') \cap \Proj(\Polarlag(L^{\mA}, \u, \dalgo))$,
  $\poldelta'_\u(\x)=\poldelta^\mA(\x) \Lambda(\u,\x)$ is non-zero.
\end{lemma}
\begin{proof}
  We need to prove that neither $\poldelta^\mA$ nor $\Lambda(\u,\X)$
  vanishes at $\x$.  Because $\Proj(\Polarlag(L^{\mA}, \u, \dalgo))$
  is contained in $\Proj(L)^\mA$, and $\Open(\polmu')$ is contained in
  $\Open(\polmu^\mA)$, $\x$ is in
  $\Open(\polmu^\mA) \cap \Proj(L)^\mA$; so $\poldelta^\mA$ does not
  vanish at $\x$, by ${\lnf_5}$ for $L$ --- as claimed.

  Since $\polmu'(\x)\poldelta^\mA(\x)$ is not zero, the matrix
  $\mM(\u,\x)$ of Eq.~\eqref{eq:mM} is well-defined. Suppose that its
  determinant is zero, or equivalently that $\Lambda(\u,\x)=0$: this
  means that the rows of the matrix $\mM(\u,\x)$ are dependent. Thus,
  there exists $\v \in \C^P$ non-zero such that
  $\v^t \mM(\u,\x)=[0~\cdots~0]$.

  Because $\x$ is in $\Proj(\Polarlag(L^{\mA}, \u, \dalgo))$, there
  exists $\bell$ in $\C^{N'-n}$ such that $\F'_\u(\x,\bell)=0$. Recall
  from the proof of Lemma~\ref{lemma:GLS:Hprime} that the system
  $\F'_\u$ involves in particular linear equations in the unknowns
  $L_{k+1,1},\dots,L_{k+1,P}$, with matrix $\mM(\u,\X)$ and right-hand
  side $[0~\dots~0~1~0\cdots~0]^t$, with $1$ at entry $c$.  After
  evaluation at $\x,\bell$ and left-multiplication by $\v^t$, we
  deduce that $v_c=0$. As a result, the matrix $\mM(\U,\x)$ itself is
  singular, or in other words $\Lambda(\U,\x)=0$.  However, since $\x$
  is in $\Open(\polmu') \cap \Proj(L)^\mA$, this contradicts
  Lemma~\ref{lemma9}.
\end{proof}

We are now going to prove that for a generic choice of $\u$, the
previous construction gives a local normal form of $\Polarlag(L^{\mA},
\u, \dalgo)$; we start by defining the Zariski open subset of $\C^P$ where
this will be the case.

First, we define a finite set of points associated to
$W=\polar(e,\dalgo,V^\mA)$. Let $\closedZ_1, \ldots, \closedZ_\ell$ be
the irreducible components of $W$, and assume without loss of
generality that $\closedZ_1, \ldots, \closedZ_{\ell'}$ are those
irreducible components of $W$ that have a non-empty intersection with
$\Open(\polmu')-S^\mA$; by assumption, $\ell' \ge 1$, since
$\Open(\polmu')\cap W-S^\mA$ is not empty. Now, recall that
$\polmu'=\polmu^\mA m' m''$, so for $i$ in $\{1,\dots,\ell'\}$, we
have in particular that $\closedZ_i$ has a non-empty intersection with
$\Open(\polmu^\mA)-S^\mA$.  Thus, by assumption, $\closedZ_i$ has a
non-empty intersection with $\Open(\polmu^\mA
\poldelta^\mA)-S^\mA$.
Because $\closedZ_i$ is irreducible, we deduce that
$\Open(\polmu' \poldelta^\mA)\cap \closedZ_i-S^\mA$ is not empty. We
thus let $\z_i$ be an element in this set, for $i$ in
$\{1,\dots,\ell'\}$, and we let $\scrX(W)=\{\z_1,\dots,\z_{\ell'}\}$.
Remark that $\ell'\ge 1$ means that $\scrX(W)$ is not empty.

Recall as well that we are given a finite subset $\scrX$ of
$\Open(\polmu' \poldelta^\mA)\cap V^\mA-S^\mA$. We can then
define $\scrX'=\scrX(W) \cup \scrX$. This is a finite subset of
$\Open(\polmu' \poldelta^\mA)\cap V^\mA-S^\mA$.

Any $\z$ in $\scrX'$ is in
$\Open(\polmu^\mA \poldelta^\mA) \cap V^\mA-S^\mA$, and thus (by
Lemma~\ref{lemma:conseq0}) in
$\Open(\polmu^\mA \poldelta^\mA) \cap \Proj(L)^\mA-S^\mA$, and
eventually in $\Open(\polmu') \cap \Proj(L)^\mA$, so
Lemma~\ref{lemma9} implies that the polynomial $\Lambda(\U,\z)$ is not
identically zero. We let
$\mathscr{I}(L,\phi,\mA,m',m'',\scrX) \subset \C^P$ be the non-empty
Zariski open set defined as
$\C^P-V(\Lambda(\U,\z_1)\cdots \Lambda(\U,\z_s))$, where we write
$\scrX' =\{\z_1,\dots,\z_s\}$. Since $\scrX(W)$ is not empty,
$s \ge 1$.

\begin{lemma}\label{sec:lagrange:lemma:asstert3bislocal}
  Suppose that $\u$ belongs to
  $\mathscr{I}(L,\phi,\mA,m',m'',\scrX)$. Then
  $\Open(\polmu')\cap\Clos{(\Polarlag(L^{\mA}, \u,
  \dalgo))}-S^\mA=\Open(\polmu')\cap W -S^\mA.$
\end{lemma}
\begin{proof}
  Because $s \ge 1$ and $\u$ belongs to
  $\mathscr{I}(L,\phi,\mA,m',m'',\scrX)$, $\Lambda(\u,\z_1) \ne 0$,
  which implies that the polynomial $\Lambda(\u,\X)$ is non-zero.
  We can thus apply Lemma~\ref{lemma:GLS:Hprime}, which implies that 
  $$\Open(\polmu'\poldelta'_\u)\cap \fbr(V(\F'_\u), Q)=
  \Open(\polmu'\poldelta'_\u)\cap \fbr(V(\H'_\u), Q),$$
  where the $\Open(\ )$ notation denotes here open subsets of
  $\C^{N'}$.  Since $\polmu' \poldelta'_\u$ is in $\QQ[\X]$, we deduce the
  equality
 $$\Open(\polmu'\poldelta'_\u)\cap \Pi_\X(\fbr(V(\F'_\u), Q))-S^\mA=
 \Open(\polmu'\poldelta'_\u)\cap \Pi_\X(\fbr(V(\H'_\u), Q))-S^\mA,$$ where
 the $\Open(\ )$ now denote open subsets of $\C^{n}$, as usual. 

 By definition,
 $\Proj(\Polarlag(L^{\mA}, \u, \dalgo))=\Pi_\X(\fbr(V(\F'_\u),
 Q))-S^\mA$.
 Also, remark that $\H'_\u$ is in normal form and $\h'$ is the
 $\X$-component of $\H'_\u$; consequently, we
 have
 $$\Open(\polmu'\poldelta'_\u)\cap
 \Proj(\Polarlag(L^{\mA}, \u, \dalgo))=\Open(\polmu'\poldelta'_\u)\cap \fbr(V(\h'), Q)-S^\mA.$$
 By Lemma~\ref{sec:lagrange:lemma:L5}, this can be rewritten as
  $$\Open(\polmu')\cap \Proj(\Polarlag(L^{\mA}, \u, \dalgo))=\Open(\polmu'\poldelta'_\u)\cap \fbr(V(\h'),
  Q)-S^\mA.$$
  On the other hand, since we suppose that
  $\Open(\polmu')\cap W-S^\mA$ is not empty, and that $\mA$ is in the
  open set $\scrGpolarchart(\psi,V,Q,S,\dalgo)$ defined in
  Lemma~\ref{sec:chart:lemma:polarchart}, that lemma shows that
  $(\polmu',\h')$ is a chart of $(W,Q,S^\mA)$, so that we have the
  equality
  $$\Open(\polmu' \poldelta'_\u)\cap W-S^\mA=\Open(\polmu' \poldelta'_\u)\cap \fbr(V(\h'),
  Q)-S^\mA.$$ Combining the former two equalities, we thus deduce
  \begin{equation}\label{eq:WProj}
  \Open(\polmu')\cap \Proj(\Polarlag(L^{\mA}, \u, \dalgo))=\Open(\polmu'\poldelta'_\u)\cap W-S^\mA.    
  \end{equation}
  We are going to relate the left- and right-hand sides of this
  equality to those appearing in the statement of the lemma.
 
  Let $A$ be the union of the irreducible components of
  $\Clos{(\Polarlag(L^{\mA}, \u, \dalgo))}$ which have a non-empty
  intersection with $\Open(\polmu')$, so that we have, by an immediate
  verification:
  \begin{enumerate}
  \item [${\sf a_1}.$]
    $\Open(\polmu')\cap A=\Open(\polmu')\cap \Clos{(\Polarlag(L^{\mA},
    \u, \dalgo))}$, \smallskip
  \item [${\sf a_2}.$]
    $A=\overline{\Open(\polmu')\cap \Proj(\Polarlag(L^{\mA}, \u,
      \dalgo))}$,
    because $\Clos{(\Polarlag(L^{\mA}, \u, \dalgo))}$ is the Zariski
    closure of $\Proj(\Polarlag(L^{\mA}, \u, \dalgo))$.
  \end{enumerate}
  Similarly, let $B$ be the union of the irreducible components of $W$
  which have a non-empty intersection with
  $\Open(\polmu')-S^\mA$; in other words, using the notation
  given prior to this lemma, $B=Z_1 \cup \cdots \cup Z_{\ell'}$. We
  claim that $B$ is also the union of the irreducible components of
  $W$ which have a non-empty intersection with
  $\Open(\polmu'\poldelta'_\u)-S^\mA$. Consider indeed an index
  $i$ in $\{1,\dots,\ell'\}$. By construction of $\z_i$,
  $\poldelta^\mA(\z_i)$ is non-zero, and by assumption on $\u$,
  $\Lambda(\u,\z_i)$ is non-zero; thus, $\poldelta'_\u$ does not
  vanish at $\z_i$.  Our claim is thus proved (since the converse
  inclusion is immediate), so as above, we have
  \begin{enumerate}
  \item [${\sf b_1}.$]
    $\Open(\polmu')\cap B-S^\mA=\Open(\polmu')\cap
    W-S^\mA$,
\smallskip
  \item [${\sf b_2}.$]
    $B=\overline{\Open(\polmu'\poldelta'_\u)\cap W-S^\mA}$
    (where we use the second characterization of $B$).
  \end{enumerate}
  Using Eq.~\eqref{eq:WProj}, as well as ${\sf a_2}$ and ${\sf b_2}$,
  we deduce that $A=B$. Finally, using ${\sf a_1}$ and ${\sf b_1}$, we conclude that
  $$\Open(\polmu')\cap \Clos{(\Polarlag(L^{\mA}, \u, \dalgo))}-S^\mA=\Open(\polmu')\cap W-S^\mA,$$
  as claimed.
\end{proof}

We can now conclude the proof of Proposition \ref{prop:local}.  Take $\u$ in
$$\mathscr{I}(L,\phi,\mA,m',m'',\scrX)\cap \QQ^P.$$ As we saw in the
proof of the previous lemma, $\Lambda(\u,\X)$ is non-zero, so
$\poldelta'_\u$ is non-zero and $\H'_\u$ is well-defined. We now prove
that $\phi'_\u=(\polmu',\poldelta'_\u, \h',\H'_\u)$ is a local normal
form for $\Polarlag(L^{\mA}, \u, \dalgo)$.
\begin{enumerate}
\item[${\lnf_1}.$] By construction, $\polmu'$ and $\poldelta'_\u$ are in 
  $\QQ[\X]-\{0\}$ and $\H'_\u$ is in normal form, with $\X$-component~$\h'$.
\smallskip
\item[${\lnf_2}.$] On one hand, we have $|\H'_\u| =
  |\H|+n-e-c-\dalgo+1+P$. On the other hand, Lemma~\ref{lemma:GLS:typeW}
  shows that $|\F'_\u|=P+N-e-\dalgo+1$. By ${\lnf_2}$ for $L$, we know
  that $|\H|+n-c=N$, so that $|\H'_\u|=|\F'_\u|$.
\smallskip
\item[${\lnf_3}.$] We proved in Lemma~\ref{lemma:GLS:Hprime} that
  the equality $$\langle \F'_\u,I \rangle = \langle \H'_\u,I \rangle$$
  holds in $\QQ[\X,\L]_{\polmu' \poldelta'_\u}$.
\smallskip
\item[${\lnf_4}.$] Since $\Open(\polmu')\cap W-S^\mA$ is not empty,
  Lemma~\ref{sec:chart:lemma:polarchart} shows that $(\polmu',\h')$ is
  a chart of $(W,Q,S^\mA)$.
  Lemma~\ref{sec:lagrange:lemma:asstert3bislocal} shows that
  $\Open(\polmu')\cap \Clos{(\Polarlag(L^{\mA}, \u, \dalgo))} -S^\mA =
  \Open(\polmu')\cap W -S^\mA$,
  so $(\polmu',\h')$ is also a chart of
  $(\Clos{(\Polarlag(L^{\mA}, \u, \dalgo))},Q,S^\mA)$.  \smallskip
\item[${\lnf_5}.$] This is a restatement of Lemma~\ref{sec:lagrange:lemma:L5}.
\end{enumerate}
The last point is to prove that $\poldelta'_\u$ vanishes nowhere on
$\scrX$. Indeed, by construction, for all $\z$ in $\scrX$,
$\poldelta^\mA(\z)$ is non-zero (by assumption on $\scrX$) and
$\Lambda(\u,\z)$ is non-zero (by definition of
$\mathscr{I}(L,\phi,\mA,m',m'',\scrX)$).

%%%%%%%%%%%%%%%%%%%%%%%%%%%%%%%%%%%%%%%%%%%%%%%%%%%%%%%%%%%%

\subsection{Proof of the proposition}\label{ssec:proof:prop:transfer:polar}

The rest of this paragraph is devoted to prove
Proposition~\ref{sec:lagrange:prop:transfer:polar}.
We start by defining the family of local normal forms we will use for
the generalized Lagrange system $\Polarlag(\dalgo, L^{\mA}, \u)$.  Let
the global normal form $\bphi$ of $(L; W^{\mA^{-1}}, \scrY)$ be
written as $\bphi=(\phi_i)_{1 \le i \le s}$, with
$\phi_i=(\polmu_i,\poldelta_i,\h_i,\H_i)$ for all $i$. For $i$ in
$\{1,\dots,s\}$, we let $\psi_i=(\polmu_i,\h_i)$ be the chart of
$(V,Q,S)$ associated with $\phi_i$, so that
$\bpsi=(\psi_i)_{1 \le i \le s}$.

For all $(i,m',m'')$, where $i$ is in $\{1,\dots,s\}$ and $m',m''$ are
respectively a $c$-minor of $\jac(\h_i^\mA,e)$ and a $(c-1)$-minor of
$\jac(\h_i^\mA,e+\dalgo)$, we let
$(\polmu'_{i,m',m''},\h'_{i,m',m''})=\chartpolar(\psi_i^\mA,m',m'')$
be the polynomials introduced in
Definition~\ref{sec:chart:notation:polar}; in particular,
$\polmu'_{i,m',m''}=\polmu_i^\mA m'm''$. We define $\zeta$ as the set
of all these $(i,m',m'')$, such that
$\Open(\polmu'_{i,m',m''}) \cap W-S^\mA$ is not empty. Note that
$\zeta$ is empty if $W$ is empty.

Let $(i, m', m'')$ be in $\zeta$ and let $Z_1,\dots,Z_\ell$ be the
irreducible components of the sets $Y^\mA_1,\dots,Y^\mA_r$ such that
$Z_j \subset W$ and $\Open(\polmu'_{i,m',m''})\cap Z_j - S^\mA$
is not empty (note that the $Z_j$'s, as well as the index $\ell$,
depend on $(i,m',m'')$, although our notation does not reflect this).
For $j$ in $\{1,\dots,\ell\}$,
$\Open(\polmu_i^\mA) \cap Z_j -S^\mA$ is in particular not
empty; as a result, applying ${\gnf_3}$ to $Z_j^{\mA^{-1}}$ shows
that $\Open(\polmu_i^\mA \poldelta_i^\mA) \cap Z_j-S^\mA$ is not
empty. Because $Z_j$ is irreducible, this finally implies that
$\Open(\polmu'_{i,m',m''} \poldelta_i^\mA) \cap Z_j-S^\mA$ is
not empty; we thus let $\z_j$ be an element in this set and we set
$\scrX_{i,m',m''}=\{\z_1,\dots,\z_{\ell}\}$.

When $\zeta$ is empty, we set $\mathscr{I}(L,\bphi,\mA,\scrY)$ to be
the whole $\C^P$. When $\zeta$ is not empty,
$\mathscr{I}(L,\bphi,\mA,\scrY)$ will be defined using Proposition
\ref{prop:local}.  Let us first verify that for any $(i,m',m'')$ in
$\zeta$, the assumptions of Proposition~\ref{prop:local} are
satisfied. 

We take $(i,m',m'')$ as above. The definition of
$\scrGpolar(\bpsi,V,Q,S,{\dalgo})$ given in the proof of
Proposition~\ref{prop:ch4} proves that $\mA$ is in the non-empty
Zariski open set $\scrGpolarchart(\psi_i,V,Q,S,\dalgo)$ defined in
Lemma~\ref{sec:chart:lemma:polarchart}. The global normal form
assumption shows that for each irreducible component $Z$ of
$W^{\mA^{-1}}$ such that $\Open(\polmu_i)\cap Z-S$ is not empty,
$\Open(\polmu_i \poldelta_i)\cap Z-S$ is not empty. By construction of
$\zeta$, $\Open(\polmu'_{i,m',m''})\cap W-S^\mA$ is not
empty. Finally, $\scrX_{i,m',m''}$ is contained in
$\Open(\polmu'_{i,m',m''} \poldelta_i^\mA) \cap V^\mA-S^\mA$.

Applying Proposition~\ref{prop:local}, we deduce that there exists a
non-empty Zariski open subset
$$\mathscr{I}(L,\phi_i,\mA,m',m'',\scrX_{i, m',m''}) \subset \C^P$$ such
that for $\u$ in $\mathscr{I}(L,\phi_i,\mA,m',m'',\scrX_{i, m',
  m''})$, the following holds:
\begin{itemize}
\item there exists a non-zero $\poldelta'_{i,m',m'',\u}$ in $\QQ[\X]$
  and polynomials $\H'_{i,m',m'',\u}$ in
  $\QQ[\X,\L_{k+1}]_{\polmu'_{i,m',m''} \poldelta'_{i,m',m'',\u}}$
  such that
  $$\phi'_{i,m',m'',\u}= (\polmu'_{i,m',m''}, \poldelta'_{i,m',m'',\u},
  \h'_{i,m',m''}, \H'_{i,m',m'',\u})$$
  is a local normal form for $\Polarlag(L^\mA,\u,\dalgo)$;
\smallskip
\item $\poldelta'_{i,m',m'',\u}$ vanishes nowhere on $\scrX_{i,m',m''}$;
\smallskip
\item the sets
  $\Open(\polmu'_{i,m',m''})\cap \Clos{(\Polarlag(L^{\mA}, \u, \dalgo))}
  -S^\mA$ and $\Open(\polmu'_{i,m',m''})\cap W-S^\mA$ coincide.
\end{itemize}
Finally, let $\mathscr{I}(L,\bphi,\mA,\scrY)$ be the intersection
of all $\mathscr{I}(L,\phi_i,\mA,m',m'',\scrX_{i, m', m''})$, for $(i,
m',m'')$ in $\zeta$; this is a non-empty Zariski open subset of
$\C^P$. In what follows, we take $\u$ in
$\mathscr{I}(L,\bphi,\mA,\scrY)$ and we prove the assertions
in the proposition. We start with an easy lemma. 
\begin{lemma}\label{chap:GLS:polar:zeta}
  With the above notation,
  $\Open(\polmu'_{i,m',m''}) \cap \Clos{(\Polarlag(L^{\mA}, \u,
  \dalgo))}-S^\mA$
  is not empty if and only if $(i,m',m'')$ is in $\zeta$.
\end{lemma}

\begin{proof}
  Suppose first that $(i,m',m'')$ is in $\zeta$. By assumption on
  $\u$, the three items above hold; the third one, and the fact that
  $(i,m',m'')$ is in $\zeta$, imply that
  $\Open(\polmu'_{i,m',m''}) \cap \Clos{(\Polarlag(L^{\mA}, \u,
  \dalgo))}-S^\mA$ is not empty.

  Conversely, suppose now that
  $\Open(\polmu'_{i,m',m''}) \cap \Clos{(\Polarlag(L^{\mA}, \u,
  \dalgo))}-S^\mA$
  is not empty. Because $\Clos{(\Polarlag(L^{\mA}, \u, \dalgo))}$ is the
  Zariski closure of $\Proj(\Polarlag(L^{\mA}, \u, \dalgo))$, we
  deduce that
  $\Open(\polmu'_{i,m',m''}) \cap \Proj(\Polarlag(L^{\mA}, \u,
  \dalgo))-S^\mA$
  is not empty. Take $\x$ in this set.  Because
  $\Proj(\Polarlag(L^{\mA}, \u, \dalgo))$ is contained in
  $\Proj(L)^\mA$, we deduce from ${\lnf_5}$ applied to $\phi_i^\mA$
  that $\poldelta^\mA_i$ does not vanish at
  $\x$. Lemma~\ref{lemma:GLS:lagF} then implies that $\x$ cancels
  $\h'_{i,m',m''}$, so that $\x$ is in $\fbr(V(\h'_{i,m',m''}), Q)$.
  The first item in Lemma~\ref{sec:chart:lemma:polarchart} implies
  that $\x$ is in $W$, so we are done.
\end{proof}

\begin{lemma}
  For $\u$ in $\mathscr{I}(L,\bphi,\mA,\scrY)$, the equality
  $\Clos{(\Polarlag(L^{\mA}, \u, \dalgo))}=W$ holds.
\end{lemma}
\begin{proof}
  For all $i$ in $\{1,\dots,s\}$, let $\zeta'_i$ be the set of all
  triples $(i, m',m'')$, where $m'$ and $m''$ are respectively
  $c$-minors of $\jac(\h_i^\mA,e)$ and $(c-1)$-minors of
  $\jac(\h_i^\mA,e+\dalgo)$, and let $\zeta_i$ be the subset of $\zeta'_i$
  for which $\Open(\polmu'_{i,m',m''})\cap W-S^\mA$ is not empty. In
  particular, $\zeta$ is the union of all $\zeta_i$; similarly, we let
  $\zeta'$ be the union of all $\zeta'_i$.

  By Lemma \ref{chap:GLS:polar:zeta},
  $\Open(\polmu'_{i,m',m''}) \cap \Clos{(\Polarlag(L^{\mA}, \u,
  \dalgo))}-S^\mA$
  is not empty if and only if $(i,m',m'')$ is in $\zeta$.  We are
  going to use this remark to prove first that
  $\Clos{(\Polarlag(L^{\mA}, \u, \dalgo))}-S^\mA=W-S^\mA$.

  %% First, assume that $W$ is empty. By Lemma \ref{chap:GLS:polar:zeta}
  %% we deduce that $V'_\u-S^\mA$ is empty: indeed if this is not the
  %% case, $\zeta$ would be non-empty which contradicts the emptiness of
  %% $W$.
  %% Assume now that $W$ is non-empty. 

  Let $i$ be in $\{1,\dots,s\}$. We know from the third item in
  Lemma~\ref{sec:chart:lemma:polarchart} that the sets
  $\Open(\polmu'_{i,m',m''}) -S^\mA$, for $(m',m'')$ in $\zeta'_i$,
  cover $\Open(\polmu_i^\mA)\cap V^\mA-S^\mA$. Because
  $\bpsi^\mA$ is an atlas of $(V^\mA,Q,S^\mA)$, the sets
  $\Open(\polmu_i^\mA)\cap V^\mA-S^\mA$ themselves cover
  $V^\mA-S^\mA$, and we deduce that the sets
  $\Open(\polmu'_{i,m',m''}) -S^\mA$, for $(i,m',m'')$ in
  $\zeta'$, cover $V^\mA-S^\mA$.

  Since both $\Clos{(\Polarlag(L^{\mA}, \u, \dalgo))}$ and $W$ are
  subsets of $V^\mA$, these sets cover in particular
  $\Clos{(\Polarlag(L^{\mA}, \u, \dalgo))}-S^\mA$ and
  $W-S^\mA$. However, we saw above that the only triples $(i,m',m'')$
  for which the intersections $\Open(\polmu'_{i,m',m''}) \cap W-S^\mA$
  or
  $\Open(\polmu'_{i,m',m''}) \cap \Clos{(\Polarlag(L^{\mA}, \u,
    \dalgo))}-S^\mA$
  are not empty are those in $\zeta$ (this is by construction of
  $\zeta$ for $W$ and by Lemma~\ref{chap:GLS:polar:zeta} for
  $\Clos{(\Polarlag(L^{\mA}, \u, \dalgo))}$). Thus, we deduce that the
  sets $\Open(\polmu'_{i,m',m''}) -S^\mA$, for $(i,m',m'')$ in
  $\zeta$, cover both $\Clos{(\Polarlag(L^{\mA}, \u, \dalgo))}-S^\mA$
  and $W-S^\mA$.

  On the other hand, due to our choice of $\u$, we have seen that the
  following holds for all $(i, m', m'')$ in $\zeta$:
  $$\Open(\polmu'_{i, m', m''})\cap \Clos{(\Polarlag(L^{\mA}, \u, \dalgo))} - 
  S^\mA=\Open(\polmu'_{i, m', m''})\cap W-S^\mA.$$
  The last two paragraphs imply that
  $\Clos{(\Polarlag(L^{\mA}, \u, \dalgo))}-S^\mA=W-S^\mA$, as
  claimed. Since $\Clos{(\Polarlag(L^{\mA}, \u, \dalgo))}$ is the
  Zariski closure of $\Proj(\Polarlag(L^{\mA}, \u, \dalgo))$, which
  does not intersect $S^\mA$, we deduce that
  $\Clos{(\Polarlag(L^{\mA}, \u, \dalgo))}$ is also the Zariski
  closure of $\Clos{(\Polarlag(L^{\mA}, \u, \dalgo))}-S^\mA$.

  If $W$ is empty, we are done (since then
  $\Clos{(\Polarlag(L^{\mA}, \u, \dalgo))}-S^\mA$ is empty, and thus
  its Zariski closure $\Clos{(\Polarlag(L^{\mA}, \u, \dalgo))}$ is
  empty as well).  On the other hand, if $W$ is not empty, the facts
  that $V$ is equidimensional of dimension $d$, with finitely many
  singular points, and that $\mA$ is in
  $\scrGpolar(\bpsi,V,Q,S,\dalgo)$ show that one can apply
  Proposition~\ref{prop:ch4} and deduce that $W$ is
  $(\dalgo-1)$-equidimensional. Since $\dalgo \ge 2$ (so that
  $\dalgo-1 \ge 1$) and $S^\mA$ is finite, $W$ is the Zariski closure
  of $W-S^\mA$. The lemma is proved.
\end{proof}

We can now conclude the proof of the proposition.  For $\u$ in
$\mathscr{I}(L,\bphi,\mA,\scrY)$, we already know that
$\Polarlag(L^\mA,\u,\dalgo)$ is a generalized Lagrange system, and the
previous lemma shows that $\Clos{(\Polarlag(L^{\mA}, \u, \dalgo))}$ is
equal to $W$. Now, we assume that $W$ is not empty; it remains to
construct a global normal form for it.

Let $\bphi'_\u$ be the set of all local normal forms
$\phi'_{i,m',m'',\u}$ defined above, for $(i,m',m'')$ in $\zeta$. We
prove that $\bphi'_\u$ is a global normal form for
$(\Polarlag(L^\mA,\u,\dalgo);\scrY^\mA)$, and that
$\atlaspolar(\bpsi^\mA,V^\mA,Q,S^\mA,\dalgo)$ is the associated atlas
of $(W,Q,S^\mA)$.
\begin{enumerate}
\item [${\gnf_1}.$] We saw above that all $\phi'_{i,m',m'',\u}$ are
  local normal forms for $\Polarlag(L^\mA,\u,\dalgo)$.
\smallskip

\item [${\gnf_2}.$] We must now prove that the sets
  $$\psi'_{i,m',m''}=(\polmu'_{i,m',m''},\h'_{i,m',m''}),$$ for
  $(i,m',m'')\in\zeta$, form an atlas of
  $(\Clos{(\Polarlag(L^\mA,\u,\dalgo))},Q,S^\mA)$, or equivalently of
  $(W,Q,S^\mA)$. Remark that this family precisely
  defines $$\atlaspolar(\bpsi^\mA,V^\mA,Q,S^\mA,\dalgo).$$ Recall that
  $V$ is $d$-equidimensional, with finitely many singular points, and
  that $\mA$ is in $\scrGpolar(\bpsi,V,Q,S,\dalgo)$; hence, all
  assumptions of Proposition~\ref{prop:ch4} are satisfied. That
  proposition~\ref{prop:ch4} proves that
  $\atlaspolar(\bpsi^\mA,V^\mA,Q,S^\mA,\dalgo)$ is an atlas of
  $(W,Q,S^\mA)$, so our claim is proved.  \smallskip
  
\item[${\gnf_3}.$] Recall that we write $\scrY=Y_1,\dots,Y_r$.  Let
  $Z$ be an irreducible component of $Y^\mA_j$, for some $j$ in
  $\{1,\dots,r\}$. Suppose that $Z$ is contained in $W$, and let
  $(i,m',m'') \in \zeta$ be such that
  $\Open(\polmu'_{i,m',m''}) \cap Z-S^\mA$ is not empty. We have
  to prove that $\poldelta'_{i,m',m'',\u}$ does not vanish identically
  on $Z$.

  By construction, for such a $Z$, there exists an element $\z$ in the
  finite set $\scrX_{i,m',m''} \cap Z$. We saw previously that for our
  choice of $\u$, $\poldelta'_{i,m',m'',\u}$ vanishes nowhere on
  $\scrX_{i,m',m''}$; as a result, $\poldelta'_{i,m',m'',\u}$ does not
  vanish at $\z$, and thus does not vanish identically on $Z$.
\end{enumerate}

%%%%%%%%%%%%%%%%%%%%%%%%%%%%%%%%%%%%%%%%%%%%

%%%%%%%%%%%%%%%%%%%%%%%%%%%%%%%%%%%%%%%%%%%%
%\input{proof5.22}

%%%%%%%%%%%%%%%%%%%%%%%%%%%%%%%%%%%%%%%%%%%%%%%%%%%%%%%%%%%%
%%%%%%%%%%%%%%%%%%%%%%%%%%%%%%%%%%%%%%%%%%%%%%%%%%%%%%%%%%%%
%%%%%%%%%%%%%%%%%%%%%%%%%%%%%%%%%%%%%%%%%%%%%%%%%%%%%%%%%%%%

\section{Proof of Proposition~\ref{sec:lagrange:prop:transfer:fiber}}\label{chapter:constructionsfiber}

In this section, we prove
Proposition~\ref{sec:lagrange:prop:transfer:fiber}, whose statement is
as follows: {\em Let $Q \subset \C^e$ be a finite set and let
  $V \subset \C^n$ and $S\subset \C^n$ be algebraic sets lying over
  $Q$, with $S$ finite.  Suppose that $V$ is equidimensional of
  dimension $d$, with finitely many singular points.
  
  Let $\bpsi$ be an atlas of $(V,Q,S)$, let $\dalgo$ be an integer in
  $\{2,\dots,d\}$ such that $\dalgo \le (d+3)/2$, and let $\mA\in
  \GL(n,e)$ be in the open set $\scrGfiber(\bpsi,V,Q,S,\dalgo)$
  defined in Proposition~\ref{sec:atlas:prop:summary1}; write
  $W=\polar(e,\dalgo,V^\mA)$.

  Let $\param2$ and $\paramsing2$ be zero-dimensional parametrizations
  with coefficients in $\QQ$ that respectively define a finite set
  $\fiber2 \subset \C^{e+\dalgo-1}$ lying over $Q$ and the set
  $\fibersing2 = \fbr(S^\mA \cup W,\fiber2)$, and let
  $\Vfiber=\fbr(V^\mA,\fiber2)$.

  Let $L=(\Gamma,\scrQ,\scrS)$ be a generalized Lagrange system such
  that $V=\Clos{(L)}$, $Q=\Zeroes(\scrQ)$ and $S=\Zeroes(\scrS)$. Let
  $\scrY=(Y_1,\dots,Y_r)$ be algebraic sets in $\C^n$ and let finally
  $\bphi$ be a global normal form for $(L;({\Vfiber}^{\mA^{-1}},
  \scrY))$ such that $\bpsi$ is the associated atlas of $(V,Q,S)$.
  Then the following holds:
  \begin{itemize}
  \item $\Fiberlag(L^\mA, \param2, \paramsing2)$ is a generalized
    Lagrange system which defines $\Vfiber;$
\smallskip
  \item if $\Vfiber$ is not empty,
    $(\Fiberlag(L^\mA, \param2, \paramsing2); \scrY^\mA)$ admits a
    global normal form whose atlas is
    $\atlasfiber(\bpsi^\mA,V^\mA,Q,S^\mA,\fiber2)$ (Definition
    \ref{sec:atlas:notation:fiber}).
 \end{itemize}
}

As we did in the previous section, we start with a local analysis
which we use to prove the global statement.

%%%%%%%%%%%%%%%%%%%%%%%%%%%%%%%%%%%%%%%%%%%%%%%%%%%%%%%%%%%%

\subsection{Local analysis}

In this paragraph, we consider a local normal form
$\phi=(\polmu,\poldelta,\h,\H)$ of $L$. We show how to deduce a local
normal form for $\Fiberlag(L^\mA,\param2,\paramsing2)$, for a suitable
choice of $\paramsing2$.

\begin{proposition}\label{sec:lagrange:prop:local:fibers}
  Let $Q \subset \C^e$ be a finite set and let $V \subset \C^n$ and
  $S\subset \C^n$ be algebraic sets lying over $Q$, with $S$ finite.
  Suppose that $V$ is $d$-equidimensional with finitely many singular
  points.

  Let $L=(\Gamma,\scrQ,\scrS)$ be a generalized Lagrange system of
  type $(k,\n,\p,e)$ such that $V=\Clos{(L)}$, $Q=\Zeroes(\scrQ)$ and
  $S=\Zeroes(\scrS)$. Let $\phi=(\polmu,\poldelta,\h,\H)$ be a local
  normal form for $L$ and let $\psi=(\polmu,\h)$ be the associated
  chart of $(V,Q,S)$. Let $\dalgo$ be an integer in $\{2,\dots,d\}$,
  such that $\dalgo \le (d+3)/2$, let $\mA \in \GL(n,e)$ be in the
  open set $\scrGfiberchart(\psi,V,Q,S,\dalgo)$ defined in
  Lemma~\ref{lemma:chart:fiber} and let $W=\polar(e,\dalgo,V^\mA)$.

  Let $\param2$ and $\paramsing2$ be zero-dimensional parametrizations
  with coefficients in $\QQ$, that respectively define a finite set
  $\fiber2 \subset \C^{e+\dalgo-1}$ lying over $Q$ and the set
  $\fibersing2 = \fbr(S^\mA \cup W,\fiber2)$, and let
  $\Vfiber=\fbr(V^\mA,\fiber2)$.  If $\Open(\polmu^\mA)\cap
  \Vfiber-\fibersing2$ is not empty, then $\phi^\mA$ is a local normal
  form for $\Fiberlag(L^\mA,\param2,\paramsing2)$.
\end{proposition}

In what follows, we write $\F$ for the polynomials
computed by $\Gamma$. Suppose that
$\Open(\polmu^\mA)\cap \Vfiber-\fibersing2$ is not empty and let
$\mA$, and all further notation, be as in the proposition; note in
particular that $\phi^\mA=(\polmu^\mA,\poldelta^\mA,\h^\mA,\H^\mA)$.
The following items check the validity of ${\lnf_1},\dots,{\lnf_5}$.
\begin{enumerate}
\item[${\lnf_1}.$] Because $\phi$ is a local normal form for $L$,
  $\phi^\mA$ is a local normal form for $L^\mA$. Then, since
  ${\lnf_1}$ concerns only the polynomials in $\phi^\mA$, it continues
  to hold here.
\smallskip
\item[${\lnf_2}.$] For the same reason, and because the defining
  equations in $\Fiberlag(L^\mA,\param2,\paramsing2)$ are simply
  $\F^\mA$, $\lnf_2$ remains valid.
\smallskip

\item[${\lnf_3}.$] Property $\lnf_3$ for $L$ states that
  $\langle \F,I \rangle = \langle \H,I \rangle $ in
  $\QQ[\X,\L]_{\polmu \poldelta}$; it implies the equality
  $\langle \F^\mA,I \rangle = \langle \H^\mA,I \rangle $ in
  $\QQ[\X,\L]_{\polmu^\mA \poldelta^\mA}$. Let then $I' \subset \QQ[\X]$ be
  the defining ideal of $\fiber2$.  Adding $I'$ to both sides of the
  former equality gives the requested
  $\langle \F^\mA,I' \rangle = \langle \H^\mA,I' \rangle $ in
  $\QQ[\X,\L]_{\polmu^\mA \poldelta^\mA}$, since $ I \subset I'$.
\smallskip

\item[${\lnf_4}.$] Because
  $\Open(\polmu^\mA)\cap \Vfiber-\fibersing2$ is not empty and $\mA$
  is in $\scrGfiberchart(\psi,V,Q,S,\dalgo)$, Lemma~\ref{lemma:chart:fiber}
  shows that $\psi^\mA=(\polmu^\mA,\h^\mA)$ is a chart of
  $(\Vfiber,\fiber2,\fibersing2)$.
\smallskip

\item[${\lnf_5}.$] By construction, $\Proj(\Fiberlag(L^\mA,
  \param2, \paramsing2))$
  is contained in $\Proj(L^\mA)$. Applying ${\lnf_5}$ for
  $L^\mA$, we deduce that
  $\Open(\polmu^\mA) \cap U^\mA = \Open(\polmu^\mA
  \poldelta^\mA) \cap \Proj(L^\mA)$.
  Intersecting with $\Proj(\Fiberlag(L^\mA,
  \param2, \paramsing2))$ proves ${\lnf_5}$.
\end{enumerate}

%%%%%%%%%%%%%%%%%%%%%%%%%%%%%%%%%%%%%%%%%%%%%%%%%%%%%%%%%%%%

\subsection{Proof of the proposition}\label{ssec:proof:prop:transfer:fiber}

We can now prove Proposition~\ref{sec:lagrange:prop:transfer:fiber}.
Since we assumed that $\mA$ is in $\scrGfiber(\bpsi,V,Q,S,\dalgo)$,
all assumptions of Proposition~\ref{sec:atlas:prop:summary1} are
satisfied and we deduce that $\Vfiber$ is either empty or $\Vfiber$
equidimensional of dimension $d-(\dalgo-1)$, with finitely many
singular points.

We already know that $\Fiberlag(L^\mA, \param2, \paramsing2)$ is a
generalized Lagrange system; the next lemmas then prove that
$\Vfiber=\Clos{(\Fiberlag(L^\mA, \param2, \paramsing2))}$.  Below, we
write $\bpsi=(\polmu_i,\h_i)_{1 \le i \le s}$ and
$\bphi=(\phi_1,\dots,\phi_s)$, with
$\phi_i=(\polmu_i,\poldelta_i,\h_i,\H_i)$ for $i$ in $\{1,\dots,s\}$.

\begin{lemma}\label{lemma:clot}
  $\Vfiber$ is the Zariski closure of $\fbr(\Proj(L)^\mA,\fiber2)$.
\end{lemma}
\begin{proof}
  Since $\Proj(L)^\mA$ is contained in $V^\mA$,
  $\fbr(\Proj(L)^\mA,\fiber2)$ is contained in
  $\Vfiber=\fbr(V^\mA,\fiber2)$; the Zariski closure of
  $\fbr(\Proj(L)^\mA,\fiber2)$ is then contained in $\Vfiber$ as
  well. Thus, we have to prove the converse inclusion. This is
  immediate when $\Vfiber$ is empty. Now we will assume that $\Vfiber$
  is not empty, so that it is equidimensional of dimension
  $d-(\dalgo-1)$. Since we assumed $2 \le \dalgo \le d$ and
  $\dalgo \le (d+3)/2$, we deduce that $d-(\dalgo-1) \ge 1$.

  Let $Z$ be an irreducible component of $\Vfiber$. Because $Z$ has
  positive dimension $\dalgo-1$, there exists $\x$ in $Z-S^\mA$, and
  thus there exists $\x'=\x^{\mA^{-1}}$ in $Z^{\mA^{-1}}-S$. Because
  $\bpsi=(\polmu_i,\h_i)_{1 \le i \le s}$ is an atlas of $(V,Q,S)$,
  and $\x'$ is in $V$, we deduce that there exists $i$ in
  $\{1,\dots,s\}$ such that $\x'$ is in $\Open(\polmu_i)$. As a
  consequence, $\Open(\polmu_i)\cap Z^{\mA^{-1}} - S$ is not empty.
 
  Remark that $Z^{\mA^{-1}}$ is an irreducible component of
  ${\Vfiber}^{\mA^{-1}}$, and is thus contained in $V$. Because
  $(L;{\Vfiber}^{\mA^{-1}}, \scrY)$ has the global normal form property,
  property ${\gnf_3}$ and the statement in the last paragraph imply
  that $Z'=\Open(\polmu_i \poldelta_i)\cap Z^{\mA^{-1}} - S$ is not
  empty.  In particular, $Z'$ is a Zariski dense open subset of
  $Z^{\mA^{-1}}$, and thus ${Z'}^{\mA}$ is Zariski dense in~$Z$.

  On the other hand, $Z^{\mA^{-1}}$ is contained in $V$, so $Z'$ is
  contained in $\Open(\polmu_i \poldelta_i)\cap V-S$.  By
  Lemma~\ref{lemma:conseq0}, $Z'$ is thus contained in
  $\Open(\polmu_i \poldelta_i)\cap \Proj(L)$, and thus in $\Proj(L)$;
  as a result, ${Z'}^{\mA}$ is contained in $\Proj(L)^\mA$. Since $Z$,
  and thus ${Z'}^{\mA}$, lie over $\fiber2$, we deduce that ${Z'}^\mA$
  is contained in $\fbr(\Proj(L)^\mA,\fiber2)$.  Taking Zariski
  closures, we deduce that $Z$ itself is contained in the Zariski
  closure of $\fbr(\Proj(L)^\mA,\fiber2)$. Proceeding in this manner
  with all irreducible components of $\Vfiber$, we finish the proof.
\end{proof}

\begin{lemma}\label{transfer:fiber:lemma:easy}
  $\Vfiber=\Clos{(\Fiberlag(L^\mA,\param2,\paramsing2))}$.
\end{lemma}
\begin{proof}
  We have to prove that $\Vfiber$ is the Zariski closure of
  $\Proj(\Fiberlag(L^\mA, \param2, \paramsing2))$. By construction,
  $$\Proj(\Fiberlag(L^\mA, \param2, \paramsing2))=\fbr(\Proj(L)^\mA,\fiber2)-\fibersing2.$$ This
  implies the inclusions
  $$\Proj(\Fiberlag(L^\mA, \param2, \paramsing2)) \subset \fbr(\Proj(L)^\mA,\fiber2) \subset U' \cup \fibersing2.$$
  Let us temporarily denote by $U'$ the Zariski closure
  $\Clos{(\Fiberlag(L^\mA,\param2,\paramsing2))}$ of
  $\Proj(\Fiberlag(L^\mA, \param2, \paramsing2))$. Since $\fibersing2$
  is finite, the previous inclusions and the previous lemma show that
  $U' \subset \Vfiber \subset U' \cup \fibersing2$.  Because
  $\fibersing2$ is finite and $\Vfiber$ is equidimensional of positive
  dimension, the right-hand inclusion implies that
  $\Vfiber \subset U'$, from which the requested equality
  $\Vfiber=U'$ follows.
\end{proof}

We can now prove the proposition. The first item follows from Lemma
\ref{transfer:fiber:lemma:easy}, and when $\Vfiber$ is empty, there is
nothing more to prove.

If we assume that $\Vfiber$ is not empty, it remains to show how to
construct a global normal form for it. We first define the local
normal forms we will use for the generalized Lagrange system
$\Fiberlag(L^\mA, \param2, \paramsing2)$. Up to reordering $\bphi$, we
can suppose that there exists $s' \in \{0,\dots,s\}$ such that
$\Open(\polmu^\mA_i) \cap \Vfiber-\fibersing2$ is not empty for
$1 \le i \le s'$, and empty for $i> s'$. We let
$\bphi'=(\phi^\mA_1,\dots,\phi^\mA_{s'})$.  We prove now that $\bphi'$
satisfies properties $\gnf_1, \gnf_2$ and~$\gnf_3$.
\begin{enumerate}
\item[${\gnf_1}.$] We saw in
  Proposition~\ref{sec:lagrange:prop:local:fibers} that for all
  $\phi^\mA_i$, with $i \le s'$, are local normal forms for
  $\Fiberlag(L^\mA,\param2,\paramsing2)$.
\smallskip
\item[${\gnf_2}.$] Let $\bpsi'=(\psi_i^\mA)_{1 \le i \le s'}$; we
  need to prove that $\bpsi'$ is an atlas
  of
  $$(\Clos{(\Fiberlag(L^\mA,\param2,\paramsing2))}, \fiber2,
  \fibersing2),$$ or equivalently, by
  Lemma~\ref{transfer:fiber:lemma:easy}, of $(\Vfiber, \fiber2,
  \fibersing2)$.  Definition~\ref{sec:atlas:notation:fiber} shows that
  $\bpsi'$ is none other than the set of polynomials
  $\atlasfiber(\bpsi^\mA,V^\mA,Q,S^\mA,\fiber2)$.  Since all
  assumptions of Proposition~\ref{sec:atlas:prop:summary1} are
  satisfied, that proposition proves that $\bpsi'$ is indeed an atlas
  of $(\Vfiber,\fiber2,\fibersing2)$, so our claim is proved.
  \smallskip
\item[${\gnf_3}.$] Recall that we write $\scrY=Y_1,\dots,Y_r$.  Let
  $Z$ be an irreducible component of $Y^\mA_j$, for some $j$ in
  $\{1,\dots,r\}$. Suppose that $Z$ is contained in $\Vfiber$, and let
  $i$ in $\{1,\dots,s'\}$ be such that
  $\Open(\polmu^\mA_i) \cap Z-\fibersing2$ is not empty. We have
  to prove that
  $\Open(\polmu^\mA_i\poldelta_i^\mA) \cap Z-\fibersing2$ is not
  empty.

  Let $\x$ be in $\Open(\polmu^\mA_i) \cap Z-\fibersing2$. Because
  $\x$ is in $Z$, and thus in $\Vfiber$, $\x$ lies over $\fiber2$. In
  particular, $\x$ is not in $S^\mA$ (since if it were, it would
  belong to $\fbr(S^\mA,\fiber2)$, and thus to $\fibersing2$). In
  other words, $\x$ is in $\Open(\polmu^\mA_i) \cap Z-S^\mA$.

  Then, $\x'=\x^{\mA^{-1}}$ belongs to
  $\Open(\polmu_i) \cap Z^{\mA^{-1}}-S$, so that
  $\Open(\polmu_i) \cap Z^{\mA^{-1}}-S$ is not empty. Besides,
  $Z^{\mA^{-1}}$ is an irreducible component of $Y_j$, and it is
  contained in $V$. We deduce (by applying ${\gnf_3}$ to $L$) that
  $\Open(\polmu_i\poldelta_i) \cap Z^{\mA^{-1}}-S$ is not empty, and
  thus that $\Open(\polmu^\mA_i\poldelta^\mA_i) \cap Z-S^\mA$ is not
  empty.

  To summarize, both $\Open(\polmu^\mA_i) \cap Z-\fibersing2$ and
  $\Open(\polmu^\mA_i\poldelta^\mA_i) \cap Z-S^\mA$ are non-empty open
  subsets of the irreducible set $Z$, so their intersection
  $\Open(\polmu^\mA_i\poldelta^\mA_i) \cap Z-\fibersing2$ is non-empty
  as well.
\end{enumerate}

%%%%%%%%%%%%%%%%%%%%%%%%%%%%%%%%%%%%%%%%%%%%

%%%%%%%%%%%%%%%%%%%%%%%%%%%%%%%%%%%%%%%%%%%%
%\input{proof6.2}
\section{Proof of Proposition~\ref{sec:posso:prop2}}\label{chap:degreebounds}

The main goal of this section is to prove
Proposition~\ref{sec:posso:prop2}, whose statement is as follows: {\em
  Consider polynomials $\F=(F_1,\dots,F_P)$ in
  $\C[\X,\L_1,\dots,\L_k]$, with $n-e,n_1,\dots,n_k$ variables in the
  respective blocks $\X,\L_1,\dots,\L_k$, and having degrees in
  $\X,\L_1,\dots,\L_k$ respectively bounded by
$$
\begin{array}{cl}
(D_1,0, 0,\dots,0) & \text{for  $F_1,\dots,F_{p}$} \\
(D_2,1,0,\dots,0) & \text{for $F_{p+1},\dots,F_{p+p_1}$} \\
\vdots & \vdots \\
(D_2,1,1,\dots,1) & \text{for $F_{p+\cdots+p_{k-1}+1},\dots,F_{p+\cdots+p_k}$},
\end{array}$$
the total number of variables being $N-e$, with $N=n+n_1+\cdots+n_k$.
Write furthermore $N_i = n+n_1+\cdots+n_i$ and $P_i =
p+p_1+\cdots+p_i$, for $i=0,\dots,k$, and suppose that the following holds for 
all $i=0,\dots,k$:
\begin{itemize}
\item  $n_i$ and $p_i$ are positive,
\smallskip
\item $ N_i-e \ge P_i$.
\end{itemize}
Let finally $\Delta$ be the ideal generated by all $P$-minors of
$\jac(\F)$ and  consider the Zariski closure $V$ of $V(\F)-V(\Delta)$.
Then for $i$
in $\{1,\dots,P\}$, $V_i$ has degree at most $\DF(k, e,\n,\p,D_1,
D_2)$, with  
$$
\DF(k, e, \n,\p,D_1, D_2)=(P_k+1)^k D_1^{p} D_2^{n-e-p}  \prod_{i=0}^{k-1} N_{i+1}^{N_i-e-P_i}.
$$}
This proposition is proved in the second half of this section; we
start by proving a general multi-homogeneous bound that is a variant of
classical ones (see e.g.~\cite{vanderWaerden29,vanderWaerden78}),
adapted to our setting.

%%%%%%%%%%%%%%%%%%%%%%%%%%%%%%%%%%%%%%%%%%%%%%%%%%%%%%%%%%%%
%%%%%%%%%%%%%%%%%%%%%%%%%%%%%%%%%%%%%%%%%%%%%%%%%%%%%%%%%%%%
%%%%%%%%%%%%%%%%%%%%%%%%%%%%%%%%%%%%%%%%%%%%%%%%%%%%%%%%%%%%

\subsection{A multi-homogeneous B\'ezout bound}

As above, consider blocks of variables $\X,\L_1,\dots,\L_k$ of
respective lengths $n-e,n_1,\dots,n_k$, and let $N=n+n_1+\cdots+n_k$,
so that the total number of variables is $N-e$. We say that a
polynomial $f$ in $\C[\X,\L_1,\dots,\L_k]$ has multi-degree bounded by
$(D_0,D_1,\dots,D_k)$ if its degree in the group of variables $\X$,
resp.\ $\L_i$, is at most $D_0$, resp.\ $D_i$, for $1 \le i \le
k$. Our goal here is to give an upper bound on the degree of algebraic
sets defined by polynomials in $\C[\X,\L_1,\dots,\L_k]$ in terms of
their multi-degrees.

All along, we let $\m$ be the ideal $\langle
\zeta_0^{n-e+1},\zeta_1^{n_1+1},\dots,\zeta_k^{n_k+1}\rangle$ in
$\Z[\zeta_0,\zeta_1\dots,\zeta_k]$. If $A$ is a polynomial in
$\Z[\zeta_0,\zeta_1,\dots,\zeta_k]$, $|A|_\infty$ is the maximum of the
absolute values of its coefficients, and $|A|_1$ is the sum of the
absolute values of its coefficients. If $I$ is an ideal in
$\C[\X,\L_1,\dots,\L_k]$, $V(I)$ will denote its zero-set in $\C^N$.

\begin{proposition}\label{sec:posso:prop1}
  Let $F_1,\dots,F_P$ be polynomials in $\C[\X,\L_1,\dots,\L_k]$ of
  mul\-ti-degrees respectively bounded by $(D_{i,0},D_{i,1},\dots,D_{i,k})$, for
  $i=1,\dots,P$. Let $V\subset \C^{N-e}$ be the equidimensional component
  of $V(F_1,\dots,F_P)$ of dimension $N-e-P$. Let further
  $$A=\prod_{i=1}^P (D_{i,0} \zeta_0 + D_{i,1} \zeta_1 +\cdots + D_{i,k} \zeta_k) \bmod \m.$$
  Then $\deg(V) \le |A|_1$.
\end{proposition}
This paragraph is devoted to prove Proposition \ref{sec:posso:prop1}.
This result is in essence the calculation of an intersection product
in the Chow ring of the multi-projective space $\P^{n-e} \times  \P^{n_1} \times \cdots
\times \P^{n_k}$, which is indeed $\Z[\zeta_0,\zeta_1,\dots,\zeta_k]/\m$.
However, the proof does not require familiarity with the techniques of
intersection theory; we rely on the aforementioned results of van der
Waerden and a theorem of~\cite{MoSo87} for these aspects.

Let $X_{0},L_{1,0},\dots,L_{k,0}$ be homogenization variables and let
$\X'$ and $\L'_1,\dots,\L'_k$ be the blocks of variables obtained by
adding respectively $X_{0},L_{1,0},\dots,L_{k,0}$ to $\X$ and
$\L_1,\dots,\L_k$. To a polynomial $f$ in $\C[\X,\L_1,\dots,\L_k]$, we
associate $f^H$ obtained by homogenizing $f$ in each block of
variables separately. To an ideal $I$ in $\C[\X,\L_1,\dots,\L_k]$, we
associate the ideal $I^H$ generated by the polynomials $\{f^H \ | \ f
\in I\}$. Conversely, for $F$ in $\C[\X',\L'_1,\dots,\L'_k]$,
$\varphi(F)$ is the polynomial obtained from $F$ by evaluating 
$X_0$ and all $L_{i,0}$ at $1$.

In what follows, we let $I$ be the radical of the ideal $\langle
F_1,\dots,F_P\rangle\subset \C[\X,\L_1,\dots,\L_k]$ and let $I=\Prime_1 \cap
\dots \cap \Prime_t$ be its prime decomposition. We further let $t' \le t$
and $I'= \Prime_1 \cap \dots \cap \Prime_{t'}$ be the intersection of the
components of dimension $d=N-e-P$ (reordering may be needed); thus, we
have
\begin{equation}
  \label{eq:degV}
\deg(V)=\deg(V(\Prime_1)) + \cdots + \deg(V(\Prime_{t'})).  
\end{equation}
\begin{lemma}\label{lemma7}
  The ideal ${I'}^H$ is radical and $\Prime_1^H \cap \dots \cap \Prime_{t'}^H$
  is its prime decomposition.
\end{lemma}
\begin{proof}
  First, we establish the following easy facts:
  \begin{enumerate}
  \item\label{f1} If $f$ is in $\C[\X,\L_1,\dots,\L_k]$, then $\varphi(f^H)=f$.
\smallskip
  \item\label{f3} If $J$ is an ideal of $\C[\X,\L_1,\dots,\L_k]$ and $F$
    is in $J^H$, $\varphi(F)$ is in $J$.
  \end{enumerate}
  The first item is obvious. To prove~\eqref{f3}, note that the
  assumption says that $F$ is a polynomial combination of polynomials
  $f^H$, for $f$ in $J$; apply $\varphi$ to conclude, using
  fact~\eqref{f1}.

  Now we can prove that all ideals $\Prime_i^H$ are prime, and that
  for all $i \ne i'$ in
  $\{1,\dots,t'\}$,$(\Prime_i \cap \Prime_{i'})^H= \Prime_i^H \cap
  \Prime_{i'}^H$
  and $\Prime_i^H \not \subset \Prime_{i'}^H$. The first two
  statements are~\cite[Proposition~4.3.10.b--d]{KrRo05}.  For the last
  one, suppose that $\Prime_i^H \subset \Prime_{i'}^H$, and let $f$ be
  in $\Prime_i$. Then, $f^H$ is in $\Prime_i^H$, so $f^H$ is in $\Prime_{i'}^H$;
  applying $\varphi$, $f=\varphi(f^H)$ is in $\Prime_{i'}$
  (facts~\eqref{f1} and~\eqref{f3}). This proves that
  $\Prime_i \subset \Prime_{i'}$, a contradiction.

  Iterating the second property above,
  ${I'}^H = \Prime_1^H \cap \cdots \cap \Prime_{t'}^H$; by the first
  property, all $\Prime_i^H$ are prime (so ${I'}^H$ is radical) and by the
  last one, $\Prime_i^H \not \subset \Prime_j^H$ holds for all $i\ne j$. This
  proves the lemma.
\end{proof}

If $J$ is a {\em homogeneous} ideal of $\C[\X',\L'_1,\dots,\L'_k]$,
$V^h(J)$ will denote the projective algebraic set it defines in
$\P^{N-e+k}$. If $Z$ is a projective algebraic set in $\P^{N-e+k}$, we
denote by $\deg(Z)$ its {\em degree}, which is defined as in the
affine case.

Finally, note that if $J$ is an ideal in $\C[\X,\L_1,\dots,\L_k]$,
$J^H \subset \C[\X',\L'_1,\dots,\L'_k]$ is multi-homogeneous, and thus
homogeneous in $N-e+k+1$ variables, so $V^h(J^H)\subset \P^{N-e+k}$ is
well-defined.

\begin{lemma}\label{lemma8}
  If $\Prime$ is a prime ideal in $\C[\X,\L_1,\dots,\L_k]$, the inequality
  \[\deg(V(\Prime)) \le \deg(V^h(\Prime^H))\] holds.
\end{lemma}
\begin{proof}
  Consider the affine cone $C$ defined by $\Prime^H$ in $\C^{N-e+k+1}$.
  By construction, the degree of $C$ equals $\deg(V^h(\Prime^H))$.
  
  Intersecting with the linear space
  $V(X_{0}-1,L_{1,0}-1,\dots,L_{k,0}-1)$ yields an algebraic set $C'$,
  with $\deg(C')\le \deg(C)$; note as well that $C'$ is defined by
  $\Prime$ and all linear equations
  $X_{0}-1,L_{1,0}-1,\dots,L_{k,0}-1$. Finally, projecting on
  $\C^{N-e}$, we obtain that $\deg(V(\Prime)) \le \deg(C')$, and we
  are done.
 \end{proof}

 If $J'$ is a multi-homogeneous ideal in $\C[\X',\L'_1,\dots,\L'_k]$,
 $\Vmp(J')$ will denote the multi-projective algebraic set it defines
 in $\P^{n-e} \times \P^{n_1}\times \cdots \times \P^{n_k}$ (the super-script $mp$
 indicates that the set lies in a multi-projective set).

 The dimension of a multi-projective algebraic set $Z$ in
 $\P^{n-e} \times \P^{n_1}\times \cdots \times \P^{n_k}$ is the Krull
 dimension of $\C[\X',\L'_1,\dots,\L'_k]/\Ideal(Z)$ minus $(k+1)$, where
 $\Ideal(Z)$ is the multi-homogeneous defining ideal of
 $Z$. By~\cite[Par. 12, pp. 754]{vanderWaerden29}, if $\Prime$ is a
 prime ideal in $\C[\X,\L_1,\dots,\L_k]$,
 $\dim(V(\Prime))=\dim(\Vmp(\Prime^H))$. Equidimensional
 multi-projective algebraic sets are defined as in the affine or
 projective cases.

For any integer $\ell$, let ${\mathfrak{R}}(\ell)$ be the set of
$(k+1)$-uples of integers $$\mm=(m_0,m_1,\dots,m_k) \in \N^{k+1}$$ such
that $|\mm|=\ell$, where we write $|\mm|=m_0+m_1+\cdots+m_k$.  Let then
$Z \subset \P^{n-e} \times \P^{n_1}\times \cdots \times \P^{n_k}$ be an
$\ell$-equidimensional multi-projective algebraic set. The {\em
  multi-degree} of $Z$ is a vector
$\bdelta(Z)=(\delta(Z,\mm))_{\mm \in {\mathfrak{R}}(\ell)}$: for any
such $\mm$, $\delta(Z,\mm)$ is the number of intersection points of
$Z$ with $m_0,\dots,m_k$ generic hyperplanes in respective coordinates
$\X',\L'_1,\dots,\L'_k$.

We can now return to the proof of our proposition. Recall that $I'$ is
the defining ideal of~$V$, and that $\Prime_1,\dots,\Prime_{t'}$ are its prime
components.
\begin{lemma}\label{lemma9a}
  The multi-projective set $\Vmp({I'}^H)$ is equidimensional of dimension
  $d=N-e-P$ and satisfies 
$$\deg(V) \le \sum_{\mm \in {\mathfrak{R}}(d)} \delta(\Vmp({I'}^H),\mm).$$
\end{lemma}
\begin{proof}
  By the remark above, each $\Vmp(\Prime_i^H)$ has dimension $d=N-e-P$.  Because
  all $\Prime_i^H$ are prime, we can use Van der Waerden's result~\cite{vanderWaerden78}
stating that
  $$ \deg(V^h(\Prime_i^H)) = \sum_{\mm \in {\mathfrak{R}}(d)} \delta(\Vmp(\Prime_i^H),\mm).$$
  Combining this with the bound in Lemma~\ref{lemma8}, we obtain
  $$\deg(V(\Prime_i)) \le \sum_{\mm \in {\mathfrak{R}}(d)} \delta(\Vmp(\Prime_i^H),\mm).$$
  Finally, we sum over $i=1,\dots,t'$. On the left,
  from~\eqref{eq:degV}, we get $\deg(V)$.  On the right, we get
  $$\sum_{i \le t'} \sum_{\mm \in {\mathfrak{R}}(d)} \delta(\Vmp(\Prime_i^H),\mm) =
  \sum_{\mm \in {\mathfrak{R}}(d)} \sum_{i \le t'} \delta(\Vmp(\Prime_i^H),\mm).$$
  Now, $\Vmp({I'}^H)$ is equidimensional of dimension $d$ and thus,
  for all $\mm$,
  $$ \sum_{i \le t'} \delta(\Vmp(\Prime_i^H),\mm) = \delta(\Vmp({I'}^H),\mm).$$
  This proves the lemma.
\end{proof}

Recall now that our input polynomials are denoted by $F_1,\dots,F_P$.
In the following lemma, if $Z$ is a multi-projective algebraic set in
$\P^{n-e} \times \P^{n_1}\times \cdots \times \P^{n_k}$, $Z_d$ will denote the union
of the irreducible components of $Z$ of dimension $d$.
\begin{lemma}\label{lemma10}
  Let $J$ be the ideal $J=\langle F_1^H,\dots, F_P^H\rangle$.  Then
  $$\deg(V) \le \sum_{\mm \in {\mathfrak{R}}(d)} \delta(\Vmp(J)_d,\mm).$$
\end{lemma}
\begin{proof}
  Fix a multi-index $\mm$ such that $|\mm|=d$. Recall that $I$ is the
  radical of the ideal $\langle F_1, \ldots, F_P\rangle$ and that $I'$
  is the intersection of those prime components of $I$ which have
  dimension $d=N-e-P$.

  We are going to prove the inequalities
  $$\delta(\Vmp({I'}^H),\mm) = \delta(\Vmp(I^H)_d,\mm) 
  \quad\text{and}\quad
  \delta(\Vmp(I^H)_d,\mm) \le \delta(\Vmp(J)_d,\mm).$$
\begin{itemize}
\item Lemma~\ref{lemma7} shows that
  $\Prime_1^H \cap \dots \cap \Prime_{t'}^H$ is the prime
  decomposition of ${I'}^H$; similarly,
  $\Prime_1^H \cap \dots \cap \Prime_{t}^H$ is the prime decomposition
  of ${I}^H$. For $j > t'$, the dimension of $\Vmp(\Prime_j^H)$ is greater
  than $d$; we deduce that $\Vmp(I^H)_d=\Vmp({I'}^H)$, and the first
  equality follows.  \smallskip
\item Let $K$ be the ideal $\langle F_1,\dots,F_P \rangle$, so that
  $I=\sqrt{K}$. Proposition~4.3.10.c of~\cite{KrRo05} shows that
  $I^H=\sqrt{K^H}$, so that $\Vmp(I^H)=\Vmp(K^H)$ and $\Vmp(I^H)_d=\Vmp(K^H)_d$.
  On the other hand, Corollary~4.3.8 of~\cite{KrRo05} shows that $K^H =
  J : (X_0 L_{1,0}\cdots L_{k,0})^\infty$. This 
  implies $\delta(\Vmp(K^H)_d,\mm) \le \delta(\Vmp(J)_d,\mm)$ 
  and thus gives the second claimed inequality.
\end{itemize}
The conclusion immediately follows from Lemma~\ref{lemma9a}.
\end{proof}

For $\mm=(m_0,m_1,\dots,m_k)$ in ${\mathfrak{R}}(d)$, recall that
$\delta(\Vmp(J)_d,\mm)$ is the number of intersection points of
$\Vmp(J)_d$ with $m_0,m_1,\dots,m_k$ generic hyperplanes
$H_{0,1},\dots,H_{k,m_k}$ in respective coordinates
$\X',\L'_1,\dots,\L'_k$. Because $d=N-e-P$, this is thus also the generic
number of isolated solutions of
$F_1^H,\dots,F_P^H,H_{0,1},\dots,H_{k,m_k}$ in
$\P^{n-e}\times \P^{n_1}\times \cdots \times \P^{n_k}$ (the intersections of
higher-dimensional components of $\Vmp(J)$ with
$H_{0,1},\dots,H_{k,m_k}$ have positive dimension). Let $A_0$ be the
polynomial
$$A_0=\prod_{i=1}^P (D_{i,0} \zeta_0 + D_{i,1} \zeta_1 + \cdots + D_{i,k} \zeta_k).$$
By the multi-homogeneous B\'ezout theorem given in~\cite{MoSo87}, we
deduce that
\begin{eqnarray*}
\delta(\Vmp(J)_d)_\mm &\le & {\rm coeff}( A_0 \zeta_0^{m_0} \cdots \zeta_k^{m_k}, \zeta_0^{n} \cdots \zeta_k^{n_k}) \\
 &\le& {\rm coeff}(A_0,\zeta_0^{n-m_0} \cdots \zeta_k^{n_k-m_k}).
\end{eqnarray*}
We deduce from Lemma~\ref{lemma10} the inequality
$$\deg(V) \le \sum_{\mm \in {\mathfrak{R}}(d)} {\rm
  coeff}(A_0,\zeta_0^{n-m_0} \cdots \zeta_k^{n_k-m_k}).$$ To
conclude the proof of Proposition~\ref{sec:posso:prop1}, it suffices to observe that the last sum equals $|A|_1$,
with $A = A_0 \bmod \m$.

%%%%%%%%%%%%%%%%%%%%%%%%%%%%%%%%%%%%%%%%%%%%%%%%%%%%%%%%%%%%
%%%%%%%%%%%%%%%%%%%%%%%%%%%%%%%%%%%%%%%%%%%%%%%%%%%%%%%%%%%%
%%%%%%%%%%%%%%%%%%%%%%%%%%%%%%%%%%%%%%%%%%%%%%%%%%%%%%%%%%%%

\subsection{Proof of the proposition}\label{ssec:proof:posso:prop2}

We can now prove Proposition~\ref{sec:posso:prop2}. Consider a
non-negative integer $e$, polynomials $\F=(F_1,\dots,F_P)$ in
$\C[\X,\L_1,\dots,\L_k]$, with $n-e,n_1,\dots,n_k$ variables in the
respective blocks $\X,\L_1,\dots,\L_k$, and having multi-degrees
bounded by
$$
\begin{array}{cl}
(D_1,0, 0,\dots,0) & \text{for  $F_1,\dots,F_{p}$} \\
(D_2,1,0,\dots,0) & \text{for $F_{p+1},\dots,F_{p+p_1}$} \\
\vdots & \vdots \\
(D_2,1,1,\dots,1) & \text{for $F_{p+\cdots+p_{k-1}+1},\dots,F_{p+\cdots+p_k}$,}
\end{array}
$$
and we assume 
\begin{equation}\label{eq:N-P2}
  N_i-e \ge P_i, \quad\text{with}\quad N_i = n+\cdots+n_i\quad\text{and}\quad P_i = p+\cdots+p_i.
\end{equation}
Let $\Delta$ be the ideal generated by all $P$-minors of
$\jac(\F)$, and for $i \le P$, let $V_i$ be the Zariski closure of
$V(F_1,\dots,F_i)-V(\Delta)$. Our goal is to prove 
that   for $i$ in $\{1,\dots,P\}$, $V_i$ has degree at most $\DF(k, e,\n,\p,D_1,
D_2)$, with
$$
\DF(k, e, \n,\p,D_1, D_2)=(P_k+1)^k D_1^{p} D_2^{n-e-p}  \prod_{i=0}^{k-1} N_{i+1}^{N_i-e-P_i}.
$$
In what follows, as in the previous section, $\m$ is the
ideal $\langle \zeta_0^{n-e+1},\zeta_1^{n_1+1},\dots,\zeta_k^{n_k+1}\rangle$ in
$\Z[\zeta_0,\dots,\zeta_k]$.
\begin{lemma}\label{lemmaA}
  Suppose that all inequalities in~\eqref{eq:N-P2} hold.  Let $0\le
  i\le k$ and let $A$ be a homogeneous polynomial in
  $\Z[\zeta_0,\dots,\zeta_i] \subset \Z[\zeta_0,\dots,\zeta_k]$ with
  non-negative coefficients, of degree less than $P_i$, and reduced
  with respect to $\m$. Let also $b=d_0 \zeta_0 + \cdots + d_i
  \zeta_i$, with all $d_i$ positive integers and $B = Ab \bmod
  \m$. Then, $|A|_\infty \le |B|_\infty$.
\end{lemma}
\begin{proof}
  Let $z=\zeta_0^{u_0}\cdots \zeta_i^{u_i}$ be a monomial that appears
  in $A$ with a non-zero coefficient, so that $z$ is reduced with
  respect to $\m$. We will prove that there exists $\ell \le i$ such
  that $z'=z \zeta_\ell$ is reduced with respect to $\m$. Since
  all $d_i$'s and all coefficients of $A$ are positive integers, this
  implies that the coefficient of $z$ in $A$ is less than or equal to
  that of $z'$ in $B$, and the claim $|A|_\infty \le |B|_\infty$
  follows.

  We argue by contradiction, assuming that for all $\ell \le i$, $z
  \zeta_\ell$ is not reduced with respect to $\m$.
  
  First, remark that since $A$ is reduced with respect to $\m$, we
  have $u_0 \le n-e$ and $u_\ell \le n_\ell$ holds for
  $\ell=1,\dots,i$. On the other hand, if $z \zeta_\ell$ is not
  reduced with respect to $\m$, we have either $u_0 + 1 > n-e$ (if
  $\ell=0$) or $u_\ell +1 > n_\ell$ (otherwise), since $\zeta_\ell$ is
  the only variable whose exponent changes; in view of the
  inequalities above, this implies that $u_0 = n-e$ (if $\ell=0$) or
  $u_\ell = n_\ell$ (otherwise).  If this is the case for {\em all}
  values of $\ell$, $z$ has total degree $n-e+n_1+\cdots +n_i=N_i-e$; this
  is impossible, since $z$ has total degree less than $P_i$ and $P_i
  \le N_i-e$, by~\eqref{eq:N-P2}.
\end{proof}

Let
$$A = (D_1 \zeta_0)^{p} (D_2 \zeta_0 + \zeta_1)^{p_1} \cdots  
(D_2 \zeta_0 + \zeta_1 + \cdots + \zeta_k)^{p_k} \bmod \m.$$ 
The next lemma shows that it will be enough to prove an upper bound on the
coefficients of $A$.
\begin{lemma}
  Suppose that all inequalities in~\eqref{eq:N-P2} hold.  For all $0\le
  i\le k$, the inequality $\deg(V_{i}) \le (P_k+1)^k |A|_\infty$
  holds.
\end{lemma}
\begin{proof}
  Define $a_0=D_1\zeta_0$ and for $\ell=1,\dots,k$, $a_\ell = (D_2\zeta_0 + \zeta_1+\cdots
  +\zeta_\ell)$. Let $P_{-1}=0$ and, for $\ell=-1,\dots,k-1$ and $j=1,\dots,p_{\ell+1}$, define
further
$$A_{\ell,j} = a_0^{p} \cdots a_{\ell}^{p_{\ell}} a_{\ell+1}^j \bmod
\m;$$ remark that this polynomial has degree $P_\ell+j$, and that $A=A_{k-1,p_k}$.  

Fix now $i$ in $\{1,\dots,P\}$. There exists a unique $\ell$ in
$\{-1,\dots, k-1\}$ such that $P_\ell < i \le P_{\ell+1}$; let then
$j=i-P_\ell$, so that $i=P_\ell+j$; note that $0 < j \le p_{\ell+1}$
and that $A_{\ell,j}$ has degree~$i$.  
Proposition~\ref{sec:posso:prop1} gives the bound $\deg(V_{i}) \le
|A_{\ell,j}|_1$ (since $V_{i}$ is the union of some of the minimum
dimensional components defined by the first $i$ equations). Remark
next that for all $\ell,j$, $A_{\ell,j}$ has total degree at most
$P_k$, so it has at most $(P_k+1)^k$ non-zero coefficients. As a
consequence, we get $\deg(V_{i}) \le (P_k+1)^k |A_{\ell,j}|_\infty$.

It remains to give an upper bound on $ |A_{\ell,j}|_\infty$.  Fix
$\ell$ in $\{-1,\dots, k-1\}$, and take first $j$ in
$\{1,\dots,p_{\ell+1}-1\}$. Then, $A_{\ell,j+1}=A_{\ell,j} a_{\ell+1}
\bmod \m$. Since $A_{\ell,j}$ lies in
$\Z[\zeta_0,\dots,\zeta_{\ell+1}]$, has degree $P_{\ell}+j <
P_{\ell+1}$, and $a_{\ell+1}=D_2\zeta_0 + \zeta_1+\cdots
  +\zeta_{\ell+1}$ has positive coefficients, 
Lemma~\ref{lemmaA} shows that $|A_{\ell,j}|_\infty \le |A_{\ell,j+1}|_\infty$.

Consider now $\ell$ in $\{-1,\dots, k-2\}$ and $j=p_{\ell+1}$, so that
$$A_{\ell+1,1}=A_{\ell,p_{\ell+1}} a_{\ell+2} \bmod \m.$$ Now,
$A_{\ell,p_{\ell+1}}$ has degree $P_{\ell+1} < P_{\ell+2}$, lies in
$\Z[\zeta_0,\dots,\zeta_{\ell+1}] \subset
\Z[\zeta_0,\dots,\zeta_{\ell+2}]$, and $a_{\ell+2}=D_2\zeta_0 +
\zeta_1+\cdots +\zeta_{\ell+2}$ has positive coefficients. Thus, as
before, we deduce from Lemma~\ref{lemmaA} that
$|A_{\ell,p_{\ell+1}}|_\infty \le |A_{\ell+1,1}|_\infty$. Altogether,
this proves that for all $\ell,j$, $|A_{\ell,j}|_\infty \le
|A|_\infty$, as claimed.
\end{proof}

The inequality in the next lemma is then sufficient to prove
Proposition~\ref{sec:posso:prop2}.

\begin{lemma}
  The inequality $|A|_\infty \le D_1^{p} D_2^{n-e-p}
  \prod_{i=0}^{k-1} N_{i+1}^{N_i-e-P_i}$
  holds.
\end{lemma}
\begin{proof}
 The polynomial $A$ is homogeneous of total degree $P_k =
  p+\cdots+p_k$, so all its monomials have the form $\zeta_0^{u_0}
  \cdots \zeta_k^{u_k}$, with $u_0 + \cdots + u_k = p + \cdots +
  p_k$, $u_0 \le n-e$ and $u_\ell \le n_\ell$ for $\ell \ge 1$. Then, considering
  successively $\zeta_k,\dots,\zeta_0$, we see that the coefficient of
  this monomial in $A$ is
{\small $$D_1^{p}D_2^{p_1 + \cdots +p_k- (u_1 + \cdots + u_k)}
\binom{p_1 + \cdots +p_k - u_2 - \cdots - u_k}{u_1} \cdots
\binom{p_{k-1}+p_k-u_k}{u_{k-1}} \binom{p_k}{u_k}. $$ }
Since $u_0 + \cdots + u_k = p + \cdots + p_k$, 
this equals
\begin{equation}
  \label{eq:coeff}
D_1^{p}D_2^{u_0-p} \binom{p_1 + \cdots +p_k - u_2 - \cdots - u_k}{u_1}\cdots \binom{p_{k-1}+p_k-u_k}{u_{k-1}} \binom{p_k}{u_k}.  
\end{equation}
Next, we use the fact that
$$p + \cdots + p_k = u_0 + \cdots + u_k$$
to deduce
$$p_\ell + \cdots +p_k-u_\ell - \cdots - u_k
= u_0 + \cdots + u_{\ell-1}- p - \cdots - p_{\ell-1}$$ 
and 
$$
p_\ell + \cdots +p_k- u_{\ell+1} - \cdots - u_k =   u_0+\cdots+u_\ell-p-\cdots-p_{\ell-1}.
$$
This implies respectively
$$p_\ell + \cdots +p_k-u_\ell - \cdots - u_k
\le n-e + n_1 + \cdots + n_{\ell-1}- p - \cdots - p_{\ell-1} =
N_{\ell-1}-e-P_{\ell-1}$$
and 
\begin{eqnarray*}
  p_\ell + \cdots +p_k-u_{\ell+1} - \cdots - u_k
& \le & n-e + n_1 \cdots + n_{\ell-1} + n_\ell - p - \cdots - p_{\ell-1} \\
& \le & n_\ell + N_{\ell-1}-e-P_{\ell-1}\\
& \le & N_\ell-e \\
& \le & N_\ell. 
\end{eqnarray*}
Finally, since ${\binom{a}{b}}\leq a^{a-b}$, we have thus proved the
inequality
$$\binom{p_\ell + \cdots +p_k - u_{\ell+1} - \cdots - u_k}{u_\ell} \le 
N_\ell^{N_{\ell-1}-e-P_{\ell-1}}.$$ Using this upper bound and $u_0\leq n-e$
in~\eqref{eq:coeff} proves our claim.
\end{proof}

%%%%%%%%%%%%%%%%%%%%%%%%%%%%%%%%%%%%%%%%%%%%

%%%%%%%%%%%%%%%%%%%%%%%%%%%%%%%%%%%%%%%%%%%%
% Solving polynomial systems
%\input{geosolve-new}
\section{Solving polynomial systems}\label{chap:posso}

The contents of this section is independent from most previous ones:
we revisit algorithms for solving polynomial systems, with a focus on
dimension zero and dimension one.

Finite sets of points will be encoded by zero-dimensional
parametrizations: we discuss basic algorithms for this data structure
in Subsection~\ref{sec:basicroutinesparam0}; curves will be
represented by a one-dimensional analogue, which is the subject of
Subsection~\ref{sec:basicroutinesparam1}. In
Subsections~\ref{sec:prodfields} and~\ref{sec:param}, we present
extensions of these questions to computations over {\em products of
  fields}, which will be needed later on. Finally, the longest
paragraph in this section is Subsection~\ref{sec:possobasepoints}; it
presents an adaptation of the geometric resolution algorithm
of~\cite{GiLeSa01} (which
follows~\cite{GiHeMoPa95,GiHeMoPa97,GiHeMoMoPa98}) to systems with
coefficients in a product of fields. The ideas we use to solve this
question are well-known (dynamic evaluation techniques), but
controlling their complexity is not straightforward. The final
subsection uses these results to describe an algorithm called ${\sf
  SingularPoints}$ that was mentioned in the main text.

In all algorithms below, we count arithmetic operations
$\{+,-,\times,\div\}$ in $\QQ$ at unit cost. To state our complexity
estimates we use the $\softO(\ )$ notation, so logarithmic factors are
omitted: $f$ is in $\softO(g)$ if there exists a constant $a$ such
that $f$ is in $O(g\log^a(g))$.  For instance, over $\QQ[X]$, polynomial
multiplication, Euclidean division, extended GCD computation and
squarefree factorization in degree $D$ can all be done using
$\softO(D)$ operations in $\QQ$~\cite{GaGe99}.

For most algorithms involving solving systems of multivariate
polynomial equations, we will use a {\em straight-line program}
encoding for the input, as was already done for generalized Lagrange
systems.

Many algorithms below are probabilistic, in the sense that they use
random elements in $\QQ$. Every time a random vector $v$ is chosen in
some parameter space $\QQ^i$, there will exist a non-zero polynomial
$\Delta$ such that the choice leads to success as soon as
$\Delta(\gamma)\ne 0$. Most such algorithms are Monte Carlo, since we
are not always able to verify correctness in an admissible amount of
time. If we are able to detect some cases of failure, we return the
string ${\sf fail}$ (but even when we do not return ${\sf fail}$, we
do not guarantee that the output is correct).

%%%%%%%%%%%%%%%%%%%%%%%%%%%%%%%%%%%%%%%%%%%%%%%%%%%%%%%%%%%%
%%%%%%%%%%%%%%%%%%%%%%%%%%%%%%%%%%%%%%%%%%%%%%%%%%%%%%%%%%%%
%%%%%%%%%%%%%%%%%%%%%%%%%%%%%%%%%%%%%%%%%%%%%%%%%%%%%%%%%%%%

\subsection{Zero-dimensional parametrizations}\label{sec:basicroutinesparam0}

Let $\KK$ be a field of characteristic zero and $\CC$ be its algebraic
closure. A zero-dimensional parametrization
$\scrQ=((q,v_1,\dots,v_N),\pollambda)$ with coefficients in $\KK$
consists in a sequence of polynomials $(q,v_1,\dots,v_N)$, such that
$q\in \KK[T]$ is squarefree and all $v_i$ are in $\KK[T]$ and satisfy
$\deg(v_i) < \deg(q)$, and in a $\KK$-linear form $\pollambda$ in
variables $X_1,\dots,X_N$, such that $\pollambda(v_1,\dots,v_N)=T$. We
already used several times the fact that the corresponding algebraic
set, denoted by $\Zeroes(\scrQ)\subset \CC^N$, is defined by
$$q(\roottau) = 0, \qquad X_i = v_i(\roottau) \ \ (1 \le i \le N);$$ the
constraint on $\pollambda$ says that the roots of $q$ are precisely the
values taken by $\pollambda$ on $\Zeroes(\scrQ)$. The {\em degree} of $\scrQ$
is then defined as $\degQ=\deg(q)$, and we call $q$ the {\em minimal
  polynomial} of $\scrQ$. By convention, when $N=0$, $\scrQ$ is the
empty sequence; it defines $\{\bullet\}\subset\C^0$ and we set $\degQ=1$.

Zero-dimensional parametrizations are used in our algorithms to
represent zero-dimen\-sional algebraic sets. In the following
paragraphs, we describe a few elementary operations on
zero-dimensional algebraic sets defined by such an encoding.  All
zero-dimen\-sional parametrizations used in this section have
coefficients in $\KK=\QQ$; we will use $\KK=\C$ as well in the next sections.

%%%%%%%%%%%%%%%%%%%%%%%%%%%%%%%%%%%%%%%%%%%%%%%%%%%%%%%%%%%%

We first mention a concept that will appear, implicitly or explicitly,
on several occasions. If $\scrQ=((q,v_1,\dots,v_N),\pollambda)$ is a
zero-dimensional parametrization with coefficients in $\QQ$, we call
{\em decomposition} of $\scrQ$ the data of parametrizations
$\scrQ_1,\dots,\scrQ_s$, with
$\scrQ_i=((q_i,v_{i,1},\dots,v_{i,N}),\pollambda)$, such that
$q=q_1 \cdots q_s$ and for all $i,j$, $v_{i,j} = v_j \bmod q_i$.
Geometrically, this means that we have decomposed $\Zeroes(\scrQ)$ as
the disjoint union of $\Zeroes(\scrQ_1),\dots,\Zeroes(\scrQ_s)$.

We can now continue with our basic algorithms, starting from an
algorithm performing linear changes of variables on zero-dimensional
parametrizations.
\begin{lemma}\label{lemma:complexity0:changevar}
  Let $\scrQ$ be a zero-dimensional parametrization of degree $\degQ$,
  with $\Zeroes(\scrQ)\subset \C^N$, and let $\mA$ be in $\GL(N,\QQ)$.
  There exists an algorithm ${\sf ChangeVariables}$ which takes as
  input $\scrQ$ and $\mA$ and returns a zero-dimensional
  parametrization $\scrQ^\mA$ such that
  $\Zeroes(\scrQ^\mA)=\Zeroes(\scrQ)^\mA$ using $\softO(N^2\degQ+N^3)$
  operations in $\QQ$.
\end{lemma}
\begin{proof}
  Suppose that the input parametrization $\scrQ$ consists in
  polynomials $(q, v_1, \ldots, v_N)$ in $\QQ[T]$ and a linear form
  $\pollambda$. First, we compute $\mA^{-1}$ in time $O(N^3)$. Then,
  computing a parametrization of
  $\Zeroes(\scrQ)^\mA=\varphi_\mA(\Zeroes(\scrQ))$, with
  $\varphi_\mA: \x \mapsto \mA^{-1} \x$, is simply done by multiplying
  $\mA^{-1}$ by the vector $[v_1, \ldots, v_N]^t$, and multiplying
  $\mA^t$ by the vector of coefficients of $\pollambda$, so the
  running time is $\softO(N^2 \degQ)$ operations in $\QQ$.
\end{proof}

Next, we consider set-theoretic operations such as union, intersection
and difference.  The first operation of this kind takes as input
zero-dimensional parametrizations $\scrQ$ and $\scrQ'$ encoding finite
sets of points in $\C^N$; it computes a zero-dimensional
parametrization encoding $\Zeroes(\scrQ)-\Zeroes(\scrQ')$. The
algorithm is described in Lemma~3 in~\cite{PoSc13} and leads to the
following result. This result is probabilistic (the algorithm chooses
at random a linear form in $X_1,\dots,X_N$ that must take pairwise
distinct values on the points of both $\Zeroes(\scrQ)$ and
$\Zeroes(\scrQ')$).

\begin{lemma}\label{sec:basicroutinesparam:lemma:discard}
  Let $\scrQ$ and $\scrQ'$ be zero-dimensional parametrizations, with
  $\Zeroes(\scrQ)$ and $\Zeroes(\scrQ')$ in $\C^N$ of respective
  degrees $\degQ$ and $\degQ'$. There exists a probabilistic algorithm
  ${\sf Discard}$ which takes as input $\scrQ$ and $\scrQ'$ and
  returns either a zero-dimensional parametrization $\scrQ''$ or
  ${\sf fail}$ using $\softO(N\max(\degQ,\degQ')^2)$ operations in
  $\QQ$.  In case of success,
  $\Zeroes(\scrQ'')=\Zeroes(\scrQ)-\Zeroes(\scrQ')$.
\end{lemma}

Algorithm ${\sf Union}$ below takes as input a sequence of
zero-dimensional paramet\-rizations $\scrQ_1, \ldots, \scrQ_s$ and it
returns a parametrization encoding
$\Zeroes(\scrQ_1)\cup\cdots\cup \Zeroes(\scrQ_s)$. The algorithm is given in
Lemma~3 of~\cite{PoSc13} as well, for the case $s=2$; the general case
is dealt with in the same manner, and gives the following result.

\begin{lemma}\label{sec:main:lemma:union}
  Let $\scrQ_1,\dots,\scrQ_s$ be zero-dimensional parametrizations,
  the sum of whose degrees being at most $\degQ$, with
  $\Zeroes(\scrQ_i)\subset \C^N$ for all $i$. There exists a
  probabilistic algorithm ${\sf Union}$ which takes as input
  $\scrQ_1, \ldots, \scrQ_s$ and returns either a zero-dimensional
  parametrization $\scrQ$ or ${\sf fail}$ using $\softO(N\degQ^2)$
  operations in $\QQ$. In case of success,
  $\Zeroes(\scrQ)=\Zeroes(\scrQ_1)\cup\cdots\cup \Zeroes(\scrQ_s)$.
\end{lemma}

The next algorithm takes as input a zero-dimensional parametrization
$\scrQ$ and a polynomial $G$. It returns a zero-dimensional
parametrization encoding $\Zeroes(\scrQ)\cap V(G)$.  We will actually
not use this algorithm as it is, but rather an extension of it with
coefficients in a product of fields; we give this simpler version
first as a starting point for the product of fields version.

\begin{lemma}\label{sec:basicroutinesparam:lemma:intersect}
  Let $\scrQ$ be a zero-dimensional parametrization of degree $\degQ$,
  with $\Zeroes(\scrQ)\subset \C^N$, and let
  $G \in \QQ[X_1, \ldots, X_N]$ a polynomial given by a straight-line
  program $\Gamma$ of length $E$.  There exists an algorithm
  ${\sf Intersect}$ which takes as input $\scrQ$ and $\Gamma$ and
  returns a zero-dimensional parametrization of
  $\Zeroes(\scrQ)\cap V(G)$ using $\softO((N+E)\degQ)$ operations in~$\QQ$.
\end{lemma}
\begin{proof}
  We are given an input parametrization $\scrQ$ consisting in
  polynomials $(q, v_1, \ldots, v_N)$ in $\QQ[T]$ and in a linear form
  $\pollambda$, and a straight-line program $\Gamma$ that computes a
  polynomial $G$. The output consists in polynomials
  $((r,w_1,\dots,w_N),\pollambda)$, with $r=\gcd(q,G(v_1,\dots,v_N))$ and
  $w_i = v_i \bmod r$ for all $i$.  To compute $r$, we rewrite it as
  $r=\gcd(q,G(v_1,\dots,v_N) \bmod q)$.  First, we compute
  $$G(v_1,\dots,v_N) \bmod q$$ by evaluating the straight-line program
  for $G$ at $v_1,\dots,v_N$, doing all operations modulo $q$; this
  takes $\softO(E\degQ)$ operations in $\QQ$. The subsequent GCD takes
  $\softO(\degQ)$ operations in $\QQ$, and the Euclidean divisions
  used to compute $w_1,\dots,w_N$ cost $\softO(N\degQ)$ operations in~$\QQ$.
\end{proof}

Finally, we deal with projections and their fibers.  Given a
zero-dimensional parametrization $\scrQ$ encoding $Q=\Zeroes(\scrQ)\subset
\C^N$ and an integer $e$, we now want to compute a zero-dimensional
parametrization encoding $\pi_e(Q)$. The following result is an
immediate consequence of \cite[Lemma~4]{PoSc13}.

\begin{lemma}\label{sec:posso:lemma:projection}
  Let $\scrQ$ be a zero-dimensional parametrization of degree $\degQ$,
  with $\Zeroes(\scrQ)\subset \C^N$. There exists a probabilistic
  algorithm ${\sf Projection}$ which takes as input $\scrQ$ and $e$
  and returns either a zero-dimensional parametrization $\scrQ'$ or
  ${\sf fail}$ using $\softO(N^2\degQ^2)$ operations in $\QQ$.  In
  case of success, $\Zeroes(\scrQ')=\pi_e(Q)$.
\end{lemma}

In the converse direction, algorithm ${\sf Lift}$ below takes as input
two zero-dimen\-sional parametrizations $\scrQ$ and $\scrR$ encoding
respectively $Q=\Zeroes(\scrQ)\subset \C^N$ and
$R=\Zeroes(\scrR)\subset \C^e$ with $e\le N$. It returns a
zero-dimensional parametrization of the fiber
$\fbr(Q,R)=Q \cap\pi_e^{-1}(R)$.
\begin{lemma}\label{lemma:subroutine:Lift}
  Let $\scrQ$ and $\scrR$ be zero-dimensional parametrizations of
  degrees at most $\degQ$ with $\Zeroes(\scrQ)\subset \C^N$,
  $\Zeroes(\scrR)\in \C^e$ and $e\leq N$. There exists a probabilistic
  algorithm ${\sf Lift}$ which takes as input $\scrQ$ and $\scrR$ and
  returns a zero-dimensional parametrization $\scrQ'$ using
  $\softO(N\degQ^2)$ operations in $\QQ$. In case of success,
  $\Zeroes(\scrQ')=\Zeroes(\scrQ)\cap\pi_e^{-1}(\Zeroes(\scrR))$.
\end{lemma}
\begin{proof}
  We let $\scrQ=((q, v_1, \ldots, v_N), \pollambda)$ and $\scrR=((r, w_1,
  \ldots, w_e), \nu)$ with $\pollambda=\pollambda_1 X_1+\cdots+\pollambda_N
  X_N$ and $\nu=\nu_1 X_1+\cdots+\nu_e X_e$.  We replace $\nu$ by a
  new random linear form, for a cost of $\softO(e\degQ^2)$,
  using~\cite[Lemma 6]{GiLeSa01}. Since $\nu$ is randomly chosen, we can
  assume that it separates the elements of $\Zeroes(\scrR) \cup
  \pi_e(\Zeroes(\scrQ))$, that is, that it takes pairwise different values
  on the points of that set.

  Let $s=\gcd(q,r(\nu_1 v_1 + \cdots + \nu_e v_e))$.  We claim that if
  $\roottau$ is a root of $q$, then $s(\roottau)=0$ if and only if the point
  $\x=(v_1(\roottau),\dots,v_N(\roottau))\in \Zeroes(\scrQ)$ satisfies
  $\pi_e(\x)\in \Zeroes(\scrR)$. Indeed, if $\pi_e(\x)$ is in
  $\Zeroes(\scrR)$, then
  $\sigma=\nu(\pi_e(\x))= \nu_1 v_1(\roottau) + \cdots + \nu_e v_e(\roottau)$
  is a root of $r$, and thus
  $r(\nu_1 v_1 + \cdots + \nu_e v_e)(\roottau)=0$. Conversely, suppose
  that $s(\roottau)=0$, so that
  $r(\nu_1 v_1 + \cdots + \nu_e v_e)(\roottau)=0$. In other words,
  $\nu_1 v_1(\roottau) + \cdots + \nu_e v_e(\roottau)=\nu(\pi_e(\x))$ is a
  root of $r$.  Write $\sigma=\nu(\pi_e(\x))$, and let
  $\y=(w_1(\sigma),\dots,w_e(\sigma))\in \Zeroes(\scrR)$.  By
  construction, $\nu(\y)=\sigma$, so $\nu(\y)=\nu(\pi_e(\x))$. By our
  assumption on $\nu$, this means that $\y=\pi_e(\x)$, so $\pi_e(\x)$
  is in $\Zeroes(\scrR)$, as claimed.

  We first compute $r(\nu_1 v_1 + \cdots + \nu_e v_e) \bmod q$, by
  evaluating it at $\nu_1 v_1 + \cdots + \nu_e v_e$ is
  $\softO(\degQ^2)$ operations.  Then, the previous discussion shows
  that it is enough to return $((s,t_1,\dots,t_N),\pollambda)$, where
  $t_i = v_i \bmod s$ for all $i$; these are computed using $\softO(N
  \degQ)$ operations.
\end{proof}

%%%%%%%%%%%%%%%%%%%%%%%%%%%%%%%%%%%%%%%%%%%%%%%%%%%%%%%%%%%%
%%%%%%%%%%%%%%%%%%%%%%%%%%%%%%%%%%%%%%%%%%%%%%%%%%%%%%%%%%%%
%%%%%%%%%%%%%%%%%%%%%%%%%%%%%%%%%%%%%%%%%%%%%%%%%%%%%%%%%%%%

\subsection{One-dimensional parametrizations}\label{sec:basicroutinesparam1}

Next, we discuss the one-dimensional analogue of the parametrizations
seen above. As above, let us first consider an arbitrary field $\KK$
of characteristic zero. A {\em
  one-dimensional parametrization}
$\scrQ=((q,v_1,\dots,v_N),\pollambda,\pollambda')$ with coefficients in
$\KK$ consists  in the following:
\begin{itemize}
\item polynomials $(q,v_1,\dots,v_N)$, such that $q\in \KK[U,T]$ is
  squarefree and monic in both $U$ and $T$, together with additional
  degree constraints explained below, and such that all $v_i$ are in
  $\KK[U,T]$ and satisfy $\deg(v_i,T) < \deg(q,T)$
\smallskip
\item linear forms $\pollambda,\pollambda'$ in $X_1,\dots,X_N$, such that
  $$\pollambda\left (v_1,\dots,v_N\right)=U \frac{\partial q}{\partial T}
  \bmod q \quad\text{and}\quad \pollambda'\left (v_1,\dots,v_N\right)=T
  \frac{\partial q}{\partial T} \bmod q.$$
\end{itemize}
This can thus be seen as a one-dimensional analogue of a
zero-dimen\-sional paramet\-rization. 

The corresponding algebraic set, denoted by
$\Zeroes(\scrQ)\subset \CC^N$, is now defined as the Zariski closure
of the locally closed set given by
$$q(\eta,\rootrho) = 0, \qquad \frac{\partial q}{\partial T}(\eta,\rootrho)
\ne 0, \qquad X_i = \frac{v_i(\eta, \rootrho)}{\frac{\partial q}{\partial
    T}(\eta, \rootrho)} \ \ (1 \le i \le N).$$
Remark that $\Zeroes(\scrQ)$ is one-equidimensional and that the
condition on $\pollambda$ and $\pollambda'$ means that the plane curve
$V(q)$ is the Zariski closure of the image of $\Zeroes(\scrQ)$ through
the projection $\x \mapsto (\pollambda'(\x),\pollambda(\x))$.

We define the {\em degree} $\degQ$ of $\scrQ$ as the degree of
$\Zeroes(\scrQ)$. Due to our assumption on $\pollambda$ and
$\pollambda'$, and using for instance~\cite[Theorem~1]{Schost03}, we
deduce that all polynomials $q,v_1,\dots,v_N$ have total degree at
most $\degQ$.

The additional degree constraint mentioned in the first item above is
that $q$ has degree {\em exactly} $\degQ$ in both $T$ and $U$ (so
under this assumption, we can simply read off $\degQ$ from $q$). This
constraint is actually very weak: because $\KK$ is infinite, {\em any}
algebraic curve in $\CC^N$ and defined over $\KK$ can be written as
$\Zeroes(\scrQ)$, for a suitable one-dimensional parametrization~$\scrQ$,
simply by choosing $\pollambda$ and $\pollambda'$ as random linear forms in
$X_1,\dots,X_N$ with coefficients in $\KK$~\cite{GiLeSa01}.

In the following paragraphs, we always take $\KK=\QQ$; we use $\KK=\C$
in the next sections. We describe a few elementary operations on
algebraic curves defined by such an encoding. As a preliminary remark,
note that if $\scrQ$ has degree $\degQ$, storing $\scrQ$ involves $O(N
\degQ^2)$ elements of $\QQ$, as each bivariate polynomial in $\scrQ$
has total degree at most $\degQ$.

\begin{lemma}\label{lemma:complexity1:changevar}
  Let $\scrQ$ be a one-dimensional parametrization of degree at most
  $\degQ$, with $\Zeroes(\scrQ)\subset \C^N$, and let $\mA$ be in
  $\GL(N,\QQ)$.  There exists an algorithm ${\sf ChangeVariables}$
  that takes as input $\scrQ$ and $\mA$ and returns a one-dimensional
  parametrization $\scrQ^\mA$ such that
  $\Zeroes(\scrQ^\mA)=\Zeroes(\scrQ)^\mA$ using
  $\softO(N^2\degQ^2+N^3)$ operations in $\QQ$.
\end{lemma}
\begin{proof}
  The proof is similar to that of Lemma~\ref{lemma:complexity0:changevar};
  it suffices to work on bivariate polynomials instead of univariate 
  ones, whence the extra cost.
\end{proof}

\begin{lemma}\label{sec:main:lemma:union1}
  Let $\scrQ$ and $\scrQ'$ be one-dimensional parametrizations, with
  $\Zeroes(\scrQ)$ and $\Zeroes(\scrQ')$ in $\C^N$ of respective degrees $\degQ$
  and $\degQ'$. There exists a probabilistic algorithm ${\sf Union}$
  which takes as input $\scrQ$ and $\scrQ'$ and returns either a
  one-dimensional parametrization $\scrQ''$ or ${\sf fail}$ using
  $\softO(N \max(\degQ,\degQ')^3)$ operations in~$\QQ$. In case of
  success, $\Zeroes(\scrQ'')=\Zeroes(\scrQ) \cup \Zeroes(\scrQ')$.
\end{lemma}
\begin{proof}
  First, we ensure that the pairs of linear forms associated to
  $\scrQ$ and $\scrQ'$ are the same; then, we use extended GCD
  techniques to combine them.
  
  For the first step, we pick two new random linear forms $\linearmu,\linearmu'$
  in $X_1,\dots,X_N$, and compute two new parametrizations $\scrS$ and
  $\scrS'$, both having $\linearmu$ and $\linearmu'$ as associated linear forms
  and such that $\Zeroes(\scrS)=\Zeroes(\scrQ)$ and $\Zeroes(\scrS')=\Zeroes(\scrQ')$.

  Suppose that the linear forms associated to $\scrQ$ are called
  $\pollambda$ and $\pollambda'$, and let us explain how to replace
  the second linear form $\pollambda'$ by $\linearmu'$ in $\scrQ$. We
  proceed as in Lemma~\ref{sec:main:lemma:union} (up to the harmless
  fact that the parametrizations of $X_1,\dots,X_N$ now take the form
  $X_i = {v_i}/{\frac{\partial q}{\partial T}}$)d, but working over
  the base field $\QQ(U)$. Using the results
  of~\cite[Lemma~2]{PoSc13}, this takes $\softO(N \degQ^2)$ operations
  in $\QQ(U)$. Letting $q$ denote the minimal polynomial of $\scrQ$,
  the fact that $\deg(q,U)=\deg(\Zeroes(\scrQ))$ implies that the
  projection $\Zeroes(\scrQ) \to \C$ given by $\x \mapsto
  \pollambda(\x)$ is finite; as a result, as in~\cite{GiLeSa01}, for a
  generic choice of $\linearmu'$, in the output of this step, all
  coefficients are in $\QQ[U]$.
  
  In order to keep the cost of computing with the extra variable $U$
  under control, we work using truncated power series in
  $\QQ[[U-u_0]]$ instead of rational functions. We choose randomly the
  point of expansion $u_0$ for our power series. For all choices of
  $u_0$, except finitely many of them, we can run the former algorithm
  with coefficients in $\QQ[[U-u_0]]$ and not encounter any division
  by a series with positive valuation (if we do, we return ${\sf
    fail}$). The degrees in $U$ of all coefficients in the output are
  at most $\degQ = \deg(q,U) = \deg(q,T)$, so is it enough to truncate
  all power series modulo $(U-u_0)^{\degQ+1}$. As a result, the total
  cost is $\softO(N \degQ^3)$ operations in $\QQ$, instead of
  $\softO(N \degQ^2)$ for the algorithm of
  Lemma~\ref{sec:main:lemma:union}.

  This process gives us a one-dimensional parametrization $\scrR$. We
  then proceed similarly to replace $\pollambda$ by $\linearmu$ in $\scrR$,
  obtaining a parametrization $\scrS$; this mainly amounts to
  exchanging the roles of $U$ and $T$, taking into account the
  particular form of denominator that appears in the
  parametrizations. We then follow the same steps with $\scrQ'$,
  obtaining a one-dimensional parametrization $\scrS'$, for a total of
  $\softO(N {\degQ'}^3)$ operations.

  In the second stage, we compute the union of $\Zeroes(\scrS)$ and
  $\Zeroes(\scrS')$. As above, we want to follow the algorithm given in
  Lemma~\ref{sec:main:lemma:union}, but with coefficients in
  $\QQ(U)$. We apply the same techniques of computations with
  truncated power series coefficients; this induces the same overhead
  $\softO(\max(\degQ,\degQ'))$ as it did in the previous paragraphs,
  so the cost is again $\softO(N \max(\degQ,\degQ')^3)$ operations in $\QQ$.
\end{proof}

Next, we deal with projections and their fibers.  Given a
one-dimensional parametrization $\scrQ$ encoding $V=\Zeroes(\scrQ)\subset
\C^N$ and an integer $e \le N$, we may want to compute a
one-dimensional parametrization encoding the Zariski closure of
$\pi_e(V)$. Remark however that $\pi_e(V)$ may not be purely
one-dimensional: some irreducible components of $V$ may project onto
isolated points (with thus infinite fibers). These points will not be
part of the output; only the one-dimensional component will be.

\begin{lemma}\label{sec:posso:lemma:projection1}
  Let $\scrQ$ be a one-dimensional parametrization of degree at most
  $\degQ$, with $V=\Zeroes(\scrQ)\subset \C^N$, and let $e$ be in
  $\{2,\dots,N\}$. There exists a probabilistic algorithm
  ${\sf Projection}$ which takes as input $\scrQ$ and $e$ and returns
  either a one-dimensional parametrization $\scrQ'$ or ${\sf fail}$
  using $\softO(N^2 \degQ^3)$ operations in $\QQ$.  In case of
  success, $\Zeroes(\scrQ')$ is the one-dimensional component of $\pi_e(V)$.
\end{lemma}
\begin{proof}
  We start from $\scrQ=((q,v_1,\dots,v_N),\pollambda,\pollambda')$, and we
  first apply an algorithm similar to that of
  Lemma~\ref{sec:posso:lemma:projection}, with polynomials in
  $\QQ(U)[T]$ instead of $\QQ[T]$. This computes polynomials
  $(r,w_1,\dots,w_e)$ and linear forms $\pollambda$ (given as input) and
  $\linearmu'$, where the latter depends only on $X_1,\dots,X_e$.  As in
  Lemma~\ref{sec:main:lemma:union1}, we circumvent the problem of
  computing with rational functions by working with power series in
  $U-u_0$, for a randomly chosen $u_0$; we need power series of
  precision $O(\degQ)$, so the total cost increases to
  $\softO(N^2\degQ^3)$. This part of the algorithm may return
  ${\sf fail}$ (if we attempt a division by a power series of positive
  valuation); otherwise, it returns a one-dimensional parametrization.

  At this stage, we have replaced $\pollambda'$ by a new linear form $\linearmu'$,
  that depends only on $X_1,\dots,X_e$. This does not give a
  one-dimensional parametrization of $\pi_e(V)$ yet, since $\pollambda$
  still involves all variables. As a second step, we follow the same 
  routine, working this time in $\QQ(T)[U]$. The cost is again 
  $\softO(N^2\degQ^3)$.
\end{proof}

The final operation is somewhat similar to algorithm ${\sf Discard}$ 
introduced for zero-dimen\-sional parametrizations, with a slight twist:
given a one-dimensional parametrization $\scrQ$ that defines a curve
$V=\Zeroes(\scrQ) \subset \C^N$, and given points $S$ in $\C^e$, for some 
$e \le N$, we want to compute a parametrization for the 
Zariski closure of $V-\pi_e^{-1}(S)$.

\begin{lemma}\label{sec:basicroutinesparam:lemma:discarddim1}
  Let $\scrQ$ be a one-dimensional parametrization of degree at most
  $\degQ$, with $\Zeroes(\scrQ)\subset \C^N$, and let $\scrR$ be a
  zero-dimensional parametrization of degree at most $\degQ'$, with
  $\Zeroes(\scrR) \subset \C^e$. There exists a probabilistic algorithm
  ${\sf Discard}$ which takes as input $\scrQ$ and $\scrR$ and returns
  either a one-dimensional parame\-trization $\scrQ'$ or ${\sf fail}$
  using $\softO(N \degQ \max(\degQ,\degQ')^2)$ operations in $\QQ$. In
  case of success, $\Zeroes(\scrQ')$ is the Zariski closure of
  $\Zeroes(\scrQ)-\pi_e^{-1}(\Zeroes(\scrR))$. 
\end{lemma}
\begin{proof}
  Let us write $\scrQ=((q,v_1,\dots,v_N),\pollambda,\pollambda')$ and
  $\scrR=((r, w_1, \ldots, w_e), \nu)$, with all polynomials in
  $\scrQ$ in $\QQ[U,T]$ and all polynomials in $\scrR$ in
  $\QQ[X]$. The parametrization we are looking for has the form
  $\scrQ'=((q',v_1',\dots,v_N'),\pollambda,\pollambda')$, for some factor
  $q'$ of $q$, and with $v'_i=v_i \bmod q'$ for all $i$.

  Suppose without loss of generality that $q$ has positive degree in
  $T$ (if $q=1$, there is nothing to do; if $q$ is in $\QQ[U]$,
  exchange $T$ and $U$). Then, we obtain the result by running the
  zero-dimensional algorithms ${\sf Lift}$ from
  Lemma~\ref{lemma:subroutine:Lift} and ${\sf Discard}$ from
  Lemma~\ref{sec:basicroutinesparam:lemma:discard}, with input $\scrQ$
  and $\scrR$; the coefficients should be taken in $\QQ(U)$, but as
  above, we use power series in $U$ of precision $O(\degQ)$.  The cost
  estimate follows from the results in these two lemmas, up to an
  $\softO(\degQ)$ overhead due to the fact that we work with power
  series of precision $O(\degQ)$.
\end{proof}

%%%%%%%%%%%%%%%%%%%%%%%%%%%%%%%%%%%%%%%%%%%%%%%%%%%%%%%%%%%%
%%%%%%%%%%%%%%%%%%%%%%%%%%%%%%%%%%%%%%%%%%%%%%%%%%%%%%%%%%%%
%%%%%%%%%%%%%%%%%%%%%%%%%%%%%%%%%%%%%%%%%%%%%%%%%%%%%%%%%%%%

\subsection{Working over a product of fields: basic operations} \label{sec:prodfields}\label{ssec:prodfieldsbasic}

In the next subsections, we will deal with zero-dimensional and
one-dimensional paramet\-rizations with coefficients in a {\em product of
fields} instead of $\QQ$; these will be well suited to handle algebraic
sets lying over a given finite set $Q$. In this paragraph, we review
de\-finitions and describe several basic operations for polynomials
over a product of fields.

Let $q$ be a monic, squarefree polynomial of degree $\degQ$ in $\QQ[T]$
and define $\A=\QQ[T]/\langle q \rangle$. Because we do not assume
that $q$ is irreducible, $\A$ may not be a field; it is the product of
the fields $\A_1=\QQ[T]/\langle c_1\rangle, \ldots,
\A_\ell=\QQ[T]/\langle c_\ell\rangle$, where $c_1,\dots,c_\ell$ are
the irreducible factors of $q$. 

We describe here how complexity results for basic computations over
$\QQ$ can be extended to computations over $\A$. If $q$ were
irreducible, it would be straightforward to deduce that working in
$\A$ induces an overhead of the form $\softO(\degQ)$. For a general $q$,
one workaround would be to factor it into irreducibles and work modulo
all factors independently; however, we do not allow the use of
factorization algorithms in $\QQ[T]$: they may not be available over
$\QQ$, or too costly. The results below show that for many questions,
we will be able to bypass factorization algorithms and pay roughly the
same overhead   $\softO(\degQ)$ as if $q$ were irreducible.

Regardless of the factorization of $q$, addition, subtraction and
multiplication in $\A$ can be done in $\softO(\degQ)$ operations in
$\QQ$. Similarly, addition, subtraction and multiplication of
polynomials of degree $D$ in $\A[X]$ can be done within $\softO(D \degQ
)$ operations in~$\QQ$.

However, because $\A$ may not be a field, some notions need to be
adapted. The first obvious remark is that a non-zero element $u$
in $\A$ may not be invertible; however, we can test whether $u$
is a unit in $\A$, and if so compute its inverse, using
$\softO(\degQ)$ operations in $\QQ$, by means of an extended GCD
computation in $\QQ[T]$ between $q$ and the canonical lift of $u$
to $\QQ[T]$. In Lemma~\ref{lemma:inter}, we will need the following
straightforward extension of this result to inversion in extension
rings of $\A$ (the degrees we use here are those that will be needed
when we apply this result).

\begin{lemma}\label{geosolve:lemma:euclideinrings1}
  Let $F,G$ be polynomials in $\A[Y,X]$, with degree at most $\delta$
  in $X$ and $Y$ and with $F$ monic in $X$. Suppose that for any root
  $\roottau$ of $q$ in $\C$, the polynomials $F(\roottau,Y,X)$ and
  $G(\roottau,Y,X)$ are coprime in $\C(Y)[X]$. Then, for all $u \in
  \QQ$ except a finite number, and for any integer $D$, $G$ is
  invertible in $\A[Y,X]/\langle (Y-u)^{\delta D},F\rangle$ and
  one can compute its inverse using $\softO(D \degQ \delta^2)$
  operations in $\QQ$.
\end{lemma}
\begin{proof}
  Our assumption implies that for any root $\roottau$ of $q$,
  the polynomial $G(\roottau,Y,X)$ is invertible in $\C[Y,X]/\langle
  (Y-u),F(\roottau,Y,X)\rangle$ for all values of $u$ except for
  a finite number. Taking all roots of $q$ into account, we deduce
  that, except for a finite number of values of $u$, $G$ is
  invertible in $\A[Y,X]/\langle (Y-u),F(Y,X)\rangle$; when it
  is, Proposition~6 in~\cite{DaJiMoSc08} shows that its inverse can be
  computed in $\softO(\degQ \delta)$ operations in $\QQ$. Using Newton
  iteration modulo the powers of
  $(Y-u)$~\cite[Chapter~9]{GaGe99}, the claim of the lemma
  follows.
\end{proof}

The notion of greatest common divisor (GCD) in $\A[X]$ requires a more
significant adaptation: we require GCD's to be monic; as a result, we
may have to {\em split} $q$ into factors and output several
polynomials that will play the role of GCD's modulo the factors of $q$.
Explicitly, if $F,G$ are in $\A[X]$, a GCD of $(F,G)$ consists in
pairs $(q_1,H_1),\dots,(q_r,H_r)$, with $q_i$ monic in $\QQ[T]$ and
$H_i$ monic in $\QQ[T]/\langle q_i\rangle[X]$, such that $q=q_1 \cdots
q_r$ and such that the ideals $\langle q_i, H_i\rangle$ and $\langle
q_i, F, G\rangle$ coincide for all $i$. Note that $q_1,\dots,q_r$ are
not necessarily irreducible, so that such a GCD may not be unique.

To compute a GCD as above, we run the fast extended GCD algorithm in
$\A[X]$, as if $\A$ were a field, but using dynamic evaluation
techniques~\cite{D5}: if we are led to attempt to invert a
zero-divisor in $\A$, knowing this zero-divisor allows us to split $q$
into two factors; we can then continue with further computations in
two branches independently. These ideas were studied from the
complexity viewpoint in~\cite{ADCM03,DaMMMScXi06}, leading to the
following result.

\begin{lemma}\label{geosolve:lemma:GCD}
  Let $F,G$ be in $\A[X]$ of degree at most $\delta$. Then, one can
  compute a GCD $(q_1,H_1),\dots,(q_r,H_r)$ of $F$ and $G$ using
  $\softO(\degQ \delta)$ operations in $\QQ$.
\end{lemma}

As an application, we discuss how to define and compute a squarefree
part of a polynomial $F$ in $\A[X]$. As above, we impose the output to
be monic. Then, a {\em squarefree part} of such an $F$ consists in
pairs $(q_1,H_1),\dots,(q_r,H_r)$, such that $q=q_1\cdots q_r$ and for
all $i$, $H_i$ is monic in $\QQ[T]/\langle q_i\rangle[X]$, and the
ideal $\langle q_i,H_i\rangle$ is the radical of the ideal $\langle
q_i,F\rangle$ in $\QQ[T,X]$; as for GCD's, this squarefree part is not
uniquely defined.  Using the GCD algorithm above, we deduce easily the
following cost estimate for squarefree part computation.

\begin{lemma}\label{geosolve:lemma:SQF}
  Let $F$ be in $\A[X]$ of degree at most $\delta$. Then, one can
  compute a squarefree part $(q_1,H_1),\dots,(q_r,H_r)$ of $F$ using
  $\softO(\degQ \delta)$ operations in $\QQ$.
\end{lemma}

In a similar vein, we will say that $F\in\A[X]$ is {\em squarefree} if
the ideal $\langle q,F\rangle$ is radical. This definition will carry
over to multivariate polynomials $F$ with coefficients in $\A$ (we
will need $F$ bivariate, at most).

Finally, we discuss the computation of resultants. For this question,
there will be no splitting involved in the output, since the resultant
can be defined over any ring. However, in the algorithm of
Subsection~\ref{sec:possobasepoints}, we will need further a rather
complex setup: we compute resultants of polynomials, not over $\A$,
but over a power series ring over $\A$. Explicitly, we work over the
ring
$$\B=\A[t,t_1,\dots,t_N,U]/\langle
(t,t_1,\dots,t_N)^2,(U-u_0)^{D\delta+1} \rangle,$$ for some new
variables $t,t_1,\dots,t_N,U$ and $u_0 \in \QQ$ and integers
$D,\delta$; remark that storing an element of $\B$ uses $O(\degQ N D
\delta)$ elements of $\QQ$.  Remark as well that $\B$ is the product
of the rings $\B_\roottau$, for $\roottau$ a root of $q$, with
$$\B_\roottau=\C[t,t_1,\dots,t_N,U,T]/\langle
(t,t_1,\dots,t_N)^2,(U-u_0)^{D\delta+1},(T-\roottau) \rangle.$$ For a
polynomial $F$ in $\B[X]$ and a root $\roottau$ of $q$, we denote by
$F_\roottau$ the image of $F$ in $\B_\roottau[X]$ obtained by evaluating $T$
at $\roottau$. Finally, in the following lemma, we use {\em subresultants}
of two polynomials, for which we use the definition
of~\cite[Chapter~6]{GaGe99} (these are elements of $\B$; they are
sometimes called {\em principal} subresultants).

\begin{lemma}\label{geosolve:lemma:euclideinrings2}
  Let $F,G$ be in $\B[X]$ with $F$ monic of degree $\delta$ and
  $\deg(G) < \delta$. Suppose that for every root $\roottau$ of $q$, every
  non-zero subresultant of $F_\roottau$ and $G_\roottau$ is a unit in
  $\B_\roottau$. Then, one can compute the resultant of $F$ and $G$ using
  $\softO(N D  \degQ \delta^2)$ operations in $\QQ$.
\end{lemma}
\begin{proof}
  As a preliminary, remark that additions and multiplications in $\B$
  can be done using $\softO(N D \degQ \delta)$ operations in $\QQ$
  (power series arithmetic in $N+1$ variables induces an extra $O(N)$
  factor; computations modulo $(U-u_0)^{D\delta+1}$ induce an
  additional $\softO(D \delta)$). Inversions (when feasible) could be
  done for a similar cost, but we will not use this fact directly.

  One can compute the resultant of polynomials with coefficients in a
  field in quasi-linear time using the fast resultant algorithm
  of~\cite[Chapter~11]{GaGe99}. For more general coefficient rings,
  this may not be the case anymore, but workarounds exist in some
  cases. 

  Precisely, we will use the fact that the former algorithm can still
  be applied to polynomials over any ring, provided all the non-zero
  subresultants of the input polynomials are units. Indeed, when it is
  the case, Theorem~11.13 in~\cite{GaGe99} implies that all remainders
  in the Euclidean remainder sequence have invertible leading
  coefficients, so this sequence is well-defined (the proof uses a
  formula established over a field in Lemma~11.12 of that reference,
  which actually holds over any ring); the fast resultant algorithm
  can then be executed.
  
  When the base ring is a product of fields such as $\A$, we can
  always reduce to such a situation through splittings. This may not
  be enough for us in general (as $\B$ is not a product of fields),
  but under the assumptions of the lemma, we will see that we can ensure 
  such a property.
  
  Consider first the polynomials $F_0$ and $G_0$ lying in $\A[X]$
  obtained by evaluating $U$ at $u_0$ and $t,t_1,\dots,t_N$ at zero
  in $F$ and $G$. As said above, one can compute the resultant of such
  polynomials by adapting the resultant algorithm
  of~\cite[Chapter~11]{GaGe99} to work over $\A$, similarly to the
  adaptation of the fast GCD algorithm used in
  Lemma~\ref{geosolve:lemma:GCD}.  As in
  Lemma~\ref{geosolve:lemma:GCD}, the total time of this step is
  $\softO(\degQ \delta)$ operations in~$\QQ$.

  Splittings may occur, yielding a result lying in a product of the
  form $\A_1 \times \cdots \times \A_s$, with $\A_i$ of the form
  $\A_i=\QQ[T]/\langle q_i \rangle$ for all $i$ and with $q= q_1\cdots
  q_s$. Due to these splittings, modulo each $q_i$, the whole Euclidean
  remainder sequence is well-defined (that is, all remainders have
  invertible leading terms); by means again of the formulas
  in~\cite[Theorem~11.13]{GaGe99}, we deduce that all non-zero
  subresultants of $F_0 \bmod q_i$ and $G_0 \bmod q_i$ are invertible
  in $\A_i$.

  For $i$ in $\{1,\dots,s\}$, we are going to compute the resultant
  $R_i$ of $F_i$ and $G_i$ in $\B_i[X]$, where
  $$\B_i=\A_i[t,t_1,\dots,t_N,U]/\langle (t,t_1,\dots,t_N)^2,
  (U-u_0)^{D \delta+1} \rangle$$ and where $(F_i,G_i)$ are the
  images of $(F,G)$ modulo $q_i$ (computing these remainders takes
  $\softO( N D \degQ \delta^2)$ operations in $\QQ$ by fast simultaneous
  modular reduction~\cite[Chapter~10]{GaGe99}).  The last operation
  will then be to apply the Chinese Remainder theorem, in order to
  recover a result in $\B$, rather than in the product of the
  $\B_i$'s. The cost of that step will be $\softO( N D \degQ \delta)$.

  Thus, we can focus on the computation of a single resultant $R_i$.
  Fixing an index $i$ in $\{1,\dots,s\}$, we claim that we can follow
  the same subresultant algorithm, but with coefficients now in
  $\B_i$, and that all non-zero subresultants of  $F_i$ and $G_i$
  are units in $\B_i$: this is proved in the last two paragraphs. If this is
  the case, then the running time will be $\softO(\delta)$ times the
  cost of arithmetic operations $(+,\times,\div)$ in $\B_{i}$, which
  is $\softO( N D \degQ_i \delta)$, with $\degQ_i=\deg(q_i)$. The total
  is $\softO( N D \degQ_i \delta^2)$ per index $i$, for a grand total
  of $\softO( N D \degQ \delta^2)$; this will prove our claim on the 
cost of the calculation.

  Let $F_{i,0}$ and $G_{i,0}$ be the polynomials in $\A_i[X]$ obtained
  by evaluating $U$ at $u_0$ and $t,t_1,\dots,t_N$ at zero in $F_i$
  and $G_i$, or equivalently by reducing $F_{0}$ and $G_{0}$ modulo
  $q_i$. Recall that we pointed out earlier that all the non-zero
  subresultants of $F_{i,0}$ and $G_{i,0}$ are units in $\A_i$.
  
  Let $\sigma\in \B_{i}$ be one of the non-zero subresultants of $F_i$
  and $G_i$, say $\sigma=\det(S_k(F_i,G_i))$ for some index $k\le
  \deg(G_i)$ using the notation of~\cite[Chapter~6]{GaGe99}; we have
  to prove that $\sigma$ is a unit in $\B_i$.  Because $\sigma$ is
  non-zero, there exist a root $\roottau$ of $q_i$ such that $\sigma(\roottau)
  \in \B_\roottau$ is non-zero, with $\B_\roottau$ as defined above this
  lemma. But $\sigma(\roottau)$ is then a non-zero subresultant of $F_\roottau$
  and $G_\roottau$ (since $F$ is monic). By assumption, this implies that
  $\sigma(\roottau)$ is a unit in $\B_\roottau$. In particular, we obtain that
  the image of $\sigma(\roottau)$ is non-zero in $\B_\roottau/\langle
  t,t_1,\dots,t_N,U-u_0\rangle$, which implies that the image of
  $\sigma$ itself is non-zero in $\B_i/\langle
  t,t_1,\dots,t_N,U-u_0\rangle=\A_i$.  But, because $F$ is monic,
  $\sigma \bmod \langle t,t_1,\dots,t_N,U-u_0\rangle \in \A_i$ is a
  subresultant of $F_{i,0}$ and $G_{i,0}$, so the remark in the
  previous paragraph implies that it is a unit in $\A_i$. Thus, by
  Hensel's lemma, we deduce that $\sigma$ is a unit in $\B_{i}$.
\end{proof}

%%%%%%%%%%%%%%%%%%%%%%%%%%%%%%%%%%%%%%%%%%%%%%%%%%%%%%%%%%%%
%%%%%%%%%%%%%%%%%%%%%%%%%%%%%%%%%%%%%%%%%%%%%%%%%%%%%%%%%%%%
%%%%%%%%%%%%%%%%%%%%%%%%%%%%%%%%%%%%%%%%%%%%%%%%%%%%%%%%%%%%

\subsection{Equations over a product of fields}\label{sec:param}

In this paragraph, we show how one can make sense of systems of
equations with coefficients in a product of fields, and we explain how
the notions of parametrizations seen before can be extended to include
the case of coefficients in a product of field. The last subsection
shows how to use these data structures to design an intersection
algorithm that will be central to our general polynomial system
solving algorithm.

{\em In all this section, $q$ is a monic squarefree polynomial in
  $\QQ[T]$, and we define the product of fields $\A=\QQ[T]/\langle q
  \rangle$. We let $\degQ$ denote the degree of $q$.}

%%%%%%%%%%%%%%%%%%%%%%%%%%%%%%%%%%%%%%%%%%%%%%%%%%%%%%%%%%%%

\subsubsection{Systems of equations}\label{sec:D5eqs}

Consider polynomials $\F=(F_1,\dots,F_s)$ in the ring
$\A[X_{e+1},\dots,X_N]$ (the choice of indices in the variables will
turn out to be natural in our applications below). To a root $\roottau$ of
$q$ in $\C$, we associate the evaluation mapping $\phi_\roottau: \A \to
\C$, naturally defined as $\phi_\roottau(f)=f(\roottau)$; this mapping carries
over to polynomial rings over $\A$.

We can then define the polynomials $\F_\roottau=(\phi_\roottau(F_i))_{1 \le i
  \le s}$, so that each $\F_\roottau$ is a vector of $s$ polynomials in
$\C[X_{e+1},\dots,X_N]$. Finally, to our system $\F$, we can then
associate the algebraic sets $(V_\roottau)_{q(\roottau)=0}$, where each
$V_\roottau=V(\F_\roottau)$ lies in $\C^{N-e}$.

A prominent example of this situation is when we are given a whole
zero-dimensional parametrization
$\scrQ=((q,v_1,\dots,v_e),\pollambda)$, together with polynomials $\f$
in $\QQ[X_1,\dots,X_N]$. We can then define the
polynomials $$\F=\f(v_1,\dots,v_e,X_{e+1},\dots,X_N) \bmod q$$ which
lie in $\A[X_{e+1},\dots,X_N]$, and the associated algebraic sets
$(V_\roottau=V(\F_\roottau))_{q(\roottau)=0}$. On the other hand, defining as
usual $Q=\Zeroes(\scrQ) \subset \C^e$, the zero-set
$$V=\fbr(V(\f),Q) \subset \C^N$$ can be decomposed as the disjoint
union of the sets $V_\x$, for $\x$ in $Q$. For any such
$\x=(x_1,\dots,x_e)$, $\roottau=\pollambda(\x)$ is a root of $q$, such
that $x_i=v_i(\roottau)$ for $i=1,\dots,e$, and one verifies that $V_\x$
can be rewritten as $(x_1,\dots,x_e) \times V_\roottau$, for
$V_\roottau \subset \C^{N-e}$ as defined above.

In the same context, we may as well be interested in the set
$V'=\freg(\f,Q)$, which was defined in
Subsection~\ref{chap:prelim:sec:definitions} as the Zariski closure of the
set of all points in $\fbr(V(\f),Q)$ where $\jac(\f,e)$ has full rank.
Then, $V'$ is the disjoint union of the sets $V'_\x$, for
$\x=(x_1,\dots,x_e)$ in $Q$, with $V'_\x$ of the form
$V'_\x=(x_1,\dots,x_e) \times V'_\roottau$, where $\roottau=\pollambda(\x)$ is
the root of $q$ corresponding to $\x$ and $V'_\roottau$ is defined as
$V'_\roottau=\freg(\F_\roottau)$.

In terms of data structures, we will often assume that polynomials
$\F$ are given by means of a straight-line program, say $\Gamma$. In
this context of computations over $\A$, we will assume that $\Gamma$
{\em has coefficients in $\A$}: this means that $\Gamma$ has input
variables $X_{e+1},\dots,X_N$, operations $+,-,\times$ and uses
constants from $\A$ instead of $\QQ$. As before, the length of $\Gamma$
is the number of operations it performs.

%%%%%%%%%%%%%%%%%%%%%%%%%%%%%%%%%%%%%%%%%%%%%%%%%%%%%%%%%%%%

\subsubsection{Dimension zero}\label{sec:D5param0}

Let $q$ and $\A$ be as above. A {\em zero-dimensional paramet\-rization}
$\scrR=((r,w_{e+1},\dots,w_N),\linearmu)$ with coefficients in $\A$ consists
in polynomials $(r,w_{e+1},\dots,w_N)$ such that $r\in \A[X]$ is monic
and squarefree (in the sense of Subsection~\ref{ssec:prodfieldsbasic})
and all $w_i$ are in $\A[X]$ and satisfy $\deg(w_i) < \deg(r)$, and in
a linear form $\linearmu$ in $X_{e+1},\dots,X_N$ with coefficients in $\QQ$,
such that $\linearmu(w_{e+1},\dots,w_N)=X$. The {\em degree} of $\scrR$ is
defined as that of $r$. 

For any root $\roottau$ of $q$, we can then define $\scrR_\roottau$ as the
zero-dimensional parametrization with coefficients in $\C$, obtained
by applying the evaluation map $\phi_\roottau$ defined above to the
coefficients of all polynomials in $\scrR$. The algebraic sets
associated to $\scrR$ are then naturally defined as the family
$(\Zeroes(\scrR_\roottau))_{q(\roottau)=0}$, where each $\Zeroes(\scrR_\roottau)$ is a subset
of $\C^{N-e}$.

\begin{lemma}\label{lemma:descent0}
  Let $q$ and $\scrR$ be as above, let $\degQ$ be the degree of $q$
  and $\degR$ be the degree of $\scrR$. There exists a probabilistic
  algorithm ${\sf Descent}$ which takes as input $q$ and $\scrR$ and
  returns either a zero-dimensional parametrization $\scrR'$ with
  coefficients in $\QQ$ or ${\sf fail}$ using
  $\softO(N \degQ^2 \degR^2)$ operations in $\QQ$. In case of success,
  $\Zeroes(\scrR')=\cup_{q(\roottau)=0} \Zeroes(\scrR_\roottau)$
 in $\C^{N-e}$.
\end{lemma}
\begin{proof}
  First, we replace $\linearmu$ by a new random linear form, say
  $\linearmu'=\linearmu'_1 X_{e+1}+\cdots+\linearmu'_N X_N$; this is done using the
  algorithm of~\cite[Lemma~2]{PoSc13} with coefficients in $\A$.  The
  algorithm involves only operations $(+,\times)$, except for a
  squarefreeness test; in our case, this test is done using
  Lemma~\ref{geosolve:lemma:SQF} (if the output is false, we return
  ${\sf fail}$). Altogether, the cost of this first step is
  $\softO(N\degQ\degR^2)$ operations in $\QQ$. Call
  $((r',v'_{e+1},\dots,v'_N),\linearmu')$ the resulting parametrization with
  coefficients in $\A$.

  Then, we compute the minimal polynomial of $\scrR'$ by applying the
  bivariate change-of-order algorithm of~\cite{PaSc06} to $q$ and
  $r'$, this time with coefficients in $\QQ$; this takes
  $\softO(\degQ^2 \degR^2)$ operations in $\QQ$ (choosing $\linearmu'$ random 
  ensures that the output polynomial is indeed squarefree).
Computing the
  parametrizations that describe the values of $X_{e+1},\dots,X_N$ is
  then done by modular compositions on the polynomials
  $v'_{e+1},\dots,v'_N$, as in~\cite{PoSc13}, in time $\softO(N
  \degQ^2 \degR^2)$.
\end{proof}

Often, we will actually know more than $q$: we will be given a
zero-dimensional paramet\-rization $\scrQ=((q,v_1,\dots,v_e),\pollambda)$
with coefficients in $\QQ$. In this case, we can define
$\Zeroes(\scrQ,\scrR)$ as the finite set defined by
$$q(\roottau)=0,\quad r(\roottau,\rootrho)=0,\quad X_i=v_i(\roottau)\ \ (1 \le i \le e), \quad X_i = w_i(\roottau,\rootrho)\ \ (e+1 \le i \le N).$$
In other words, $\Zeroes(\scrQ,\scrR)$ is the disjoint union of the
finite sets $(v_1(\roottau),\dots,v_e(\roottau))\times \Zeroes(\scrR_\roottau)$,
for $\roottau$ a root of $q$. In this situation, we can deduce a
zero-dimensional parametrization with coefficients in $\QQ$ for this
set.

\begin{lemma}\label{lemma:descent0-2}
  Let $\scrQ$ and $\scrR$ be as above, let $\degQ$ be the degree of
  $\scrQ$ and $\degR$ the degree of $\scrR$. There exists a
  probabilistic algorithm ${\sf Descent}$ which takes as input $\scrQ$
  and $\scrR$ and returns either a zero-dimensional parametrization
  $\scrR'$ with coefficients in $\QQ$ or ${\sf fail}$ using
  $\softO(N \degQ^2 \degR^2)$ operations in $\QQ$. In case of success,
  $\Zeroes(\scrR')=\Zeroes(\scrQ,\scrR)$.
\end{lemma}
\begin{proof}
  The algorithm is entirely similar to that of
  Lemma~\ref{lemma:descent0}, except that in the last stage, we also
  apply modular compositions to the polynomials $v_1,\dots,v_e$ in
  order to obtain a description of the values of $X_1,\dots,X_e$. The
  overall analysis does not change.
\end{proof}

Not {\em any} family of finite algebraic sets $(V_\roottau)_{q(\roottau)=0}$,
with $V_\roottau \subset \C^{N-e}$ for all $\roottau$, may be described as
$V_\roottau=\Zeroes(\scrR_\roottau)$, for some zero-dimensional parametrization
$\scrR$ with coefficients in $\A$. For instance, since we require that
$r$ be monic and squarefree in $\A[X]$, all $V_\roottau$'s must have
the same cardinality.

Thus, to represent a family of finite algebraic sets
$(V_\roottau)_{q(\roottau)=0}$, with $V_\roottau \subset \C^{N-e}$ for all $\roottau$,
we will use a sequence of pairs $(q_1,\scrR_1),\dots,(q_s,\scrR_s)$
with, for all $i$, $q_i$ monic in $\QQ[T]$ and $\scrR_i$ a
zero-dimensional parametrization with coefficients in
$\A_i=\QQ[T]/\langle q_i \rangle$, and with $q = q_1 \cdots q_s$, such
that the following holds. For any root $\roottau$ of $q$, there exists a
unique $i$ in $\{1,\dots,s\}$ such that $q_i(\roottau)=0$. Then
$\scrR_{i,\roottau}$ is well-defined, and we require that $V_\roottau =
\Zeroes(\scrR_{i,\roottau})$. We will call $(q_1,\scrR_1),\dots,(q_s,\scrR_s)$
{\em zero-dimensional parametrizations over $\A$ for
  $(V_\roottau)_{q(\roottau)=0}$}.

Even then, not every family of algebraic sets $(V_\roottau)_{q(\roottau)=0}$
can be represented by zero-dimensional parametrizations over $\A$,
since the fields of definitions of the various sets $V_\roottau$ also
matter. There is however one class of examples where we can assert it
will be the case, and which encompasses all examples we will see
below: take two families of polynomials $\F$ and $\G$ in
$\A[X_{e+1},\dots,X_N]$ and, for any root $\roottau$ of $q$, define
$V_\roottau\subset \C^{N-e}$ as the set of isolated points of the Zariski
closure of $V(\F_\roottau)-V(\G_\roottau)$. We claim that in this situation,
there do exist zero-dimensional parametrizations over $\A$ for
$(V_\roottau)_{q(\roottau)=0}$: simply take $q_1,\dots,q_s$ as the irreducible
factors of $q$, and let $\scrR_i$ be the zero-dimensional
parametrizations for the ideal that defines the isolated points of the
Zariski closure of $V(\F)-V(\G)$ over the fraction field of
$\QQ[T]/\langle q_i \rangle$. Of course, the algorithms below will
avoid factoring $q$ into irreducibles.

We continue with some algorithms to perform elementary set-theoretic
operations on sets $(V_\roottau)_{q(\roottau)=0}$ using such a
representation. First, we give a cost estimate for applying a linear
change of variables.

\begin{lemma}\label{lemma:complexity0POF:changevar}
  Let $(q_1,\scrR_1),\dots,(q_s,\scrR_s)$ be zero-dimensional
  parametrizations over $\A$ that define algebraic sets
  $(V_\roottau)_{q(\roottau)=0}$, let $\degQ$ be the degree of $\scrQ$ and
  $\degR$ be the maximum of the degrees of $\scrR_1,\dots,\scrR_s$,
  and let $\mA$ be in $\GL(N-e,\QQ)$.

  There exists an algorithm ${\sf ChangeVariables}$ which takes as
  input
  $$(q_1,\scrR_1),\dots,(q_s,\scrR_s)$$ and $\mA$ and returns
  zero-dimensional parametrizations
  $(q_1,\scrR^\mA_1),\dots,(q_s,\scrR^\mA_s)$ over $\A$ that define
  the algebraic sets $(V_\roottau^\mA)_{q(\roottau)=0}$ using
  $\softO(N^2\degQ \degR+N^3)$ operations in $\QQ$.
\end{lemma}
\begin{proof}
  For $i=1,\dots,s$, we can apply Algorithm ${\sf ChangeVariables}$
  from Lemma \ref{lemma:complexity0:changevar} with coefficients in
  $\A_i=\QQ[T]/\langle q_i\rangle$, since this algorithm only involves
  operations $(+,\times)$ in $\A_i$ and inversions in $\QQ$.  The cost
  is thus $\softO(N^2 \degR+N^3)$ operations in $\A_i$, which is
  $\softO(N^2\degQ_i \degR+N^3)$ operations in $\QQ$, and the
  conclusion of the lemma follows by summing over all $i$.
\end{proof}

As announced prior to
Lemma~\ref{sec:basicroutinesparam:lemma:intersect}, we will also need
below an algorithm to intersect finite algebraic sets of the form
$(V_\roottau)_{q(\roottau)=0}$ with a hypersurface. We assume that the
algebraic sets $(V_\roottau)_{q(\roottau)=0}$ are represented by means of
zero-dimensional parametrizations $(q_1,\scrR_1),\dots,(q_s,\scrR_s)$
over $\A$, and that the hypersurface is defined by a polynomial $G$ in
$\A[X_{e+1},\dots,X_N]$.  As done before, we will assume that $G$ is
given by a straight-line program $\Gamma$ with coefficients in $\A$.

\begin{lemma}\label{sec:basicroutinesparam:lemma:intersect2}
  Let $(q_1,\scrR_1),\dots,(q_s,\scrR_s)$ be zero-dimensional
  parametriza\-tions over $\A$ that define algebraic sets
  $(V_\roottau)_{q(\roottau)=0}$, let $\degQ$ be the degree of $\scrQ$ and
  $\degR$ the maximum of the degrees of $\scrR_1,\dots,\scrR_s$. 

  Let further $G$ be a polynomial in $\A[X_{e+1},\dots,X_N]$, given by
  a straight-line program $\Gamma$ of length $E$.

  There exists an algorithm ${\sf Intersect}$ which takes as input
  $(q_1,\scrR_1),\dots,(q_s,\scrR_s)$ and $\Gamma$ and returns
  zero-dimensional parametrizations
  $(q'_1,\scrR'_1),\dots,(q'_t,\scrR'_t)$ over $\A$ that define the
  algebraic sets $(V'_\roottau)_{q(\roottau)=0}$, with
  $V'_\roottau=V_\roottau \cap V(G)$ for all $\roottau$, using
  $\softO( (E+N)\degQ \degR)$ operations in~$\QQ$.
\end{lemma}
\begin{proof}
  As in Lemma~\ref{sec:basicroutinesparam:lemma:intersect}, we first
  compute $g=G(w_{e+1},\dots,w_N) \bmod r$; this requires $\softO( E
  \degQ \degR)$ operations in $\QQ$. We can then compute a GCD
  $$(q_1,h_1),\dots,(q_s,h_s)$$ of $r$ and $g$ in $\A[X]$; the cost is
  $\softO( \degQ \degR)$ by Lemma~\ref{geosolve:lemma:GCD}.

  We conclude by computing $v_{i,j} = v_i \bmod q_j$ (for
  $i=1,\dots,e$ and $j=1,\dots,s$) and $w_{i,j}=w_i \bmod \langle
  q_j,h_j\rangle$ (for $i=e+1,\dots,N$ and $j=1,\dots,s$), all in
  $\softO( N \degQ \degR )$ operations. Finally, we return the pairs
  $\scrQ_j=((q_j,v_{1,j},\dots,v_{e,j}),\pollambda)$ and
  $\scrR_j=((h_j,w_{e+1,j},\dots,w_{N,j}),\linearmu)$.
\end{proof}

%%%%%%%%%%%%%%%%%%%%%%%%%%%%%%%%%%%%%%%%%%%%%%%%%%%%%%%%%%%%

\subsubsection{Dimension one} 

The previous idea can be extended to represent curves. A {\em
  one-dimensional parametrization}
$\scrR=((r,w_{e+1},\dots,w_N),\linearmu,\linearmu')$ with coefficients in $\A$
consists in the following:
\begin{itemize}
\item polynomials $(r,w_{e+1},\dots,w_N)$, such that $r\in \A[U,X]$ is
  squarefree (in the sense of Subsection~\ref{ssec:prodfieldsbasic}) and
  monic in both $U$ and $X$, all $w_i$ are in $\A[U,X]$ and satisfy
  $\deg(w_i,X) < \deg(q,X)$; we will impose the same degree constraint 
as  in Subsection~\ref{sec:basicroutinesparam1} (detailed below);
\smallskip
\item linear forms $\linearmu,\linearmu'$ in $X_{e+1},\dots,X_N$ with coefficients in
  $\QQ$ such that, as in Subsection~\ref{sec:basicroutinesparam1}, we have
  $$\linearmu\left (w_{e+1},\dots,w_N\right)=U \frac{\partial r}{\partial X}
  \bmod r \quad\text{and}\quad \linearmu'\left (w_{e+1},\dots,w_N\right)=X
  \frac{\partial r}{\partial X} \bmod r.$$ 
\end{itemize}

As in dimension zero, we will mostly be interested in the situation
where we know a zero-dimensional parametrization of the form
$\scrQ=(q,(v_1,\dots,v_e),\pollambda)$.  We can then define
$\Zeroes(\scrQ,\scrR)$ as the Zariski closure of the locally closed set
defined by
$$q(\roottau)=0,\qquad r(\roottau,\eta,\rootrho) = 0,\qquad  \frac{\partial r}{\partial X}(\roottau,\eta,\rootrho) \ne 0$$
and
$$X_i = v_i(\roottau) \ \ (1 \le i \le e),\qquad X_{i} =
\frac{w_i(\roottau,\eta,\rootrho)}{\frac{\partial r}{\partial
    X}(\roottau,\eta,\rootrho)} \ \ (e+1 \le i \le N).$$ When $q$ or $r$ is
constant, $\Zeroes(\scrQ,\scrR)$ is empty.  Else, it is an algebraic curve
that lies over $\Zeroes(\scrQ)$; furthermore, it is the disjoint union of
the finitely many curves $Z_\x$, for $\x$ in $\Zeroes(\scrQ)$, where $Z_\x$
is defined as $Z_\x=\fbr(\Zeroes(\scrQ,\scrR),\x)$ and thus lies over $\x$.

Equivalently, for any root $\roottau$ of $q$, we define $\scrR_\roottau$ as
the one-dimensional paramet\-rization with coefficients in $\C$ obtained
by applying the evaluation map $\phi_\roottau$ to the coefficients of all
polynomials in $\scrR$. Then, also associated to $\scrR$ are the
algebraic sets $(\Zeroes(\scrR_\roottau))_{q(\roottau)=0}$, where each
$\Zeroes(\scrR_\roottau)$ is a subset of $\C^{N-e}$. For $\x=(x_1,\dots,x_e)$ in
$\Zeroes(\scrQ)$, $Z_\x=(x_1,\dots,x_e)\times \Zeroes(\scrR_\roottau)$, where
$\roottau=\pollambda(\x)$ is the root of $q$ corresponding to $\x$.

In terms of degree, for $\roottau$ a root of $q$, we let $\degR_\roottau$ be
the degree of curve $\Zeroes(\scrR_\roottau)$, and let $\degR$ be
the maximum of all $\degR_\roottau$. Using ~\cite[Theorem~1]{Schost03}, we
deduce that for any root $\roottau$ of $q$, $\phi_\roottau(r)$ has degree at
most $\degR_\roottau$ in both $U$ and $X$, and similarly for the
polynomials $w_i$. Thus, $r$ and all $w_i$'s have degree at most
$\degR$ in both $U$ and~$X$.

Our last constraint, mentioned above, is that for all $\roottau$,
$r(\roottau,U,X)$ has degree $\degR_\roottau$ in both $U$ and $X$; since we
assumed that $r$ is monic in both $U$ and $X$, this actually implies
that $\degR_\roottau=\degR$ holds for all $\roottau$.

\begin{lemma}\label{lemma:descent1}
  Let $q$ and $\scrR$ be as above, let $\degQ$ be the degree of
  $\scrQ$ and $\degR$ the degree of $\scrR$. There exists a
  probabilistic algorithm ${\sf Descent}$ which takes as input $\scrQ$
  and $\scrR$ and returns either a one-dimensional parametrization
  $\scrR'$ with coefficients in $\QQ$ or ${\sf fail}$ using
  $\softO(N \degQ^3 \degR^3)$ operations in $\QQ$. In case of success,
  $\Zeroes(\scrR')=\cup_{q(\roottau)=0} \Zeroes(\scrR_\roottau)$.
\end{lemma}
\begin{proof}
  As we did several times in Subsection~\ref{sec:basicroutinesparam1}, we
  follow the zero-dimensional version of the algorithm (which was in
  this case Lemma~\ref{lemma:descent0}), with the intent of doing all
  computations over $\QQ(U)$; the algorithm chooses a new linear form
  in $X_{e+1},\dots,X_N$ at random, and for a generic choice, the output
  coefficients will actually be in $\QQ[U]$.
  
  In order to avoid computations with rational functions in $U$, we
  replace them by power series in $U-u_0$, for a randomly chosen
  $u_0$. Since the output has degree at most $\degQ \degR$ in $U$, the
  overhead compared to the zero-dimensional case is
  $\softO(\degQ \degR)$, and the cost increases to
  $\softO(N \degQ^3 \degR^3)$ operations in $\QQ$.
\end{proof}

Continuing the analogy with the case of dimension zero, we may not be
able to represent any family of algebraic curves
$(V_\roottau)_{q(\roottau)=0}$ as $V_\roottau=\Zeroes(\scrR_\roottau)$, for a
one-dimensional parametrization $\scrR$ with coefficients in $\A$. The
workaround will be the same: we consider a sequence of pairs
$(q_1,\scrR_1),\dots,(q_s,\scrR_s)$ with, for all $i$, $q_i$ monic in
$\QQ[T]$ and $\scrR_i$ a one-dimensional parametrization with
coefficients in $\A_i=\QQ[T]/\langle q_i \rangle$, and with $q = q_1
\cdots q_s$, such that the following holds. For any root $\roottau$ of
$q$, there exists a unique $i$ in $\{1,\dots,s\}$ such that
$q_i(\roottau)=0$. Then $\scrR_{i,\roottau}$ is well-defined, and we require
that $V_\roottau = \Zeroes(\scrR_{i,\roottau})$. We will call
$(q_1,\scrR_1),\dots,(q_s,\scrR_s)$ {\em one-dimensional
  parametrizations over $\A$ for $(V_\roottau)_{q(\roottau)=0}$}. As in
dimension zero, an arbitrary family $(V_\roottau)_{q(\roottau)=0}$ may not
admit such a representation; in all cases of interest to us, though,
it will be the case.

We conclude with a cost estimate for applying a change of variables,
in precisely this context.
\begin{lemma}\label{lemma:complexity1POF:changevar}
  Let $(q_1,\scrR_1),\dots,(q_s,\scrR_s)$ be one-dimensional
  parametrizations over $\A$ that define algebraic sets
  $(V_\roottau)_{q(\roottau)=0}$, let $\degQ$ be the degree of $\scrQ$ and
  $\degR$ the maximum of the degrees of $\scrR_1,\dots,\scrR_s$,
  and let $\mA$ be in $\GL(N-e,\QQ)$.

  There exists an algorithm ${\sf ChangeVariables}$ which takes as
  input
  $$(q_1,\scrR_1),\dots,(q_s,\scrR_s)$$ and $\mA$ and returns
  one-dimensional parametrizations
  $(q_1,\scrR^\mA_1),\dots,(q_s,\scrR^\mA_s)$ over $\A$ that define
  the algebraic sets $(V_\roottau^\mA)_{q(\roottau)=0}$ using $\softO(N^2\degQ
  \degR^2+N^3)$ operations in $\QQ$.
\end{lemma}
\begin{proof}
  The proof is similar to that of
  Lemma~\ref{lemma:complexity1:changevar}, but working over the rings
  $\A_i=\QQ[T]/\langle q_i \rangle$ instead of $\QQ$.
\end{proof}

%%%%%%%%%%%%%%%%%%%%%%%%%%%%%%%%%%%%%%%%%%%%%%%%%%%%%%%%%%%%

\subsubsection{An intersection algorithm} \label{ssec:inter}

Finally, we describe the main step for the algorithms of the next
paragraphs, following~\cite{GiLeSa01,Lecerf2000}.  We are interested
in ``computing'' an intersection such as $V\cap V(G)$, or such as the
Zariski closure of $V \cap V(G) - V(H)$, for an algebraic set $V$ and
polynomials $G,H$. Following the philosophy of those references, that
goes back to~\cite{GiHeMoPa95,GiHeMoPa97,GiHeMoMoPa98}, both input and
output will be represented by means of hyperplane sections, since this
is sufficient to perform the required tasks (in a numerical context,
similar ``witness points'' feature prominently in algorithms based on
homotopy continuation methods, see~\cite{SoWa05} and references
therein).

The algorithms below are direct extensions of those
in~\cite{GiLeSa01}; the main difference is that here, all computations
are done over a product of fields.

As in the previous paragraphs, $q$ is a monic squarefree polynomial in
$\QQ[T]$, and $\A$ is product of fields $\A=\QQ[T]/\langle q \rangle$.
As usual, we fix two integers $N$ and $e$, and in what follows we work
in $\C^{N-e}$ (these will be the actual choices of dimensions when we
use this algorithm in the next paragraph). As in
Subsection~\ref{sec:D5eqs}, for a root $\roottau$ of $q$ and a family of
polynomials $\F$ in $\A[X_{e+1},\dots,X_N]$, we write $\F_\roottau$ for
the polynomials in $\C[X_{e+1},\dots,X_N]$ obtained from $\F$ through
the evaluation map $\phi_\roottau: \A \to \C$.

The algorithm relies on the following assumptions.
\begin{enumerate}
\item[${\sf g_1}.$] $(V_\roottau)_{q(\roottau)=0}$ is a family of
  algebraic sets, with each $V_\roottau$ either empty or
  $d$-equidimensional in $\C^{N-e}$.
\smallskip
\item[${\sf g_2}.$] $\F=(F_1,\dots,F_{P})$, with $P=N-e-d$, are
  polynomials in $\A[X_{e+1},\dots,X_N]$ such that for each $\roottau$ root
  of $q$, if $V_\roottau$ is not empty, it is contained in $V(\F_\roottau)$,
and the matrix $\jac(\F_\roottau)$ has generically full rank
  $P$ on all the irreducible components of $V_\roottau$.
\end{enumerate}
In addition, we consider two further polynomials $G$ and $H$ in
$\A[X_{e+1},\dots,X_N]$. For $\roottau$ root of $q$, we define $V'_\roottau =
V_\roottau \cap V(G) \subset \C^{N-e}$; our next assumption is then
the following:
\begin{enumerate}
\item [${\sf g_3}.$] each $V'_\roottau$ is either empty or $(d-1)$-equidimensional.
\end{enumerate}
We can finally define $V''=(V''_\roottau)_{q(t)=0}$ by letting $V''_\roottau$
be the Zariski closure of $V'_\roottau-V(H)$ for any root $\roottau$  of $q$.

To analyze the upcoming algorithm, we let $\degQ$ be the degree of
$q$, $\delta$ be the maximum of the degrees of the algebraic sets
$V_\roottau$, for $\roottau$ a root of $q$ and $D=\max(\deg(G),\deg(H))$. In
terms of data representation, we will suppose that $\F,G,H$ are given
by a straight-line program $\Gamma$ with coefficients in $\A$, as
defined in Subsection~\ref{sec:D5eqs}; we denote by $E$ an upper bound on
the length of it.

Finally, we use the following short-hand in all this paragraph: if
$\y=(y_1,\dots,y_d)$ is in $\C^d$, we write
$\pi(\y)=(y_1,\dots,y_{d-1})\in \C^{d-1}$. Then, the main result
of this paragraph is the following.

\begin{proposition}\label{prop:inter}
  There exists a probabilistic algorithm ${\sf SolveIncremental}$
  which takes as input \new{$\F$, $G$ and $H$ as above and}
  zero-dimensional parametrizations
  $(q_1,\scrR_1),\dots,(q_s,\scrR_s)$ over $\A$, and returns either
  zero-dimensional parametrizations
  $(q''_1,\scrR''_1),\dots,(q''_t,\scrR''_t)$ over $\A$ or
  ${\sf fail}$ using $\softO(N(E+N^3) D \degQ \delta^2 )$ operations
  in $\QQ$, and with the following characteristics.

  Suppose that ${\sf g_1}$, ${\sf g_2}$, ${\sf g_3}$ hold.  There
  exist a non-empty Zariski open subset $\mathscr{N}$ of $\GL(N-e)$, and,
  for $\mA$ in $\mathscr{N}$, a non-empty Zariski open subset $\mathscr{N}_\mA$
  of $\C^d$, such that if $\y\in \mathscr{N}_\mA$, and if the input
  $(q_1,\scrR_1),\dots,(q_s,\scrR_s)$ describes
  $(\fbr(V^\mA_\roottau,\y))_{q(\roottau)=0}$, then in case of success, the
  output $(q''_1,\scrR''_1),\dots,(q''_t,\scrR''_t)$ of ${\sf
    SolveIncremental}$ describes
  $(\fbr({V''_\roottau}^\mA,\pi(\y)))_{q(\roottau)=0}$.
\end{proposition}

The proof of this proposition will occupy the rest of this paragraph.
We start by dimension and degree properties.

\begin{lemma}\label{lemma:dimension}
  Suppose that ${\sf g_1}$, ${\sf g_2}$ and ${\sf g_3}$ hold. There
  exists a non-empty Zariski open subset $\ZOomega$ of $\GL(N-e)$, such
  that for $\mA$ in $\ZOomega$, and for every root $\roottau$ of $q$, the
  following holds.
  There exists a non-empty Zariski open subset $\ZOomega_{\mA,\roottau}$ of 
  $\C^d$ such that for $\y$ in $\ZOomega_{\mA,\roottau}$, we have:
  \begin{itemize}
  \item the fiber $\fbr(V_\roottau^\mA, \y)$ is empty or of dimension
    zero, and has the same degree as ${V_\roottau}$,
\smallskip
  \item the fiber $\fbr(V_\roottau^\mA, \pi(\y))$ is empty or one-equidimensional,
    and has the same degree as ${V_\roottau}$,
\smallskip
  \item the fibers $\fbr({V'_\roottau}^\mA, \pi(\y))$ and
    $\fbr({V''_\roottau}^\mA, \pi(\y))$ are empty or of dimension zero, and
    have the same degree as respectively $V'_\roottau$ and ${V''_\roottau}$.
  \end{itemize}
\end{lemma}
\begin{proof}
  Fix a root $\roottau$ of $q$. If $V_\roottau$ is empty, all assertions obviously
  hold, so we will assume that we are not in this case. By ${\sf g_1}$, 
  we deduce that $V_\roottau$ is $d$-equidimensional.

  Then, for a generic change of variables $\mA$ in $\GL(N-e)$,
  $V_\roottau^{\mA}$ is in Noether position with respect to the projection
  on the first $d$ variables. For such choices, all fibers for the
  projection on these $d$ variables are zero-dimensional, and all of
  them in a Zariski dense subset of $\C^d$ have degree $\deg(V_\roottau)$.
  Similarly, all fibers for the projection on the first $d-1$
  variables are one-equidimensional, and all of them in a
  Zariski dense subset of $\C^{d-1}$ have degree $\deg(V_\roottau)$ (for
  all this, see for instance~\cite[Corollary~2.5]{DuLe08}). The same
  argument applies to the set $V'_\roottau$ and $V''_\roottau$ (which are
  either $(d-1)$-equidimensional or empty by ${\sf g_3}$) to prove the
  third point.
\end{proof}

%% The subsets $\omega'_{\mA,\roottau}$ and $\omega''_{\mA,\roottau}$ give some
%% of the genericity conditions we will need for 
%% Taking the finitely many roots of $q$ into consideration, the previous
%% lemma proves the first two statements of Proposition~\ref{prop:inter}.

Algorithm ${\sf SolveIncremental}$ follows the intersection process
of~\cite{GiLeSa01}; the only nontrivial difference is that our
computations take place with coefficients taken modulo $q$, or factors
of it. If $q$ were irreducible, we could simply point out that the
algorithm of~\cite{GiLeSa01} still applies over the field
$\A=\QQ[T]/\langle q \rangle$, and we would be done. Without this
assumption, the only steps that require attention are those involving
inversions in $\A$.

The length of the exposition in~\cite{GiLeSa01} prevents us from
giving all details of the algorithms, let alone proofs of correctness:
we briefly revisit the main steps in the algorithm and indicate the
necessary modifications. First, starting from zero-dimensional
parametrizations over $\A$ for the finite sets $(\fbr(V_\roottau^\mA,\y))$, we
recover one-dimensional parametrizations over $\A$ for the curves
$(\fbr(V_\roottau^\mA,\pi(\y)))$ (Lemma~\ref{lemma:lift} below, to be
compared to~\cite[Lemma~3]{GiLeSa01}). Then, we perform an
intersection process (Lemma~\ref{lemma:inter} below, to be compared
to~\cite[Lemma~16]{GiLeSa01}). Altogether, we simply lose a factor
$\softO(\degQ)$ in the running time, and combining these two lemmas proves
Proposition~\ref{prop:inter}.

\begin{lemma}\label{lemma:lift}
  There exists an algorithm ${\sf SolveIncremental\text{-}Lift}$ that takes
  as input zero-dimensional parametrizations
  $(q_1,\scrR_1),\dots,(q_s,\scrR_s)$ over $\A$, and returns either
  one-dimen\-sional parametrizations
  $(q'_1,\scrR'_1),\dots,(q'_s,\scrR'_s)$ over $\A$ or ${\sf fail}$
  using $\softO( N(E+N^3)\degQ \delta^2)$ operations in~$\QQ$, and
  with the following characteristics.

  Suppose that ${\sf g_1}$, ${\sf g_2}$, ${\sf g_3}$ hold.  For $\mA$
  in $\ZOomega$, there exists a non-empty Zariski open subset
  $\ZOomega'_\mA$ of $\C^d$, such that if $\y \in \ZOomega'_\mA$, and if
  the input $(q_1,\scrR_1),\dots,(q_s,\scrR_s)$ describes
  $(\fbr(V^\mA_\roottau,\y))_{q(\roottau)=0}$, then in case of success, the
  output $(q'_1,\scrR'_1),\dots,(q'_s,\scrR'_s)$ of
  ${\sf SolveIncremental\text{-}Lift}$ describes
  $(\fbr({V_\roottau}^\mA,\pi(\y)))_{q(\roottau)=0}$.
\end{lemma}
\begin{proof}
  The first restriction is that $\mA$ should satisfy the assumptions
  of the previous lemma. Further restrictions on $\y$ are needed: for
  any root $\roottau$ of $q$, the fiber $\fbr(V_\roottau^\mA,\y)$ should have
  the same degree as $V_\roottau$ itself (see the previous lemma), and the
  square Jacobian matrix $\jac(\F_\roottau,d)$ should be invertible on all
  points of $\fbr(V_\roottau^\mA,\y)$. Proposition~4.3 in~\cite{DuLe08}
  shows that under assumption ${\sf g_2}$, this is the case for a
  generic choice of $\y$. Taking all roots $\roottau$ into considerations
  defines the set $\ZOomega'_\mA$.

  Let then $(q_1,\scrR_1),\dots,(q_s,\scrR_s)$ be the input
  zero-dimensional parametrizations over $\A$ for
  $(\fbr(V_\roottau^\mA,\y))_{q(\roottau)=0}$, with for all $i$,
  $\scrR_i=((r_i,w_{i,e+1},\dots,w_{i,N}),\linearmu_i)$, all polynomials in
  $\scrR_i$ having coefficients in $\A_i=\QQ[T]/\langle q_i \rangle$.
  Remark that $\deg(r_i) \le \delta$ holds for all $i$ and that
  $\degQ_1 + \cdots + \degQ_s=\degQ$, with $\degQ_i=\deg(q_i)$ for all
  $i$.

  First, we restrict our attention to those roots $\roottau$ of $q$ for
  which $V_\roottau$ is not empty. Since we assume that $V_\roottau$ and
  $\fbr(V_\roottau^\mA,\y)$ have the same degree, it suffices to discard
  those pairs $(q_i,\scrR_i)$ for which $\scrR_i$ defines the empty
  set, {\it i.e.} for which $r_i=1$. At the end of the process, we
  will then re-introduce some ``dummy'' pairs for those indices, of
  the form $(q_i,\scrR'_i)$, where $\scrR'_i$ is a one-dimensional
  parametrization of the form (say) $((1,0,\dots,0), \linearmu_i, \linearmu'_i)$
  that defines the empty set. In order to avoid introducing further
  notation, we still write $(q_1,\scrR_1),\dots,(q_s,\scrR_s)$ for the
  remaining objects.

  We are going to work with all pairs $(q_i,\scrR_i)$
  independently. For this, we first have to transform the
  straight-line program $\Gamma$ that computes $\F$ into straight-line
  programs $\Gamma_1,\dots,\Gamma_s$, where $\Gamma_i$ has
  coefficients in $\A_i$: for a given $i$, this is done by replacing
  all constants in $\A$ that appear in $\Gamma$ by their images modulo
  $q_1,\dots,q_s$; altogether, this take $\softO(E \degQ)$ operations
  in $\QQ$. Then, for $i=1,\dots,s$, we follow Algorithm~2
  from~\cite{GiLeSa01}, with coefficients in $\A_i$. This consists in
  two steps:
  \begin{itemize}
  \item inverting the matrix
    $\jac(\F,d)(w_{i,e+1},\dots,w_{i,N})$ over $\B_i=\QQ[T,X]/\langle q_i, r_i \rangle$;
\smallskip
  \item using this inverse, applying a version of Newton iteration, to
    compute a one-dimensional parametrization $\scrR'_i$ with
    coefficients in $\A_i$.
  \end{itemize}
  In the first step, we compute the matrix $\jac(\F,d)$ evaluated at
  $(w_{i,e+1},\dots,w_{i,N})$ and its determinant (the cost is
  subsumed by the cost of lifting given below). The assumption made
  above on $\y$ implies that the inversion we attempt is indeed
  feasible (if not, we return ${\sf fail}$). Then, as explained
  in~\cite[Proposition~6]{DaJiMoSc08}, the determinant can be inverted
  using $\softO(\degQ_i \delta)$ operations in $\QQ$.

  The second part of the algorithm is the lifting per se; this part
  does not require any inversion, so the analysis
  in~\cite[Lemma~3]{GiLeSa01} carries over to our situation over
  $\A_i$, giving a running time of $\softO(N(E+N^3) \delta^2)$
  operations $(+,\times)$ in $\A_i$, or $\softO( N(E+N^3)\degQ_i
  \delta^2)$ operations in $\QQ$. Summing over all $i$ concludes the
  proof of the lemma.
\end{proof}

Combining ${\sf SolveIncremental\text{-}Lift}$ and algorithm
${\sf SolveIncremental\text{-}Intersect}$ below is enough to prove
Proposition~\ref{prop:inter}.

\begin{lemma}\label{lemma:inter}
  There exists an algorithm ${\sf SolveIncremental\text{-}Intersect}$ that
  takes as input one-dimensional parametrizations
  $(q'_1,\scrR'_1),\dots,(q'_s,\scrR'_s)$ over $\A$, and returns
  either zero-dimensional parametrizations
  $(q''_1,\scrR''_1),\dots,(q''_t,\scrR''_t)$ over $\A$ or
  ${\sf fail}$ using $\softO( N(E+N^2) D \degQ\delta^2)$ operations
  in~$\QQ$, and with the following characteristics.

  Suppose that ${\sf g_1}$, ${\sf g_2}$, ${\sf g_3}$ hold. Then, there
  exist a non-empty Zariski open subset $\ZOomega''$ of $\GL(N-e)$, and,
  for $\mA$ in $\ZOomega''$, a non-empty Zariski open subset
  $\ZOomega''_\mA$ of $\C^d$, such that if $\y$ in $\ZOomega''_\mA$, and
  if the input $(q'_1,\scrR'_1),\dots,(q'_s,\scrR'_s)$ describes
  $(\fbr(V^\mA_\roottau,\pi(\y)))_{q(\roottau)=0}$, then in case of success,
  the output $(q''_1,\scrR''_1),\dots,(q''_t,\scrR''_t)$ of
  ${\sf SolveIncremental\text{-}Intersect}$ describes the set
  $(\fbr({V''_\roottau}^\mA,\pi(\y)))_{q(\roottau)=0}$.
\end{lemma}
\begin{proof} 
  The first assumptions on $(\mA,\y')$ are that all sets
  $\fbr(V_\roottau^\mA,\y')$ are empty or one-equidi\-men\-sional and have
  the same degree as $V_\roottau$; similarly, all sets
  $\fbr({V''_\roottau}^\mA,\y')$ must be empty or zero-dimensional and
  have the same degree as $V''_\roottau$ (see Lemma
  \ref{lemma:dimension}). The algorithm requires further assumptions
  on $(\mA,\y')$, which are mentioned in \cite[Lemma~16]{GiLeSa01} and
  discussed in detail in \cite[Proposition~4.3]{DuLe08}. We shall not
  need to give them in detail here; using
  \cite[Proposition~4.3]{DuLe08}, it is enough to note that they hold
  for generic choices of $\mA$ and $\y'$ as above, which leads to the
  existence of the open sets $\ZOomega''$ and $\ZOomega''_\mA$.
  
  Let $(q'_1,\scrR'_1),\dots,(q'_s,\scrR'_s)$ be the input
  one-dimensional parametrizations over $\A$ for the sets
  $(\fbr(V_\roottau^\mA,\y'))_{q(\roottau)=0}$, with for all $i$,
  $\scrR'_i=((r_i,w_{i,e+1},\dots,w_{i,N}),\linearmu_i,\linearmu'_i)$, where $r_i$
  is in $\A_i[U,X]$, with $\A_i=\QQ[T]/\langle q'_i\rangle$. Now
  we write $\degQ_i=\deg(q'_i)$ and we remark that $\degQ_1 + \cdots +
  \degQ_s=\degQ$. Up to discarding all $(q'_i,\scrR'_i)$ for which
  $r_i=1$, we may assume that none of the sets $\fbr(V_\roottau^\mA,\y')$
  is empty; at the end of the process, we will reintroduce pairs
  $(q'_i,\scrR''_i)$ for those pairs we discarded, with
  $\scrR''_i=((1,0\dots,0),\nu_i)$, for some linear form $\nu_i$.

  The algorithm starts as in the previous lemma, replacing $\Gamma$ by
  straight-line programs $\Gamma_1,\dots,\Gamma_s$ having coefficients
  in respectively $\A_1,\dots,\A_s$. The cost of this preparation will
  be negligible compared to what follows.

  We will work independently with all pairs $(q'_i,\scrR'_i)$; this
  time, we follow~\cite[Algorithm 11]{GiLeSa01}. Let us thus fix $i$
  in $\{1,\dots,s\}$. Algorithm~11 in~\cite{GiLeSa01} relies on four
  subroutines, which are called (in that order) Algorithms~8,~7,~9
  and~10 in that reference. We review them briefly and underline the
  steps that require adaptation when working over a product of fields
  (that is, those steps that involve inversions).
  \begin{itemize}
  \item In the first one (Algorithm 8), the only difficulty arises
    when we invert $\partial r_i/\partial X$ modulo the ideal
    $\langle (U-u_0)^{D \delta+1}, r_i \rangle$ in $\A_i[U,X]$, for
    a randomly chosen $u_0\in\QQ$. Our genericity assumptions on
    $\mA$ and $\y$ imply that this inversion is feasible and that we
    are under the assumptions of
    Lemma~\ref{geosolve:lemma:euclideinrings1}; in view of that lemma,
    this can be done using $\softO( D \degQ_i \delta^2)$ operations in
    $\QQ$; all other steps in Algorithm~8 carry over to arithmetic
    over $\A_i$ without modification and their costs add up to
    $\softO( N^2 D \degQ_i \delta^2)$ operations in $\QQ$. If the
    inversion is impossible, we return ${\sf fail}$.

\smallskip

    The output of this step is a sequence of polynomials
    $$R_i,V_{i,e+1},\dots,V_{i,N}$$ in $\B_i[X]$, with
    $$\B_i=\A_i[t,t_{e+1},\dots,t_N,U]/\langle (t,t_{e+1},\dots,t_N)^2,
    (U-u_0)^{D \delta+1}\rangle,$$ where $t,t_{e+1},\dots,t_N$ 
    are new variables.

  \item In the second subroutine (Algorithm 7), we perform a similar
    inversion as in the previous step, but with coefficients in a ring
    of the form $$\A_i[t,t_{e+1},\dots,t_N]/\langle
    (t,t_{e+1},\dots,t_N)^2 \rangle$$ instead of $\A_i$: this can be
    done by first computing the inverse over $\A_i$ (as in the
    previous step, so we can again apply the result of
    Lemma~\ref{geosolve:lemma:euclideinrings1}), then doing one step
    of Newton iteration to lift the inverse modulo $$\langle
    (t,t_{e+1},\dots,t_N)^2 \rangle.$$ This results in an overhead of
    $O(N)$, for a total of $\softO( N D\degQ_i \delta^2)$ operations in
    $\QQ$.

\smallskip

    Then, we compute the resultant $S_i$ of two polynomials of degree
    at most $\delta$ in $\B_i[X]$, with as above
    $$\B_i=\A_i[t,t_{e+1},\dots,t_N,U]/\langle (t,t_{e+1},\dots,t_N)^2,
    (U-u_0)^{D \delta+1}\rangle.$$ These polynomials are derived
    from $G$ and from the output $$R_i,V_{i,e+1},\dots,V_{i,N}$$ of the
    previous step; using the straight-line program $\Gamma_i$ for $G$,
    they are computed in $\softO( N (E+N^2)D \degQ_i\delta^2)$
    operations in~$\QQ$.  
\smallskip

    The discussion in~\cite[Section~6.3]{GiLeSa01} then shows that for
    a choice of $\mA$ and $\y$ satisfying the genericity assumptions
    mentioned in the preamble, the assumptions of
    Lemma~\ref{geosolve:lemma:euclideinrings2} are satisfied; as a
    result, the running time of the resultant computation is
    $\softO( N^2 D \degQ_i\delta^2)$ operations in $\QQ$. If these
    assumptions are not satisfied,
    Lemma~\ref{geosolve:lemma:euclideinrings2} will attempt a division
    by a power series of positive valuation; if this is detected, we
    return ${\sf fail}$.

\smallskip

    The cost of all other operations, which involve no inversion in
    $\A_i$, adds up to a similar $\softO( N^2 D\degQ_i \delta^2)$. The
    total for this subroutine is thus $\softO( N (E+ N^2) D \degQ_i
    \delta^2)$ operations in $\QQ$.
\smallskip

  \item Next subroutine is Algorithm~9, where we compute a squarefree
    part of a polynomial (derived from polynomial $S_i$ above) of
    degree at most $D \delta$ in $\A_i[U]$, followed by $O(N)$
    simpler operations on such polynomials (Euclidean divisions). We
    handle the squarefree part computation using
    Lemma~\ref{geosolve:lemma:SQF} using $\softO( D \degQ_i \delta)$
    operations in $\QQ$; the Euclidean divisions take $\softO( N D \degQ_i
    \delta)$ operations in $\QQ$.
\smallskip
    
    Invoking Lemma~\ref{geosolve:lemma:SQF} may induce a factorization
    of $q'_i$ into polynomials $q'_{i,1},\dots,q'_{i,j_i}$; we continue
    the computations modulo each $q'_{i,k}$ separately. This requires
    reducing the coefficients of $O(N)$ polynomials of degree $D
    \delta$ with coefficients in $\A_i$ modulo $q'_{i,k}$: this is done
    by fast modular reduction using a total $\softO( N D \degQ_i
    \delta)$ operations in $\QQ$.

\smallskip

    For $k=1,\dots,j_i$, Algorithm~9 further requires an inversion in
    the ring $\A_{i,k}[U]/\langle M_{i,k} \rangle$, with
    $\A_{i,k}=\QQ[T]/\langle q'_{i,k}\rangle$, where $M_{i,k}$ is a
    monic polynomial of degree at most $D \delta$ derived from the
    outcome of the above squarefree computation.  For a choice of
    $\mA$ and $\y$ satisfying the genericity assumptions in the
    preamble, it is proved in~\cite{GiLeSa01} that all these
    inversions are feasible; using
    again~\cite[Proposition~6]{DaJiMoSc08}, each of them is seen to
    cost $\softO( D \degQ_{i,k}\delta)$ operations in $\QQ$, where
    $\degQ_{i,k}$ is the degree of $q'_{i,k}$. The total for these
    inversions is $\softO( D \degQ_i \delta)$ and altogether, the cost
    of Algorithm~9 is $\softO( N D \degQ_i\delta)$ operations in
    $\QQ$.

\smallskip

    If some inversion turns out to be not feasible, we return
    ${\sf fail}$.

\smallskip

  \item For $k=1,\dots,j_i$, Algorithm~10 finally entails the
    evaluation of our input polynomial $H$ at elements of residue
    class rings of the form $\A_{i,k}[U]/\langle M'_{i,k} \rangle$,
    with $\A_{i,k}$ as above and all $M'_{i,j}$ of degree at most $D
    \delta$ (derived from the polynomials $M_{i,k}$ above), followed
    by a GCD computation in degree $D \delta$ in the rings
    $\A_{i,k}[U]$ and $O(N)$ Euclidean divisions in similar degrees.
    The output of the algorithm is then directly deduced from these
    results.

\smallskip

    For a given index $k$, the cost of evaluating $H$ is
    $\softO( E D \degQ_{i,k} \delta)$ operations in $\QQ$. The GCD computation
    is handled using Lemma~\ref{geosolve:lemma:GCD}, for a cost of
    $\softO( D \degQ_{i,k} \delta)$; the cost of all Euclidean
    divisions is then $\softO( N D \degQ_{i,k}\delta)$. In total,
    the cost for a given index $i$ is $\softO( (E+N) D\degQ_{i} \delta))$.

\smallskip

    {\em In the next section, we will use again this last subroutine; 
      as in~\cite{GiLeSa01}, we will refer to it as Algorithm {\sf Clean}}
  \end{itemize}
  Altogether, the cost for a given index $i$ is $\softO( N(E+N^2) D
  \degQ_i \delta^2)$; the total is thus $\softO( N(E+N^2) D \degQ
  \delta^2)$ operations in $\QQ$.
\end{proof}

%%%%%%%%%%%%%%%%%%%%%%%%%%%%%%%%%%%%%%%%%%%%%%%%%%%%%%%%%%%%
%%%%%%%%%%%%%%%%%%%%%%%%%%%%%%%%%%%%%%%%%%%%%%%%%%%%%%%%%%%%
%%%%%%%%%%%%%%%%%%%%%%%%%%%%%%%%%%%%%%%%%%%%%%%%%%%%%%%%%%%%

\subsection{Polynomial system solving}\label{sec:possobasepoints}

We now reach the main part of this section: some algorithms
for solving systems of polynomial equations. As before, we consider
$N-e$ coordinates $X_{e+1},\dots,X_N$ and let $q$ be a squarefree
polynomial of degree $\degQ$ in $\QQ[T]$.

Our main results in this paragraph are
Propositions~\ref{geosolve:prop:variant1a} (in
Subsection~\ref{ssec:solvingf}) and \ref{geosolve:prop:variant2} (in
Subsection~\ref{ssec:solvingfg}); these are estimates on the cost of
solving equations with coefficients in $\A=\QQ[T]/\langle q \rangle$,
respectively of the form $\F(\x)=0$ (under some regularity assumptions)
and $\F=\G=0$ (under regularity assumptions only on $\F$). All are
based on the geometric resolution algorithm in \cite{GiLeSa01} and its
variant in~\cite{Lecerf2000}. The only difference is that computations
are run modulo $q$ (or factors of it), whereas in previous references
the same results were given over $\QQ$; thus, we have to rely on the
algorithm described in the previous paragraph.

%%%%%%%%%%%%%%%%%%%%%%%%%%%%%%%%%%%%%%%%%%%%%%%%%%%%%%%%%%%%

\subsubsection{Basic definitions}\label{ssec:definitionsF}

Let $\F=(F_1, \ldots, F_P)$ be polynomials in the ring $\A[X_{e+1},\dots,X_N]$,
with $P \le N-e$. In this short paragraph, we define the objects
associated to $\F$ that will play a prominent role in the sequel.

For $\roottau$ a root of $q$, we define polynomials $\F_\roottau \in
\C[X_{e+1},\dots,X_N]$ as in Subsection~\ref{sec:D5eqs}; we will feel
free to use the same notation for further families of polynomials.  We
will be interested in the family of algebraic sets
$(V_\roottau)_{q(\roottau)=0}$, where each algebraic set
$V_\roottau=\freg(\F_\roottau) \subset \C^{N-e}$ is as in
Subsection~\ref{sec:D5eqs}. As was pointed out in
Subsection~\ref{chap:prelim:sec:definitions}, by the Jacobian criterion
(\cite[Theorem 16.19]{Eisenbud95}, or
  Lemma~\ref{lemma:prelim:locallyclosed}),
each $V_\roottau$ is either equidimensional of dimension $d=N-e-P$ or empty.

Defining the set $\Delta$ of maximal minors of $\jac(\F)$, which thus
have size $P$, and the Zariski open sets
$\mathcal{O}_\roottau=\C^{N-e}-V(\Delta_\roottau)$, $V_\roottau=\freg(\F_\roottau)$ is
by definition the Zariski closure of
$V(\F_\roottau)\cap \mathcal{O}_\roottau$.

The algorithm will solve the whole system $\F$ by considering all
intermediate systems it defines. For $1\leq i \leq P$, we thus denote
by $\F_i$ the sequence $(F_1, \ldots, F_i)$; if $\roottau$ is a root of
$q$, we then let $V_{i,\roottau}$ the Zariski closure of
$V(\F_{i,\roottau})\cap \mathcal{O}_\roottau$; when $i=P$, we recover
$V_\roottau=V_{P,\roottau}$.

\begin{lemma}\label{sec:geosolve:lemma:equidim}
  For each root $\roottau$ of $q$, the following holds:
  \begin{itemize}
  \item for $1\leq i \leq P$, the matrix $\jac(\F_{i,\roottau})$ has 
    generically full rank $i$ on each irreducible component of $V_{i,\roottau}$;
\smallskip
  \item for $1\leq i \leq P$, $V_{i,\roottau}$ is either empty or
    equidimensional of dimension $N-e-i$;
\smallskip
  \item for $1 \le i < P$, $V_{i,\roottau} \cap V(F_{i+1,\roottau})$ is either
    empty or equidimensional of dimension $N-e-i-1$.
  \end{itemize}
\end{lemma}
\begin{proof}
  Fix a root $\roottau$ of $q$; suppose that $i \le P$ and that
  $V_{i,\roottau}$ is not empty.

  Let $\Delta_{i,\roottau}$ be the set of maximal ($i\times i$) minors of
  $\jac(\F_{i,\roottau})$. If all the minors in $\Delta_{i,\roottau}$ vanish
  at a point $\x \in \C^{N-e}$, then all the minors in $\Delta_\roottau$
  vanish at $\x$, so $V(\Delta_{i,\roottau})$ is contained in
  $V(\Delta_\roottau)$, and thus $V(\F_{i,\roottau})-V(\Delta_\roottau)$ is
  contained in $V(\F_{i,\roottau})-V(\Delta_{i,\roottau})$.  Letting $\tilde
  V_{i,\roottau}$ be the Zariski closure of
  $V(\F_{i,\roottau})-V(\Delta_{i,\roottau})$, we deduce that $V_{i,\roottau}$ is
  the union of the irreducible components of $\tilde V_{i,\roottau}$ not
  contained in $V(\Delta_\roottau)$.  By the Jacobian criterion, $\tilde V_{i,\roottau}$ is
  $(N-e-i)$-equidimensional or empty. This implies that all irreducible
  components of $V_{i,\roottau}$ have the same dimension $N-e-i$, so the
  first two items are proved.

  Suppose further that $i < P$.  Because $V_{i,\roottau}$ is
  equidimensional of dimension $N-e-i$, any irreducible component of
  $V_{i,\roottau} \cap V(F_{i+1,\roottau})$ has dimension either $N-e-i$ or
  $N-e-i-1$. Let us prove that the latter necessarily holds. Assume
  that there exists such an irreducible component $Z$ of dimension
  $N-e-i$. Then, $Z$ must be an irreducible component of $V_{i,\roottau}$
  itself, and $F_{i+1,\roottau}$ vanishes identically on $Z$.

  Because $Z$ is contained in $V_{i,\roottau}$, it is contained in
  $V(\F_{i,\roottau})$, and because $F_{i+1,\roottau}$ is zero on $Z$, $Z$ is
  actually contained in $V(\F_{i+1,\roottau})$. As a consequence,
  $Z-V(\Delta_\roottau)$ is contained in
  $V(\F_{i+1,\roottau})-V(\Delta_\roottau)$.  Because $Z$ is an irreducible
  component of $V_{i,\roottau}$, we know that the Zariski closure of
  $Z-V(\Delta_\roottau)$ is $Z$ itself, so that $Z$ is contained in
  $V_{i+1,\roottau}$. This is a contradiction, since $V_{i+1,\roottau}$ has
  dimension $N-e-i-1$.
\end{proof}

The cost of our algorithms will depend on the degree of the
intermediate algebraic sets $V_{i,\roottau}$. The actual notion we will
use is the following, taken from~\cite{GiHeMoMoPa98}.

\begin{definition}\label{def:gdeg}
   For $1\leq i \leq P$, we denote by $\delta_i$ the maximum of the
   degrees of the sets $V_{i,\roottau}$, for $\roottau$ a root of $q$. We call
   $\delta=\max(\delta_1, \ldots, \delta_P)$ the {\em geometric
     degree} of $\F$.
\end{definition}

%%%%%%%%%%%%%%%%%%%%%%%%%%%%%%%%%%%%%%%%%%%%%%%%%%%%%%%%%%%%

\subsubsection{Solving \texorpdfstring{$\F=0$}{$\F=0$}}\label{ssec:solvingf}

With notation as above, our first goal is to give an algorithm that
solves equations $\F=0$, with $\F=(F_1,\dots,F_P)$ in
$\A[X_{e+1},\dots,X_N]$. More precisely, we restrict our attention to
dimension zero or one, and we compute zero, resp.\ one-dimensional
parametrizations of the family $(V_\roottau)_{q(\roottau)=0}$, with
$V_\roottau=\freg(\F_\roottau)$. In other words, we focus on the cases $P=N-e$
and $P=N-e-1$.

\begin{proposition}\label{geosolve:prop:variant1a}
  There exists a probabilistic algorithm ${\sf Solve\_F}$ that takes as
  input a squarefree polynomial $q$ and a straight-line program
  $\Gamma$ with coefficients in $\A$, with the following
  characteristics: Suppose that $\Gamma$ has length $E$, computes
  polynomials $\F$ of degree at most $D$, that $q$ has degree $\degQ$
  and let $\delta$ be the geometric degree of $\F$. Then,
  \begin{itemize}
  \item when $P=N-e$, ${\sf Solve\_F}(q,\Gamma)$ outputs either
    zero-dimensional paramet\-rizations over $\A$ or ${\sf fail}$
    using $\softO( N^3(E+N^3) D \degQ \delta^2)$ operations in $\QQ$.
    In case of success, the output describes the family
    $(V_\roottau)_{q(\roottau)=0}$, where $V_\roottau=\freg(\F_\roottau)$ for all
    $\roottau$.
\smallskip
  \item when $P=N-e-1$, ${\sf Solve\_F}(q,\Gamma)$ outputs either
    one-dimensional paramet\-rizations over $\A$ or ${\sf fail}$ using
    $\softO( N^3(E+N^3) D \degQ \delta^2)$ operations in $\QQ$.  In
    case of success, the output describes the family
    $(V_\roottau)_{q(\roottau)=0}$, where $V_\roottau=\freg(\F_\roottau)$ for all
    $\roottau$.
  \end{itemize}
\end{proposition}

The proof of this proposition will occupy this paragraph. Given an
$(N-e)\times P$ matrix $\mS$ with entries in $\QQ$, we will denote by
$J_\mS$ the determinant of $\jac(\F)\mS$. Given such an $\mS$, for
$1\leq i \leq P$ and for $\roottau$ a root of $q$, we denote by $V_{i,
  \mS,\roottau} \subset \C^{N-e}$ the Zariski closure of
$V(\F_{i,\roottau})-V(J_{\mS,\roottau})$, with $\F_{i,\roottau}$ as defined in the
previous paragraph.  The algebraic sets $V_{i,\mS,\roottau}$ are simpler to
define than the sets $V_{i,\roottau}$ (we do not need to involve all
determinants in $\Delta_\roottau$); the following lemma shows that they
coincide for a generic choice of~$\mS$.

\begin{lemma}\label{lemma:VVi}
  There exists a non-empty Zariski open subset $\mathfrak{S}$ of
  $\C^{(N-e)P}$ such that for $\mS$ in $\mathfrak{S}$, for all $i$ in
  $\{1,\dots,P\}$ and all roots $\roottau$ of $q$,
    $V_{i,\mS,\roottau}=V_{i,\roottau}$ holds.
\end{lemma}
\begin{proof}
  Let us first fix a root $\roottau$ of $q$ and $i$ in
  $\{1,\dots,P\}$. Recall that by construction, $V_{i,\roottau}$ is the
  Zariski closure of $V(\F_{i,\roottau})-V(\Delta_\roottau)$, where
  $\Delta_\roottau$ is the ideal generated by all $P$-minors of
  $\jac(\F_\roottau)$, and $V_{i,\mS,\roottau}$ is the Zariski closure of
  $V(\F_{i,\roottau})-V(J_{\mS,\roottau})$. In what follows, we prove the
  slightly more general result: {\em let $\closedZ$ be {\em any} algebraic
    set in $\C^{N-e}$. Then, for a generic choice of $\mS$, the
    Zariski closures $\closedZ'$ and $\closedZ''$ of respectively $\closedZ-V(\Delta_\roottau)$
    and $\closedZ-V(J_{\mS,\roottau})$ coincide.}

  Let $\closedZ_{1},\dots,\closedZ_{\lambda(\roottau)}$ be the decomposition of
  $\closedZ$ into irreducible components. Then, $\closedZ'$ is the union of
  those $U_{k}$ that are not contained in $V(\Delta_\roottau)$, whereas
  $\closedZ''$ is the union of those that are not contained in
  $V(J_{\mS,\roottau})$. Thus, we have to prove that for a generic choice
  of $\mS$, for all $k$, $\closedZ_{k}$ is contained in $V(\Delta_\roottau)$ if
  and only if it is contained in $V(J_{\mS,\roottau})$.

  Suppose first that $\closedZ_{k}$ is contained in $V(\Delta_\roottau)$
  and let $\x$ be in $\closedZ_{k}$. By assumption, the Jacobian
  matrix $\jac(\F_\roottau)$ has rank less than $P$ at $\x$; thus, it is
  also the case for $\jac(\F_\roottau)\mS$, for any $\mS$ in
  $\QQ^{(N-e)P}$, so $\closedZ_{k}$ is contained in $V(J_\mS)$.  In
  other words, for {\em any} $\mS$, if $\closedZ_{k}$ is contained in
  $V(\Delta)$, it is contained in $V(J_\mS)$.

  Conversely, suppose that $\closedZ_{k}$ is not contained in
  $V(\Delta_\roottau)$, so there exists $\x$ in $\closedZ_{k}$ such that
  $\jac(\F_\roottau)$ has rank $P$ at $\x$. This implies that there exists
  $\mS$ in $\QQ^{(N-e)P}$ such that $\jac(\F_\roottau)\mS$ still has rank
  $P$ at $\x$, so for this particular choice of $\mS$, $\closedZ_{k}$
  is not contained in $V(J_{\mS,\roottau})$. The set of $\mS$ for which
  this holds is a Zariski open subset $\mathfrak{S}_{k,\roottau}$ of
  $\C^{(N-e)P}$ (because $J_\mS(\x)$ is a polynomial in $\mS$), that
  is non empty in view of the previous remark.

  Taking for ${\mathfrak{S}}$ the intersection of the finitely many
  Zariski open subsets
  $$\mathfrak{S}_{1,\roottau},\dots,\mathfrak{S}_{\pollambda(\roottau),\roottau},$$ for
  all roots $\roottau$ of $q$, proves our claim and hence the lemma.
\end{proof}

If $\mS$ satisfies the assumptions of the previous lemma, we obtain
the following alternative description for $V_{i+1,\roottau}$ from
$V_{i,\roottau}$. This shows that we will be able to apply the algorithm
of Subsection~\ref{ssec:inter} to the present situation.

\begin{lemma}\label{lemma:VVi2}
  Suppose that $\mS$ belongs to $\mathfrak{S}$. Then, for $0 \le i < P$,
  and for every root $\roottau$ of $q$,
  $V_{i+1,\roottau}$ is the Zariski closure of $V_{i,\roottau} \cap V(F_{i+1,\roottau})-V(J_{\mS,\roottau})$.
\end{lemma}
\begin{proof}
  Fix a root $\roottau$ of $q$ and $i$ in $\{1,\dots,P-1\}$. Under our
  assumption on $\mS$, the previous lemma shows that $V_{i,\roottau}$ and
  $V_{i+1,\roottau}$ are the Zariski closures of respectively
  $V(\F_{i,\roottau})-V(J_{\mS,\roottau})$ and
  $V(\F_{i+1,\roottau})-V(J_{\mS,\roottau})$.

  Let us write $V(\F_{i,\roottau})$ as $A \cup B$, where $A$, resp.\ $B$,
  is the union of the irreducible components of $V(\F_{i,\roottau})$ where
  $J_{\mS,\roottau}$ vanishes identically, resp.\ is not identically zero.
  As a result, $V_{i,\roottau}=B$. On the other hand, we deduce that
  $V(\F_{i+1,\roottau})= (A \cap V(F_{i+1,\roottau})) \cup (B \cap
  V(F_{i+1,\roottau}))$, so that $V_{i+1,\roottau}$ is the Zariski closure of
  $B \cap V(F_{i+1,\roottau})-V(J_{\mS,\roottau})$. Since we have seen that
  $B=V_{i,\roottau}$, the lemma is proved.
\end{proof}

The bulk of Algorithm ${\sf Solve\_F}$ is an incremental intersection
process: for $i=0,\dots,P-1$, we start from zero-dimensional
parametrizations over $\A$ for the sets $\fbr(V_{i,\roottau}^\mA,\y_i)$,
for some random $\y_i$ in $\QQ^{N-e-i}$ and $\mA$ in $\GL(N-e,\QQ)$
and deduce one-dimensional parametrizations over $\A$ for the sets
$\fbr(V_{i+1,\roottau}^\mA,\y_{i+1})$, where $\y_{i+1}$ is obtained from
$\y_i$ by discarding its last entry.

Assuming that $\mS$ belongs to $\mathfrak{S}$, the operation above
will be done by applying Algorithm ${\sf SolveIncremental}$ of
Proposition~\ref{prop:inter} to the sets $(V_{i,\roottau})$, the system
$\F_i$, $G=F_{i+1}$ and $H=J_\mS$; indeed,
Lemmas~\ref{sec:geosolve:lemma:equidim},~\ref{lemma:VVi}
and~\ref{lemma:VVi2} show that we are then under the assumptions of
this proposition. There is a slight difference, however, for $i=0$:
then, there are no equations to use for the lifting step of that
algorithm; in that case, it is straightforward to bypass the lifting
step and directly enter the intersection step.

As input, the algorithm of Proposition~\ref{prop:inter} requires
zero-dimensional paramet\-rizations over $\A$ for the sets
$\fbr(V_{i,\roottau}^\mA,\y_i)$, together with a straight-line program
that evaluates $F_1,\dots,F_{i}$, $G$ and $H$. What we are given is a
straight-line program $\Gamma$ of length $E$ for
$\F=F_1,\dots,F_P$. However, due to the definition of $J_\mS$, it is
easy to deduce a straight-line program $\Gamma'$ that computes $\F$
and $J_\mS$ of length $E'=O(NE+N^4)=O(N(E+N^3))$, where the first term
gives the cost of computing $\F$ and its Jacobian matrix, and the
extra $O(N^4)$ steps amount to computing the determinant giving
$J_\mS$ (which has degree at most $N D$). As a result, the cost of one
call to Proposition~\ref{prop:inter} is $\softO( N^2 (E+N^3) D \degQ
\delta^2)$.

Applying this $P$ times, we obtain zero-dimensional parametrizations
over $\A$ for the sets $(\fbr(V_\roottau^\mA,\y))_{q(\roottau)=0}$, for some
$\mA$ in $\GL(N-e,\QQ)$ and $\y$ in $\QQ^{N-e-P}$ using $\softO( P
N^2 (E+N^3) D \degQ \delta^2)$ operations in $\QQ$, which is $\softO(
N^3 (E+N^3) D \degQ \delta^2)$.

If $P=N-e$, each $V_\roottau$ is either zero-dimensional or empty, and the
set $\fbr(V_\roottau^\mA,\y)$ is simply equal to $V_\roottau^\mA$
itself. Thus, we can finally undo the change of variables $\mA$ by
using Algorithm ${\sf ChangeVariables}$ from
Lemma~\ref{lemma:complexity0POF:changevar}, using a negligible
$\softO(N^2\degQ \delta+N^3)$ operations in $\QQ$.  This proves the
first part of Proposition~\ref{geosolve:prop:variant1a}.

If $P=N-e-1$, each $V_\roottau$ is an algebraic curve, or it is
empty. Starting from the zero-dimensional parametrizations for the
sets $\fbr(V_\roottau^\mA,\y)$, where $\y$ is in $\QQ$, we first apply
Lemma~\ref{lemma:lift} in order to obtain one-dimensional
parametrizations over $\A$ for the sets $V_\roottau^\mA$ (the cost is
within the bounds given above). As above, we conclude with a change of
variables, using Algorithm ${\sf ChangeVariables}$ from
Lemma~\ref{lemma:complexity1POF:changevar}. The cost is
$\softO(N^2\degQ \delta^2+N^3)$ operations in $\QQ$ which is
negligible. This concludes the proof of
Proposition~\ref{geosolve:prop:variant1a}.

%%%%%%%%%%%%%%%%%%%%%%%%%%%%%%%%%%%%%%%%%%%%%%%%%%%%%%%%%%%%

\subsubsection{Solving \texorpdfstring{$\F=\G=0$}{$\F=\G=0$}}\label{ssec:solvingfg}

In this second paragraph, we discuss a refinement of the previous
question. In addition to $q$ and to the polynomials
$\F=(F_1,\dots,F_P)$ introduced previously, we also consider a family
of new polynomials $\G=(G_1, \ldots, G_t)$ in $\A[X_{e+1},\dots,X_N]$,
where we write as before $\A=\QQ[T]/\langle q\rangle$. Notation
for polynomials
$\F_\roottau$ or $\G_\roottau$ is as in the previous paragraphs.

Recall from Subsection~\ref{ssec:preliminaries:fixing} that for a root
$\roottau$ of $q$, $\oreg(\F_\roottau)$ is the set of all
$\x=(x_{e+1},\dots,x_N)$ in $V(\F_\roottau)$ where $\jac(\F_\roottau)$ has
full rank $P$. We are interested here in describing the sets
$(Y_\roottau)_{q(\roottau)=0}$, where for any root $\roottau$ of $q$, $Y_\roottau$ is
the set of {\em isolated points} of $\oreg(\F_\roottau) \cap V(\G_\roottau)
\subset \C^{N-e}$.

\begin{proposition}\label{geosolve:prop:variant2}
  There exists a probabilistic algorithm ${\sf Solve\_FG}$ that takes
  as input a squarefree polynomial $q$ and a straight-line program
  $\Gamma'$ with coefficients in $\A$, with the following
  characteristics.

  Suppose that $\Gamma'$ has length $E'$, computes polynomials $\F$
  and $\G$ of degree at most $D$, resp.\ $D'$, that $q$ has degree
  $\degQ$; let $\delta$ be the geometric degree of $\F$ and
  $\delta'=\delta {D'}^{N-e-P}$. Then ${\sf Solve\_FG}(q,\Gamma')$
  outputs either zero-dimensional parametrizations over $\A$ or
  ${\sf fail}$ using $\softO( N^3(tE'+tN+N^3)D''\degQ {\delta'}^2 )$
  operations in $\QQ$, with $D''=\max(D,D')$. In case of success, the
  output describes $(Y_\roottau)_{q(\roottau)=0}$, where $Y_\roottau$ is the set
  of isolated points of $\oreg(\F_\roottau) \cap V(\G_\roottau)$ for all
  $\roottau$.

  In addition, the degree of each set $Y_\roottau$ is bounded by $\delta'$.
\end{proposition}

In order to prove Proposition~\ref{geosolve:prop:variant2}, the
results of the previous paragraph cannot be applied directly, as we
do not restrict ourselves anymore to the points where the Jacobian of
the whole system $\F,\G$ has full rank. However, the fact that we only want
isolated solutions will allow us to find a workaround.

We start with the degree bound.  Let us first define as in
Subsection~\ref{ssec:definitionsF} the algebraic sets
$(V_\roottau)_{q(\roottau)=0}$, where $V_\roottau=\freg(\F_\roottau) \subset
\C^{N-e}$. In addition, we recall that for a root $\roottau$ of $q$,
$\mathcal{O}_\roottau$ is the Zariski open set $\C^{N-e}-V(\Delta_\roottau)$,
where $\Delta_\roottau$ is the set of $P$-minors of $\jac(\F_\roottau)$. Then, we
can establish the following easy statement.

\begin{lemma}\label{lemma:Ytau}
  For any root $\roottau$ of $q$, $Y_\roottau$ is the set of isolated points
  of $V_\roottau \cap V(\G_\roottau) \cap \mathcal{O}_\roottau$.
\end{lemma}
\begin{proof}
  By definition, $Y_\roottau$ is the set of isolated points of
  $\oreg(\F_\roottau) \cap V(\G_\roottau)$. Starting from the definition of
  $V_\roottau$ as the Zariski closure of $\oreg(\F_\roottau)=V(\F_\roottau) \cap
  \mathcal{O}_\roottau$, we obtain $V_\roottau \cap \mathcal{O}_\roottau =
  \oreg(\F_\roottau)$.  This implies that $V_\roottau\cap V(\G_\roottau) \cap
  \mathcal{O}_\roottau= \oreg(\F_\roottau) \cap V(\G_\roottau)$, and looking at
  the set of isolated points on both sides proves our claim.
\end{proof}

For any root $\roottau$ of $q$, $V_\roottau$ has by construction degree at
most $\delta$, and Lemma~\ref{sec:geosolve:lemma:equidim} shows that
it is either equidimensional of dimension $N-e-P$ or empty. As a
consequence, Proposition~2.3 in~\cite{HeSc80} implies that the degree
of $V_\roottau \cap V(\G_\roottau)$ is at most $\delta {D'}^{N-e-P}$. Using
the lemma above, this proves the first point in
Proposition~\ref{geosolve:prop:variant2}.

Let $\a=(a_{1,1},\dots,a_{N-e-P,t})$ be in $\QQ^{t(N-e-P)}$ and, for
$i$ in $\{1,\dots,N-e-P\}$, define
$$G'_i=a_{i,1} G_1 + \cdots+a_{i,t} G_t;$$ remark that in all that
follows, polynomials $G'_i$ and the algebraic sets they define depend
on the choice of $\a$, but we chose not to add a subscript to our
notation.

For any root $\roottau$ of $q$, we denote by $Y_{P,\roottau}$ the algebraic
set $V_\roottau=\freg(\F_\roottau)$ and, for $1\leq i \leq N-e-P$, we denote
by $Y_{P+i,\roottau}$ the union of the irreducible components of
$V_\roottau\cap V(G'_{1,\roottau}, \ldots, G'_{i,\roottau})$ of dimension
$N-e-(P+i)$ that have a non-empty intersection with $\mathcal{O}_\roottau$
(as before, the subscript indicates relative codimension). In
particular, for $i=N-e-P$, $Y_{N-e,\roottau}$ has dimension zero; we will
prove below that for a generic choice of $\a$, the equality
$Y_\roottau=Y_{N-e,\roottau} \cap V(\G_\roottau)$ holds.

For $i$ in $\{0,\dots, N-e-P\}$, the set $Y_{P+i,\roottau}$ is further
decomposed into
$$Y_{P+i,\roottau}^R \quad\text{and}\quad Y_{P+i,\roottau}^I,$$ where
$Y_{P+i,\roottau}^R$ (the regular part) is the union of all irreducible
components of $Y_{P+i,\roottau}$ that are not contained in
$V(G'_{i+1,\roottau})$ and $Y_{P+i,\roottau}^I$ (the irregular part) is the
union of all other irreducible components.

In what follows, we rely on the choice of an $(N-e)\times P$-matrix
$\mS$ with entries in $\QQ$, as in the previous paragraph.
\begin{lemma}\label{lemma:interG}
  For a generic choice of $\mS$, and for $i$ in $\{0,\dots,N-e-P-1\}$,
  the following holds for each root $\roottau$ of $q$:
  \begin{itemize}
  \item $Y_{P+i,\roottau}^R \cap V(G'_{i+1,\roottau})$ is either empty or
    equidimensional of dimension $N-e-(P+i+1)$;
\smallskip
  \item $Y_{P+i+1,\roottau}$ is the Zariski closure of $Y_{P+i,\roottau}^R \cap
    V(G'_{i+1,\roottau})-V(J_{\mS,\roottau})$;
\smallskip
  \item if $i < N-e-P-1$, $Y^R_{P+i+1,\roottau}$ is the Zariski closure of
    $Y_{P+i,\roottau}^R \cap V(G'_{i+1,\roottau})-V(J_{\mS,\roottau} G'_{i+2,\roottau})$.
  \end{itemize}
\end{lemma}
\begin{proof}
  In all that follows, we fix a root $\roottau$ of $q$.
  The first item is a direct consequence of the definition of
  $Y^R_{P+i,\roottau}$. Next, for $i=1,\dots,N-e-P-1$, write
  $$V_\roottau \cap V(G'_{1,\roottau},\dots,G'_{i,\roottau}) = Y^{R}_{P+i,\roottau}
  \cup Y^I_{P+i,\roottau} \cup Y^{\mathcal{O}_\roottau}_{P+i,\roottau} \cup
  Y^{d}_{P+i,\roottau},$$ where $Y^R_{P+i,\roottau}$ and $Y^I_{P+i,\roottau}$ are
  as above, $Y^{\mathcal{O}_\roottau}_{P+i,\roottau}$ is the union of the
  irreducible components of $Y_{P+i,\roottau}$ that do not intersect the
  open set $\mathcal{O}_\roottau$ and $Y^{d}_{P+i,\roottau}$ are all other
  irreducible components, which must have dimension greater that
  $N-e-(P+i)$.  Intersecting with $V(G'_{i+1,\roottau})$, we obtain that
  $V_\roottau \cap V(G'_{1,\roottau},\dots,G'_{i+1,\roottau})$ is the union of the
  following sets:
  $$ Y^R_{P+i,\roottau}\cap V(G'_{i+1,\roottau}),\ Y^I_{P+i,\roottau} \cap
  V(G'_{i+1,\roottau}),\ Y^{\mathcal{O}_\roottau}_{P+i,\roottau}\cap V(G'_{i+1,\roottau}),\ Y^{d}_{P+i,\roottau}\cap
  V(G'_{i+1,\roottau}).$$ The set $Y_{P+i+1,\roottau}$ is obtained by keeping only the
  irreducible components of the above sets that have dimension
  $N-e-(P+i+1)$ and that intersect $\mathcal{O}_\roottau$. The last three terms
  above do not contribute to this construction, so we deduce that
  $Y_{P+i+1,\roottau}$ is the union of the irreducible components of
  $Y^R_{P+i,\roottau}\cap V(G'_{i+1,\roottau})$ that intersect $\mathcal{O}_\roottau$.

  Because $\mathcal{O}_\roottau=\C^{N-e}-V(\Delta_\roottau)$, we deduce that
  $Y_{P+i+1,\roottau}$ is the Zariski closure of $Y^R_{P+i,\roottau}\cap
  V(G'_{i+1,\roottau})-V(\Delta_\roottau)$. As we saw in the proof of
  Lemma~\ref{lemma:VVi}, this means that $Y_{P+i+1,\roottau}$ is the
  Zariski closure of $Y^R_{P+i,\roottau}\cap
  V(G'_{i+1,\roottau})-V(J_{\mS,\roottau})$, for a generic choice of
  $\mS$. This proves the second item.

  If $i < N-e-P-1$, the definition of $Y^R_{P+i+1,\roottau}$ implies that it is
  obtained by discarding from $Y_{P+i+1,\roottau}$ all irreducible components
  on which $G'_{i+2,\roottau}$ vanishes identically; the last item follows.
\end{proof}

The previous lemma holds for any choice of $\a$. For a generic choice
of $\a$, the following lemma further gives a description of the sets
$V_\roottau \cap V(G'_{1,\roottau},\dots,G'_{i,\roottau})$.

\begin{lemma}\label{lemma:A87}
  For a generic choice of $\a$, the following holds for any root
  $\roottau$ of $q$. Let $i$ be in $\{1,\dots,N-e-P\}$ and let $Z$ be an
  irreducible component of $V_\roottau \cap
  V(G'_{1,\roottau},\dots,G'_{i,\roottau})$. Then, either $Z$ is contained in
  $V_\roottau \cap V(\G_\roottau)$, or the following two properties hold:
  \begin{itemize}
  \item $\dim(Z) = N-e-(P+i)$;
\smallskip
  \item for $\x$ in $Z\cap \mathcal{O}_\roottau-V(\G_\roottau)$,
    $\jac(\F_\roottau,G'_{1,\roottau},\dots,G'_{i,\roottau})$ has full rank $P+i$ at $\x$.
  \end{itemize}
\end{lemma}
\begin{proof}
  This is a restatement of the first two items of Theorem~A.8.7
  in~\cite{SoWa05} taking into account that for $\roottau$ as above, a
  point $\x$ in $V_\roottau \cap \mathcal{O}_\roottau$ is a regular point on $V_\roottau$.
\end{proof}

When $\a$ satisfies the assumptions of the previous lemma, the first
item in this lemma shows that for any root $\roottau$ of $q$, $V_\roottau \cap
V(G'_{1,\roottau},\dots,G'_{i,\roottau})$ is the union of $V_\roottau \cap
V(\G_\roottau)$ and (possibly) of some algebraic set of pure dimension
$N-e-(P+i)$. For $i=N-e-P$, we obtain in particular the following
result, as announced above.

\begin{lemma}\label{lemma:YfromYNe}
  For a generic choice of $\a$, and for any root $\roottau$ of $q$, the
  equality $Y_\roottau=Y_{N-e,\roottau} \cap V(\G_\roottau)$ holds.
\end{lemma}
\begin{proof}
  As usual, we fix a root $\roottau$ of $q$. Recall that we proved in
  Lemma~\ref{lemma:Ytau} that $Y_\roottau$ is the set of isolated points
  of $V_\roottau \cap V(\G_\roottau) \cap \mathcal{O}_\roottau$.

  On the other hand, taking $i=N-e-P$ in Lemma~\ref{lemma:A87}, we deduce
  that $V_\roottau \cap V(G'_{1,\roottau},\dots,G'_{N-e-P,\roottau})$ is the union
  of $V_\roottau \cap V(\G_\roottau)$ and of finitely many isolated
  points. Since $Y_{N-e,\roottau}$ is the set of isolated points in
  $V_\roottau \cap V(G'_{1,\roottau},\dots,G'_{N-e-P,\roottau}) \cap
  \mathcal{O}_\roottau$, we deduce that $Y_{N-e,\roottau}$ is the union of the
  finite set $Y_\roottau$ we are interested in and of some isolated
  points, say $Y'_\roottau$, that are not in $V(\G_\roottau)$. The conclusion
  follows.
\end{proof}

As a result, we are now going to show how to compute a description of
the sets $Y_{N-e,\roottau}$, since filtering out the undesired extra
points will raise no difficulty.  To this end, we follow the
intersection process of Subsection~\ref{ssec:inter}.

To start the process, we deal with equations $\F$ only. This is done
using the algorithm ${\sf Solve\_F}$ given in the previous paragraph;
we obtain zero-dimensional parametrizations over $\A$ for the finite
sets $\fbr(V_{P,\roottau}^\mA,\y)$, for $\roottau$ a root of $q$,
and for some $\mA$ in $\GL(N-e)$ and $\y$ in $\QQ^{N-e-P}$, using
$\softO( N^3 (E'+N^3) D''\degQ \delta^2)$ operations in $\QQ$. We then
remove all those points that cancel the polynomials $G'_{1,\roottau}$,
for $\roottau$ as above. For a generic choice of $\mA$ and $\y$, the
remaining points define the sets $\fbr({Y^R_{P,\roottau}}^\mA,\y)$.

This hardly impacts the running time: this last step is done using 
Algorithm {\sf Clean} of~\cite{GiLeSa01}, which we already used in
the proof of Lemma~\ref{lemma:inter}. The analysis made in 
that proof remains valid, and shows that this step takes
$\softO( (E'+t+N) D\degQ \delta))$ operations in $\QQ$,
since the cost of evaluating $G'_1$ is $O(E'+t)$.
We will bound the cost so far by 
$\softO( N^3 (E'+t+N^3) D''\degQ \delta^2)$.

Using the last claim in Lemma~\ref{lemma:interG}, the same process
allows us to compute zero-dimensional parametrizations over $\A$ of
witness points for the families of algebraic sets
$Y^R_{P,\roottau},\dots,Y^R_{N-e-1,\roottau}$; the last step is done
by applying the second claim in that lemma instead, giving us
zero-dimensional parametrizations for the sets
$(Y_{N-e,\roottau})_{q(\roottau)=0}$. Let us verify that at every
stage, we are indeed under the assumptions of
Proposition~\ref{prop:inter}:
\begin{itemize}
\item By construction, for any root $\roottau$ of $q$, $Y^R_{P+i,\roottau}$ is either empty or equidimensional
  of dimension $N-e-(P+i)$.
\smallskip
\item For any such $\roottau$, the polynomials
  $\F_\roottau,G'_{1,\roottau},\dots,G'_{i,\roottau}$ vanish on $Y^R_{P+i.\roottau}$,
  and we claim that for a generic choice of $\a$, the matrix
  $\jac(\F_\roottau,G'_{1,\roottau},\dots,G'_{i,\roottau})$ has generically full
  rank $P+i$ on each irreducible component $Z$ of $Y^R_{P+i,\roottau}$.
  The second item in Lemma~\ref{lemma:A87} ensures it: $Z$ cannot be
  contained in $V_\roottau\cap V(\G_\roottau)$ (otherwise, it would be
  contained in $V(G'_{i+1,\roottau})$, which we assume is not the case)
  and $Z \cap \mathcal{O}_\roottau-V(\G_\roottau)$ is non empty, so there
  exists $\x$ in $Z \cap \mathcal{O}_\roottau-V(\G_\roottau)$ where said
  Jacobian matrix has full rank.
\smallskip
\item $Y^R_{P+i,\roottau}\cap V(G'_{i+1,\roottau})$ is either empty or
  $(N-e-(P+1))$-equidimensional: this is the first item in
  Lemma~\ref{lemma:interG}.
\end{itemize}

In terms of complexity, remark that all $G'_1,\dots,G'_{N-e-P}$ can be
computed by a straight-line program of length $O(E'+tN)$, and that for
all $i \le N-e-P$ and for any root $\roottau$ in $q$, $Y^R_{P+i,\roottau}$ has
degree at most $\delta'=\delta {D'}^{N-e-P}$ (using again
Proposition~2.3 in~\cite{HeSc80}). As a result, the total cost is
$\softO( N^3 (E'+tN+N^3) D''\degQ {\delta'}^2)$ operations in $\QQ$.

At this stage, we have obtained a description of the sets
$(Y_{N-e,\roottau}^\mA)_{q(\roottau)=0}$ by means of pairs
$(q_1,\scrR_1),\dots,(q_s,\scrR_s)$. In view of
Lemma~\ref{lemma:YfromYNe}, we keep only the points on the sets
$Y_{N-e,\roottau}^\mA$ where $G_{1,\roottau}^\mA,\dots,G_{t,\roottau}^\mA$ all
vanish; this is done by applying $t$ times the Algorithm
${\sf Intersect}$ from
Lemma~\ref{sec:basicroutinesparam:lemma:intersect2}. The cost is
$\softO(t\degQ \delta' (E'+N^2))$, since evaluating $\G^\mA$ induces
an $O(N^2)$ additional cost in the straight-line program for $\G$;
this is negligible compared to the previous cost.

We are thus left with pairs of the form
$(q'_1,\scrR'_1),\dots,(q'_v,\scrR'_v)$ that form zero-dimensional
parametrizations over $\A$ for the sets $(Y_\roottau^\mA)_{q(\roottau)=0}$.
As in the previous paragraph, we use algorithm ${\sf ChangeVariables}$
from Lemma~\ref{lemma:complexity0POF:changevar} in order to obtain
zero-dimensional parametrizations over $\A$ for the sets
$(Y_\roottau)_{q(\roottau)=0}$, using $\softO(N^2\degQ \delta+N^3)$ operations
in $\QQ$, which is negligible. This concludes the proof of
Proposition~\ref{geosolve:prop:variant2}.

%%%%%%%%%%%%%%%%%%%%%%%%%%%%%%%%%%%%%%%%%%%%%%%%%%%%%%%%%%%%

\subsubsection{An application}\label{sec:posso:singularpoints}

We end this paragraph with a first application of the routine ${\sf
  Solve\_FG}$. Let $\f=(f_1, \ldots, f_p)\subset \QQ[X_1, \ldots,
X_n]$ be a reduced regular sequence defining an algebraic set $V(\f)
\subset \C^n$ such that $\sing(V(\f))$ is finite. We apply ${\sf
  Solve\_FG}$ to compute a zero-dimensional para\-metrization of
$\sing(V(\f))$. 

One possible approach would be to solve the system consisting of $\f$
and all $p$-minors of its Jacobian matrix. In the following
proposition, we use Lagrange systems instead, since it allows us to
obtain a slightly better cost.

\begin{proposition}\label{prop:singularpoints}
  Let $\Gamma$ be a straight-line program of length $E$ that computes
  a reduced regular sequence $\f=(f_1, \ldots, f_p)$, with
  $\deg(f_i)\leq D$ for all $i$, and such that $\sing(V(\f))$ is
  finite. Suppose that $D \ge 2$.

  There exists a probabilistic algorithm ${\sf SingularPoints}$ which
  takes as input $\f$ and either returns ${\sf fail}$ or returns a
  zero-dimensional parametrization using $\softO( E D^{4n+1} )$
  operations in $\QQ$. In case of success, the output describes
  $\sing(V(\f))$ and it has degree bounded by $n D^{2n}$.
\end{proposition}
\begin{proof}
  Consider new indeterminates $\L=(L_1,\dots,L_p)$, and the system
  $\G$ consisting of $\f$ and $\lag(\f,0,\L)$, where the second term
  denotes the entries of the matrix $[L_1~\cdots~L_p]\cdot \jac(\f)$.  The
  set we want to compute is the projection on the $\X$-space of the
  solutions of the system $\G=0$, $(L_1,\dots,L_p)\ne (0,\dots,0)$.
  We are going to reduce the solution of this set of equations and
  inequations to several instances of systems that can be solved by 
  means of Algorithm ${\sf Solve\_FG}$. 
  
  Let us partition $V(\f)$ into subset $(V_i)_{0 \le i \le p}$, where
  $V_i$ is the subset of all $\x$ in $V(\f)$ where $\jac(\f)$ has rank $i$;
  we are thus interested in describing $V_0,\dots,V_{p-1}$.
  Fix $i$ in $\{0,\dots,p-1\}$: at any such point, the solution set
  $S_\x$ of $\lag(\f,0,\L)$ is a linear subspace of $\C^p$
  of dimension $p-i$, so that the intersection of $S_\x$ with $(p-i-1)$ random
  linear forms $(\u_j \cdot \L=0)_{1 \le j \le p-i-1}$ and $1$ random
  affine form $\u_0 \cdot \L=1$ is a single point~$\bell_\x$.

  Let us thus introduce the systems $\G_i$, for $i=0,\dots,p-1$, where
  $\G_i$ consists of the $2p+n-i$ equations $\f$, $\lag(\f,0,\L)$,
  $(\u_j \cdot \L=0)_{1 \le j \le p-i-1}$ and $\u_0 \cdot \L=1$.
  We claim that for a generic choice of all $\u_j$'s, the isolated 
  points of $V(\G_i) \subset \C^{n+p}$ are precisely those points 
  $(\x,\bell_\x)$, for $\x$ in $V_i$. 

  Take a point $(\x,\bell)$ in $V(\G_i)$. If the Jacobian matrix
  $\jac(\f)$ had full rank at $\x$, we would necessarily have
  $\bell=0$, a contradiction with the constraint $\u_0 \cdot
  \L=1$. Hence, $\x$ is in $\sing(V(\f))$. Suppose in addition that
  $(\x,\bell)$ is isolated in $V(\G_i)$: this implies that $\bell$ is
  an isolated solution of the linear system $\lag(\f,0,\L)|_{\X=\x}$,
  $(\u_j \cdot \L=0)_{1 \le j \le p-i-1}$, $\u_0 \cdot \L=1$: since the 
  $\u_j$'s are chosen generic, this implies that $\jac(\f)$ has rank
  $p-i$ at $\x$, and $\x$ is indeed in $V_i$.

  Conversely, the discussion of the previous paragraphs shows that any
  point $(\x,\bell_\x)$, for $\x$ in $V_i$, is indeed a solution of
  $\G_i$; we have to prove that it is isolated. We saw above that any
  point $(\x',\bell')$ in $V(\G_i)$ is in $\sing(V(\f))$, and $\x$ is 
  isolated in  $\sing(V(\f))$ (as this set is finite). By construction,
  $\bell_\x$ is isolated among the solutions of $\lag(\f,0,\L)|_{\X=\x}$,
  $(\u_j \cdot \L=0)_{1 \le j \le p-i-1}$, $\u_0 \cdot \L=1$, so we are done 
  with the proof of our claim.

  Let $\F$ be the empty set. The algorithm calls Algorithm ${\sf Solve\_FG}$
  of Proposition \ref{geosolve:prop:variant2} $p$ times, with inputs
  (say) $q=X$, $e=0$ and straight-line programs $\Gamma_0,\dots,\Gamma_{p-1}$
  that respectively evaluate the polynomials $\G_0,\dots,\G_{p-1}$. Since 
  there are no polynomials $\F$, in each
  case, we obtain the isolated solutions of $V(\G_i)$; then, we project
  them on the $\X$-space and return the union of the corresponding 
  finite sets of points.

  For a given index $i$, the polynomials in $\G_i$ involve $n+p \le
  2n$ variables, have total degree at most $D$, and can be computed by
  a straight-line program of length $O(nE + n^2)$, where the first
  term corresponds to the overhead induced by the calculation of all
  partial derivatives of $\f$, and the second one to all dot products.
  Because we assume $D\ge 2$, we can neglect polynomials in $n$
  compared to terms of the form $D^n$ in our soft-O estimates.  For
  each index $i$, the cost of Proposition~\ref{geosolve:prop:variant2}
  then becomes $\softO( E D^{4n+1} )$, and the bound on the degree of
  the output is $D^{2n}$; in particular, the sum of the output degrees
  is at most $nD^{2n}$.  The total time spent in the subsequent
  projection and union operations (Lemmas~\ref{sec:main:lemma:union}
  and~\ref{sec:posso:lemma:projection}) is then $\softO(D^{2n})$.
\end{proof}

%%%%%%%%%%%%%%%%%%%%%%%%%%%%%%%%%%%%%%%%%%%%

%%%%%%%%%%%%%%%%%%%%%%%%%%%%%%%%%%%%%%%%%%%%
%\input{proof6.3}
\section{Proof of Proposition~\ref{chap:solvelagrange:prop:basicsolve}}
\label{proof:solvelagrange:prop:basicsolve}

In this section, we prove
Proposition~\ref{chap:solvelagrange:prop:basicsolve}.  We consider a
generalized Lagrange system $L=(\Gamma, \scrQ, \scrS)$ of type $(k,
\n, \p, e)$, where $\Gamma$ is a straight-line program of length $E$
that computes polynomials $\F=(\f, \f_1, \ldots, \f_k)$, with
$\f\subset \QQ[\X]$ and $\f_i\subset \QQ[\X, \L_1, \ldots, \L_i]$ for
$1 \le i \le k$. As in Definition~\ref{def:GLS}, we write $d=N-e-P$;
we let $D$ denote the maximum degree of the polynomials in $\f$,
$\delta=\DF(k, e,\n,\p,D, D-1)$ is as in
Definition~\ref{sec:solvelagrange:notationsNPdelta}. Finally,
we write 
$Q=\Zeroes(\scrQ)\subset \C^e$ and $S=\Zeroes(\scrS)\subset \C^n$, 
as well as
$\degQ=\deg(\scrQ)$ and $\degS = \deg(\scrS)$.

With this notation, we prove the following: {\em There exists a
  probabilistic algorithm ${\sf SolveLagrange}$ which takes as input a
  generalized Lagrange system $L$ as above, such that $N-e-P=1$, and
  returns either a one-dimensional parametrization with coefficients
  in $\QQ$ or ${\sf fail}$ using
  $$\softO(N^3(E+N^3) (D+k) \degQ^3\delta^3 + N\degQ
  \delta\degS^2)$$ operations in~$\QQ$, using the notation introduced
  above. If either
  \begin{itemize}
  \item $\Clos{(L)}$ is empty,
\smallskip
  \item or $L$ has a global normal form,
  \end{itemize}
  then in case of success, the output of ${\sf SolveLagrange}$
  describes $\Clos{(L)}$.  In addition, $\Clos{(L)}$ has degree at most
  $\degQ \delta$.} 

%%%%%%%%%%%%%%%%%%%%%%%%%%%%%%%%%%%%%%%%%%%%%%%%%%%%%%%%%%%%

\subsection{Algorithm ${\sf IsEmpty}$}

We start by an auxiliary function for testing emptiness.

\begin{proposition}\label{prop:isempty}
  There exists a probabilistic algorithm ${\sf IsEmpty}$ which takes
  as input a generalized Lagrange system $L$ and returns either
  ${\sf true}$, ${\sf false}$ or ${\sf fail}$ using
  $\softO( N^3(E+N^3) (D+k) \degQ \delta^2+ N \degQ^2 \delta^2 + N
  \degS^2)$
  operations in~$\QQ$, using the notation introduced above. If either
  \begin{itemize}
  \item $\Clos{(L)}$ is empty,
\smallskip
  \item or $L$ has a global normal form,
  \end{itemize}
  then in case of success, ${\sf IsEmpty}$ decides whether $\Clos{(L)}$
  is empty.
\end{proposition}

Before proving this proposition, we introduce notation that will be 
useful below. Let us write $\scrQ=((q,v_1,\dots,v_e),\pollambda)$, define
$\A=\QQ[T]/\langle q \rangle$, and let $\tilde \F$ be the polynomials
$\F(v_1,\dots,v_e,X_{e+1},\dots,X_{N})$, that lie in
$\A[X_{e+1},\dots,X_{N}]$.  Recall that we assume that polynomials
$\F$ are given by a straight-line program $\Gamma$; replacing all
inputs $X_1,\dots,X_e$ by $v_1,\dots,v_e$ in $\Gamma$, we obtain a
straight-line program $\tilde \Gamma$ with coefficients in $\A$ that computes
the polynomials $\tilde \F$.  The following lemma gives an upper bound
on the geometric degree (see Definition~\ref{def:gdeg}) of these polynomials in terms of $\delta$.

\begin{lemma}\label{lemma:geometricdegree}
  The geometric degree of $\tilde \F$ is at most $\delta$.
\end{lemma}
\begin{proof}
  The definition of generalized Lagrange systems implies that
  all inequalities in~\eqref{eq:N-P2} are satisfied. Thus, applying
  Proposition~\ref{sec:posso:prop2} to the systems $\tilde
  \F_\roottau=\phi_\roottau(\tilde \F)$ (as defined in
  Section~\ref{sec:D5eqs}), for $\roottau$ a root of $q$, proves our 
  inequality.
\end{proof}

The other notation we will need is the following. Let $\Delta$ be the
set of maximal minors of $\jac(\F, e)$, let $\mathcal{O}$ be the Zariski
open set $\C^{N}-V(\Delta)$ and let finally $V=\freg(\F,Q)$ be the
Zariski closure of $\fbr(V(\F), Q)\cap \mathcal{O}$.  Recall as well that
we denote by $\pi_\X:\C^{N} \to \C^n$ the projection on the
$\X$-space.

\begin{proof}[of Proposition~\ref{prop:isempty}]
 Choose $d$ random linear forms $\Lambda$ with coefficients in $\QQ$
 in all variables $\X,\L_1,\dots,\L_k$, and let $\F'$ be the system
 obtained by adjoining $\Lambda$ to $\F$. Just as we defined $V$ as
 the Zariski closure of $\fbr(V(\F),Q) \cap \mathcal{O}$, we define
 $V'=\freg(\F',Q)$ as the Zariski closure of $\fbr(V(\F'),Q) \cap
 \mathcal{O'}$, where $\mathcal{O}'$ is the Zariski open set
 $\C^{N}-V(\Delta')$ and $\Delta'$ is the set of maximal minors of
 $\jac(\F', e)$. Remark that $\F'$ consists of $P+d = N -e$
 equations, so that $\jac(\F', e)$ is actually square of size $N-e$,
 and $\Delta'$ simply consists in the determinant of that matrix. In
 particular, by Proposition~\ref{geosolve:prop:variant1a}, $V'$ is a
 finite set, so we can alternatively define it as $V'=\fbr(V(\F'),Q)
 \cap \mathcal{O'}$.

 Under the assumptions that either $\Clos{(L)}$ is empty or $L$ has a
 global normal form, we are going to prove that for a generic choice
 of $\Lambda$, $V'$ is contained in $\pi_\X^{-1}(S)$ if and only if
 $\Clos{(L)}$ is empty. The condition on $V'$ will be tested using
 Algorithm ${\sf Solve\_F}$ introduced in Section~\ref{chap:posso}.

 Suppose first that $\Clos{(L)}$ is empty. In this case, $\Cons(L)$ is
 empty as well, which implies that $\fbr(V(\F),Q)$ is contained in
 $\pi_\X^{-1}(S)$. As a result, $V'$, which is a subset of
 $\fbr(V(\F),Q)$, is contained in $\pi_\X^{-1}(S)$ as well.

 Suppose on the other hand that $L$ has a global normal form. By
 Lemma~\ref{lemma:degree:clos}, $V$ is equidimensional of dimension
 $d$ and it does not lie over $S$ (since otherwise, the third
 equality in that lemma would imply that $\Clos{(L)}$ is empty, whereas
 it establishes that $\Clos{(L)}$ is $d$-equidimensional). As a
 consequence, for a generic choice of $d$ linear forms $\Lambda$, $V
 \cap V(\Lambda)$ is a non-empty finite set, not contained in
 $\pi_\X^{-1}(S)$. To conclude this discussion, we will now prove that
 in this case, for generic $\Lambda$, $V'=V \cap V(\Lambda)$ (so that,
 as claimed above, $V'$ is not contained in $\pi_\X^{-1}(S)$).

 Take $\x$ in $V'$, so that $\x$ is in $\fbr(V(\F'),Q)$ and
 $\jac(\F',e)$ has full rank $N-e$ at $\x$. This implies that $\x$
 is in $\fbr(V(\F),Q)$ and that $\jac(\F,e)$ has full rank
 $N-e-d=P$ at $\x$, so $\x$ is in $\fbr(V(\F), Q)\cap \mathcal{O}$,
 and thus in $V$. Since $\x$ also cancels the linear forms $\Lambda$,
 $\x$ is in $V \cap V(\Lambda)$.  Conversely, for a generic choice of
 $\Lambda$, every point $\x$ in $V \cap V(\Lambda)$ is non-singular on
 $V$, and $V(\Lambda)$ intersects $V$ transversally at $\x$ (this is
 for instance a consequence of~\cite[Theorem~A.8.7]{SoWa05}). For such
 an $\x$, $\T_\x V$ is the nullspace of $\jac(\F,e)$ at $\x$, so the
 transversality condition means that $\jac(\F',e)$ has full rank
 $N-e$ at $\x$.  This proves that $\x$ is in $V'$.

 As announced above, the discussion in the last paragraphs shows that
 for a generic choice of $\Lambda$, and under the assumption that
 either $\Clos{(L)}$ is empty or $L$ has a global normal form, $V'$ is
 contained in $\pi_\X^{-1}(S)$ if and only if $\Clos{(L)}$ is
 empty. Algorithm ${\sf IsEmpty}$ is then simple.  Starting from
 polynomials $\F'$, we define
 $\tilde \F'=\F'(v_1,\dots,v_e,X_{e+1},\dots,X_{N})$, so that these
 polynomials lie in $\A[X_{e+1},\dots,X_{N}]$. As was pointed out in
 Section~\ref{sec:D5eqs}, $V'$ is the disjoint union of the sets
 $\x \times V'_\roottau$, for $\x$ in $Q$, where $\roottau=\pollambda(\x)$ is a
 root of $q$ and $V'_\roottau=\freg(\tilde \F'_\roottau)$.

 Thus, we use Algorithm ${\sf Solve\_F}$ of
 Proposition~\ref{geosolve:prop:variant1a}, with input $q$ and (a
 straight-line program for) $\tilde \F'$. Upon success, the output is
 a family of zero-di\-men\-sio\-nal parametri\-zations over $\A$ of
 the form $(q_1,\scrR_1),\dots,(q_s,\scrR_s)$ for the sets
 $(V'_\roottau)_{q(\roottau)=0}$, where each $\scrR_i$ has the form
 $\scrR_i=((r_i,w_{i,e+1},\dots,w_{i,N}),\linearmu_i)$, and has coefficients
 in $\A_i=\QQ[T]/\langle q_i\rangle$. We can then define the
 zero-dimensional parametrizations
 \[\scrR'_i=((r_i,v_1 \bmod q_i,\dots,v_e \bmod
 q_i,w_{i,e+1},\dots,w_{i,N}),\linearmu_i),\]
 for $1\leq i \leq s$ so that $(q_1,\scrR'_1),\dots,(q_s,\scrR'_s)$
 are zero-dimensional parametrizations over $\A$ for the sets
 $$((v_1(\roottau),\dots,v_e(\roottau))\times V'_\roottau)_{q(\roottau)=0}.$$  Using
 Algorithms ${\sf Descent}$ from Lemma~\ref{lemma:descent0} and
 ${\sf Union}$ from Lemma~\ref{sec:main:lemma:union}, we obtain a
 zero-dimensional parametrization $\scrR'$ of degree $\degQ\delta$
 with coefficients in $\QQ$ that defines the union of these sets, that
 is, $V'$. Finally, we can test whether $V'=\Zeroes(\scrR')$ is contained in
 $\pi_\X^{-1}(S)$ using Algorithm ${\sf Lift}$ from
 Lemma~\ref{lemma:subroutine:Lift}.

 Let us give the cost of all these steps. The system $\F'$ can be
 computed by a straight-line program $\Gamma'$ of length
 $E'=E+O(N^2)$, where the second term stands for the cost of
 computing linear forms $\Lambda$. From this, we can deduce a
 straight-line program $\tilde \Gamma'$ that computes polynomials
 $\tilde \F'$ with the same number of steps, by replacing all inputs
 $X_1,\dots,X_e$ by $v_1,\dots,v_e$ in $\Gamma'$.

 If all polynomials $\f$ have degree at most $D$, then all polynomials
 in $\F$ and $\F'$ have degree at most $D+k$. Finally, the geometric
 degree $\delta'$ of $\tilde\F'$ is less than or equal that of 
 $\tilde\F$, since all additional equations are linear. Since
 we saw above that the latter is at most  $\delta$, we deduce that the
 cost of calling ${\sf Solve\_F}(q, \tilde \Gamma')$ is
 $\softO( N^3(E+N^3) (D+k) \degQ \delta^2)$ operations in $\QQ$.  The
 total cost of all calls to ${\sf Descent}$, ${\sf Union}$ and
 ${\sf Lift}$ is $\softO(N \degQ^2 \delta^2 + N \degS^2)$.
\end{proof}

%%%%%%%%%%%%%%%%%%%%%%%%%%%%%%%%%%%%%%%%%%%%%%%%%%%%%%%%%%%%

\subsection{Proof of the proposition}

We can now prove Proposition~\ref{chap:solvelagrange:prop:basicsolve}.
First, we call ${\sf IsEmpty}$ (Proposition~\ref{prop:isempty}): if
the output is true, we simply return the one-dimen\-sional
parametrization that defines the empty set; the cost
$\softO( N^3(E+N^3) (D+k) \degQ \delta^2+ N \degQ^2 \delta^2 + N
\degS^2)$
will be negligible compared to that of other steps. Else, we may
assume that there exists a global normal form for $L$. Then, by
Lemma~\ref{lemma:degree:clos}, $\Clos{(L)}$ is the Zariski closure of
$\pi_\X(V-\pi_\X^{-1}(S))$, with  $V=\freg(\F,Q)$. By definition of the geometric degree
(Definition~\ref{def:gdeg}),
and using Lemma~\ref{lemma:geometricdegree}, we
obtain that $V$ has degree at most $\degQ\delta$; as a
consequence, the degree of $\Clos{(L)}$ admits the same upper bound.

  In order to compute a one-dimensional parametrization of $\Clos{(L)}$,
  we first apply the routine ${\sf Solve\_F}$ given in
  Proposition~\ref{geosolve:prop:variant1a} to $q$ and the
  straight-line program $\tilde \Gamma$ that computes $\tilde \F$.
  This gives us one-dimensional parametrizations over $\A$ for the
  sets $(V_\roottau)_{q(\roottau)=0}$, with $V_\roottau = \freg(\tilde \F_\roottau)$,
  and the cost is $\softO( N^3(E+N^3) (D+k)\degQ \delta^2)$ operations
  in $\QQ$. As in the proof of the previous lemma, we apply next
  Algorithms ${\sf Descent}$ and ${\sf Union}$, but in their
  one-dimensional versions (Lemmas~\ref{lemma:descent1}
  and~\ref{sec:main:lemma:union1}); the cost is
  $\softO(N \degQ^3 \delta^3)$ operations in~$\QQ$.

  As output, we obtain a one-dimensional parametrization of $V$ with
  coefficients in $\QQ$, and we saw above that it has degree at most
  $\degQ\delta$. Discarding those points in $V$ whose image by
  $\pi_\X$ lies in $S$ is done using the routine ${\sf Discard}$ of
  Lemma~\ref{sec:basicroutinesparam:lemma:discarddim1}. This requires
  $\softO(N\max(\degQ \delta, \degS)^2)$ arithmetic operations in
  $\QQ$ at most and the extra cost is bounded by
  $\softO(N \degQ^3 \delta^3 + N \degQ \delta \degS^2)$.

  The last step of this algorithm applies projection $\pi_\X$, by
  means of algorithm ${\sf Projection}$ from
  Lemma~\ref{sec:posso:lemma:projection1}; the cost is
  $\softO(N\degQ^3 \delta^3)$ operations in $\QQ$. The cost given in
  this lemma is an upper bound on all costs seen so far.
%% \end{proof}

%%%%%%%%%%%%%%%%%%%%%%%%%%%%%%%%%%%%%%%%%%%%

%%%%%%%%%%%%%%%%%%%%%%%%%%%%%%%%%%%%%%%%%%%%
%\input{proof6.4}
\section{Proof of Proposition~\ref{sec:main:critical}}
\label{proof:main:critical}

We prove now Proposition~\ref{sec:main:critical}. The setup is exactly
as in the previous section: we consider a generalized Lagrange system
$L=(\Gamma, \scrQ, \scrS)$ of type $(k, \n, \p, e)$, where $\Gamma$ is
a straight-line program of length $E$ that computes polynomials
$\F=(\f, \f_1, \ldots, \f_k)$, with $\f\subset \QQ[\X]$ and
$\f_i\subset \QQ[\X, \L_1, \ldots, \L_i]$ for $1 \le i \le k$. We
write $d=N-e-P$, $D$ is the maximum degree of the polynomials in $\f$,
$\delta=\DF(k, e,\n,\p,D, D-1)$.  Finally, we write $Q=\Zeroes(\scrQ) \subset \C^e$
and $S=\Zeroes(\scrS)\subset \C^n$, as well as $\degQ=\deg(\scrQ)$ and $\degS =
\deg(\scrS)$.

Then, we prove the following: {\em There exists a probabilistic
  algorithm ${\sf W}_1$ which takes as input a generalized Lagrange
  system $L$ as above and returns either a zero-dimensional
  parametrization with coefficients in $\QQ$ or ${\sf fail}$ using
  $$\softO( (k+1)^{2d+1}  N^{4d+8}E D^{2d+1} \degQ^2 \delta^2 +N\degS^2)$$
 operations in~$\QQ$. If either  $\Clos{(L)}$ is empty, or
  \begin{itemize}
    \item $\Clos{(L)}$ is $d$-equidimensional (so that $\polar(e, 1,
    \Clos{(L)})$ is well-defined),
\smallskip
  \item $\polar(e, 1, \Clos{(L)})$ is finite,
\smallskip
  \item $(L; \polar(e, 1, \Clos{(L)}))$ has a global normal form,
  \end{itemize}
  then in case of success, the output of ${\sf W}_1$ describes
  $\polar(e,1,\Clos{(L)})-S$.  In addition, the finite set
  $\polar(e,1,\Clos{(L)})-S$ has degree at most $\degQ\delta N^d (D-1+k)^d$.
}

\begin{lemma}\label{sec:lagrange:prop:critoflagrange}\label{chap:lagrange:defcrit}
  Let $Q \subset \C^e$ be a finite set and let $V \subset \C^n$ and
  $S\subset \C^n$ be algebraic sets lying over $Q$, with $S$
  finite. Suppose that $V$ is $d$-equidimensional with finitely many
  singular points.

  Let further $L=(\Gamma,\scrQ,\scrS)$ be a generalized Lagrange
  system, with $V=\Clos{(L)}$, $Q=Z(\scrQ)$, $S=Z(\scrS)$, $\F$ in
  $\QQ[\X, \L]$ as in Definition \ref{def:GLS} and define $d=N-e-P$.
  Suppose that $(L; \polar(e,1,V))$ has the global normal form property and
  that $\polar(e,1,V)$ is finite, and let $\G$ be the set of $P$-minors of
  $\jac(\F, e+1)$. Let finally $Z$ be the isolated points of
  $\oreg(\F,Q) \cap V(\G)$.  Then, $\polar(e,1,V)- S=\pi_\X(Z)-S.$
\end{lemma}
\begin{proof}
  We denote by $\lcLambda$ the locally closed set
  $$\fbr(V(\F,\G),Q) - \pi_\X^{-1}(S) = \Cons(L) \cap V(\G).$$ First,
  we prove that $\polar(e,1,V)-S = \pi_\X(\lcLambda)$.
By assumption, there exists a global normal form
$$\bphi=(\phi_i)_{1\leq i \leq s}$$ of $(L; \polar(e,1,V))$ with
$\phi_i=(\polmu_i, \poldelta_i, \h_i, \H_i)$. We claim that
$\polar(e,1,V)-S$ is contained in the union of the open sets
$\Open(\polmu_i \poldelta_i)$.  Indeed, take $\x$ in
$\polar(e,1,V)-S$, so $\x$ is in particular in $\polar(e,1,V)$. Since
$\polar(e,1,V)$ is by assumption finite, $\x$ is actually an
irreducible component of $\polar(e,1,V)$. Besides, since $\x$ is in
$V-S$, ${\gnf_2}$ implies that there exists $i$ in $\{1,\dots,s\}$
such that $\x$ is actually in $\Open(\polmu_i)\cap V-S$; by
${\gnf_3}$, this implies that $\poldelta_i$ does not vanish at $\x$,
as claimed.

We start by proving that $\pi_\X(\lcLambda)\subset \polar(e,1,V)$;
this will actually prove that
$\pi_\X(\lcLambda)\subset \polar(e,1, V)-S$, since the projection
$\pi_\X(\lcLambda)$ avoids $S$.  Let thus $(\x, \bell)$ be in
$\lcLambda$.  Then, $(\x, \bell)$ is in $\Cons(L)$, and $\x$ is in
$\new{\Proj(L)} \subset V-S$.  We deduce by $\gnf_2$ and $\lnf_5$ that
there exists $i \in \{1,\dots,s\}$ such that $\x$ is in
$\Open(\polmu_i \poldelta_i)\cap \new{\Proj(L)}$.

Denote by $I$ the defining ideal of $Q$. By Lemma~\ref{lemma:mS},
there exists a $(P\times P)$ matrix $\mS$ with entries in
$\QQ[\X]_{\polmu_i \poldelta_i}$ such that
$\jac(\H_i, e)=\mS\,\jac(\F, e)$ over
$\QQ[\X,\L]_{\polmu_i\poldelta_i}/\langle \F , I\rangle$. Since, by
definition of $\lcLambda$, $\jac(\F, e+1)$ has rank less than $P$ at
$(\x, \bell)$, we deduce that $\jac(\H_i, e+1)$ also has rank less
than $P$ at $(\x, \bell)$. Since $\H_i$ is in normal form, we conclude
that $\jac(\h_i,e+1)$ has rank less than $c$ at $\x$.  As a result,
since $\x$ is in particular in $\Open(\polmu_i)\cap V-S$,
Lemma~\ref{sec:atlas:chartpolar} shows that $\x$ is in
$\polar(e,1,V)$.

Conversely, we prove that $\polar(e,1, V)-S$ is contained in
$\pi_\X(\lcLambda)$. Let thus $\x$ be in $\polar(e,1, V)-S$. In view
of our preliminary remarks, we know that there exists
$i \in \{1,\dots,s\}$ such that $\x$ is in
$\Open(\polmu_i \poldelta_i)$. Since $\x$ is also in $V-S$,
Lemma~\ref{lemma:conseq0} implies that $\x$ is in $\new{\Proj(L)}$. As
a result, there exists $\bell$ such that $(\x,\bell)$ is in
$\Cons(L)$. It remains to prove that $\jac(\F, e+1)$ has rank less
than $P$ at $(\x, \bell)$.

By ${\lnf_3}$, $(\x,\bell)$ is in $\fbr(V(\H_i),Q)$. On the other
hand, as we saw above, there exists a $(P\times P)$ matrix $\mS$ with
entries in $\QQ[\X]_{\polmu_i \poldelta_i}$ such that
$\jac(\H_i, e)=\mS\,\jac(\F, e)$ over
$\QQ[\X,\L]_{\polmu_i \poldelta_i}/\langle \F , I\rangle$. Thus, to
prove that $\jac(\F, e+1)$ has rank less than $P$ at $(\x, \bell)$, it
is enough to prove that
\begin{itemize}
\item the determinant of $\mS$ does not vanish at $\x$;
\smallskip
\item $\jac(\H_i, e+1)$ has rank less than $P$ at $(\x, \bell)$. 
\end{itemize}
We start with the first assertion. By properties $\lnf_4$ and
$\sfC_4$, we deduce that $\jac(\h_i, e)$ has full rank $c$ at $\x$;
the last statement in Lemma~\ref{lemma:mS} then implies that
$\det(\mS)$ is non-zero at $\x$, as claimed. We now prove the second
assertion. Because $(\polmu_i, \h_i)$ is a chart of $(V, Q, S)$, and
$V$ is $d$-equidimensional with finitely many singular points, one can
apply Lemma \ref{sec:atlas:chartpolar} to $V$ and deduce that
$\jac(\h_i, e+1)$ has rank less than $c$ at $\x$. Using again the fact
that $\h_i$ is the $\X$-component of $\H_i$, and that $\H_i$ is in
normal form, we deduce that $\jac(\H_i, e+1)$ has rank less than $P$
at $(\x,\bell)$, as requested.

At this stage, we have proved that $$\polar(e,1,V)-S =
\pi_\X(\lcLambda),\qquad \text{ with }\qquad \lcLambda=\fbr(V(\F,\G),Q) -
\pi_\X^{-1}(S).$$ Next, we prove that $\lcLambda$ is finite and that
$\jac(\F,e)$ has full rank $P$ at every point in~$\lcLambda$.

We saw above that $\polar(e, 1, V)-S$ is contained the union of the
open sets $\Open(\polmu_i \poldelta_i)$ and thus (by
Lemma~\ref{lemma:conseq0}) in $\new{\Proj(L)}$.  Using again the global normal form
property, one can apply Proposition~\ref{sec:lagrange:propertyP} and deduce
that $\pi_\X$ induces a bijection between $\polar(e, 1, V)-S$ and its
preimage $\pi_\X^{-1}(\polar(e, 1, V)-S) \cap \Cons(L)$, so that in
particular, $\pi_\X^{-1}(\polar(e, 1, V)-S) \cap \Cons(L)$ is finite;
that lemma proves as well that $\jac(\F, e)$ has maximal rank at any
point of that set.  Applying $\pi_\X^{-1}$ to both sides of the
equality $\polar(e,1,V)- S=\pi_\X(\lcLambda)$, and using the fact that
$\lcLambda$ is contained in $\Cons(L)$, we deduce that
$\pi_\X^{-1}(\polar(e, 1, V)-S) \cap \Cons(L)=\lcLambda$, so we are
done with the claims above.

The fact that that $\jac(\F,e)$ has full rank $P$ at every point
in~$\lcLambda$ implies that $\lcLambda$ can be rewritten as
$\lcLambda = \oreg(\F,Q)\cap V(\G)-\pi_\X^{-1}(S)$. Now, the locally
closed set $\oreg(\F,Q)\cap V(\G)$ can be written as
$\oreg(\F,Q)\cap V(\G) = Z \cup T$, with $Z$ being its isolated points
and $T$ the union of all components of positive dimension, and where
the union is disjoint.  As a consequence, we have
$\lcLambda = (Z - \pi_\X^{-1}(S))\ \cup\ (T - \pi_\X^{-1}(S))$. Now,
if $T - \pi_\X^{-1}(S)$ is not empty, it must be infinite, so
$\lcLambda$ being finite implies that
$\lcLambda = Z - \pi_\X^{-1}(S)$, and we are done.
\end{proof}

%% \begin{proof}[of Proposition \ref{sec:main:critical}]
As in the previous section, we define $\A=\QQ[T]/\langle q \rangle$,
and let $\tilde \F$ be the polynomials
$\F(v_1,\dots,v_e,X_{e+1},\dots,X_{N})$, that lie in
$\A[X_{e+1},\dots,X_{N}]$.  Recall that we assume that polynomials
$\F$ are given by a straight-line program $\Gamma$; replacing all
inputs $X_1,\dots,X_e$ by $v_1,\dots,v_e$ in $\Gamma$, we obtain a
straight-line program $\tilde \Gamma$ with coefficients in $\A$ that
computes the polynomials $\tilde \F$.

The algorithm starts by checking whether $\Clos{(L)}$ is empty, using
algorithm ${\sf IsEmpty}$ (Proposition~\ref{prop:isempty}); the cost
$\softO( N^3(E+N^3) (D+k) \degQ \delta^2 + N \degQ^2 \delta^2 + N
\degS^2)$
of this step will be negligible (or of the same order) compared to
that of what follows. If $\Clos{(L)}$ is empty, we return the
zero-dimensional parametrization $(1)$ that defines (by convention) the 
empty set, and we are done.

We can thus assume that $\Clos{(L)}$ lies over $Q$ and is
$d$-equidimensional, so that
$\polar(e, 1, \Clos{(L)})$ is well-defined; we also assume that
$\polar(e, 1, \Clos{(L)})$ is finite and that
$(L; \polar(e, 1, \Clos{(L)}))$ has a global normal form.  
In particular, all singular points of $\Clos{(L)}$ are 
contained in $S=\Zeroes(\scrS)$ by Lemma~\ref{sec:coro:lemma:singSX}, 
so they are in finite number. 

Let $\G$ be
the set of $P$-minors of $\jac(\F, e+1)$ and denote by $Z$ the
isolated points of $\oreg(\F,Q) \cap V(\G)$; then,
Lemma~\ref{chap:lagrange:defcrit} shows that
  $$ \polar(e, 1, \Clos{(L)})-S= \pi_\X(Z)-S.  $$ Let us define the
  polynomials $$\tilde \G = \G(v_1,\dots,v_e,X_{e+1},\dots,X_{N})$$
  which lie in $\A[X_{e+1},\dots,X_{N}]$, as do the polynomials
  $\tilde \F$. The definition of $Z$ then shows that it can be written
  as the disjoint union of the sets
  $Z_\roottau = \x(\roottau) \times \zeta_\roottau$, where $\roottau$ is a root of $q$
  and $\x(\roottau)=(v_1(\roottau),\dots,v_e(\roottau))$, and $\zeta_\roottau$ is the
  set of isolated points of
  $\oreg(\tilde \F_\roottau) \cap V(\tilde \G_\roottau)$.
  
  To compute a zero-dimensional parametrization of
  $\polar(e, 1, \Clos{(L)})-S$, we first call the routine ${\sf Solve\_FG}$
  of Proposition~\ref{geosolve:prop:variant2} with input $q$ and a
  straight-line program that evaluates $\tilde \F$ and $\tilde \G$;
  this outputs zero-dimensional parametrizations over $\A$ for the
  sets $(\zeta_\roottau)_{q(\roottau)=0}$, of the form
  $(q_1,\scrR_1),\dots,(q_s,\scrR_s)$; each $\scrR_i$ has the form
  $\scrR_i=((r_i,w_{i,e+1},\dots,w_{i,N}),\linearmu_i)$.

  As in Proposition~\ref{prop:isempty}, we can then define the
  zero-dimensional parametrizations
  $$\scrR'_i=((r_i,v_1 \bmod q_i,\dots,v_e\bmod q_i,w_{i,e+1},\dots,w_{i,N}),\linearmu_i),$$ so
  that $(q_1,\scrR'_1),\dots,(q_s,\scrR'_s)$ are zero-dimensional
  parametrizations over $\A$ for the sets $(Z_\roottau)_{q(\roottau)=0}$.

  Using Algorithms ${\sf Descent}$ from Lemma \ref{lemma:descent0} and
  ${\sf Union}$ from Lemma \ref{sec:main:lemma:union}, we obtain a
  zero-dimensional parametrization $\scrR'$ with coefficients in $\QQ$
  that defines the union of these sets, that is, $Z$.

  Next, we use the routine ${\sf Projection}$
  of~Lemma~\ref{sec:posso:lemma:projection} to obtain a
  zero-dimensional pa\-ra\-met\-rization of $\pi_\X\left (Z\right )$.
  Finally, we use the routine ${\sf Discard}$
  of~Lemma~\ref{sec:basicroutinesparam:lemma:discard} to compute a
  ze\-ro-dimensional parametrization of $\pi_\X(Z)-S$.

  First, we establish the degree bound on $\polar(e, 1,
  \Clos{(L)})-S$.
  Note that the degrees of the polynomials in $\G$ and $\Delta$ are at
  most $D'=N(D+k-1)$, since $\G$ and $\Delta$ are minors of size at
  most $N$ of matrices with polynomial entries of degrees at most
  $D+k-1$. By Proposition~\ref{geosolve:prop:variant2}, we deduce that
  each $\zeta_t$, or equivalently each $Z_t$, has degree at most
  $\delta {D'}^{d}$. Then, the finite set $Z$ has degree at most
  $\degQ \delta {D'}^{d}$; the same holds for $\pi_\X(Z)-S$, and thus
  for $\polar(e, 1,\Clos{(L)})-S$. This concludes the proof for our
  degree bounds.

  By differentiating every step in $\Gamma$, we deduce from it a
  straight-line program that computes both $\F$ and its Jacobian
  matrix using $O(N E)$ operations. There are 
  $$t={\binom{N-e-1}{P}} \leq (N-e-1)^{N-e-1-P}\leq N^d$$
  polynomials in $\G$. Using Berkowitz' determinant algorithm (which
  evaluates any minor in $\G$ using $O(N^4)$ steps), we obtain a
  straight-line program $\Gamma'$ evaluating $\F$ and $\G$ of length
  $E'=O( N^{d + 4} + N E)$. As in the previous propositions, we
  evaluate $X_1,\dots,X_e$ at $v_1,\dots,v_e$ in $\Gamma'$; this
  results in a straight-line program $\tilde \Gamma'$ of length $E'$,
  with coefficients in $\A$, for the polynomials $\tilde \F$ and
  $\tilde \G$.  Using Proposition~\ref{geosolve:prop:variant2} we
  deduce that we can run Algorithm ${\sf Solve\_FG}$ with input $q$
  and $\tilde \Gamma'$ in
  $$
  \softO( N^3(tE'+tN+N^3)D''\degQ \delta^2  {D'}^{2d}  )
  $$ operations in $\QQ$, with $D''=\max(D, D')=\max(D, N(D-1+k))$.
  Since $ t \le N^d$, we deduce that $tN\leq N^{d+1}$ and
  $tE'=O(N^{2d+4} + N^{d+1}E)$. Using the obvious inequality $D+k-1 \leq (k+1) D$
  that holds for $k \ge 0$ and $D \ge 1$, 
  and its consequence $D' \le (k+1) D N $, we
  obtain
  $$N^3(tE'+tN+N^3){D''} = O( k (N^{2d+8} + N^{d+5}E ) D) = O (k  N^{2d+8} E D)$$
  and
  $$
  {D'}^{2d}\leq (k+1)^{2d} N^{2d} D^{2d}. 
  $$
Incorporating these inequalities in the above
complexity estimate and using some straightforward simplifications, we
obtain that the cost of the first step is bounded by 
$$
\softO( (k+1)^{2d+1}  N^{4d+8}E D^{2d+1}\degQ^2 \delta^2). 
$$ Denoting by $(q_1,\scrR_1),\dots,(q_s,\scrR_s)$ the
zero-dimensional parametrizations returned by the first step, the
degree estimates given above show that each $\scrR_i$ has degree at
most $\delta {D'}^d$.  We deduce that the cost of applying Algorithm
${\sf Descent}$ to any given pair $(q_i,\scrR_i)$ is
$\softO(N \degQ_i^2 \delta^2 {D'}^{2d})$, with $\degQ_i=\deg(q_i)$;
the total cost adds up to a negligible
$\softO(N \degQ^2 \delta^2 {D'}^{2d})$.  The same estimate holds for
applying Algorithm ${\sf Union}$; for ${\sf Projection}$, the total
cost is $\softO(N^2 \degQ^2 \delta^2 {D'}^{2d})$.

At this stage, we have a zero-dimensional parametrization of
$\pi_\X(Z)$.  Finally, Lemma
\ref{sec:basicroutinesparam:lemma:discard} shows that removing those
points in $Z$ that lie in $S$ can be done in
$\softO(N\max(\degQ\delta {D'}^{d}, \degS)^2)$ operations in $\QQ$;
the extra cost is thus $\softO(N\degS^2)$.  Summing up these
estimates, we obtain the announced cost.
%% \end{proof}

%%%%%%%%%%%%%%%%%%%%%%%%%%%%%%%%%%%%%%%%%%%%

%%%%%%%%%%%%%%%%%%%%%%%%%%%%%%%%%%%%%%%%%%%%
%\input{proof6.5}
\section{Proof of Proposition~\ref{sec:main:prop:fiber}}
\label{proof:prop:fiber}

In this section, we prove Proposition~\ref{sec:main:prop:fiber}.  Let
us repeat the definition of the main objects it deals with: we
consider a generalized Lagrange system $L=(\Gamma, \scrQ, \scrS)$ of
type $(k, \n, \p, e)$, where $\Gamma$ is a straight-line program of
length $E$ that computes polynomials $\F=(\f, \f_1, \ldots, \f_k)$,
with $\f\subset \QQ[\X]$ and $\f_i\subset \QQ[\X, \L_1, \ldots, \L_i]$
for $1 \le i \le k$. We let $d=N-e-P$, $D$ be the maximum degree of
the polynomials in $\f$ and $\delta=\DF(k, e,\n,\p,D, D-1)$.  We write
$Q=\Zeroes(\scrQ) \subset \C^e$ and $S=\Zeroes(\scrS)\subset \C^n$, as
well as $\degQ=\deg(\scrQ)$ and $\degS = \deg(\scrS)$.

With these definitions, we prove the following:  {\em There exists a probabilistic algorithm ${\sf Fiber}$
  which takes as input a generalized Lagrange system
  $L=(\Gamma,\scrQ,\scrS)$ of type $(k, \n, \p, e)$ and a
  zero-dimensional parametrization $\scrQ''$ of degree $\degQ''$,
  defining a finite set of points $Q''\subset \C^{e+d}$ lying over
  $Q=\Zeroes(\scrQ)$, and which returns either a zero-dimensional
  parametrization with coefficients in $\QQ$ or ${\sf fail}$ using
  $$\softO( N^3(N E+N^3)D {\degQ''}^2 \delta^2 + N\degS^2 )$$
  operations in~$\QQ$, using the notation introduced
  above. If either
  \begin{itemize}
  \item $\Clos{(L)}$ is empty,
\smallskip
  \item or $\fbr(\Clos{(L)}, Q'')$ is finite and $(L; \fbr(\Clos{(L)}, Q''))$
    has a global normal form,
  \end{itemize}
  then in case of success, the output of ${\sf Fiber}$ describes
  $\fbr(\Clos{(L)}, Q'')-S$.  In addition, $\fbr(\Clos{(L)}, Q'')-S$ has
  degree at most $\degQ''\delta$.
}

\begin{lemma}\label{chap:lagrange:computefibers}
  Let $Q \subset \C^e$ be a finite set and let $V \subset \C^n$ and
  $S\subset \C^n$ be algebraic sets lying over $Q$, with $S$ finite. 

  Let further $L=(\Gamma,\scrQ,\scrS)$ be a generalized Lagrange
  system, with $V=\Clos{(L)}$, $Q=Z(\scrQ)$, $S=Z(\scrS)$, $\F$ in
  $\QQ[\X, \L]$ as in Definition \ref{def:GLS} and define $d=N-e-P$.

  Let $\fiber2 \subset \C^{e+d}$ be a finite set lying over $Q$ and
  suppose that $\fbr(V, \fiber2)$ is finite and that $(L; \fbr(V, \fiber2))$ has
  the global normal form property.  Let finally $Z'$ be the isolated
  points of $\fbr(\oreg(\F,Q),\fiber2)$.  Then, $\fbr(V, \fiber2)-S=\pi_\X(Z')-S$.
\end{lemma}
\begin{proof}
  Let $\lcLambda$ be the locally closed set
  $$\fbr(\fbr(V(\F),Q),\fiber2)-\pi_\X^{-1}(S).$$ We first prove that
  $\fbr(V, \fiber2)-S=\pi_\X(\lcLambda)$. Note from the outset that
  $\lcLambda$ can be rewritten as $\lcLambda=\fbr(\Cons(L),\fiber2)$.

  Since there exists a global normal form for $(L; \fbr(V, \fiber2))$
  and $\fbr(V, \fiber2)$ is finite, we can prove as in
  Lemma~\ref{sec:lagrange:prop:critoflagrange} that
  $\fbr(V, \fiber2)-S$ is contained in $\Proj(L)$, and thus that
  $\fbr(V, \fiber2)-S$ is contained in $\fbr(\Proj(L),\fiber2)$. On
  the other hand, $\Proj(L)$ is contained in $V-S$, so that
  $\fbr(\Proj(L),\fiber2)$ is contained in $\fbr(V,\fiber2)-S$; we can
  thus conclude that $\fbr(V, \fiber2)-S=\fbr(\Proj(L), \fiber2)$.  As
  a consequence, we get, as claimed above:
  \begin{eqnarray*}
    \fbr(V, \fiber2)-S &= & \fbr(\Proj(L), \fiber2) \\
    & = & \fbr(\pi_\X(\Cons(L)), \fiber2)\\ 
    & = & \pi_\X(\fbr(\Cons(L), \fiber2))\\ 
    & = & \pi_\X(\lcLambda).
  \end{eqnarray*}
  To conclude, it will thus be enough to prove that
  $\lcLambda=Z'-\pi_\X^{-1}(S)$. We start by proving that proving that
  $\lcLambda$ is finite and that $\jac(\F,e)$ has full rank $P$ at
  every point in~$\lcLambda$.

  Using again the global normal form property, one can apply
Proposition~\ref{sec:lagrange:propertyP}, to deduce that
  $\fbr(\Cons(L), \fiber2)$ is in one-to-one correspondence with
  $\fbr(\Proj(L), \fiber2)$. Since
  $\fbr(\Proj(L), \fiber2)=\fbr(V, \fiber2)-S$, and $\fbr(V,\fiber2)$
  is finite by assumption, we deduce that
  $\lcLambda=\fbr(\Cons(L), \fiber2)$ is finite.  Using again
  Proposition~\ref{sec:lagrange:propertyP}, we also conclude that
  $\jac(\F, e)$ has maximal rank at any point in $\Cons(L)$ and thus
  in particular at every point in $\lcLambda$; our claims above are
  thus proved.

  As in the proof of Lemma~\ref{chap:lagrange:defcrit}, the latter fact implies that
  we can rewrite $\lcLambda$ as
  $\lcLambda=\fbr(\oreg(\F,Q),\fiber2)-\pi_\X^{-1}(S)$, and the fact that
  $\lcLambda$ is finite allows us to prove that
  $\lcLambda=Z'-\pi_\X^{-1}(S)$, where $Z'$ is the set of isolated
  points of $\fbr(\oreg(\F,Q),\fiber2)$.
\end{proof}

%% \begin{proof}[of Proposition \ref{sec:main:prop:fiber}]
  As in the previous propositions, we start by checking whether
  $\Clos{(L)}$ is empty, using algorithm ${\sf IsEmpty}$; the cost is
  $\softO(N^3(E+N^3) (D+k) \degQ \delta^2 + N \degQ^2 \delta^2 + N
  \degS^2)$.
  If $\Clos{(L)}$ is empty, we return the zero-dimensional
  parametrization that defines the empty set, and we are done.

  Else, we can assume that $\fbr(\Clos{(L)}, Q'')$ is finite and that
  $(L; \fbr(\Clos{(L)}, Q''))$ has a global normal form. We are thus
  under the assumptions of Lemma~\ref{chap:lagrange:computefibers}. If
  we define as in that lemma the set $Z' \subset \C^{N}$ as the set
  of isolated points of $\fbr(\oreg(\F,Q),Q'')$, then that lemma shows
  that $\fbr(\Clos{(L)}, Q'')-S=\pi_\X(Z')-S$. Because $Q''$ lies over
  $Q$, the set $\fbr(\oreg(\F,Q),Q'')$ can be rewritten as the set of
  all points in $V(\F)$ that lie over $Q''$ and at which $\jac(\F,e)$
  has full rank $P$.

  Let us write $\scrQ''=((q',v'_1,\dots,v'_{e+d}),\pollambda')$, and
  define the product of fields $\A'=\QQ[T]/\langle q'\rangle$, as well
  as the polynomials $\bar \F=\F(v'_1,\dots,v'_e,X_{e+1},\dots,X_{N})$
  in $\A'[X_{e+1},\dots,X_{N}]$. We also define the polynomials $\bar
  \G=(\bar G_{e+1},\dots,\bar G_{e+d})$, with, for all $i$, $\bar G_i
  = X_i - v'_i \in \A'[X_{e+1},\dots,X_{N}]$.  For a root $\roottau$ of
  $q'$, let us then write $\zeta'_\roottau \subset \C^{N-e}$ for the set
  of isolated points of $\oreg(\bar \F_\roottau)\cap V(\bar \G_\roottau)$, and
  write $Z'_\roottau = (v'_1(\roottau),\dots,v'_e(\roottau)) \times \zeta'_\roottau
  \subset \C^{N}$.  Then, using the last remark in the previous
  paragraph, one verifies that $Z'$ is the disjoint union of the sets
  $Z'_\roottau$, for $\roottau$ a root of $q'$.

  Since all polynomials $\bar \G$ have degree $1$,
  Proposition~\ref{geosolve:prop:variant2} applied to $\bar \F$ and
  $\bar \G$ implies that each $\zeta'_\roottau$ has degree at most
  $\delta$; this is thus also the case for the sets $Z'_\roottau$, so that
  $Z'$ has degree at most $\degQ'' \delta$. This implies that the same
  inequality also holds for $\fbr(\Clos{(L)}, Q'')-S$, as claimed.

  To compute a zero-dimensional parametrization encoding
  $\fbr(\Clos{(L)}, Q'')-S$, we first call the routine ${\sf Solve\_FG}$
  of Proposition~\ref{geosolve:prop:variant2} with input $q'$ and a
  straight-line program that evaluates $\bar \F$ and $\bar \G$; this
  outputs zero-dimensional parametrizations over $\A'$ for the sets
  $(\zeta'_\roottau)_{q'(\roottau)=0}$, of the form
  $(q'_1,\scrR_1),\dots,(q'_s,\scrR_s)$; each $\scrR_i$ has the form
  $\scrR_i=((r_i,w_{i,e+1},\dots,w_{i,N}),\linearmu_i)$.

  We continue as in the previous proposition: we define the
  zero-dimensional parametrizations $\scrR'_i=((r_i,v'_1 \bmod
  q'_i,\dots,v'_e\bmod q'_i,w_{i,e+1},\dots,w_{i,N}),\linearmu_i)$, so that
  $(q'_1,\scrR'_1),\dots,(q'_s,\scrR'_s)$ are zero-dimensional
  parametrizations over $\A'$ for the sets $(Z'_\roottau)_{q'(\roottau)=0}$.

  Using Algorithms ${\sf Descent}$ from Lemma~\ref{lemma:descent0} and
  ${\sf Union}$ from Lemma~\ref{sec:main:lemma:union}, we obtain a
  zero-dimensional parametrization $\scrR'$ with coefficients in $\QQ$
  that defines the union $Z'$ of these sets. Next, we use routine
  ${\sf Projection}$ of~Lemma~\ref{sec:posso:lemma:projection} to
  obtain a zero-dimensional parametrization of
  $\pi_\X\left (Z'\right )$, and ${\sf Discard}$ of
  Lemma~\ref{sec:basicroutinesparam:lemma:discard} to compute a
  zero-dimensional parametrization of $\pi_\X(Z')-S$.

From the straight line program $\Gamma$ for $\F$, we can deduce a
straight-line program $\bar \Gamma$ over $\A'$ for both $\bar \F$ and
$\bar \G$: we substitute as usual $X_1,\dots,X_e$ by
$v'_1,\dots,v'_e$, and we add $O(N)$ operations that compute the
equations $X_i-v'_i$, for $i=e+1,\dots,e+d$.  Since all polynomials in
$\bar \F$ and $\bar \G$ have degree at most $D$, and since $\bar \G$
contains at most $N$ polynomials, the cost given by
Proposition~\ref{geosolve:prop:variant2} is $\softO( N^3(N
E+N^3)D \degQ'' \delta^2 )$ operations in $\QQ$.

Because all parametrizations $\scrR_i$ have degree at most $\delta$,
the cost of applying ${\sf Descent}$ and ${\sf Union}$ is
$\softO(N {\degQ''}^2 \delta^2)$, and the cost of applying
${\sf Projection}$ is $\softO(N^2 {\degQ''}^2 \delta^2)$. Applying
${\sf Discard}$ takes $\softO(N \max({\degQ''} \delta, \degS)^2)$
operations in $\QQ$ at most which is boun\-ded by
$\softO(N ({\degQ''} \delta+ \degS)^2)$.  Summing up the costs of all
these steps yields the announced result.
%% \end{proof}

%%%%%%%%%%%%%%%%%%%%%%%%%%%%%%%%%%%%%%%%%%%%

%%%%%%%%%%%%%%%%%%%%%%%%%%%%%%%%%%%%%%%%%%%%
% \input{main-algo}
%%%%%%%%%%%%%%%%%%%%%%%%%%%%%%%%%%%%%%%%%%%%

%%%%%%%%%%%%%%%%%%%%%%%%%%%%%%%%%%%%%%%%%%%%
%\input{proof7.2}
\section{Proof of Proposition~\ref{prop:correctnessglobal}}\label{chap:correctness}

This section is devoted to prove of
Proposition~\ref{prop:correctnessglobal}, which establishes the
correctness of algorithm ${\sf MainRoadmapLagrange}$. In
Subsection~\ref{ssec:correctness}, we defined a binary tree $\scrT$
that describes the trace of algorithm ${\sf RoadmapRec}$, with nodes
denoted by $\nodetau$. We reuse this construction for
Proposition~\ref{prop:correctnessglobal}, whose statement is as
follows.

{\em   
  Consider polynomials $\f=f_1,\dots,f_p$ in $\QQ[X_1,\dots,X_n]$,
  given by a straight-line program $\Gamma$, that define a reduced
  regular sequence.

  Suppose that $V=V(\f) \subset \C^n$ has finitely many singular
  points and that $V(\f)\cap \R^n$ is bounded. Consider also a
  zero-dimensional parametrization $\scrC_0$ that describes a finite
  set $C_0 \subset \C^n$.

  Suppose that the matrices $(\mA_\nodetau)_{\nodetau  \text{~internal node of~} \scrT}$
  satisfy the assumptions of Theorem~\ref{THEO:MAINABSTRACT}. Then,
  there exists a family of non-empty Zariski open sets
  $\scrIopen_\nodetau\subset \C^{P_\nodetau}$, for $\nodetau$ 
  an internal node of $\scrT$, such that the following holds.

  Consider vectors $(\u_\nodetau)_{\nodetau \text{~internal node of~}
    \scrT}$, with $\u_\nodetau$ in $\QQ^{P_\nodetau}$ for all
  $\nodetau$.  If, for all internal nodes $\nodetau$ of $\scrT$,
  $\u_\nodetau$ is in $\scrIopen_\nodetau$, $\mA_\nodetau$
and $\u_\nodetau$ are used in 
  the corresponding recursive call of ${\sf RoadmapRecLagrange}$, 
 and if all calls to
  subroutines such as ${\sf Union}$, ${\sf Projection}$, ${\sf W_1}$,
  ${\sf Lift}$ are successful, then ${\sf MainRoadmapLagrange}(\Gamma,
  \scrC_0)$ returns a roadmap of $(V, C_0)$.
 }

\medskip

The algorithm ${\sf MainRoadmapLagrange}$ performs a call to ${\sf
  RoadmapRecLagrange}$, just as the abstract algorithm ${\sf MainRoadmap}$
does to ${\sf RoadmapRec}$. We already established correctness of
${\sf RoadmapRec}$ through Theorem~\ref{THEO:MAINABSTRACT}, where we
defined the Zariski open sets $\scrOpen_\nodetau\subset \GL(n,
e_\nodetau)$ for $\nodetau$ an internal node of $\scrT$.

The strategy of our proof of correctness for ${\sf
  RoadmapRecLagrange}$ is then to prove that it computes the same
objects as ${\sf RoadmapRec}$, assuming in the whole section that we
take $\dalgo=\lfloor(d+3)/2\rfloor$.  We prove that this claim holds
if $\mA_\nodetau$ is in $\scrOpen_\nodetau$ for all internal nodes
$\nodetau$ of $\scrT$, and if the vector $\u_\nodetau$ is
well-chosen. As we previously did, we proceed by induction on the
depth of $\nodetau$. We will introduce an induction assumption which
is the counterpart of the induction assumption $\sfTa$ given in
Subsection \ref{chap:abstractalgo:objects}; proving this new property
at a node $\nodetau$ will now depend on the choice of vector
$\u_\nodetau$.

%%%%%%%%%%%%%%%%%%%%%%%%%%%%%%%%%%%%%%%%%%%%%%%%%%

\subsection{Basic constructions}

Let us start by reviewing the construction of the objects attached to
the binary tree $\scrT$. Let $\Gamma$ and $\scrC_0$ be the input of
${\sf MainRoadmapLagrange}$, where $\Gamma$ computes polynomials
$\f=(f_1,\dots,f_p)$ in $\QQ[X_1, \ldots, X_n]$, that define
$V=V(\f)\subset \C^n$. We suppose that $\f$ forms a reduced regular
sequence, that $\sing(V)$ is finite and $V \cap \R^n$ is bounded.  Let
finally $d=n-p$ and $\bpsi$ be the atlas of $(V, \bullet, \sing(V))$
given by $\bpsi=(\psi)$, with $\psi=(1,\f)$.

As in ${\sf MainRoadmapLagrange}$, we define
$$\scrS={\sf SingularPoints}(\Gamma)\qquad \text{ and }\qquad \scrC =
{\sf Union}(\scrC_0,\scrS),$$ so that $\Gamma$ and $\scrC$ are the
input to the recursive algorithm ${\sf RoadmapRecLagrange}$; thus, we
have that $C=\Zeroes(\scrC)$ satisfies $C = C_0 \cup \sing(V)$, with
$C_0=\Zeroes(\scrC_0)$. Accordingly, on input $(V,C_0)$, algorithm
${\sf MainRoadmap}$ indeed calls ${\sf RoadmapRec}$ with input $V$ and
$C$.

Each node $\nodetau$ of the tree $\scrT$ is labelled by integers
$(d_\nodetau, e_\nodetau)$.  Let now $(\mA_\nodetau)_{\nodetau
  \text{~internal node of~} \scrT}$ be a family of matrices, with
$\mA_\nodetau$ in $\GL(n, e_\nodetau, \QQ)$ for all $\nodetau$. 
 We
saw in the proof of Theorem~\ref{THEO:MAINABSTRACT} that there exist
non-empty Zariski open sets $\scrOpen_\nodetau\subset \GL(n,
e_\nodetau)$ for all internal nodes $\nodetau$ of $\scrT$, with the
following properties: Suppose that $\mA_{\nodetau}$ belongs to
$\scrOpen_{\nodetau}$ for all internal nodes $\nodetau$ of
$\scrT$. Then, we associate to each node $\nodetau$ of $\scrT$ the
objects
$(V_{\nodetau},Q_\nodetau,S_\nodetau,C_\nodetau,\bpsi_\nodetau)$,
which satisfy the following:
\begin{enumerate}
\smallskip
\item[${\sf t_{1}.}$] $Q_\nodetau$ is a finite subset of $\C^{e_\nodetau}$ and $S_\nodetau,C_\nodetau$
  are finite subsets of $\C^n$;
\smallskip
\item[${\sf t_{2}.}$] $V_\nodetau,S_\nodetau,C_\nodetau$ lie over $Q_\nodetau$;
\smallskip
\item[${\sf t_{3}.}$] either $V_\nodetau$ is empty, or $V_\nodetau$ lies
  over $Q_\nodetau$ and is $d_\nodetau$-equidimensional with finitely many
  singular points, in which case $\bpsi_\nodetau$ is an atlas of
  $(V_\nodetau,Q_\nodetau,S_\nodetau)$;
\smallskip
\item[${\sf t_{4}.}$] the inclusion $S_\nodetau \subset C_\nodetau$ holds.
\end{enumerate}
In addition, in these conditions, algorithm ${\sf MainRoadmap}$
returns a roadmap of its input $(V,C_0)$. In algorithm ${\sf
  RoadmapRec}$, we also defined algebraic sets
$B_\nodetau,Q''_\nodetau,C'_\nodetau,C''_\nodetau,W_\nodetau=\polar(e_\nodetau,\dalgo_\nodetau,V_\nodetau^{\mA_\nodetau})$
and $V''_\nodetau=\fbr(V_\nodetau^{\mA_\nodetau}, Q''_\nodetau)$.

In what follows, as in the statement of
Proposition~\ref{PROP:CORRECTNESSGLOBAL}, we assume that
$\mA_{\nodetau}$ indeed belongs to $\scrOpen_{\nodetau}$ for all
internal nodes $\nodetau$ of $\scrT$, so that the above conclusions
hold.

For the analysis of ${\sf RoadmapRecLagrange}$, we now associate to
each node $\nodetau$ of $\scrT$ a family of algebraic sets
$\mathscr{Y}_{\nodetau}$, all contained in $V_\nodetau$; this will allow us to 
specify some global normal form properties that will be needed below
(see property ${\sf t_{3}'}$).

We start by leaves, since it is then straightforward: for these nodes,
$\scrY_\nodetau$ is empty. Consider next two internal nodes $\nodetau,\kappa$
in $\scrT$, such that $\kappa$ is one of the descendants of $\nodetau$ (we
count $\nodetau$ as one of its own descendants), and let
$\nodetau_1=\nodetau,\dots,\nodetau_m=\kappa$ be the path from $\nodetau$ to $\kappa$
in $\scrT$. Let further $\mB_{\nodetau,\kappa}=\mA_{\nodetau_1} \cdots
\mA_{\nodetau_{m}} \in \GL(n,\QQ)$ be the product of all matrices from
$\nodetau$ to $\kappa$, so that applying the inverse of
$\mB_{\nodetau,\kappa}$ puts the geometric objects associated to $\kappa$
in the coordinate system considered at $\nodetau$. Then, we define
$$\mathscr{Y}_{\nodetau,\kappa} = \left \{
W_\kappa^{\mB_{\nodetau,\kappa}^{-1}}, \quad
\polar(e_\kappa,1,W_\kappa)^{\mB_{\nodetau,\kappa}^{-1}},\quad
{\fbr(W_\kappa,Q''_\kappa)}^{\mB_{\nodetau,\kappa}^{-1}}, \quad
{V''_\kappa}^{\mB_{\nodetau,\kappa}^{-1}} \right \}. $$ Finally, for a
given node $\nodetau$ of $\scrT$, we denote by $\mathscr{Y}_{\nodetau}$ the
union of all $\mathscr{Y}_{\nodetau,\kappa}$, for $\kappa$ a descendant of
$\nodetau$. By construction, $\mathscr{Y}_\nodetau$ is thus a finite family of
algebraic sets, that are all contained in~$V_\nodetau$. It is important to
note that the sets $\mathscr{Y}_\nodetau$ only depend on the input $(V,C)$
and the changes of variables $\mA_\nodetau$. Note as well that for an
internal node $\nodetau$, $\scrY_\nodetau$ is the union of
\begin{itemize}
\item the sets $ W_\nodetau^{\mA_{\nodetau}^{-1}},\
  %% \polar(e_\nodetau,1,V_\nodetau^{\mA_\nodetau})^{\mB_{\nodetau}^{-1}},\ 
  \polar(e_\nodetau,1,W_\nodetau)^{\mA_{\nodetau}^{-1}},\
  {\fbr(W_\nodetau,Q''_\nodetau)}^{\mA_{\nodetau}^{-1}}, \ 
  {V''_\nodetau}^{\mA_{\nodetau}^{-1}}$,
\smallskip
\item the sets $\scrY_{\nodetau'}^{\mA_\nodetau^{-1}}$ and
  $\scrY_{\nodetau''}^{\mA_\nodetau^{-1}}$, where $\nodetau'$ and $\nodetau''$ are the
  children of $\nodetau$.
\end{itemize}
In particular, if $\nodetau$ is an internal node of $\scrT$ and $\nodetau'$,
$\nodetau''$ are its children, then $\mathscr{Y}_{\nodetau'}$ and
$\mathscr{Y}_{\nodetau''}$ are both contained in
$\mathscr{Y}_\nodetau^{\mA_\nodetau}$.

%%%%%%%%%%%%%%%%%%%%%%%%%%%%%%%%%%%%%%%%%%%%%%%%%%

\subsection{Genericity assumptions}

The computations performed by ${\sf RoadmapRecLagrange}$ on input
$(\Gamma,\scrC)$ can be described using a binary tree; as one should
expect, we will verify below that this is the same tree $\scrT$ as for
${\sf RoadmapRec}$. We will indeed associate to each node $\nodetau$
of the tree $\scrT$ a type
$(k_\nodetau,\n_\nodetau,\p_\nodetau,e_\nodetau)$, defining
$k_\nodetau$, $\n_\nodetau$ and $\p_\nodetau$ inductively
($e_\nodetau$ was defined before); we will then see that, when the
random choices made in the algorithm are lucky, tracing ${\sf
  RoadmapRecLagrange}$ amounts to associating to each $\nodetau \in
\scrT$ a generalized Lagrange system $L_\nodetau$ of type
$(k_\nodetau,\n_\nodetau,\p_\nodetau,e_\nodetau)$.

Let us first define the integers $k_\nodetau$, $\n_\nodetau$ and $\p_\nodetau$.
At the root $\rho$, we set $k_\rho=0$, $\n_\rho=(n)$,
$\p_\rho=(p)$. Suppose then that $\nodetau$ has type
$(k_\nodetau,\n_\nodetau,\p_\nodetau,e_\nodetau)$, with
$\n_\nodetau=(n_{\nodetau},n_{\nodetau,1},\dots,n_{\nodetau,k_\nodetau})$ and
$\p_\nodetau=(p_{\nodetau},p_{\nodetau,1},\dots,p_{\nodetau,k_\nodetau})$, and write as
usual
$$N_\nodetau = n_{\nodetau} + n_{\nodetau,1}+\dots + n_{\nodetau,k_\nodetau}
\quad\text{and}\quad 
P_\nodetau = p_{\nodetau} + p_{\nodetau,1}+\dots + p_{\nodetau,k_\nodetau}.$$
Then, if $\nodetau$ is an internal node of $\scrT$, we define the types
at his two children as follows:
\begin{itemize}
\item the left child
  $\nodetau'$ has type $(k_{\nodetau'},\n_{\nodetau'},\p_{\nodetau'},e_{\nodetau'})$,
with
$$
\begin{array}{c}
k_{\nodetau'}=k_\nodetau+1, \quad
\n_{\nodetau'} = (n_{\nodetau},n_{\nodetau,1},\dots,n_{\nodetau,k_\nodetau},P_\nodetau),\\    
\p_{\nodetau'} = (p_{\nodetau},p_{\nodetau,1},\dots,p_{\nodetau,k_\nodetau},N_\nodetau-e_\nodetau-\dalgo_\nodetau+1)
\end{array}
$$
with $\dalgo_\nodetau=\lfloor(d_\nodetau+3)/2\rfloor$; recall that in this case, we defined
$d_{\nodetau'}=\dalgo_\nodetau-1$ and $e_{\nodetau'} = e_\nodetau$;
\smallskip
\item the right child $\nodetau''$ has type $(k_{\nodetau''},\n_{\nodetau''},\p_{\nodetau''},e_{\nodetau''})$,
with
$$k_{\nodetau''}=k_\nodetau, \quad \n_{\nodetau''} = \n_\nodetau, \quad \p_{\nodetau''} = \p_\nodetau;$$
in this case, we defined previously $d_{\nodetau''}= d_\nodetau-(\dalgo_\nodetau-1)$ and $e_{\nodetau''} = e_\nodetau + \dalgo_\nodetau-1$.
\end{itemize}
In particular, we deduce inductively that, for all $\nodetau$,
$n_{\nodetau}=n$ and $p_{\nodetau}=p$ hold, and that the indices $d_\nodetau$ and
$e_\nodetau$ associated to node $\nodetau$ satisfy $d_\nodetau =
N_\nodetau-e_\nodetau-P_\nodetau.$

As for algorithm ${\sf RoadmapRec}$, the node corresponding to the
recursive call at Step~\ref{rmplag:step:8} is the left child $\nodetau'$,
and the node corresponding to the recursive call at
Step~\ref{rmplag:step:11} is the right child $\nodetau''$. 

Consider now vectors $(\u_\nodetau)_{\nodetau \text{~internal node
    of~} \scrT}$, with $\u_\nodetau$ in $\QQ^{P_\nodetau}$ for all
$\nodetau$. Proof of existence of the Zariski open sets
$(\scrIopen_\nodetau)$ will be done by induction on the node
$\nodetau$ of $\scrT$, with the following induction assumption.
\begin{enumerate}
\item [${\sfTpa:}$] There exists a family of non-empty Zariski open sets
  $(\scrIopen_\nodeoherletter)_{\nodeoherletter\text{~proper ancestor of~} \nodetau}$,
  with $\scrIopen_{\nodeoherletter}$ in $\C^{P_\nodeoherletter}$
 for all $\nodeoherletter$,
  and with the following properties.
  Suppose that $\u_{\nodeoherletter}$ belongs to $\scrIopen_{\nodeoherletter}$ for
  all proper ancestors  $\nodeoherletter$ of $\nodetau.$ 
  Then to the node $\nodetau$ are associated the objects
  $(L_{\nodetau},\scrC_\nodetau)$, such that:
  \begin{enumerate}
  \item[${\sf t'_{1}.}$] $L_\nodetau=(\Gamma_\nodetau, \scrQ_\nodetau, \scrS_\nodetau)$ is a generalized
    Lagrange system of type $(k_\nodetau,\n_\nodetau,\p_\nodetau,e_\nodetau)$ and $\scrC_\nodetau$ is a zero-dimensional parametrization;
\smallskip
  \item[${\sf t'_{2}.}$] $V_\nodetau=\Clos{(L_\nodetau)}$, $Q_\nodetau=\Zeroes(\scrQ_\nodetau)$, $S_\nodetau=\Zeroes(\scrS_\nodetau)$
    and $C_\nodetau=\Zeroes(\scrC_\nodetau)$;
  \end{enumerate}
\smallskip
  and, if $V_\nodetau$ is not empty, then 
\smallskip
  \begin{enumerate}
  \item[${\sf t'_{3}.}$] $(L_\nodetau; \scrY_\nodetau)$ admits a global normal form $\bphi_\nodetau$;
\smallskip
  \item[${\sf t'_{4}.}$] the atlas of $(V_\nodetau,Q_\nodetau,S_\nodetau)$ associated with $\bphi_\nodetau$ is $\bpsi_\nodetau$.
  \end{enumerate}
\end{enumerate}

We claim that the root $\rho$ of $\scrT$ satisfies ${\sfTpa}$.
Indeed, following algorithm ${\sf MainRoadmapLagrange}$, we take
$L_\rho=(\Gamma,(\,),\scrS)$ and $\scrC_\rho= \scrC$. Then,
Proposition~\ref{sec:lagrange:prop:initnormalform} implies that
${\sfTpa}$ holds at the root $\rho$ of $\scrT$, with global normal
form $\bphi_\rho=((1,1,\f,\f))$.

Suppose now that that an internal node $\nodetau$ satisfies
${\sfTpa}$. We define the subset $\scrIopen_\nodetau$ of
$\C^{P_\nodetau}$ as follows:
\begin{itemize}
\item If $\u_{\nodeoherletter}$ belongs to $\scrIopen_{\nodeoherletter}$ for all
  proper ancestors $\nodeoherletter$ of $\nodetau$, and if $V_\nodetau$ is
  empty, we take $\scrIopen_\nodetau=\C^{P_\nodetau}$.

\smallskip

\item If $\u_{\nodeoherletter}$ belongs to
  $\scrIopen_{\nodeoherletter}$ for all proper ancestors
  $\nodeoherletter$ of $\nodetau$, and if $V_\nodetau$ is not empty,
  the sets $V_\nodetau,Q_\nodetau,S_\nodetau$, 
  the atlas $\bpsi_\nodetau$, the integer $\dalgo_\nodetau$, the
  change of variable $\mA_\nodetau$, the generalized Lagrange system
  $L_\nodetau$, its normal form $\bphi_\nodetau$ and the algebraic
  sets $\mathscr{Y}_\nodetau$ satisfy the assumptions of
  Proposition~\ref{sec:lagrange:prop:transfer:polar}, so that we can
  let $\scrIopen_\nodetau$ be the Zariski open set
  $\mathscr{I}(L_\nodetau,\bphi_\nodetau,\mA_\nodetau,\scrY_\nodetau)\subset
  \C^{P_\nodetau}$ defined in that proposition. Remark that the
  assumptions of this proposition require that
  $W_\nodetau^{\mA_\nodetau^{-1}}$ belong to $\mathscr{Y}_\nodetau$;
  this is the case by construction.
\smallskip

\item Else, we take $\scrIopen_\nodetau=\C^{P_\nodetau}$.
\end{itemize}

\begin{lemma}\label{sec:abstractalgo:lemma:correctnessH}
  If $\nodetau$ is an internal node that satisfies $\sfTpa$ and if the
  calls to all subroutines ${\sf Union}$, ${\sf Projection}$, ${\sf
    W}_1$, ${\sf Fiber}$, ${\sf Lift}$ are successful, the children
  $\nodetau'$ and $\nodetau''$ of $\nodetau$ satisfy~$\sfTpa$.
\end{lemma}
\begin{proof}
  To prove $\sfTpa$ at either $\nodetau'$ or $\nodetau''$, we assume
  that $\u_{\nodeoherletter}$ belongs to $\scrIopen_{\nodeoherletter}$
  for all ancestors $\nodeoherletter$ of $\nodetau$, including
  $\nodetau$ itself. In particular, we are in one of the first two
  cases in the previous case discussion.

  Because $\nodetau$ is an internal node, we know that we are not in
  the case $d \le 1$, so that we need only consider steps from 2 on in
  the algorithm. In all that follows, we assume that the calls to all
  subroutines ${\sf Union}$, ${\sf Projection}$, ${\sf W}_1$, ${\sf
    Fiber}$, ${\sf Lift}$ are successful. First, we prove that all
  objects computed by ${\sf RoadmapRecLagrange}$ match the quantities
  defined in ${\sf RoadmapRec}$.
  \begin{itemize}
  \item $V_\nodetau=\Clos{(L_\nodetau)}$, $Q_\nodetau=\Zeroes(\scrQ_\nodetau)$, $S_\nodetau=\Zeroes(\scrS_\nodetau)$
    and $C_\nodetau=\Zeroes(\scrC_\nodetau)$.

\smallskip

    These are true by assumption $\sfTpa$ for $\nodetau$.
\smallskip
  \item $L'_\nodetau$ is a generalized Lagrange system of type
    $(k_{\nodetau'},\n_{\nodetau'},\p_{\nodetau'},e_{\nodetau'})$ such that
    $\Clos{(L'_{\nodetau})}=W_\nodetau$.

\smallskip

    The claim on the type of $L'_\nodetau$ follows from our inductive
    definition of the type, together with Lemma~\ref{lemma:GLS:typeW}.
    The second claim is obtained through a case discussion:

\smallskip

    \begin{itemize}
    \item If $V_\nodetau$ is empty, $V_\nodetau'=W_\nodetau$ is empty as well; on
      the other hand, since $V_\nodetau=\Clos{(L_\nodetau)}$, the construction
      of $L'_\nodetau$ implies that $\Clos{(L'_\nodetau)}$ is empty.
\smallskip
    \item If $V_\nodetau$ is not empty, our assumption on $\u_\nodetau$ shows
      that we can apply the results of
      Proposition~\ref{sec:lagrange:prop:transfer:polar}, which
      implies the claim. In addition, if $W_\nodetau$ is not empty,
      $(L'_\nodetau; \scrY_\nodetau^{\mA_\nodetau}-\{W_\nodetau\})$ admits a global
      normal form, and the associated atlas of
      $(W_\nodetau,Q_\nodetau,S^{\mA_\nodetau}_\nodetau)$ is
      $\atlaspolar(\bpsi_\nodetau^{\mA_\nodetau},V_\nodetau^{\mA_\nodetau},Q_\nodetau,S_\nodetau^{\mA_\nodetau},\dalgo_\nodetau)$,
      that is, $\bpsi_{\nodetau'}$.
    \end{itemize}

\smallskip

  \item ${\sf W}_1(L'_\nodetau)$ is a zero-dimensional parametrization of
    $\polar(e_\nodetau, 1, W_\nodetau)-\Zeroes(\scrS_\nodetau^{\mA_\nodetau}).$

\smallskip

    All we need to do is to verify that the assumptions of
    Proposition~\ref{sec:main:critical} are satisfied, remembering
    that $\Clos{(L'_{\nodetau})}=W_\nodetau$.

\smallskip

    \begin{itemize}
    \item If $W_\nodetau$ is empty, this is clear.
\smallskip
    \item Because $B_\nodetau$ is finite
      (Lemma~\ref{sec:abstractalgo:lemma:correctnessA1}),
      $\Kpolar(e_\nodetau,1,W_\nodetau)=\Kpolar(e_\nodetau, 1, \Clos{({L'_\nodetau})})$ is finite,
      which in turn implies that $\polar(e_\nodetau, 1, \Clos{({L'_\nodetau})})$ is
      finite. The other point to verify is that
      $(L_\nodetau',\polar(e_{\nodetau},1,\Clos{(L'_\nodetau)}))$ has a global normal
      form; this is because $(L'_\nodetau;
      \scrY_\nodetau^{\mA_\nodetau}-\{W_\nodetau\})$ admits a global normal form,
      and $\scrY_\nodetau^{\mA_\nodetau}-\{W_\nodetau\}$ contains
      $\polar(e_{\nodetau},1,\Clos{(L'_\nodetau)})$.
    \end{itemize}

\smallskip

  \item  $\Zeroes(\scrB_\nodetau)=B_\nodetau$.

\smallskip

    Since we know that $\Zeroes(\scrS_\nodetau^{\mA_\nodetau})=S_\nodetau^{\mA_\nodetau}$
    and $\Zeroes(\scrC_\nodetau^{\mA_\nodetau})=C_\nodetau^{\mA_\nodetau}$, we deduce from
    the previous item that $\Zeroes(\scrB_\nodetau)$ is the union of $\polar(e_\nodetau,
    1, W_\nodetau)-S_\nodetau^{\mA_\nodetau}$ and $C_\nodetau^{\mA_\nodetau}$.  Also by
    assumption ${\sfTa}$ on $\nodetau$, $S_\nodetau$ is contained in
    $C_\nodetau$; thus, after applying $\mA_\nodetau$, we deduce that
    $\Zeroes(\scrB_\nodetau)$ is the union of $\polar(e_\nodetau, 1, W_\nodetau)$ and
    $C_\nodetau^{\mA_\nodetau}$.

\smallskip

    Now, we claim that $\sing(W_\nodetau)$ is contained in
    $S_\nodetau^{\mA_\nodetau}$ (and thus in $C_\nodetau^{\mA_\nodetau}$): this is
    obvious if $W_\nodetau$ is empty; else, using
    Lemma~\ref{sec:coro:lemma:singSX}, this is because $W_\nodetau$ is
    $(\dalgo_\nodetau-1)$-equidimensional and $\bpsi_{\nodetau'}$ is an atlas
    of $(W_\nodetau,Q_\nodetau,S_\nodetau^{\mA_\nodetau})$.

\smallskip

    The difference $\Kpolar(e_\nodetau, 1, W_\nodetau)-\polar(e_\nodetau, 1, W_\nodetau)$
    is contained in $\sing(W_\nodetau)$, and thus in
    $C_\nodetau^{\mA_\nodetau}$. As a result, we finally conclude that
    $\Zeroes(\scrB_\nodetau)$ is the union of $\Kpolar(e_\nodetau, 1, W_\nodetau)$ and
    $C_\nodetau^{\mA_\nodetau}$, that is, $B_\nodetau$.

\smallskip

  \item   $\Zeroes(\scrQ''_\nodetau)=Q''_\nodetau$.

\smallskip

    This follows from the previous item, by  projecting on $\C^{e_\nodetau+\dalgo_\nodetau-1}$.

\smallskip

  \item   $\Zeroes(\scrC'_\nodetau)=C'_\nodetau$.

\smallskip

    The right-hand side is equal to $C_\nodetau^{\mA_\nodetau} \cup
    \fbr(W_\nodetau,Q''_\nodetau)$. For the left-hand side, remember that
    $\Zeroes(\scrQ''_{\nodetau})=Q''_{\nodetau}$, and that
    $\Zeroes(\scrC'_\nodetau)=C_\nodetau^{\mA_\nodetau} \cup {\sf
      Fiber}(L'_\nodetau,\scrQ''_\nodetau)$. Let us then verify that the
    assumptions of Proposition~\ref{sec:main:prop:fiber} applied to
    $L'_\nodetau$ and $\scrQ''_\nodetau$ are satisfied, keeping in mind that
    $W_\nodetau=\Clos{(L'_\nodetau)}$:

\smallskip

    \begin{itemize}
    \item If $W_\nodetau$ is empty, this is clear.

\smallskip

    \item If $W_\nodetau$ is not empty, this is because
      $\fbr(W_\nodetau,Q''_\nodetau)$ is finite, and $$(L'_\nodetau,
      \fbr(W_\nodetau,Q''_\nodetau))$$ has the global normal form property
      (because $(L'_\nodetau; \scrY_\nodetau^{\mA_\nodetau}-\{W_\nodetau\})$ admits a
      global normal form, and $\scrY_\nodetau^{\mA_\nodetau}-\{W_\nodetau\}$
      contains $\fbr(W_\nodetau,Q''_\nodetau)$).
   \end{itemize}

\smallskip

    As a result, ${\sf Fiber}(L'_\nodetau,\scrQ''_\nodetau)$ returns a
    zero-dimensional parametrization of
    $\fbr(W_\nodetau,Q''_\nodetau)-S_\nodetau^{\mA_\nodetau}$.  Since we saw above
    that $S_\nodetau^{\mA_\nodetau}$ is contained in $C_\nodetau^{\mA_\nodetau}$, we
    conclude that $C_\nodetau^{\mA_\nodetau} \cup {\sf
      Fiber}(L'_\nodetau,\scrQ''_\nodetau)$ defines $C_\nodetau^{\mA_\nodetau} \cup
    \fbr(W_\nodetau,Q''_\nodetau)$. As was pointed out above, this is enough
    to conclude.

\smallskip

  \item    $\Zeroes(\scrC''_\nodetau)=C''_\nodetau$.

\smallskip

    This follows directly from the specifications of ${\sf Lift}$.

\smallskip

  \item  $\Zeroes(\scrS'_\nodetau)=S_\nodetau^{\mA_\nodetau} \cup \fbr(W_\nodetau,Q''_\nodetau)$.

\smallskip

    This is the same argument as in the proof that  $\Zeroes(\scrC'_\nodetau)=C'_\nodetau$,
    replacing $C_\nodetau^{\mA_\nodetau}$ by $S_\nodetau^{\mA_\nodetau}$.

\smallskip

  \item $\Zeroes(\scrS''_\nodetau)=\fbr(S_\nodetau^{\mA_\nodetau} \cup W_\nodetau,Q''_\nodetau)$.

\smallskip

    Again, this follows from the specifications of ${\sf Lift}$.

\smallskip

  \item $L''_\nodetau$ is a generalized Lagrange system of type
    $(k_{\nodetau''},\n_{\nodetau''},\p_{\nodetau''},e_{\nodetau''})$ such that
    $\Clos{(L''_{\nodetau})}=V''_\nodetau.$

\smallskip

    The claim on the type of $L''_\nodetau$ follows from our inductive
    definition of the type, together with Lemma~\ref{lemma:GLS:typeF}.
    The second claim is obtained through a case discussion:

\smallskip

    \begin{itemize}
    \item If $V_\nodetau$ is empty, then $V''_\nodetau$, which is a section of
      it, is empty as well. Since we have $V_\nodetau=\Clos{(L_\nodetau)}$, we
      deduce from Definition~\ref{def:CZPC} that
      $\fbr(V(\F_\nodetau),Q_\nodetau)$ is contained in $\pi_\X^{-1}(S_\nodetau)$,
      where $\F_\nodetau$ are the polynomials computed by
      $\Gamma_\nodetau$. We will now prove that the definition of
      $L''_\nodetau= \Fiberlag(L_\nodetau^{\mA_\nodetau},
      \scrQ''_\nodetau,\scrS''_\nodetau)$ given
      in~\ref{sec:lagrange:notation:fiber} implies that
      $\Clos{(L''_\nodetau)}$ is empty, which is what we have to establish.

\smallskip

      Since we saw that $\Zeroes(\scrQ''_\nodetau)=Q''_\nodetau$, our claim is
      equivalent to $$\fbr(V(\F^{\mA_\nodetau}_\nodetau),Q''_\nodetau)$$ being contained in
      $\pi_\X^{-1}(\Zeroes(\scrS''_\nodetau))$, where we saw that
      $$\Zeroes(\scrS''_\nodetau)=\fbr(S_\nodetau^\mA \cup W_\nodetau,Q''_\nodetau).$$

\smallskip

      By assumption ${\sf t_{2}}$ for $\nodetau$, $Q''_\nodetau$ lies over
      $Q_\nodetau$. Take $$(\x,\bell)\in
      \fbr(V(\F^{\mA_\nodetau}_\nodetau),Q''_\nodetau).$$ Then,
      $(\x^{\mA_\nodetau^{-1}},\bell)$ is in $\fbr(V(\F_\nodetau),Q''_\nodetau)$.
      Then previous remark shows that $(\x^{\mA_\nodetau^{-1}},\bell)$ is
      in $\fbr(V(\F_\nodetau),Q_\nodetau)$, so that the assumption that $V_\nodetau$
      is empty implies that $\x^{\mA_\nodetau^{-1}}$ is in $S_\nodetau$;
      equivalently, $\x$ is in $S_\nodetau^{\mA_\nodetau}$. Since $\x$ lies over 
      $Q''_\nodetau$, we deduce that $\x$ is in $\fbr(S_\nodetau^\mA,Q''_\nodetau)$,
      and thus in $\Zeroes(\scrS''_\nodetau)$, as claimed.

\smallskip

    \item If $V_\nodetau$ is not empty, the algebraic sets
      $V_\nodetau,Q_\nodetau,S_\nodetau$, the atlas $\bpsi_\nodetau$, the integer
      $\dalgo_\nodetau$, the change of variable $\mA_\nodetau$, the
      parametrizations $\scrQ''_\nodetau$ and $\scrS''_\nodetau$, the generalized
      Lagrange system $L_\nodetau$, its normal form $\bphi_\nodetau$, the
      algebraic sets $\mathscr{Y}_\nodetau$ satisfy the assumptions of
      Proposition~\ref{sec:lagrange:prop:transfer:fiber}.  (Remark that
      the assumptions of this proposition require that
      ${V''_\nodetau}^{\mA_\nodetau^{-1}}$ belong to $\mathscr{Y}_\nodetau$; this is
      the case by construction). 
      
\smallskip

      Then, that proposition proves our claim. In addition,
      $(L''_\nodetau, \scrY_\nodetau^{\mA_\nodetau}-\{V''_\nodetau\})$ admits a global
      normal form whose atlas
      is
      $$\atlasfiber(\bpsi_\nodetau^{\mA_\nodetau},V_\nodetau^{\mA_\nodetau},Q_\nodetau,S_\nodetau^{\mA_\nodetau},Q_\nodetau')$$
      that is, $\bpsi_{\nodetau''}$.
    \end{itemize}
  \end{itemize}

We can now prove that $\nodetau'$ satisfies ${\sfTpa}$. We already saw
that the type of $L_{\nodetau'}=L'_\nodetau$ is as claimed. Since in addition
we have by definition $\scrC_{\nodetau'}=\scrC'_{\nodetau}$, and this set has
dimension zero, we deduce that ${\sf t'_{1}}$ holds at $\nodetau'$.

To prove ${\sf t'_{2}}$, notice that we have already seen that
$V_{\nodetau'}=W_\nodetau$ coincides with
$\Clos{(L_{\nodetau'})}=\Clos{(L'_\nodetau)}$.  By construction,
$Q_{\nodetau'}=Q_\nodetau$, and by assumption ${\sfTa}$ for $\nodetau$,
$Q_\nodetau=\Zeroes(\scrQ_\nodetau)$; since $\scrQ_{\nodetau'}=\scrQ_\nodetau$, we
deduce that $Q_{\nodetau'}=\Zeroes(\scrQ_{\nodetau'})$. Similarly,
$S_{\nodetau'}=S_\nodetau^{\mA_\nodetau}$, and by assumption ${\sfTa}$ for $\nodetau$,
$S_\nodetau=\Zeroes(\scrS_\nodetau)$. Since
$\scrS_{\nodetau'}=\scrS_\nodetau^{\mA_\nodetau}$, we obtain
$S_{\nodetau'}=\Zeroes(\scrS_{\nodetau'})$. Finally, we saw above that
$\Zeroes(\scrC'_\nodetau)=C'_\nodetau$, or equivalently
$\Zeroes(\scrC_{\nodetau'})=C_{\nodetau'}$.  Thus, ${\sf t'_{2}}$ is proved.

Suppose finally that $W_\nodetau=V_{\nodetau'}$ is not empty. We saw above
that $(L'_\nodetau; \scrY_\nodetau^\mA-\{W_\nodetau\})$ admits a global normal
form whose atlas is $\bpsi_{\nodetau'}$. Because $\scrY_{\nodetau'}$ is
contained in $\scrY_\nodetau^\mA-\{W_\nodetau\}$, this proves at once ${\sf
  t'_{3}}$ and ${\sf t'_{4}}$. So, we are done for $\nodetau'$.

To conclude, we prove that $\nodetau''$ satisfies ${\sfTpa}$. As in the
case of $\nodetau'$, we saw above that the type of $L_{\nodetau''}=L''_\nodetau$
is as claimed. Since in addition we have
$\scrC_{\nodetau''}=\scrC''_{\nodetau}$, and this set has dimension zero, we
deduce that ${\sf t'_{1}}$ holds at $\nodetau''$.

To prove ${\sf t'_{2}}$ at $\nodetau''$, we have to establish the
equalities $V_{\nodetau''}=\Clos{(L_{\nodetau''})}$,
$Q_{\nodetau''}=\Zeroes(\scrQ_{\nodetau''})$,
$S_{\nodetau''}=\Zeroes(\scrS_{\nodetau''})$ and
$C_{\nodetau''}=\Zeroes(\scrC_{\nodetau''})$.  The first two items were proved
above.  Next, we have to prove that
$S_{\nodetau''}=\Zeroes(\scrS_{\nodetau''})$, or equivalently
$\fbr(S_\nodetau^{\mA_\nodetau} \cup W_\nodetau,
Q''_\nodetau)=\Zeroes(\scrS_{\nodetau''})$:
this was proved above as well. Finally, we need to prove that
$C_{\nodetau''}=\Zeroes(\scrC_{\nodetau''})$, or equivalently
$C''_\nodetau=\Zeroes(\scrC''_\nodetau)$: this was also proved above.  Thus,
${\sf t'_{2}}$ is proved.

Suppose in addition that $V''_\nodetau=V_{\nodetau''}$ is not empty. We saw
above that $(L''_\nodetau; \scrY_\nodetau^\mA-\{V''_\nodetau\})$ admits a global
normal form whose atlas is $\bpsi_{\nodetau''}$. Because $\scrY_{\nodetau''}$
is contained in $\scrY_\nodetau^\mA-\{V''_\nodetau\}$, this proves at once
${\sf t'_{3}}$ and ${\sf t'_{4}}$. Thus, $\nodetau''$ satisfies
${\sfTpa}$ and the lemma is proved.
\end{proof}

%% Similarly to what we did in
%% Subsection~\ref{chap:abstractalgo:objects}, the above lemma allows us
%% to introduce a global assumption that includes all internal nodes
%% $\nodetau$.
%% \begin{definition}\label{def:concrete:H'}
%%   Assume that $\scrA=(\mA_\nodetau)_{\nodetau \in \scrT}$ satisfies
%%   ${\sfTa}(V,C,\bpsi)$ (see Definition \ref{def:abstract:H}) and let
%%   further $\scrU=(\u_\nodetau)_{\nodetau \in\scrT}$, with $u_\nodetau$ in
%%   $\QQ^{P_\nodetau}$ for all $\nodetau$. We say that $\scrU$ satisfies
%%   ${\sfTp}(V, C, \bpsi, \scrA)$ if for every node $\nodetau$ of $\scrT$:
%%   \begin{itemize}
%%   \item the calls to all subroutines at $\nodetau$, such as ${\sf Union}$,
%%     ${\sf Projection}$, ${\sf W}_1$, ${\sf Fiber}$, ${\sf Lift}$, are
%%     successful;
%%   \item if $\nodetau$ is an internal node in $\scrT$, $\nodetau$ satisfies
%%     $\sfTpa$ and $\u_\nodetau$ satisfies $\sfTpb$.
%%   \end{itemize}
%% \end{definition}

%%%%%%%%%%%%%%%%%%%%%%%%%%%%%%%%%%%%%%%%%%%%%%%%%%%%%%%%%%%%

\subsection{Proof of the proposition}

Repeated applications of the previous lemma allow us to define a
family of non-empty Zariski open sets $\scrIopen_\nodetau \subset
\GL(n, e_\nodetau)$, for $\nodetau$ internal node of $\scrT$, for
which all nodes of $\scrT$ satisfy property $\sfTpa$.

If, as Proposition~\ref{PROP:CORRECTNESSGLOBAL}, we assume that for
all internal nodes $\nodetau$ of $\scrT$, $\u_\nodetau$ is in
$\scrIopen_\nodetau$, property $\sfTpa$ shows that we can associate to
any node $\nodetau$ of $\scrT$ a generalized Lagrange system
$L_\nodetau$, that defines the algebraic set $V_\nodetau$ considered
when running ${\sf RoadmapRec}$, when using the same matrices
$\mA_\nodetau$ as in ${\sf RoadmapRecLagrange}$.

The only pending point to prove is that at the leaves $\nodetau$ of
the recursion, the behavior of ${\sf RoadmapRecLagrange}$ agrees with
that of ${\sf RoadmapRec}$. Indeed, after we have reached the leaves,
going up the recursion tree simply amounts to performing changes of
variables and unions, for which there is no difficulty.
  
Let us then consider a leaf $\nodetau$. By assumption, $\nodetau$ satisfies
${\sfTpa}$, so in particular $\Clos{(L_\nodetau)}=V_\nodetau$, and either
$V_\nodetau$ is empty or $L_\nodetau$ admits a global normal form (recall that
$\scrY_\nodetau$ is empty at the leaves). We can then apply
Proposition~\ref{chap:solvelagrange:prop:basicsolve}, and deduce that
we correctly return a one-dimensional parametrization of $V_\nodetau$.

As a consequence, correctness follows from Theorem
\ref{THEO:MAINABSTRACT}, and
Proposition~\ref{prop:correctnessglobal} is proved.

%%%%%%%%%%%%%%%%%%%%%%%%%%%%%%%%%%%%%%%%%%%%

%%%%%%%%%%%%%%%%%%%%%%%%%%%%%%%%%%%%%%%%%%%%
%\input{proof7.3}
\section{Proof of Proposition~\ref{prop:complexity:mainlagrange}}\label{chap:main:sec:complexite}

Finally, we prove Proposition~\ref{prop:complexity:mainlagrange} whose
statement is as follows.

{\em   Consider polynomials $\f=f_1,\dots,f_p$ in $\QQ[X_1,\dots,X_n]$ of
  degrees bounded by $D$, given by a straight-line program $\Gamma$ of
  length $E$, that define a reduced regular sequence. 

  Suppose that $V=V(\f) \subset \C^n$ has finitely many singular
  points and that $V(\f)\cap \R^n$ is bounded. Consider also a
  zero-dimensional parametrization $\scrC_0$ of degree $\degC$ that describes a finite
  set $C_0 \subset \C^n$.

  Suppose that all matrices $\mA_\nodetau$ and all vectors $\u_\nodetau$
  satisfy the assumptions of Proposition~\ref{prop:correctnessglobal},
  and that all calls to subroutines such as ${\sf Union}$, ${\sf
    Projection}$, ${\sf W_1}$, ${\sf Lift}$ are successful. Then, ${\sf MainRoadmapLagrange}(\Gamma, \scrC_0)$ either returns
  ${\sf fail}$ or returns a one-dimensional parametrization of degree
  bounded by
\[
\softO \left (
\degC 16^{3d}  (n \log_2(n))^{2(2d+12\log_2(d))(\log_2(d)+6)}D^{(2n+1)(\log_2(d)+4)}
\right )\]
using 
\[
\softO \left (
\degC^3 16^{9d} E (n \log_2(n))^{6(2d+12\log_2(d))(\log_2(d)+7)}D^{3(2n+1)(\log_2(d)+5)}
\right ) 
\]
arithmetic operations in $\QQ$, with $d=n-p$.
}

We start by establishing some elementary bounds on the number of
variables and polynomials in the generalized Lagrange systems
considered during the recursive calls of ${\sf RoadmapRecLagrange}$.

Next, we prove uniform degree bounds on the geometric objects
represented by generalized Lagrange systems and zero-dimensional
parametrizations computed at Steps
(\ref{rmplag:step:5}--\ref{rmplag:step:9-a}) of
${\sf RoadmapRecLagrange}$. This enables us to deduce bounds on the
degree of the output roadmap and, consequently, bounds on the size of
the output.

Finally, we use these degree bounds to bound the cost of ${\sf
  RoadmapRecLagrange}$, and thus of ${\sf MainRoadmapLagrange}$. This
mainly relies on algorithms ${\sf SolveLagrange}$, ${\sf W1}$ and
${\sf Fiber}$ described in
Propositions~\ref{chap:solvelagrange:prop:basicsolve},~\ref{sec:main:critical}
and~\ref{sec:main:prop:fiber} and the basic routines dealing with
zero- and one-dimensional parametrizations given in Section
\ref{chap:posso}.

%%%%%%%%%%%%%%%%%%%%%%%%%%%%%%%%%%%%%%%%%%%%%%%%%%%%%%%%%%%%
%%%%%%%%%%%%%%%%%%%%%%%%%%%%%%%%%%%%%%%%%%%%%%%%%%%%%%%%%%%%
%%%%%%%%%%%%%%%%%%%%%%%%%%%%%%%%%%%%%%%%%%%%%%%%%%%%%%%%%%%%

\subsection{Notation and auxiliary results}\label{chap:complexite:section:notations}

We first recall notation introduced in Section~\ref{chap:correctness},
where we attached integers and data to the nodes of the tree, and
introduce further quantities. Then, we prove basic inequalities on
these quantities, that will be needed for the cost analysis.

%%%%%%%%%%%%%%%%%%%%%%%%%%%%%%%%%%%%%%%%%%%%%%%%%%%%%%%%%%%%

\subsubsection{Notation}\label{chap:complexite:subsection:notations}

In the whole section, we assume without loss of generality that the following inequalities hold:
\begin{itemize}
\item $n\geq 2$
\smallskip
\item $p\geq 1$
\smallskip
\item $n-p\geq 1$
\smallskip
\item $D\geq 2$ (else, $V\cap\R^n$ cannot satisfy the boundedness
  assumption).
\end{itemize}
Each node $\nodetau$ of $\scrT$ is labelled with the following integers:
\begin{itemize}
\item $d_\nodetau$ (defined previously; it is the dimension of the current algebraic set),
\smallskip
\item $e_\nodetau$ (defined previously; it is the number of variables assuming fixed values),
\smallskip
\item $h_\nodetau$, which we define as the height of $\nodetau$.
\end{itemize}
Since by assumption at any node $\nodetau$ of $\scrT$, $\mA_\nodetau$
is in $\scrOpen_\nodetau$ and $\u_\nodetau$ is in
$\scrIopen_\nodetau$, and since all calls to our various subroutines
are successful, to each node $\nodetau$ are also associated the
following objects and quantities:
\begin{itemize}
\item a generalized Lagrange system $L_\nodetau=(\Gamma_\nodetau, \scrQ_\nodetau,
  \scrS_\nodetau)$,
\smallskip
\item a zero-dimensional parametrization $\scrC_\nodetau$,
\smallskip
\item an integer $E_\nodetau$, which denotes the length of $\Gamma_\nodetau$.
\end{itemize}
When $\nodetau$ is not a leaf, the following objects are defined:
\begin{itemize}
\item zero-dimensional parametrizations $\scrB_\nodetau$, $\scrQ''_\nodetau$,
  $\scrC'_\nodetau$, $\scrC''_\nodetau$, $\scrS'_\nodetau$, $\scrS''_\nodetau$, that
  are computed at Steps \ref{rmplag:step:5}--\ref{rmplag:step:9-a};
\smallskip
\item one-dimensional parametrizations $\scrR'_\nodetau$, $\scrR''_\nodetau$,
  $\scrR_\nodetau$, respectively computed at Steps
  \ref{rmplag:step:8},~\ref{rmplag:step:11} and returned at Step
  \ref{rmplag:step:12};
\smallskip
\item generalized Lagrange systems $L'_\nodetau$, $L''_\nodetau$ constructed
  at Steps~\ref{rmplag:step:4} and~\ref{rmplag:step:10};
\smallskip
\item algebraic sets $\mathscr{Y_\nodetau}$ introduced in the previous
  section for the collection of all geometric objects associated to
  the descendants of $\nodetau$;
\smallskip
\item an integer $\dalgo_\nodetau=\lfloor (d_\nodetau+3)/2 \rfloor$;
\smallskip
\item an integer $k_\nodetau$ and vectors of integers
  $\n_\nodetau=(n,n_{\nodetau,1},\dots,n_{\nodetau,k_\nodetau})$ and
  $\p_\nodetau=(p,p_{\nodetau,1},\dots,p_{\nodetau,k_\nodetau})$. For $i$ in 
  $\{0,\dots,k_\nodetau\}$, we define
\smallskip
  \begin{itemize}
  \item $N_{i,\nodetau}=n+\sum_{\ell=1}^i n_{\nodetau,\ell}$, and $N_\nodetau = N_{k_\nodetau,\nodetau}$
\smallskip
  \item $P_{i,\nodetau}=p+\sum_{\ell=1}^i p_{\nodetau,\ell}$, and $P_\nodetau = P_{k_\nodetau,\nodetau}$
\smallskip
  \item $d_{i,\nodetau}=N_{i,\nodetau}-e_\nodetau-P_{i,\nodetau}$; note that we have
    $d_\nodetau = d_{k_\nodetau,\nodetau}$.
  \end{itemize}
\end{itemize}
When $\nodetau$ is a leaf, the one-dimensional parametrization computed at
Step \ref{rmplag:step:1} is denoted by~$\scrR_\nodetau$.

%% We reuse hereafter the notation introduced in
%% Subsection~\ref{ssec:complexity} for the degrees of the geometric
%% objects encoded by the data manipulated by our algorithms.

%%%%%%%%%%%%%%%%%%%%%%%%%%%%%%%%%%%%%%%%%%%%%%%%%%%%%%%%%%%%

\subsubsection{Some useful inequalities}

We start with a technical but simple and useful lemma.  It shows that
the number of equations and unknowns is at all times at most $2n^2$.
In what follows, we use notation such as $\dinit,E_\rho,\dots$ to
denote the values of the various quantities seen above at the root.

\begin{lemma}\label{chap:complexite:ineq1}
  Let $\nodetau$ be a node of $\scrT$. The following holds:
  \begin{itemize}
  \item $k_\nodetau\leq h_\nodetau\leq \lceil \log_2(\dinit)\rceil$
\smallskip
  \item $E_\nodetau\leq 4n^{4+2\log_2(\dinit)}(E_\rho+n^4)$
\smallskip
  \item for $i$ in $\{0,\dots,k_\nodetau\}$, we have:
\smallskip
  \begin{itemize}
  \item $ P_{i, \nodetau}+1\leq N_{i, \nodetau}\leq 2^in$
\smallskip
  \item $d_{i,  \nodetau} \leq \frac{\dinit}{2^{i}}+1$;
\smallskip
  \end{itemize}
  so, in particular, $d_\nodetau\leq \frac{\dinit}{2^{h_\nodetau}}+1$.
  \end{itemize}
\end{lemma}
\begin{proof}
  The fact that $h_\nodetau \le \lceil \log_2(\dinit)\rceil$ is true by
  construction, for all nodes $\nodetau$. 
  Our reasoning for the other inequalities is by increasing induction
  on the height of $\nodetau$.  We actually prove a slightly stronger form
  of the upper bound on $E_\nodetau$, which reads $$E_\nodetau\leq
  (3n^2)^{h_\nodetau} E_\rho+4^{h_\nodetau} n^{4+2h_\nodetau}.$$ Note that this
  inequality implies that \[E_\nodetau\leq (4n^2)^{h_\nodetau} E_\rho+4^{h_\nodetau}
  n^{4+2h_\nodetau}\leq (4n^2)^{h_\nodetau}(E_\rho+n^4)\leq
  4n^{4+2\log_2(\dinit)}(E_\rho+n^4),\] since $h_\nodetau\leq \lceil
  \log_2(\dinit)\rceil\leq 1+\log_2(\dinit)$.

  At the root $\nodetau=\rho$, all inequalities are immediate, except for
  the case $i=0$ of $P_{i, \nodetau}+1\leq N_{i, \nodetau}\leq 2^in$ (which is
  the only one we have to consider); this is equivalent to $n-p \ge
  1$, which is true by assumption.  

  Let now $\nodetau$ be a node of $\scrT$. Assume that it satisfies the
  induction assumption, and that it is not a leaf; then, it has a left
  child $\nodetau'$ and a right child~$\nodetau''$.

  Let us work with $\nodetau'$ first. By
  Definition~\ref{sec:lagrange:notation:polar}, we have
  $k_{\nodetau'}=k_{\nodetau}+1$; since we have $k_{\nodetau}\leq
  h_{\nodetau}$ by induction, and $h_{\nodetau'}=h_{\nodetau}+1$ by
  definition, we deduce that $k_{\nodetau'}\leq h_{\nodetau'}$. Thus,
  the first item is proved.

  Next, since $h_{\nodetau'}=h_{\nodetau}+1$, we have to establish
  $E_{\nodetau'}\leq
  (3n^2)^{h_{\nodetau}+1}E_\rho+4^{h_{\nodetau}+1}n^{4+2(h_{\nodetau}+1)}$. Propagating
  partial derivatives in the forward manner, we would obtain that one
  can evaluate $\F_\nodetau$ and all its partial derivatives within $4N_\nodetau E_{\nodetau}$
  operations; however, using the reverse mode as in Baur-Strassen's
  algorithm \cite{BaurStrassen}, the cost reduces to $3P_\nodetau E_\nodetau \le 3 N_\nodetau E_\nodetau$.

  Multiplying on the right $\jac(\F_\nodetau,e_\nodetau+\dalgo_\nodetau)$ with a vector of $P_\nodetau$
  variables costs at most $2N_\nodetau P_\nodetau$ operations; a final $2P_\nodetau$ operations
  come from the cost of computing the affine form in
  $\Polarlag(L_\nodetau,\u_\nodetau,\dalgo_\nodetau)$. Using the induction assumption, we have
  $N_\nodetau \le n^2$ and $2N_\nodetau P_\nodetau+2P_\nodetau \le 2n^4$; we deduce that
  $$E_{\nodetau'}\leq 3N_\nodetau E_\nodetau+2N_\nodetau P_\nodetau+2P_\nodetau\leq 3n^2((3n^2)^{{h_{\nodetau}}}E_\rho+4^{h_{\nodetau}} n^{4+2{h_{\nodetau}}})+2n^4,$$
  which implies that $$E_{\nodetau'}\leq (3n^2)^{{h_{\nodetau}}+1} E_\rho +3\cdot 4^{h_{\nodetau}} n^{4+2({h_{\nodetau}}+1)}
  + 2n^4.$$ 
  Now, since $n \ge 2$, we have the upper bound $2n^4 \le n^{4+2({h_{\nodetau}}+1)}$;
  using the inequality $3\cdot 4^{h_{\nodetau}}+1\leq 4^{{h_{\nodetau}}+1}$, we conclude that $E_{\nodetau'}\leq
  (3n^2)^{{h_{\nodetau}}+1}+4^{{h_{\nodetau}}+1}n^{4+2({h_{\nodetau}}+1)}$ as requested. This proves the second 
  point for $\nodetau'$.

  For the third item, using again Definition
  \ref{sec:lagrange:notation:polar}, we have $N_{i,\nodetau}=N_{i,\nodetau'}$ and
  $P_{i,\nodetau}=P_{i,\nodetau'}$ for $i$ in $\{0,\dots,k_{\nodetau}\}$, as well as $e_{\nodetau}=e_{\nodetau'}$; in
  particular, the only new inequalities we have to prove are for index
  $i=k_{\nodetau}+1$.

  We first prove that $P_{k_{\nodetau}+1,\nodetau'}+1\leq N_{k_{\nodetau}+1,\nodetau'}\leq 2^{k_{\nodetau}+1}n$.  By
  Lemma \ref{lemma:GLS:typeW}, we have
  $$ N_{k_{\nodetau}+1,\nodetau'}=N_{\nodetau}+P_{\nodetau} \qquad \text{ and }\qquad P_{k_{\nodetau}+1,{\nodetau'}}=N_{\nodetau}+P_{\nodetau}-e_{\nodetau}-\dalgo_\nodetau+1$$
  with $\dalgo_\nodetau=\lfloor\frac{{d_{\nodetau}}+3}{2}\rfloor \ge 2$ (Step
  \ref{rmplag:step:3}). We deduce that $P_{k_{\nodetau}+1,\nodetau'}+1\leq N_{k_{\nodetau}+1,\nodetau'}$. On
  the other hand, by our induction assumption $P_{\nodetau}+1 \le N_{\nodetau}\leq 2^{k_{\nodetau}}n^2$,
  we deduce that $N_{k_{\nodetau}+1,\nodetau}\leq 2^{k_{\nodetau}+1}n$. Finally, note that
  $$
  {d_{\nodetau'}}=\dalgo_\nodetau-1 = \lfloor\frac{{d_{\nodetau}}+1}{2}\rfloor \leq
  \frac{{d_{\nodetau}}}{2}+\frac{1}{2}\leq \left
  (\frac{\dinit}{2^{k_{\nodetau}+1}}+\frac{1}{2}\right )+\frac{1}{2}\leq \frac{\dinit}{2^{k_{\nodetau}+1}}+1,
  $$ as requested. Thus, we are done with $\nodetau'$.

  Proving the inequalities for $\nodetau''$ is done with a similar
  reasoning: we use instead Definition~\ref{sec:lagrange:notation:fiber} and Lemma~\ref{lemma:GLS:typeF}
  which imply that $k_{\nodetau''}=k_{\nodetau}$; since $h_{\nodetau''}={h_{\nodetau}}+1$, we obtain $k_{\nodetau''}\leq h_{\nodetau''}$.
  Next,  we need to establish that $E_{\nodetau''}\leq
  (3n^2)^{h_{\nodetau''}}+4^{h_{\nodetau''}}n^{4+2 h_{\nodetau''}}$. This is immediate since by definition
  of $L_{\nodetau''}$, we have $E_{\nodetau''}=E_{\nodetau}$ and $h_{\nodetau''}={h_{\nodetau}}+1$.

  Finally, we have $P_{i,\nodetau''}=P_{i,\nodetau}$ and $N_{i,\nodetau''}=N_{i,\nodetau}$ for $i$ in
  $\{0,\dots,k_{\nodetau}\}$, so the inequalities $ P_{i,\nodetau}+1\leq N_{i,\nodetau}\leq 2^in$
  remain true.  We also have ${d_{\nodetau''}}={d_{\nodetau}}-(\dalgo_\nodetau-1)\leq \frac {d_{\nodetau}}2$; since we
  supposed that ${d_{\nodetau}} \le \frac{\dinit}{2^{h_{\nodetau}}}$, and $h_{\nodetau''}={h_{\nodetau}}+1$, we obtain ${d_{\nodetau''}}\leq
  \frac{\dinit}{2^{h_{\nodetau''}}}+1$.
\end{proof}

\begin{lemma}\label{chap:complexite:ineq2}
  Let $\nodetau$ be an internal node of $\scrT$. Then, the
  following inequality holds:
  $$N_\nodetau^{d_\nodetau} \leq \left ( n^2 \right )^{\frac{\dinit}{2^{h_\nodetau}}+1}.
$$
\end{lemma}
\begin{proof}
  By the previous lemma, we have that $k_\nodetau \leq h_\nodetau $,
 that $N_\nodetau $ is bounded by $2^{k_\nodetau}n$, and that $d_\nodetau$ is bounded by
  $\frac{\dinit}{2^{h_\nodetau}}+1$. We deduce
  that
  \[
  N_\nodetau^{d_\nodetau}  \leq \left (2^{k_\nodetau}n\right )^{\frac{\dinit}{2^{h_\nodetau}}+1} \leq 
  \left (2^{h_\nodetau} n
  \right )^{\frac{\dinit}{2^{h_\nodetau}}+1}. 
  \]
  Now, since $\nodetau$ is an internal node, we actually have ${h_\nodetau} \le
  \lceil \log_2(\dinit)\rceil -1 \le \log_2(\dinit)$, so we have $2^{h_\nodetau}
  \le \dinit \le n$.
\end{proof}

%%%%%%%%%%%%%%%%%%%%%%%%%%%%%%%%%%%%%%%%%%%%%%%%%%%%%%%%%%%%
%%%%%%%%%%%%%%%%%%%%%%%%%%%%%%%%%%%%%%%%%%%%%%%%%%%%%%%%%%%%
%%%%%%%%%%%%%%%%%%%%%%%%%%%%%%%%%%%%%%%%%%%%%%%%%%%%%%%%%%%%

\subsection{Uniform degree bounds}\label{chap:complexity:sec:degreebounds}

We use the following notation for the degrees of various objects (when
they are defined): for any node $\nodetau$,
\begin{itemize}
\item $\degC_\nodetau$, $\degC'_\nodetau$ and $\degC''_\nodetau$ are the degrees
  of respectively $\Zeroes(\scrC_\nodetau)$, $\Zeroes(\scrC'_\nodetau)$ and
  $\Zeroes(\scrC''_\nodetau)$;
\smallskip
\item $\degQ_\nodetau$ and $\degQ''_\nodetau$ are the degrees of
  respectively $\Zeroes(\scrQ_\nodetau)$ and  $\Zeroes(\scrQ''_\nodetau)$
\smallskip
\item $\degS_\nodetau$, $\degS'_\nodetau$ and $\degS''_\nodetau$ are the degrees
  of respectively $\Zeroes(\scrS_\nodetau)$, $\Zeroes(\scrS'_\nodetau)$ and
  $\Zeroes(\scrS''_\nodetau)$;
\smallskip
\item $\degB_\nodetau$ is the degree of $\Zeroes(\scrB_\nodetau)$;
\smallskip
\item $\degfiber_\nodetau$ is the degree of ${\sf Fiber}(L'_\nodetau, \scrQ''_\nodetau)$;
\smallskip
\item $\delta_\nodetau= \DF(k_\nodetau, e_\nodetau, \n_\nodetau,\p_\nodetau,D, D-1)$ (see Definition
  \ref{sec:solvelagrange:notationsNPdelta}).
\end{itemize}
If $\nodetau$ is an internal node and $\nodetau',\nodetau''$ are its children,
then by construction, $\scrQ_{\nodetau'}=\scrQ_\nodetau$ and
$\scrS_{\nodetau'}=\scrS^{\mA_\tau}_\nodetau$, so
$(\degQ_{\nodetau'},\degS_{\nodetau'})=(\degQ_{\nodetau},\degS_{\nodetau})$;
similarly, we have $\scrQ_{\nodetau''}=\scrQ''_\nodetau$ and
$\scrS_{\nodetau''}=\scrS''_\nodetau$, so $(\degQ_{\nodetau''},\degS_{\nodetau''})=
(\degQ''_\nodetau,\degS''_\nodetau)$. Note also that
$\scrC_{\nodetau'}=\scrC'_\nodetau$ and $\scrC_{\nodetau''}=\scrC''_\nodetau$, which
implies that $\degC_{\nodetau'}=\degC'_\nodetau$ and
$\degC_{\nodetau''}=\degC''_\nodetau$.

The goal of this paragraph is to establish uniform bounds on the
degrees
$\degC_\nodetau, \degQ_\nodetau, \degfiber_\nodetau, \degB_\nodetau, \degS_\nodetau$ and
$\delta_\nodetau$, for any node $\nodetau$ of $\scrT$ where they are defined
(if $\nodetau$ is a leaf, only $\degC_\nodetau$, $\degQ_\nodetau$, $\degS_\nodetau$
and $\delta_\nodetau$ are). Our bounds are expressed in terms of the
quantities
 $$\bdelta= 16^{\dinit+2} n^{2\dinit+12 \log_2(\dinit)} D^n$$ and 
$$\bzeta=  (\degC_\rho+\degQ_\rho) 16^{2(\dinit+3)} 
(n
\log_2(n))^{2(2\dinit+12\log_2(\dinit))(\log_2(\dinit)+4)}D^{(2n+1)(\log_2(\dinit)+2)}.$$

\begin{proposition}\label{chap:complexity:prop:uniformdegreebounds}
  Let $\nodetau$ be a node of $\scrT$. Then the inequalities
  $$\delta_\nodetau\leq \bdelta \quad\text{and}\quad \degC_\nodetau,
  \degQ_\nodetau, \degS_\nodetau \le \bzeta$$ hold. If $\nodetau$ is an internal
  node, we also have $\degfiber_\nodetau, \degB_\nodetau \le \bzeta$.  If
  $\nodetau$ is a leaf, the output of ${\sf SolveLagrange}(L_\nodetau)$ has
  degree at most $\bzeta \bdelta$.
\end{proposition}

The proof of the above result will occupy most of this paragraph. We
start by proving the inequality $\delta_\nodetau \le \bdelta$ and next we
establish a recurrence formula on the quantities $\degB_\nodetau$,
$\degfiber_\nodetau$, $\degC_{\nodetau}+\degQ_{\nodetau}$, $\degS_\nodetau$ when
$\nodetau$ varies as a node of $\scrT$ (Lemma~\ref{chap:complexite:rec1}
below), as a key ingredient for the proof of
Proposition~\ref{chap:complexity:prop:uniformdegreebounds}.

\begin{lemma}\label{lemma:bornedelta}
  Let $\nodetau$ be a node of $\scrT$. Then, the inequality
  $\delta_\nodetau\leq\bdelta$ holds.
\end{lemma}
\begin{proof}
  Using the definition of $\delta_\nodetau=\DF(k_\nodetau, e_\nodetau,
  \n_\nodetau,\p_\nodetau,D, D-1)$ given in Definition
  \ref{sec:solvelagrange:notationsNPdelta}, we can rewrite the left-hand side as
  $$\delta_\nodetau=(P_\nodetau+1)^{k_\nodetau}
  D^{p}(D-1)^{n-e_\nodetau-p}\prod_{i=0}^{k_\nodetau-1}
  N_{i+1,\nodetau}^{N_{i,\nodetau}-e_\nodetau-P_{i,\nodetau}}.$$ We will prove that 
  \[
  (P_\nodetau+1)^{k_\nodetau}
  D^{p}(D-1)^{n-e_\nodetau-p}\prod_{i=0}^{k_\nodetau-1}
  N_{i+1,\nodetau}^{N_{i,\nodetau}-e_\nodetau-P_{i,\nodetau-P}} \le  16^{\dinit+2} n^{2\dinit+3\log_2(\dinit)+8} D^n;
  \]
  from that, our conclusion will follow, since $3 \log_2(\dinit)+ 8 \le 12
  \log_2(\dinit)$ holds if $\dinit \ge 2$ (if $\dinit=1$, the upper bound we wish 
  to establish is clearly true).  Since $e_\nodetau\geq 0$, we get
  $D^{p}(D-1)^{n-e_\nodetau-p}\leq D^n$. Thus, it remains to establish
  $$
  (P_\nodetau+1)^{k_\nodetau}\prod_{i=0}^{k_\nodetau-1} N_{i+1,\nodetau}^{N_{i,\nodetau}-e_\nodetau-P_{i,\nodetau}}\leq  16^{\dinit+2} n^{2\dinit+3\log_2(\dinit)+8},
  $$ which is what we do now.  Lemma \ref{chap:complexite:ineq1}
  implies that for $i$ in $\{0,\dots,k_\nodetau\}$ we have $P_{i,\nodetau}+1\leq
  N_{i,\nodetau}\leq 2^{i}n$ and $d_{i,\nodetau}\leq \frac{\dinit}{2^i}+1$, with $d_{i,\nodetau} =
  N_{i,\nodetau}-e_\nodetau-P_{i,\nodetau}$. Recall also that $N_{k_\nodetau,\nodetau}=N_\nodetau$. As a consequence, we get
  \begin{equation*}
    \begin{gathered}
      (P_\nodetau+1)^{k_\nodetau}\prod_{i=0}^{k_\nodetau-1}N_{i+1,\nodetau}^{N_{i,\nodetau}-e_\nodetau-P_{i,\nodetau}}\leq 
      N_\nodetau^{k_\nodetau}\prod_{i=0}^{k_\nodetau-1}\left (2^{i+1}n\right )^{\frac{\dinit}{2^i}+1}
      \leq \left (2^{k_\nodetau}n\right )^{k,\nodetau}\prod_{i=0}^{k_\nodetau-1}\left (2^{i+1}n\right )^{\frac{\dinit}{2^i}+1}    \\
      \leq 2^{k_\nodetau^2+k_\nodetau+\sum_{i=0}^{k_\nodetau-1}(i+1)\frac{\dinit}{2^i}} n^{2k_\nodetau+\sum_{i=0}^{k_\nodetau-1}\frac{\dinit}{2^i}}.  
    \end{gathered}  
  \end{equation*}
  Straightforward computations show that
  $$\sum_{i=0}^{k_\nodetau-1}\frac{\dinit}{2^i} \leq 2\dinit
  \quad \text{ and }\quad
  \sum_{i=0}^{k_\nodetau-1}\left (i+1 \right )\frac{\dinit}{2^i}\leq 4\dinit.$$
  We deduce that 
  \[
  (P_\nodetau+1)^{k_\nodetau}\prod_{i=0}^{k_\nodetau-1} N_{i+1,\nodetau}^{N_{i,\nodetau}-e_\nodetau-P_{i,\nodetau}} \leq  2^{4\dinit + k_\nodetau^2+k_\nodetau} n^{2\dinit+2k_\nodetau}
  \]
  and it remains to prove that $2^{4 \dinit +k_\nodetau^2+k_\nodetau} n^{2\dinit+2k_\nodetau}\leq
  16^{\dinit+2} n^{2\dinit+3\log_2(\dinit)+8}$.  Using $k_\nodetau \le \log_2(\dinit)+1$
  (Lemma \ref{chap:complexite:ineq1}), one deduces that
  $n^{2\dinit+2k_\nodetau}\leq n^{2\dinit+2\log_2(\dinit)+2}$. Using again $k_\nodetau\leq
  \log_2(\dinit)+1$, we also deduce that $2^{k_\nodetau}\leq 2n$ and $2^{k_\nodetau^2}\leq
  (2n)^{\log_2(\dinit)+1}$, which implies that $2^{k_\nodetau^2+k_\nodetau}\leq
  (2n)^{\log_2(\dinit)+2}$. This implies that
  \[
  2^{4\dinit+k_\nodetau^2+k_\nodetau}\leq 16^\dinit (2n)^{\log_2(\dinit)+2} 
  \]
  and finally, 
  \[
  2^{4\dinit+k_\nodetau^2+k_\nodetau} n^{2\dinit+2k_\nodetau} \leq  16^{\dinit+\log_2(\dinit)+2} n^{2\dinit+3\log_2(\dinit)+4}.
  \]
  Noticing that $16^{\log_2(\dinit)}\le n^4$, we are done.
\end{proof}

We can now establish the recurrence formula on the quantities
$\degB_\nodetau$, $\degfiber_\nodetau$, $\degC_{\nodetau}+\degQ_{\nodetau}$,
$\degS_\nodetau$ when $\nodetau$ varies as a node of $\scrT$.

\begin{lemma}\label{chap:complexite:rec1}
  Let $\nodetau$ be an internal node of $\scrT$, and define
  \[
  \zeta_\nodetau=\left ( n^2 \log_2(n) D\right )^{\frac{\dinit}{2^{h_\nodetau}}+1}.
  \]  
  Then, letting $\nodetau'$ and $\nodetau''$ be respectively the left and
  right child of $\nodetau$, all the quantities $\degB_\nodetau$,
  $\degfiber_\nodetau$, $\degC_{\nodetau'}+\degQ_{\nodetau'}$,
  $\degC_{\nodetau''}+\degQ_{\nodetau''}$, $\degS_\nodetau$, $\degS'_{\nodetau}$,
  $\degS''_{\nodetau}$ are at most $2 \bdelta^2 \zeta_\nodetau
  (\degC_\nodetau+\degQ_\nodetau)$.
\end{lemma}
\begin{proof}
  We let $L_\nodetau$ be the generalized Lagrange system at node $\nodetau$,
  and $L'_\nodetau$ be the one computed at Step~\ref{rmplag:step:4}.
  Remark that, as pointed out before, the quantities
  $\degC_{\nodetau'},\degC_{\nodetau''},\degQ_{\nodetau'},\degQ_{\nodetau''}$
  are respectively equal to
  $\degC'_{\nodetau},\degC''_{\nodetau},\degQ_{\nodetau},\degQ''_{\nodetau}$;
  we use the latter for the proof.
  
  \begin{itemize}
  \item $\degB_\nodetau\leq \degC_\nodetau+\degQ_\nodetau\bdelta\zeta_\nodetau$.
    
\smallskip

    By definition of $\scrB$ in Step~\ref{rmplag:step:5} of Algorithm
    ${\sf RoadmapRecLagrange}$, $\degB_\nodetau$ is bounded by the sum of degrees
    of ${\sf W}_1(L'_\nodetau)$ and $\scrC_\nodetau$ (that is, $\degC_\nodetau$).

\smallskip

    {F}rom Proposition~\ref{sec:main:critical}, we deduce that the
    zero-dimensional paramet\-rization returned by ${\sf W}_1(L'_\nodetau)$ has
    degree at most $\degQ_{\nodetau'} \delta_{\nodetau'} (N_{\nodetau'}(D-1+k_{\nodetau'}))^{d_{\nodetau'}}$.
  
\smallskip

    We saw previously that $\degQ_{\nodetau'}=\degQ_\nodetau$ and
    $k_{\nodetau'}=k_\nodetau+1$; then, we obtain that $
    D-1+k_{\nodetau'}=D+k_{\nodetau}$. We claim that we can use the
    upper bound $D + k_{\nodetau} \le \log_2(n) D$: if $n < 4$, the
    only possible value for $k_{\nodetau}$ is $k_{\nodetau}=0$, for
    which the claim clearly holds; otherwise, because $\nodetau$ is an
    internal node, $k_\nodetau \le \log_2(n)$, and the inequality $D +
    \log_2(n) \le \log_2(n) D$ holds for all $D \ge 2$.
    Moreover, we have $d_{\nodetau'}\leq  \frac{\dinit}{2^{h_\nodetau}}+1$,
    so that $(D-1+k_{\nodetau'})^{d_{\nodetau'}}$ is at most 
    $(\log_2(n) D)^{\frac{\dinit}{2^{h_\nodetau}}+1}$.

\smallskip

    Next, we prove that $N_{\nodetau'}^{d_{\nodetau'}}$ is at most
    $(n^2)^{\frac{\dinit}{2^{h_\nodetau}}+1}$. If $\tau'$ is an
    internal node, this is a consequence of Lemma
    \ref{chap:complexite:ineq2} (since the lemma proves that this
    quantity is at most $(n^2)^{\frac{\dinit}{2^{h_{\nodetau'}}}+1}$
    and $h_{\nodetau'} \ge h_\nodetau$).  Else, $\tau'$ is a leaf, so
    that $d_{\nodetau'}=1$; in that case, we have the inequality
    $N_{\nodetau'} \le 2n^2$ from Lemma~\ref{chap:complexite:ineq1},
    so our conclusion follows as well.

\smallskip

    Finally, $\delta_{\nodetau'}$ is bounded by $\bdelta$ by
    Lemma~\ref{lemma:bornedelta} so  altogether, we get that
    $\degQ_{\nodetau'} \delta_{\nodetau'}
    (N_{\nodetau'}(D-1+k_{\nodetau'}))^{d_{\nodetau'}}$ is at most
    $$ \degQ_\nodetau \bdelta \left (n^2 \log_2(n) D\right )^{\frac {\dinit}{2^{h_\nodetau}}+1}.
    $$ This proves that
    $$\degB_\nodetau \le \degC_\nodetau +  \degQ_\nodetau\bdelta \left (n^2 \log_2(n) D\right
    )^{\frac {\dinit}{2^{h_\nodetau}}+1},$$
    which we recognize as $ \degC_\nodetau +  \degQ_\nodetau \bdelta\zeta_\nodetau$.
    
\smallskip

  \item $\degQ''_{\nodetau} \le \degC_\nodetau+\degQ_\nodetau \bdelta\zeta_\nodetau$.
    
\smallskip

    We just proved that $\degB_\nodetau = \deg(\scrB_\nodetau)$ satisfies $\degB_\nodetau \le \degC_\nodetau
    + \degQ_\nodetau \bdelta\zeta_\nodetau$; on the other hand, by construction,
    $\degQ''_{\nodetau}=\deg(\scrQ''_{\nodetau})$ satisfies $\degQ''_{\nodetau} \le \degB_\nodetau$.

\smallskip
    
  \item $\degfiber_\nodetau \leq \bdelta(\degC_\nodetau+\degQ_\nodetau \bdelta\zeta_\nodetau)$.
    
\smallskip

    {F}rom Proposition \ref{sec:main:prop:fiber}, we deduce
    that $$\degfiber_\nodetau = \deg({\sf Fiber}(L'_\nodetau, \scrQ''_{\nodetau}))$$ satisfies
    $\degfiber_\nodetau\leq \delta'_\nodetau \degQ''_{\nodetau}$. Since $\delta'_\nodetau\le \bdelta$ by
    Lemma~\ref{lemma:bornedelta}, the previous bound on $\degQ''_{\nodetau}$
    implies that $\degfiber_\nodetau \le \bdelta(\degC_\nodetau+\degQ_\nodetau  \bdelta \zeta_\nodetau)$,
    as requested.
    
\smallskip

  \item $\degC'_\nodetau \le \degC_\nodetau + \bdelta(\degC_\nodetau+\degQ_\nodetau \bdelta\zeta_\nodetau)$.

\smallskip

    The set $\Zeroes(\scrC'_\nodetau)$ is the union of $\Zeroes(\scrC_\nodetau)^{\mA_\nodetau}$ and ${\sf
      Fiber}(L'_\nodetau, \scrQ''_{\nodetau})$, so its cardinality $\degC'_\nodetau$ is at most 
    $\degC_\nodetau + \degfiber_\nodetau$.

\smallskip
    
  \item $\degC''_{\nodetau} \le \degC_{\nodetau} + \bdelta(\degC_{\nodetau}+\degQ_{\nodetau} \bdelta\zeta_{\nodetau})$.
    
\smallskip

    This is because the set $\Zeroes(\scrC''_{\nodetau})$ is a subset of $\Zeroes(\scrC'_{\nodetau})$.
    
\smallskip

  \item $\degC'_{\nodetau}+\degQ_{\nodetau'} \le 2 \bdelta^2\zeta_{\nodetau} (\degC_{\nodetau}+\degQ_{\nodetau})$.
    
\smallskip

    We know that $\degC'_{\nodetau} \le \degC_{\nodetau} +
    \bdelta(\degC_{\nodetau}+\degQ_{\nodetau} \bdelta \zeta_{\nodetau}
    )$, and that $\degQ_{\nodetau'} = \degQ_{\nodetau}$, so that
    $\degC'_{\nodetau}+\degQ_{\nodetau'} \le \degC_{\nodetau} +
    \bdelta(\degC_{\nodetau}+\degQ_{\nodetau} \bdelta\zeta_{\nodetau})
    + \degQ_{\nodetau}$, which admits the upper bound given above.

\smallskip

  \item $\degC''_{\nodetau}+\degQ''_{\nodetau} \le 2 \bdelta^2\zeta_{\nodetau} (\degC_{\nodetau}+\degQ_{\nodetau})$.
    
\smallskip

    We know that $\degC''_{\nodetau} \le \degC_{\nodetau} +
    \bdelta(\degC_{\nodetau}+\degQ_{\nodetau}
    \bdelta\zeta_{\nodetau})$ and $\degQ''_{\nodetau} \le
    \degC_{\nodetau}+\degQ_{\nodetau} \bdelta\zeta_{\nodetau}$, so
    $\degC''_{\nodetau}+\degQ''_{\nodetau} \le 2\degC_{\nodetau} +
    \bdelta(\degC_{\nodetau}+\degQ_{\nodetau} \bdelta\zeta_{\nodetau})
    +\degQ_{\nodetau} \bdelta\zeta_{\nodetau}$, which admits the upper
    bound given above (since $\bdelta \ge 2$).

\smallskip

  \item $\degS_{\nodetau} \le \degC_{\nodetau}$.

\smallskip

    This is because we proved that $\Zeroes(\scrS_{\nodetau})$ is contained in $\Zeroes(\scrC_{\nodetau})$.

\smallskip

   \item $\degS'_{\nodetau} \le 2 \bdelta^2\zeta_{\nodetau} (\degC_{\nodetau}+\degQ_{\nodetau})$.

\smallskip

     This is because we proved that $\Zeroes(\scrS'_{\nodetau})$ is contained in $\Zeroes(\scrC'_{\nodetau})$,
     so $\degS'_{\nodetau} \le \degC'_{\nodetau} \le \degC'_{\nodetau}+\degQ'_{\nodetau}$.

\smallskip

  \item $\degS''_{\nodetau} \le 2 \bdelta^2\zeta_{\nodetau}  (\degC_{\nodetau}+\degQ_{\nodetau})$.

\smallskip

    Same argument as above, for the inclusion $\Zeroes(\scrS''_{\nodetau}) \subset \Zeroes(\scrC''_{\nodetau})$.
  \end{itemize}
  At this stage, we are mostly done; we only need to verify that the
  bounds given for $\degB_{\nodetau}$, $\degfiber_{\nodetau}$ and $\degS_\nodetau$ are at most
  $2 \bdelta^2\zeta_{\nodetau} (\degC_{\nodetau}+\degQ_{\nodetau})$, which is indeed the case.
\end{proof}

\begin{proof}[of Proposition \ref{chap:complexity:prop:uniformdegreebounds}]
  We proved in Lemma~\ref{lemma:bornedelta} the inequality
  $\delta_\nodetau \le \bdelta$. Let next $\nodetau$ be an internal node of
  $\scrT$, with children $\nodetau'$ and $\nodetau''$. We will prove below
  that $\degB_\nodetau,\degfiber_\nodetau,\degC_\nodetau,\degQ_\nodetau,\degS_\nodetau$,
  as well as $\degC_{\nodetau'},\degQ_{\nodetau'},\degS_{\nodetau'}$ and
  $\degC_{\nodetau''},\degQ_{\nodetau''},\degS_{\nodetau''}$ are all at most
  $\bzeta$; this is enough to conclude, since it covers the bounds for
  the two child nodes.

  Let $\gamma$ be a root-to-leaf path in $\scrT$ containing
  $\nodetau$; we denote by $\gamma'$ the path obtained from
  $\gamma$ by excluding the leaf it
  contains. Lemma~\ref{chap:complexite:rec1} implies that for any node
  $\gamma$, and in particular~$\nodetau$, all the quantities written
  above are at most
  $$(\degC_\rho+\degQ_\rho)\prod_{\nu\in \gamma'} 2 \zeta_\nu \bdelta^2.$$
  Our first step is to prove the following:
  \begin{equation}\label{eq:firststepnu}
    \prod_{\nu\in
      \gamma'} 2 \zeta_\nu \bdelta^2 \leq 2 n\bdelta^{2(\log_2(\dinit)+1)} \left ( n^2  \log_2(n) D \right )^{2\dinit+\log_2(\dinit)+1}. 
  \end{equation}
  Recall that, by Lemma \ref{chap:complexite:rec1}, 
  \[
  \zeta_\nu=  \left ( n^2 \log_2(n) D \right )^{\frac{\dinit}{2^{h_\nu}}+1},
  \]
  so that we have to give an upper bound on 
  $$\prod_{\nu \in \gamma'} 2 \bdelta^2 \left ( n^2 \log_2(n) D \right
  ) \cdot \prod_{\nu \in \gamma'} \left ( n^2 \log_2(n) D \right
  )^{\frac{\dinit}{2^{h_\nu}}}.$$
  For the first product, since the depth of $\scrT$ is at most
  $\lceil \log_2(\dinit) \rceil$, the number of nodes in $\gamma'$ is
  at most $\lceil \log_2(\dinit)\rceil \le \log_2(\dinit)+1$.  Thus,
  the first product is at most
  $$ 2n \bdelta^{2(\log_2(\dinit)+1)} \left (n^2 \log_2(n) D \right
  )^{\log_2(\dinit)+1}.$$
  For the second product, remarking that
  $\sum_{\nu \in \gamma'}\frac {\dinit}{2^{h_\nu}} \leq 2\dinit$, we
  obtain the upper bound $ ( n^2 \log_2(n) D )^{2\dinit}$, which ends
  the proof of~\eqref{eq:firststepnu}.

  Recall that $\bdelta= 16^{\dinit+2} n^{2\dinit+12\log_2(\dinit)} D^n$, so
  that 
  \begin{align*}
  \bdelta^{2(\log_2(\dinit)+1)} &=& 16^{2(\dinit+2)(\log_2(\dinit) +1)} n^{2(2\dinit+12\log_2(\dinit))(\log_2(\dinit)+1)} D^{2n(\log_2(\dinit)+1)}\\
&\leq& 16^{2(\dinit+2)}n^{8(\dinit+2)}  n^{2(2\dinit+12\log_2(\dinit))(\log_2(\dinit)+1)} D^{2n(\log_2(\dinit)+1)}.
  \end{align*}
  Using the crude upper bounds $2 \le 16^2$ and $n \le n^8$, we deduce
  that the left-hand side of~\eqref{eq:firststepnu} is at most
  {\small$$ 16^{2(\dinit+3)}n^{8(\dinit+3)}
    n^{2(2\dinit+12\log_2(\dinit))(\log_2(\dinit)+1)}
    D^{2n(\log_2(\dinit)+1)}\left ( n^2 \log_2(n) D \right
    )^{2\dinit+\log_2(\dinit)+1}. $$}
  We see that the exponent of $D$ is at most
  $(2n+1)(\log_2(\dinit)+2)$. Replacing both bases $n$ and $n^2 \log_2(n)$
  by $(n \log_2(n))^2$, we see that powers of $n$ appearing in the
  previous expression admit an upper bound of the form
  $$(n \log_2(n))^{8(\dinit+3) + 2(2\dinit+12\log_2(\dinit))(\log_2(\dinit)+1) +
    2(2\dinit+\log_2(\dinit)+1)}.$$
  The exponent is at most $2(2\dinit+12 \log_2(\dinit))(\log_2(\dinit)+4)$,
  so the proof of our upper bounds is complete.

  It remains to deal with the degrees at the leaves: this is a direct
  consequence of the degree bound in
  Proposition~\ref{chap:solvelagrange:prop:basicsolve}, together 
  with the above bounds on $\degQ_\nodetau$ and $\delta_\nodetau$.
\end{proof}

\begin{corollary}\label{chap:complexity:cor:bounds}
  Let $\nodetau$ be a node of $\scrT$. Then the following inequalities hold. 
{\small\begin{eqnarray*}
\degQ_\nodetau \delta_\nodetau & \leq & 
(\degC_\rho+\degQ_\rho)
16^{3(\dinit+3)}  (n \log_2(n))^{2(2\dinit+12\log_2(\dinit))(\log_2(\dinit)+5)}D^{(2n+1)(\log_2(\dinit)+3)}\\
\degQ_\nodetau \delta_\nodetau^2 & \leq & (\degC_\rho+\degQ_\rho)
16^{4(\dinit+3)}  (n \log_2(n))^{2(2\dinit+12\log_2(\dinit))(\log_2(\dinit)+5)}D^{(2n+1)(\log_2(\dinit)+3)}\\
\degQ_\nodetau \delta_\nodetau \degS_\nodetau^2 & \leq & 
 (\degC_\rho+\degQ_\rho)^3
16^{7(\dinit+3)} (n \log_2(n))^{6(2\dinit+12\log_2(\dinit))(\log_2(\dinit)+5)}D^{3(2n+1)(\log_2(\dinit)+3)}.
  \end{eqnarray*}}
\end{corollary}
\begin{proof}
  By Proposition~\ref{chap:complexity:prop:uniformdegreebounds},
  the quantities above admit the respective upper bounds 
  $\bzeta \bdelta$, 
  $\bzeta \bdelta^2$ and
  $\bzeta^3 \bdelta$.
  Given the definitions of $\bdelta$ and $\bzeta$, namely
  \begin{eqnarray*}
    \bdelta&=& 16^{\dinit+2} n^{2\dinit+12 \log_2(\dinit)} D^n\\
    \bzeta&=&  (\degC_\rho+\degQ_\rho) 16^{2(\dinit+3)} (n \log_2(n))^{2(2\dinit+12\log_2(\dinit))(\log_2(\dinit)+4)}D^{(2n+1)(\log_2(\dinit)+2)},
  \end{eqnarray*}
  the bounds given in the corollary follow directly, using in particular the upper
  bound $\bdelta \le 16^{\dinit+3} (n\log_2(n))^{2\dinit+12 \log_2(\dinit)}
  D^n$.
\end{proof}

%%%%%%%%%%%%%%%%%%%%%%%%%%%%%%%%%%%%%%%%%%%%%%%%%%%%%%%%%%%%
%%%%%%%%%%%%%%%%%%%%%%%%%%%%%%%%%%%%%%%%%%%%%%%%%%%%%%%%%%%%
%%%%%%%%%%%%%%%%%%%%%%%%%%%%%%%%%%%%%%%%%%%%%%%%%%%%%%%%%%%%

\subsection{Runtime estimates for ${\sf RoadmapRecLagrange}$}

The goal of this paragraph is to prove the following bounds on the
output degree and runtime for ${\sf RoadmapRecLagrange}$.

\begin{proposition}\label{chap:complexite:prop:roadmaprec}
  Let $L_\rho=(\Gamma_\rho, (\,),\scrS_\rho)$ be a generalized
  Lagrange system such that $\Clos{(L_\rho)}$ is $d$-equidimensional
  with finitely many singular points and $\Clos{(L_\rho)}\cap\R^n$ is
  bounded.  Let $\scrC_\rho$ be a zero-dimensional parametrization
  encoding a finite set of points in $\C^n$. Assume that the
  assumptions and inequalities stated in the introduction 
  of Subsection~\ref{chap:complexite:subsection:notations}
 hold, and that $\Zeroes(\scrS_\rho)$ is contained in
  $\Zeroes(\scrC_\rho)$.

  Then, ${\sf RoadmapRecLagrange}((\Gamma_\rho,(\,),
  \scrS_\rho),\scrC_\rho)$ outputs a roadmap of $(\Clos{(L_\rho)},
  \Zeroes(\scrC_\rho))$ of degree
  \[\softO \left ((\degC_\rho+\degQ_\rho)
  16^{3d}  (n \log_2(n))^{2(2d+12\log_2(\dinit))(\log_2(\dinit)+5)}D^{(2n+1)(\log_2(\dinit)+3)}\right )\]
  using 
  \[
  \softO \left (
  (\degC_\rho+\degQ_\rho)^3
  16^{9d} E_\rho (n \log_2(n))^{6(2d+12\log_2(\dinit))(\log_2(\dinit)+6)}D^{3(2n+1)(\log_2(\dinit)+4)}
  \right ) 
  \]
  operations in $\QQ$. 
\end{proposition}

Note that the number of nodes in $\scrT$ is $O(n)$, because $\scrT$ is
a binary tree of depth bounded by $\lceil\log_2(\dinit)\rceil$. Thus,
to bound the number of arithmetic operations of performed by
${\sf RoadmapRecLagrange}$, it is enough to take $n$ times a bound on
the cost of each step. Because all our bounds will involve a term that
will be at least $D^n$, since we ignore polylogarithmic factors, we
can safely omit the extra factor $n$.

We bound the cost of each step using the uniform degree bounds given
in Proposition~\ref{chap:complexity:prop:uniformdegreebounds}, the
complexity estimates of Subsection \ref{solvinglagrange:algorithms}
for solving generalized Lagrange systems and the complexity estimates
of Subsections \ref{sec:basicroutinesparam0}
and~\ref{sec:basicroutinesparam1} of Section \ref{chap:posso} for
basic routines on parametrizations.

%%%%%%%%%%%%%%%%%%%%%%%%%%%%%%%%%%%%%%%%%%%%%%%%%%%%%%%%%%%%

\subsubsection{Analysis of Step \ref{rmplag:step:1}}

\begin{lemma}\label{chap:complexite:lemma:step:1}
  Under the above notation and assumptions, the total cost of all
  calls to Step \ref{rmplag:step:1} of ${\sf RoadmapRecLagrange}$ on
  input $(L_\rho, \scrC_\rho)$ is
  \[
  \softO \left (
  (\degC_\rho+\degQ_\rho)^3
  16^{9\dinit} E_\rho (n \log_2(n))^{6(2\dinit+12\log_2(\dinit))(\log_2(\dinit)+6)}D^{3(2n+1)(\log_2(\dinit)+4)}
  \right ) 
  \]
  operations in $\QQ$. 
\end{lemma}
\begin{proof}
  It is enough to give a bound on the maximal cost of calling the
  routine ${\sf SolveLagrange}$. Since the assumptions of
  Proposition~\ref{chap:solvelagrange:prop:basicsolve} are satisfied,
  so the cost of each call to ${\sf SolveLagrange}$ is
  \begin{equation}
    \label{eq:analysis:1}
  \softO(N_\nodetau^3(E_\nodetau+N_\nodetau^3) (D+k_\nodetau) \degQ_\nodetau^3\delta_\nodetau^3 + N_\nodetau\degQ_\nodetau
  \delta_\nodetau\degS_\nodetau^2)  
  \end{equation}
  arithmetic operations in $\QQ$. By Lemma \ref{chap:complexite:ineq1}, the
  following inequalities hold.
\[ 
N_\nodetau\leq 2n^2, \qquad E_\nodetau = O\left (n^{4+2\log_2(\dinit)}(E_\rho+n^4)\right ) \qquad \text{ and } \qquad k_\nodetau\leq \lceil \log_2(\dinit)\rceil. 
\]
This shows that $\softO(N_\nodetau^3(E_\nodetau+N_\nodetau^3) (D+k_\nodetau)$ lies in 
\[
  \softO\left (n^6(n^{4+2\log_2(\dinit)}(E_\rho+n^4)) D\right ). 
\]
Using Corollary \ref{chap:complexity:cor:bounds}, we have 
\[
\degQ_\nodetau \delta_\nodetau  \leq  
(\degC_\rho+\degQ_\rho)
16^{3(\dinit+3)}  (n \log_2(n))^{2(2\dinit+12\log_2(\dinit))(\log_2(\dinit)+5)}D^{(2n+1)(\log_2(\dinit)+3)}
\]
and 
\[
\degQ_\nodetau \delta_\nodetau \degS_\nodetau^2  \le
 (\degC_\rho+\degQ_\rho)^3
16^{7(\dinit+3)} (n \log_2(n))^{6(2\dinit+12\log_2(\dinit))(\log_2(\dinit)+5)}D^{3(2n+1)(\log_2(\dinit)+3)}.
\]
As argued previously, because the above bounds involve terms at least
equal to $D^n$, polynomial factors in $n$ are omitted thanks to the
soft-Oh notation. Then, using straightforward simplifications, we
obtain that \eqref{eq:analysis:1} is 
\[
\softO \left (
(\degC_\rho+\degQ_\rho)^3
16^{9(\dinit+3)} E_\rho (n \log_2(n))^{6(2\dinit+12\log_2(\dinit))(\log_2(\dinit)+6)}D^{3(2n+1)(\log_2(\dinit)+4)}
\right ),
\]
which is
\[
\softO \left (
(\degC_\rho+\degQ_\rho)^3
16^{9\dinit} E_\rho (n \log_2(n))^{6(2\dinit+12\log_2(\dinit))(\log_2(\dinit)+6)}D^{3(2n+1)(\log_2(\dinit)+4)}
\right ). %\qedhere
\]
\end{proof}

%%%%%%%%%%%%%%%%%%%%%%%%%%%%%%%%%%%%%%%%%%%%%%%%%%%%%%%%%%%%

\subsubsection{Analysis of Steps \ref{rmplag:step:2}--\ref{rmplag:step:6}}

\begin{lemma}\label{chap:complexity:lemma2-6}
  Under the above notation and assumptions, the total cost of all
  calls to Steps \ref{rmplag:step:2}--\ref{rmplag:step:6} of
  ${\sf RoadmapRecLagrange}$ is
$$
\softO \left (
    (\degC_\rho+\degQ_\rho)^2
    16^{6\dinit} E_\rho (n \log_2(n))^{4(2\dinit+12\log_2(\dinit))(\log_2(\dinit)+6)}D^{2(2n+1)(\log_2(\dinit)+4)} \right )$$
operations in $\QQ$.  
\end{lemma}

\begin{proof}
  Steps \ref{rmplag:step:2}--\ref{rmplag:step:6} are performed for
  internal nodes of $\scrT$. Let $\nodetau$ be such a node. Steps
  \ref{rmplag:step:2}--\ref{rmplag:step:4} perform changes of
  variables and construct generalized Lagrange systems; their
  computational cost is negligible compared the cost of Steps
  \ref{rmplag:step:5} and \ref{rmplag:step:6}.

  Step \ref{rmplag:step:5} consists in computing $\scrB_\nodetau={\sf
    Union}({\sf W}_1(L'_\nodetau), \scrC_\nodetau^\mA)$. Remark that
  $L_{\nodetau'}=L'_\nodetau$.  Since the assumptions of
  Proposition~\ref{sec:main:critical} are satisfied, the call ${\sf
    W}_1(L'_\nodetau)$ uses
  \begin{equation}\label{eq:costW1}
\softO( (k_{\nodetau'}+1)^{2d_{\nodetau'}+1} D^{2d_{\nodetau'}+1}
  N_{\nodetau'}^{4d_{\nodetau'}+8}E_{\nodetau'}\degQ_{\nodetau'}^2 \delta_{\nodetau'}^2
  +N_\nodetau\degS_{\nodetau'}^2)    
  \end{equation}
  arithmetic operations in $\QQ$.  To analyze the cost of the calls to
  ${\sf Union}$ (at Step \ref{rmplag:step:5}) and ${\sf Projection}$ (at
  Step \ref{rmplag:step:6}), we use Lemmas~\ref{sec:main:lemma:union}
  and~\ref{sec:posso:lemma:projection}, which state that these calls
  use $\softO(N_\nodetau\degQ_\nodetau^2)$ and $\softO(N_\nodetau^2\degQ_\nodetau^2)$
  arithmetic operations in $\QQ$.  The costs of these calls are
  negligible compared to cost of calling ${\sf W}_1$ above.
  
  As above, thanks to the soft-Oh notation, polynomial factors in $n$
  can be omitted in complexity estimates where $n$ appears as an
  exponent, so it is enough to give an upper bound on the expression
  in~\eqref{eq:costW1}. For the same reason, as in the proof of the
  previous lemma, the contribution of $E_{\nodetau'}$ will be
  $n^{2\log_2(\dinit)}E_\rho$; similarly, since $N_{\nodetau'} \le 2n^2$ (Lemma
  \ref{chap:complexite:ineq1}), terms polynomial in it can be
  neglected. Finally, by construction, $d_{\nodetau'}$ is at most $\dinit$
  and $k_{\nodetau'}$ is at most
  $\lceil \log_2(\dinit) \rceil\leq \lceil \log_2(n) \rceil$, by Lemma
  \ref{chap:complexite:ineq1} again.

  Finally, the term $\degS_{\nodetau'}^2$ is negligible in front of
  $\degQ_{\nodetau'}^2 \delta_{\nodetau'}^2$. Plugging these bounds in the
  above complexity estimates, we obtain that the number of arithmetic
  operations used by the calls to ${\sf W}_1$ lies in
$$
\softO( (\lceil \log_2(n)\rceil+1)^{2\dinit} D^{2\dinit+1}
  (2n^2)^{4\dinit} n^{2\log_2(\dinit)} E_{\rho}\degQ_{\nodetau'}^2 \delta_{\nodetau'}^2).$$
    Using the upper bound $\lceil \log_2(n)\rceil+1 \le 2 \log_2(n)$,
    we see that this is 
$$\softO( 2^{6 \dinit}  E_{\rho}
  (n\log_2(n))^{8\dinit+2\log_2(\dinit)}  D^{2\dinit+1}\degQ_{\nodetau'}^2 \delta_{\nodetau'}^2).
    $$
    Now, we can use the first bound given in
    Corollary~\ref{chap:complexity:cor:bounds}, which states 
    that $$
    \degQ_{\nodetau'} \delta_{\nodetau'}  \leq 
    (\degC_\rho+\degQ_\rho)
    16^{3(\dinit+3)}  (n \log_2(n))^{2(2\dinit+12\log_2(\dinit))(\log_2(\dinit)+5)}D^{(2n+1)(\log_2(\dinit)+3)}.$$
    this shows that the total running time is
\[
\softO \left (
    (\degC_\rho+\degQ_\rho)^2
    16^{6\dinit} E_\rho (n \log_2(n))^{4(2\dinit+12\log_2(\dinit))(\log_2(\dinit)+6)}D^{2(2n+1)(\log_2(\dinit)+4)} \right ).%\qedhere
\]
\end{proof}

%%%%%%%%%%%%%%%%%%%%%%%%%%%%%%%%%%%%%%%%%%%%%%%%%%%%%%%%%%%%

\subsubsection{Analysis of Steps \ref{rmplag:step:7}--\ref{rmplag:step:9-a}}

\begin{lemma}\label{chap:complexity:lemma7-9}
  Under the above notation and assumptions, the total cost of all
  calls to Steps \ref{rmplag:step:7}--\ref{rmplag:step:9-a} of ${\sf
    RoadmapRecLagrange}$ is bounded from above by the total cost of all
  calls to Steps~\ref{rmplag:step:2}--\ref{rmplag:step:6}
\end{lemma}
\begin{proof}
  Steps \ref{rmplag:step:7}--\ref{rmplag:step:9-a} are performed for
  internal nodes of $\scrT$; let $\nodetau$ be such a node. Recall that
  these steps consist in computing ${\sf Fiber}(L'_\nodetau,
  \scrQ''_\nodetau)$, take its unions $\scrC'_\nodetau$ and $\scrS'_\nodetau$ with
  $\scrC^{\mA_\nodetau}_\nodetau$ and $\scrS^{\mA_\nodetau}_\nodetau$ respectively and compute
  $\scrC''_\nodetau={\sf Lift}(\scrC'_\nodetau, \scrQ''_\nodetau)$ and
  $\scrS''_\nodetau={\sf Lift}(\scrS'_\nodetau, \scrQ''_\nodetau)$.

  Denote by $\nodetau'$ and $\nodetau''$ the left and right children of $\nodetau$
  and observe that $\scrC'_\nodetau=\scrC_{\nodetau'}$,
  $\scrC''_\nodetau=\scrC_{\nodetau''}$, $\scrS^{\mA_\nodetau}_\nodetau=\scrS_{\nodetau'}$ and
  $\scrS''_\nodetau=\scrS_{\nodetau''}$. We deduce by Proposition
  \ref{chap:complexity:prop:uniformdegreebounds} that the degrees of
  all these objects are at most $\bzeta$.

  By Lemma \ref{lemma:subroutine:Lift}, the calls to ${\sf Lift}$ are
  polynomial in $N_\nodetau \leq 2n^2$ (Lemma \ref{chap:complexite:ineq1})
  and quadratic in the above degree bounds. The cost is thus at most
  that reported in the previous lemma, since the estimate
  in~\eqref{eq:costW1} involved similar (and actually higher) costs.
\end{proof}

%%%%%%%%%%%%%%%%%%%%%%%%%%%%%%%%%%%%%%%%%%%%%%%%%%%%%%%%%%%%

\subsubsection{Analysis of Step \ref{rmplag:step:12}}

\begin{lemma}\label{chap:complexity:lemma14}
  Under the above notation and assumptions, the total cost of all
  calls to Step \ref{rmplag:step:12} of ${\sf RoadmapRecLagrange}$ is
  bounded from above by the total cost of all calls to
  Step~\ref{rmplag:step:1}.
\end{lemma}
\begin{proof}
  Step \ref{rmplag:step:12} is performed for internal nodes of
  $\scrT$; let $\nodetau$ be such a node.  The call to the routine ${\sf
    Union}$ at Step \ref{rmplag:step:12} is linear in $n$ and cubic in
  the maximum of the degrees of the roadmaps computed at Steps
  \ref{rmplag:step:8} and \ref{rmplag:step:11} (Lemma
  \ref{sec:main:lemma:union1}). The cost of
  Lemma~\ref{chap:complexite:lemma:step:1} involves a cost that is at
  least as high, see Eq.~\eqref{eq:analysis:1}.
\end{proof}

%%%%%%%%%%%%%%%%%%%%%%%%%%%%%%%%%%%%%%%%%%%%%%%%%%%%%%%%%%%%

\subsubsection{Proof of Proposition \ref{chap:complexite:prop:roadmaprec}}

Let us summarize the complexity estimates established above
\begin{itemize}
\item Lemma \ref{chap:complexite:lemma:step:1}, the calls to Step
  \ref{rmplag:step:1} use
  \[
  \softO \left (
  (\degC_\rho+\degQ_\rho)^3
  16^{9\dinit} E_\rho (n \log_2(n))^{6(2\dinit+12\log_2(\dinit))(\log_2(\dinit)+6)}D^{3(2n+1)(\log_2(\dinit)+4)}
  \right ) 
  \]
operations in $\QQ$.
\smallskip
\item Lemma \ref{chap:complexity:lemma2-6} implies that all calls to
  Steps \ref{rmplag:step:2}--\ref{rmplag:step:6} use
  \[
  \softO \left (
  (\degC_\rho+\degQ_\rho)^2
  16^{6\dinit} E_\rho (n
  \log_2(n))^{4(2\dinit+12\log_2(\dinit))(\log_2(\dinit)+6)}D^{2(2n+1)(\log_2(\dinit)+4)}
  \right )
  \]
  operations in $\QQ$.
\smallskip
\item By Lemma \ref{chap:complexity:lemma7-9} and Lemma~\ref{chap:complexity:lemma14},
  all other costs can be absorbed in the above bounds.
\end{itemize}
The cost from Step~\ref{rmplag:step:1} is dominant, and gives the
total reported in Proposition \ref{chap:complexite:prop:roadmaprec}.
The bound on the output degree follows from
Proposition~\ref{chap:complexity:prop:uniformdegreebounds} and
Corollary~\ref{chap:complexity:cor:bounds}; removing 
polylogarithmic factors, it becomes
\[(\degC_\rho+\degQ_\rho)
16^{3\dinit}  (n \log_2(n))^{2(2\dinit+12\log_2(\dinit))(\log_2(\dinit)+5)}D^{(2n+1)(\log_2(\dinit)+3)},\]
as claimed.

%%%%%%%%%%%%%%%%%%%%%%%%%%%%%%%%%%%%%%%%%%%%%%%%%%%%%%%%%%%%
%%%%%%%%%%%%%%%%%%%%%%%%%%%%%%%%%%%%%%%%%%%%%%%%%%%%%%%%%%%%
%%%%%%%%%%%%%%%%%%%%%%%%%%%%%%%%%%%%%%%%%%%%%%%%%%%%%%%%%%%%

\subsection{Proof of the proposition}

We finally estimate the complexity of ${\sf MainRoadmapLagrange}$. On
input $\Gamma$ and $\scrC_0$, where
\begin{itemize}
\item $\Gamma$ is a straight-line program of length $E$ evaluating a
  sequence of polynomials
  $\f=(f_1, \ldots, f_p)\subset \QQ[X_1, \ldots, X_n]$ of degree
  $\leq D$ such that $V(\f)$ is $d$-equidimensional (with $d=n-p$)
  with finitely many singular points, $V(\f)\cap \R^n$ is bounded, and
\smallskip
\item $\scrC_0$ is a zero-dimensional parametrization of degree $\degC$
  encoding a finite set of points in $V(\f)$.
\end{itemize}
${\sf MainRoadmapLagrange}$ starts by calling the routine
${\sf SingularPoints}$ (see Proposition \ref{prop:singularpoints}) to
compute a zero-dimensional parametrization $\scrS_\rho$ encoding the
singular points of $V(\f)$ and next performs a call to
${\sf RoadmapRecLagrange}$ with input
$(\Gamma_\rho,(\,), \scrS_\rho), \scrC_\rho$, where $\scrC_\rho$
is a zero-dimensional parametrization encoding
$\Zeroes(\scrC_0)\cup \Zeroes(\scrS_\rho)$.

By Proposition \ref{prop:singularpoints}, the call to ${\sf SingularPoints}$ uses 
$$ \softO(E D^{4n} )$$ operations in $\QQ$ and returns a
zero-dimensional parametrization of degree bounded by $nD^{2n}$, so we
conclude that the degree of $\scrC_\rho$ is bounded by
$\degC+nD^{2n}$; the call to ${\sf Union}$ takes quadratic time in
this degree (and polynomial time in $n$), so we can ignore it.  Also,
by construction $\Zeroes(\scrQ_\rho)=\{\bullet\}$, hence $\degQ_\rho=1$.

Using Proposition \ref{chap:complexite:prop:roadmaprec}, and after a
few straightforward simplifications, we deduce that the call to ${\sf
  RoadmapRecLagrange}$ on input $(\Gamma_\rho, (\,),\scrS_\rho),
\scrC_\rho$ outputs a one-dimensional parametrization of degree
\[
\softO \left (
\degC 16^{3\dinit}  (n \log_2(n))^{2(2d+12\log_2(\dinit))(\log_2(\dinit)+6)}D^{(2n+1)(\log_2(\dinit)+4)}
\right )\]
using 
\[
\softO \left (
\degC^3 16^{9\dinit} E (n \log_2(n))^{6(2d+12\log_2(\dinit))(\log_2(\dinit)+7)}D^{3(2n+1)(\log_2(\dinit)+5)}
\right ) 
\]
operations in $\QQ$. Observing that $\dinit=d$ ends the proof.

%%% Local Variables: 
%%% mode: latex
%%% TeX-master: "roadmaps-JACM"
%%% End: 

%%%%%%%%%%%%%%%%%%%%%%%%%%%%%%%%%%%%%%%%%%%%

\newpage
\includepdf[pages={1-},scale=0.75]{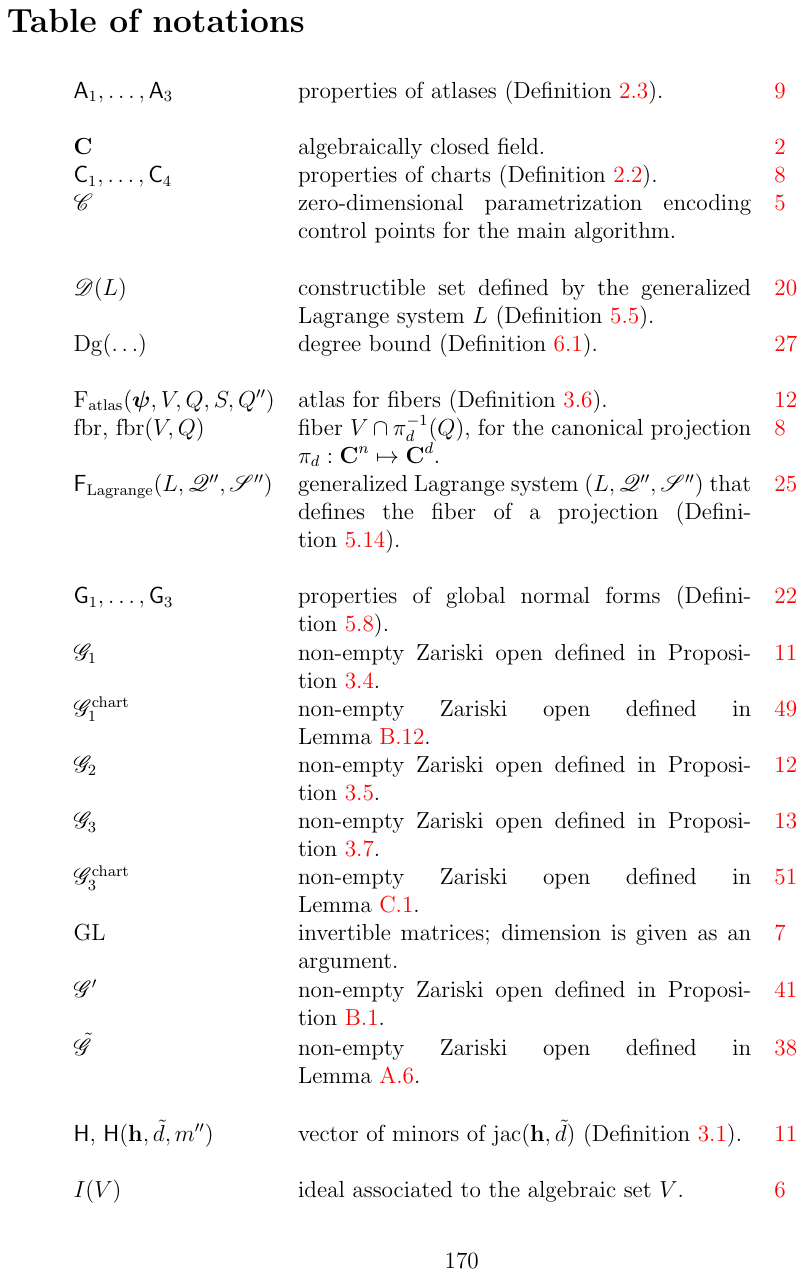}

\end{document}